\documentclass[11pt,letterpaper]{article}

\usepackage[utf8]{inputenc}

\usepackage{amsthm}

\usepackage{tikz}
\usetikzlibrary{matrix,positioning}

\usetikzlibrary{decorations.markings}

\usepackage{multicol}
\usepackage{sidecap}

\usepackage{amsmath,amssymb,amsfonts} 

\usepackage{booktabs}
\usepackage{array}       
\usepackage{arydshln} 
\usepackage{float}

\usepackage{epsfig} \usepackage{latexsym,nicefrac,bbm}

\usepackage{color,fancybox,graphicx,subfigure}
\usepackage[top=1in, bottom=1in, left=1in, right=1in]{geometry}
\usepackage{tabularx} \usepackage{hyperref}

\usepackage[linesnumbered,ruled,vlined]{algorithm2e}

\usepackage{enumitem}
\usepackage{multirow}
 \usepackage{url}

 \usepackage{tcolorbox}

\renewcommand{\epsilon}{\varepsilon}

\newcommand{\nfrac}{\nicefrac}
\newcommand{\eps}{\varepsilon}

\newtheorem{theorem}{Theorem}[section]

\newtheorem{definition}{Definition}[section]
\newtheorem{lemma}[theorem]{Lemma}
\newtheorem{remark}[theorem]{Remark}
\newtheorem{proposition}[theorem]{Proposition}
\newtheorem{corollary}[theorem]{Corollary}

\newtheorem{assumption}[theorem]{Assumption}

\def\FullBox{\hbox{\vrule width 6pt height 6pt depth 0pt}}

\def\qed{\ifmmode\qquad\FullBox\else{\unskip\nobreak\hfil
\penalty50\hskip1em\null\nobreak\hfil\FullBox
\parfillskip=0pt\finalhyphendemerits=0\endgraf}\fi}

\usepackage[T1]{fontenc}
\usepackage{lmodern}

\title{\bf Private Low-Rank Approximation for Covariance Matrices, Dyson Brownian Motion, and Eigenvalue-Gap Bounds for Gaussian Perturbations\footnote{This paper combines, extends, and presents complete proofs of the results in \cite{DBM_Neurips} and \cite{our-COLT} by the authors.}
}

 \author{Oren Mangoubi\\ Worcester Polytechnic Institute \and Nisheeth K. Vishnoi \\ Yale University}

\begin{document}
\date{}

\allowdisplaybreaks

\maketitle

\begin{abstract}

We consider the problem of approximating a $d \times d$ covariance matrix $M$ with a rank-$k$ matrix under $(\eps,\delta)$-differential privacy. 
We present and analyze a complex variant of the Gaussian mechanism and obtain upper bounds on the Frobenius norm of the difference between the matrix output by this mechanism and the best rank-$k$ approximation to $M$.
 Our analysis provides improvements over previous bounds, particularly when the spectrum of 
$M$ satisfies natural structural assumptions. 
The novel insight is to view the addition of Gaussian noise to a matrix as a continuous-time matrix Brownian motion. 
This viewpoint allows us to track the evolution of eigenvalues and eigenvectors of the matrix, which are governed by stochastic differential equations discovered by Dyson. 
These equations enable us to upper bound the Frobenius distance between the best rank-$k$ approximation of $M$ and that of a Gaussian perturbation of $M$ as an integral that involves inverse eigenvalue gaps of the stochastically evolving matrix, as opposed to a sum of perturbation bounds obtained via Davis-Kahan-type theorems.
Subsequently, again using the Dyson Brownian motion viewpoint, we show that the eigenvalues of the matrix $M$ perturbed by  Gaussian noise have large gaps with high probability.
These results also contribute to the analysis of low-rank approximations under average-case perturbations, and to an understanding of eigenvalue gaps for random matrices, both of which may be of independent interest.

\end{abstract}

\newpage
\tableofcontents
\newpage

\section{Introduction}

Given a  matrix $M \in \mathbb{R}^{d \times d}$,  consider the following basic problem of finding a rank-$k$ matrix $X$ that is closest to $M$ in Frobenius norm   \cite{bhatia2013matrix,blum2020foundations}:
$$ \min_{X: \mbox{ rank} (X)\leq k} \|M-X\|_F.$$
Of interest is the case when $M$ is the covariance matrix of a data matrix: Given a 
 matrix $A \in \mathbb{R}^{m \times d}$, consisting of $m$ individuals with $d$-dimensional features, $M=A^\top A$.
Such an $M$ is positive semi-definite (PSD) and has eigenvalues $\sigma_1 \geq \cdots \geq \sigma_d \geq 0$.
The  solution to the  optimization problem above is well-known \cite{bhatia2013matrix}: It is given by 
$M_k:= V \Gamma_k V^\top$  where $\Gamma_k := \mathrm{diag}(\sigma_1,\ldots, \sigma_k, 0,\ldots, 0)$
and $V$ is a matrix whose columns are the orthonormal 
eigenvectors of $M$.

In several applications of this low-rank approximation problem, the rows of $A$ correspond to sensitive features of individuals and the release of a low-rank approximation to $M$ may reveal their private information, e.g., as in the case of the Netflix prize problem   \cite{bennett2007netflix}. 
In such contexts, differential privacy (DP) has been employed to quantify the extent to which an algorithm preserves privacy   \cite{dwork2006calibrating} and, in particular, algorithms for low-rank covariance matrix approximation under differential privacy have been widely studied; see     \cite{blum2005practical, dwork2006calibrating, blocki2012johnson,  kapralov2013differentially, dwork2014analyze,  upadhyay2018price,  sheffet2019old,mangoubi2022private} and the references therein.\footnote{Another set of works has studied the problem of approximating a rectangular data matrix $A$ under DP  \cite{blum2005practical, achlioptas2007fast, hardt2012beating, hardt2013beyond}.
 We note that upper bounds on the utility of differentially-private mechanisms for rectangular matrix approximation problems can grow with the number of data points $m$. In contrast, those for covariance matrix approximation problems often depend only on the dimension $d$ of the covariance matrix and do not grow with $m$.}
Notions of DP studied in the literature include $(\eps, \delta)$-DP \cite{dwork2006calibrating, hardt2012beating, hardt2013beyond, dwork2014analyze} which is the notion we study in this paper, as well as pure $(\eps, 0)$-DP \cite{dwork2006calibrating, kapralov2013differentially, amin2019differentially, leake2020polynomial, mangoubi2022private}.

A randomized mechanism $\mathcal{A}$ is said to be $(\eps, \delta)$-differentially private for  privacy parameters $\eps, \delta \geq 0$  if for all ``neighboring'' matrices $M, M' \in  \mathbb{R}^{d \times d}$, and any measurable subset $S$ of the range of $\mathcal{A}$, we have 
\begin{equation}\label{DP}
 \mathbb{P}(\mathcal{A}(M) \in S) \leq e^\eps \mathbb{P}(\mathcal{A}(M') \in S) + \delta.
\end{equation}
Following \cite{blum2005practical, dwork2006calibrating}, $M$ and $M'$ are said to be neighbors if their corresponding data matrices $A,A' \in \mathbb{R}^{m \times d}$ differ by at most one row, i.e., $M'=M-uu^\top+vv^\top$ where $u$ is a row of $A$ and $v$ is a row of $A'$. 
It is assumed that each such row vector is of norm at most $1$, i.e., in the above, $\|u\|_2,\|v\|_2\leq 1$.

Various distance functions have been used in the literature to evaluate the utility of an $(\eps,\delta)$-DP mechanism for matrix approximation problems, including the Frobenius-norm based distances $\|M-\mathcal{A}(M)\|_F - \|M-M_k\|_F$ and 
$\|\mathcal{A}(M) -M_k\|_F$  (e.g. \cite{dwork2014analyze, amin2019differentially}).\footnote{For a variant of this low-rank covariance approximation problem, the subspace approximation problem, an additional metric, the Frobenius inner product utility, is used. The relation of this metric to the Frobenius distance between subspaces is discussed in Section \ref{sec:subspace}.
.}
Using triangle inequality, one can see that an upper bound on $\|\mathcal{A}(M)- M_k\|_F$ implies the same upper bound on  $\|M-\mathcal{A}(M)\|_F - \|M-M_k\|_F$ (the reverse direction is not true in general).
Moreover, the Frobenius-norm distance can be a good utility metric to use if the goal is to recover a low-rank matrix from a dataset of noisy observations (see e.g. \cite{davenport2016overview}). 
Thus, we use the Frobenius-norm distance to measure the utility of an $(\eps,\delta)$-DP mechanism.
This leads to the problem of designing an $(\eps,\delta)$-differentially private mechanism $\mathcal{A}$ that, given a covariance matrix $M$ with eigenvalues $\sigma_1 \geq \cdots \geq \sigma_d \geq 0$, outputs a rank-$k$ matrix $Y$ that minimizes $\| Y - M_k\|_F$.

\cite{dwork2014analyze} analyze a version of the Gaussian mechanism of \cite{dwork2006our}, where one perturbs the entries of $M$ by adding a symmetric matrix $E$ with i.i.d. Gaussian entries $N(0,\nfrac{\sqrt{\log\frac{1}{\delta}}}{\eps})$, to obtain an $(\eps, \delta)$-differentially private mechanism which outputs a perturbed matrix $\hat{M} = M+E$.
They then post-process this matrix $\hat{M}$ to obtain a rank-$k$ matrix $Y$ with the same top-$k$ eigenvectors and eigenvalues as $\hat{M}$.
\cite{dwork2014analyze} show that the output $Y$ of their mechanism satisfies  $\| M-Y\|_F - \|M-M_k\|_F = \tilde{O}(k\sqrt{d})$ w.h.p. (Theorem 7 in \cite{dwork2014analyze}), and also give related bounds for the spectral norm.
While their bound is independent of the number of data points $m$, it may not be tight.
For instance, when $k=d$, one can obtain a tighter bound since, by the triangle inequality, $$\|M-Y\|_F - \|M-M_k\|_F \leq \|Y- M_k\|_F = \|\hat{M}- M\|_F = \|E\|_F \leq O(d)$$ w.h.p., since $\|E\|_F$ is just the norm of a vector of $d^2$ Gaussians with variance $\tilde{O}(1)$.
  Here, the $\tilde{O}$ notation hides polynomial factors of  $\frac{1}{\epsilon}$ and $\log(\frac{1}{\delta})$; in the rest of this section it also hides factors of $(\log d)^{\log \log d}$.
\color{black}

Thus, a question arises whether  Frobenius-norm utility bounds for the rank-$k$ covariance matrix approximation can be improved. 

\paragraph{Our contributions.} We show that a {\em complex}  version of the Gaussian mechanism (Algorithm \ref{alg_quaternion_Gaussian}) satisfies $\|Y- M_k\|_F \leq \tilde{O}(\sqrt{kd})$   whenever $\sigma_{k}-\sigma_{k+1} =\Omega(\sigma_k)$ (Theorem \ref{thm_utility}).
The large $k$'th eigenvalue gap is common in the matrix approximation literature as it motivates the problem of finding a rank-$k$ approximation in the first place --
 it suggests the presence of a useful rank-$k$ ``signal'' $M_k$ which one wishes to extract from the background ``noise'' in the data (see e.g. \cite{dwork2014analyze}).
Moreover, such a gap is also necessary for good rank-$k$ approximations to exist under the stronger metric $\|Y- M_k\|_F$ (see Appendix \ref{appendix_tightness}).
Thus, for matrices with such a gap, our bound improves by a factor
of, roughly, $\sqrt{k}$ on what is implied by the bound of \cite{dwork2014analyze}.
We note that the proof of Theorem \ref{thm_utility} also implies a bound of  $\|M-Y\|_F- \|M-M_k\|_F \leq \tilde{O}(\sqrt{kd})$ without any assumption the eigenvalues of $M$, improving unconditionally on the bound of   \cite{dwork2014analyze}; see Appendix \ref{appendix_gap_free_bounds_in_weaker_metric}.

Our main technical contribution is a new bound on the difference in the Frobenius norm of the best rank-$k$ approximation to $M$  and that of $M+E$ when $E$ is a complex Gaussian matrix (Theorem \ref{thm_rank_k_covariance_approximation_new}). 
Key to this result is the following insight:  View the addition of Gaussian noise to $M$ (the Gaussian mechanism of \cite{dwork2014analyze}) as a continuous-time matrix diffusion $M + B(t)$ for $t \in [0,T]$, with $B(0)=0$ and $B(T)=E$ for an appropriate choice of $T$.
This matrix-valued Brownian motion induces a stochastic process on the eigenvalues $\gamma_1(t) \geq \cdots \geq \gamma_d(t)$ and corresponding eigenvectors $u_1(t), \ldots, u_d(t)$ of $M+B(t)$ originally discovered by Dyson and now referred to as Dyson Brownian motion, with initial values $\gamma_i(0) = \sigma_i$ and $u_i(0)$ which are the eigenvalues and eigenvectors of the initial matrix $M$ \cite{dyson1962brownian}.
We then use these stochastic differential equations to track the perturbations to each eigenvector.
Roughly speaking, these equations say that, as the Dyson Brownian motion evolves over time, every pair of eigenvalues $\gamma_i(t)$ and $\gamma_j(t)$, and corresponding eigenvectors $u_i(t)$ and $u_j(t)$, interacts with the other eigenvalue/eigenvector with the magnitude of the interaction term proportional to $ \frac{1}{\gamma_i(t) - \gamma_j(t)}$ at any given time $t$. 
We then derive a stochastic differential equation that tracks how the utility changes as the Dyson Brownian motion evolves over time and integrate this differential equation over time to obtain a bound on the (expectation of) the utility $\mathbb{E}[\|Y -M_k \|_F]$ (Lemma \ref{Lemma_integral}) as a function of the gaps $\gamma_i(t) - \gamma_j(t)$.
This viewpoint leads to a bound on the utility which includes terms of the form $\frac{(\lambda_i - \lambda_{i+1})^2}{(\gamma_i(t) -\gamma_{i+1}(t))^{2}}$  and $\frac{(\lambda_i - \lambda_{i+1})}{(\gamma_i(t) -\gamma_{i+1}(t))^{2}}$ integrated over time, where, roughly speaking, $\lambda_i = \sigma_i$ for $i \leq k$ and $\lambda_i = 0$ otherwise, where  
 $\sigma_1,\ldots, \sigma_k$ are the eigenvalues of the rank-$k$ approximation $M_k$, and the $\gamma_i(t)$ are the eigenvalues of  $M+B(t)$.
The gaps $\gamma_i(t)-\gamma_{i+1}(t)$, however, may become small for some $t$, causing the terms in the utility bound to become large.

To bypass this,  several novel steps are employed here:
1) Rather than analyzing the utility by considering the output eigenvalues $\lambda_i$ to be fixed numbers, we instead set the top-$k$ output eigenvalues $\lambda_i$ to be {\em dynamically} changing over time and equal to $\gamma_i(t)$, making the gaps in the numerators small at exactly those times when the denominators are small.
2) We then leverage the fact that our mechanism adds {\em complex} Gaussian noise, which implies that $\gamma_i(t)$s evolve by repelling each other with a stronger ``force'' than when only real noise is added, to show that the gaps between the eigenvalues satisfy a high-probability lower bound of $\mathbb{P}(\gamma_i(t) - \gamma_{i+1}(t) \leq  \nfrac{s}{\sqrt{td}}) \leq \tilde{O}(s^3)$;  see Theorem 
\ref{thm:eigenvalue_gap}.
Our bound improves, in the setting where the random matrix is Gaussian, on previous eigenvalue gap bounds of \cite{nguyen2017random} where the bound on the probability decays as $O(s^2)$, which is insufficient for our application.
We prove Theorem \ref{thm:eigenvalue_gap}  by first showing, in Lemma \ref{lemma_gap_comparison}, that one can reduce the problem of bounding the gaps $\gamma_i(t)-\gamma_{i+1}(t)$ to the special case when the initial eigenvalues are all zero, and we subsequently prove Theorem \ref{thm:eigenvalue_gap} for this special case. 

 We suspect that the techniques presented here, which view the addition of random noise by the Gaussian mechanism as a matrix-valued diffusion, will find further applications for other private matrix approximation problems.
Using the ideas in the proof of Theorem \ref{thm_utility},
 we show a result similar to it for the private rank-$k$ subspace recovery problem where the goal is to output the best rank-$k$ {\em projection} matrix that approximates $M$ (Section \ref{sec:subspace}).

There is also a large body of work that studies matrix approximation problems beyond applications to privacy. 
These include deterministic matrix perturbation bounds such as those of \cite{davis1970rotation} (see also \cite{wedin1972perturbation}), which bound the distance between the subspace spanned by the top-$k$ eigenvectors of a symmetric or Hermitian matrix $M$ and a perturbed matrix $M + E$, where $E$ may be any deterministic symmetric or Hermitian matrix.
These bounds have been widely used in many applications which involve matrices perturbed by random noise, including in statistics \cite{yu2015useful}, engineering \cite{garg2021order, hajrya2013principal}, and numerical linear algebra \cite{cullum2002lanczos, hardt2014understanding}.
  While these perturbation bounds are tight with respect to worst-case deterministic perturbations, in many of these applications the noise $E$ is a Gaussian random matrix, and available bounds are not tight with respect to Gaussian random matrix noise.
   In the setting where $E$ is Gaussian, our bounds in Theorems \ref{thm_rank_k_covariance_approximation_new} and \ref{thm_rank_k_subspace} improve over the bounds implied by \cite{davis1970rotation} (and also improve on bounds implied by previous results \cite{o2018random} specialized to random matrix perturbations) for many matrices $M$ with spectral profiles $\sigma_1 \geq \cdots \geq \sigma_d$ with specific structure (see the discussion following Theorems \ref{thm_rank_k_covariance_approximation_new} and \ref{thm_rank_k_subspace}).

      The techniques that we introduce in the proofs of Theorems  \ref{thm_utility}, \ref{thm_rank_k_covariance_approximation_new}, and \ref{thm_rank_k_subspace} may be of independent interest.
  For instance, the stochastic analysis techniques developed in our paper which analyze perturbations of Hermitian matrices via the Dyson Brownian motion eigenvector process have been extended in \cite{lai2024singular} to obtain bounds on perturbations to the singular vectors of {\em rectangular} matrices perturbed by Gaussian noise.
 There, the evolution of the singular vectors is analyzed via a related stochastic process— the Dyson-Bessel process— which governs the evolution of the singular values and singular vectors of a rectangular matrix-valued diffusion.

Moreover, there is a long line of work (see e.g., \cite{wigner1955characteristic, dyson1962statistical, ginibre1965statistical, erdHos2012rigidity}) which studies the eigenvalues of Gaussian Orthogonal Ensemble (GOE) (and Gaussian Unitary Ensemble (GUE)) random matrices-- random matrices $G+G^\ast$ where $G$ has i.i.d. (complex) Gaussian entries -- including their gap statistics  \cite{ben2013extreme, feng2018large, feng2019small, forrester2018functional, giraud2022probing, bui2018gaps}, (and, more generally, the gap statistics of Wigner random matrices \cite{tao2013asymptotic, nguyen2017random}).
As our eigenvalue gap bounds in Theorem \ref{thm:eigenvalue_gap} improve over previous eigenvalue gap bounds \cite{nguyen2017random} for GOE and GUE random matrices, Theorem \ref{thm:eigenvalue_gap}  may be of interest in random matrix theory.

Extending Theorem \ref{thm_rank_k_covariance_approximation_new} to the spectral-norm utility is left as an open problem (see Remark \ref{rem_spectral_norm} for details).
\color{black}

\section{Main results}

For any $S \in \mathbb{C}^{d \times d}$, denote by $S^\ast$ its conjugate transpose. 
$S$ is Hermitian if $S=S^\ast$.
For any Hermitian matrix $S \in \mathbb{C}^{d \times d}$, consider its spectral decomposition $S = U \Lambda U^\ast$ where $\Lambda := \mathrm{diag}(\lambda_1,\ldots, \lambda_d)$ is a diagonal matrix containing the eigenvalues $\lambda_1 \geq \cdots \geq \lambda_d$ of $S$ and $U  := [u_1, \ldots, u_d]$ is a unitary matrix containing the eigenvectors $u_1, \ldots, u_d$ of $S$.
Denote by $\Lambda_k := \mathrm{diag}(\lambda_1,\ldots, \lambda_k, 0,\ldots,0)$ and by $S_k := U\Lambda_k U^\ast$ the Frobenius-norm minimizing rank-$k$ approximation of $S$.
Denote by $U_k := [u_1, \ldots, u_k]$ the $d \times k$ matrix of the top-$k$ eigenvectors of $U$.

\subsection{Private low-rank covariance approximation} \label{sec_our_results_1}

Our first result (Theorem \ref{thm_utility}) analyzes the complex Gaussian mechanism (Algorithm \ref{alg_quaternion_Gaussian})  and provides an upper bound on the expected Frobenius distance utility of this mechanism for the problem of rank-$k$ covariance approximation.
In the following, the $\tilde{O}$ notation hides polynomial factors of $(\log d)^{\log \log d}$; when used in discussions outside of formal theorem statements and proofs, the $\tilde{O}$ notation oftentimes also hides factors of $\epsilon$ and $\log(\frac{1}{\delta})$.

 This result relies on the following assumption about the $k$'th eigenvalue gap of the input matrix M:
\begin{assumption}[\bf $(M, k, T)$ eigenvalue gap]\label{assumption_gap}
The gap between the $k'$th largest eigenvalue $\sigma_k$ and $k+1$'st largest eigenvalue $\sigma_{k+1}$ of the matrix $M$ satisfies  $\sigma_k - \sigma_{k+1} \geq 4\sqrt{Td} $.
\end{assumption}
\color{black}

\LinesNumbered
\begin{algorithm}
\caption{Complex Gaussian Mechanism} 
 \label{alg_quaternion_Gaussian}
\label{alg_general_self_concordant}
\KwIn{$\epsilon, \delta >0$, $d, k \in \mathbb{N}$. A real symmetric PSD matrix $M \in \mathbb{R}^{d \times d}$}

 \KwOut{A real symmetric matrix $Y \in \mathbb{R}^{d\times d}$}

Sample matrices $W_1, W_2 \in \mathbb{R}^{d \times d}$ with i.i.d. $N(0,1)$ entries

Set $G := (W_1 +  \mathfrak{i} W_2) + (W_1 +  \mathfrak{i} W_2)^\ast$

Set $\hat{M} := M + \sqrt{T}G$,  where $T := \frac{2\log\frac{1.25}{\delta}}{\eps^2}$

Compute the diagonalization $\hat{M} = \hat{V} \hat{\Sigma} \hat{V}^\ast$ with eigenvalues $\hat{\sigma}_1 \geq \cdots \geq \hat{\sigma}_d$

Set $\hat{M}_k := \hat{V} \hat{\Sigma}_k \hat{V}^\ast$,  where $\hat{\Sigma}_k := \mathrm{diag}(\hat{\sigma}_1, \ldots,  \hat{\sigma}_k, 0,\ldots, 0)$

Output  $Y$ to be the real part of $\hat{M}_k$

\end{algorithm}

\begin{theorem}[\bf Private low-rank covariance approximation] \label{thm_utility}
Given $\eps,\delta>0$, there is an $(\varepsilon, \delta)$-differentially private algorithm (Algorithm \ref{alg_quaternion_Gaussian}) that, on input $k>0$ and a real symmetric PSD matrix $M \in \mathbb{R}^{d \times d}$ with eigenvalues $\sigma_1 \geq \cdots \geq \sigma_d \geq 0$
satisfying Assumption \ref{assumption_gap} $(M, k,  \nicefrac{(4\log\frac{1.25}{\delta})}{\eps^2})$ and $\sigma_1 \leq d^{50}$, 
outputs a rank-$k$ matrix $Y \in \mathbb{R}^{d \times d}$ such that  
$$ \sqrt{\mathbb{E}[\|M_k - Y\|_F^2]}  \leq \tilde{O}\left(\sqrt{kd} \times \frac{\sigma_k}{\sigma_k - \sigma_{k+1}} \times  \frac{\sqrt{\log\frac{1}{\delta}}}{\eps}\right).$$
$M_k$ 
is the Frobenius-norm minimizing rank-$k$ approximation to $M$.
\end{theorem}
\noindent
The proof of Theorem \ref{thm_utility} appears in Section \ref{sec_proof_utility_privacy}.
We note that the requirement in Theorem \ref{thm_utility} that  $\sigma_1 \leq d^{50}$ is an artifact of the proof, and can be replaced with $\sigma_1 \leq d^{C}$ for any large universal constant $C>0$.

For matrices $M$ whose  $k$'th eigengap satisfies $\sigma_k - \sigma_{k+1} = \Omega(\sigma_k)$, Theorem \ref{thm_utility} improves by a factor of $\sqrt{k}$ on the (expectation of) the bound in Theorem 7 of \cite{dwork2014analyze} which says the output $Y$ of their mechanism satisfies  $\|Y - M\|_F - \|M_k - M\|_F = \tilde{O}(k\sqrt{d} \log(\frac{1}{s}))$ w.h.p $1-s$ for any $s>0$ \footnote{Note that while the growth rate of the high-probability bound is not stated explicitly in the statement of Theorem 7 in \cite{dwork2014analyze}, a logarithmic growth rate of $\log(\frac{1}{s})$ follows directly from their proof.}.
\color{black}
This is because an upper bound on $\| Y - M_k\|_F$ implies an upper bound on their utility measure by the triangle inequality.
The reason why \cite{dwork2014analyze} is independent of the gap $\sigma_k - \sigma_{k+1}$ while our bound depends on the  ratio $\frac{\sigma_k}{\sigma_k -\sigma_{k+1}}$  is due to the fact that 
if, e.g., $\sigma_k - \sigma_{k+1} = 0$ an arbitrarily small Gaussian perturbation to $M$ would lead to a perturbation of $\|\hat{V}_k - V_k\|_2 =\Omega(1)$ w.h.p., where $\hat{V}_k$ and $V_k$ are the matrices containing the top-$k$ eigenvectors of $\hat{M}$ and $M$ respectively.
 Roughly speaking, this, in turn, would lead to a perturbation of at least $\|Y - M_k \|_F  \geq \|\sigma_k\hat{V}_k \hat{V}_k^\ast - \sigma_k V_k V_k^\ast\|_2 \geq \Omega(  \sigma_k)$.
The techniques used in the proof of Theorem \ref{thm_utility} can also be used to improve this Frobenius utility to $\tilde{O}(\sqrt{kd})$ {\em without} assuming the eigengap condition; see Theorem \ref{thm_weaker_Frobenius_metric} in  Appendix \ref{appendix_gap_free_bounds_in_weaker_metric}.
For many applications, the matrix $M$ has a large $k$'th eigenvalue gap $\sigma_{k}-\sigma_{k+1}$  (e.g., $\sigma_{k}-\sigma_{k+1} = \Omega(\sigma_{k})$), and the presence of a large $k$'th eigenvalue gap is oftentimes given as the motivation for computing a low-rank approximation of a given rank $k$ (see e.g. \cite{yu2015useful, zhang2020note,  von2007tutorial,  jolliffe2002choosing, halimi2016estimating}).
That being said, for applications where the weaker metric $\|Y - M\|_F - \|M_k - M\|_F$ may be sufficient, the eigengap-free bound in Theorem \ref{thm_weaker_Frobenius_metric} may be of interest. 
 See Appendix \ref{sec_strong_weak_metric_comparison} for a discussion comparing the stronger Frobenius distance metric $\| Y - M_k\|_F$ used in Theorem \ref{thm_utility} and the weaker metric of Theorem  \ref{thm_weaker_Frobenius_metric}. 
Finally, the expectation bound in Theorem \ref{thm_utility} immediately implies a high-probability bound via Chebyshev's inequality,  $\|\hat{M}_k - M_k \|_F \leq \tilde{O}(\sqrt{kd} \frac{\sigma_k}{\sigma_k - \sigma_{k+1}}  \sqrt{T}  \frac{1}{\sqrt{s}})$ w.h.p. $1-s$ for all $s>0$.

Note that the growth factor of this high-probability bound, $\frac{1}{\sqrt{s}}$, is sublinear in the (inverse) probability parameter $\frac{1}{s}$, while the bound in \cite{dwork2014analyze} has a logarithmic growth factor, $\log(\frac{1}{s})$.
It is an interesting open problem whether our bounds can be extended to high-probability bounds which grow logarithmically in the probability parameter $\frac{1}{s}$ (see Appendix \ref{sec_high_probability} for details).
\color{black}

The privacy guarantee in Theorem \ref{thm_utility} follows directly from prior works on the (real) Gaussian mechanism (see Section \ref{sec_proof_utility_privacy}).
The utility bound in Theorem \ref{thm_utility} follows from the following ``average-case'' matrix perturbation bound for complex Gaussian random perturbations.

\subsection{Bound on change in low-rank approximations under complex Gaussian perturbations}

\begin{theorem}[\bf Frobenius bound for complex Gaussian perturbations] \label{thm_rank_k_covariance_approximation_new}
Suppose we are given $k>0$, $T>0$, and a  Hermitian matrix  $M \in \mathbb{C}^{d \times d}$ with eigenvalues  $\sigma_1 \geq \cdots \geq \sigma_d \geq 0$.   
Let $\hat{M} := M + \sqrt{T}[(W_1 + \mathfrak{i}W_2) + (W_1 + \mathfrak{i}W_2)^\ast]$ where $W_1, W_2 \in \mathbb{R}^{d \times d}$ have entries which are independent $N(0,1)$ random variables.
 Denote, respectively, by  $\sigma_1 \geq \cdots \geq \sigma_d$ and  $\hat{\sigma}_1\geq \ldots \geq \hat{\sigma}_d \geq 0$ the eigenvalues of $M$ and $\hat{M}$, and by $V$ and  $\hat{V}$ the matrices whose columns are the corresponding eigenvectors of $M$ and $\hat{M}$.
Moreover, let  $M_k := V \Gamma_k V^\ast$ and  $\hat{M}_k := \hat{V} \hat{\Gamma}_k \hat{V}^\ast$ be the rank-$k$ approximations of $M$ and $\hat{M}$, where $\Gamma_k := \mathrm{diag}(\sigma_1,\ldots, \sigma_k,0,\ldots,0)$  and $\hat{\Gamma}_k := \mathrm{diag}(\hat{\sigma}_1,\ldots, \hat{\sigma}_k,0,\ldots,0)$. 
Suppose that $M$ satisfies Assumption \ref{assumption_gap} $(M, k, T)$ and that $\sigma_1 \leq d^{50}$. 
Then we have
$$ 
\sqrt{\mathbb{E}\left[\left\|\hat{M}_k -  M_k\right \|_F^2\right]} \leq \tilde{O}\left(\sqrt{kd} \frac{\sigma_k}{\sigma_k - \sigma_{k+1}}\right) \cdot \sqrt{T}.
$$
\end{theorem}
\noindent
 The proof of Theorem \ref{thm_rank_k_covariance_approximation_new} is presented in Section \ref{sec_utility}.
The requirement  $\sigma_1 \leq d^{50}$  can be replaced with $\sigma_1 \leq d^{C}$ for any large universal constant $C>0$.
One can also compare the bound in this theorem to those obtained by deploying deterministic eigenvector perturbation bounds such as those of \cite{davis1970rotation}, which say roughly that given any Hermitian matrices $M, E$, one has 
\begin{equation}\label{eq_DK}
\left\|\hat{V}_k\hat{V}_k^\ast - V_k V_k^\ast \right\|_2 \leq \frac{\|E\|_2}{\sigma_k - \sigma_{k+1}},
\end{equation}
where $V_k$ and $\hat{V}_k$ are, respectively, the top-$k$ eigenvectors of the input matrix $M$ and the perturbed matrix $\hat{M} := M+E$.
Applying \eqref{eq_DK}, together with concentration bounds which say that the spectral norm of a random matrix $G$ with i.i.d. $N(0,1)$ entries satisfies $\|G\|_2 = O(\sqrt{d})$ w.h.p. (e.g. Theorem 4.4.5 of \cite{vershynin2018high}), one can obtain a bound on the Frobenius distance of
$$\left\|\hat{M}_k - M_k\right\|_F \leq O\left(k^{1.5}\sqrt{d} + \frac{\sigma_k}{\sigma_{k}-\sigma_{k+1}} \sqrt{k}\sqrt{d}\right) \cdot \sqrt{T}$$ w.h.p. when  $\sigma_k-\sigma_{k+1} \geq \Omega(\sqrt{Td})$ (see Section \ref{sec_deterministic_perturbation_bounds} for details). 
Theorem \ref{thm_rank_k_covariance_approximation_new} improves (in expectation) on this bound by a factor of $k$ when e.g. $\sigma_k-\sigma_{k+1} = \Omega(\sigma_k)$.

\cite{o2018random} provide eigenvector perturbation bounds for matrices $\hat{M} := M+E$ when the input matrix $M$ is a deterministic low-rank matrix of rank $r \geq k$ and the matrix $E$ is a random matrix.
In particular, their Theorem 18 improves w.h.p. on the deterministic bound \eqref{eq_DK} for certain inputs $M$ of sufficiently low rank and random matrices $E$. 
If one directly applies their bound to the setting when $E$ is a Hermitian Gaussian random matrix (e.g., by plugging in their bound in place of  \eqref{eq_DK} in Inequality \eqref{eq_DK_1} of Section \ref{sec_deterministic_perturbation_bounds}), one obtains a bound on the quantity $\|\hat{M}_k - M_k\|_F$.
Theorem \ref{thm_rank_k_covariance_approximation_new} improves (in expectation) on the resulting bound by a factor of $k^{1.5}$ whenever e.g. $\sigma_{k}-\sigma_{k+1} \geq \Omega(\sqrt{d})$.

While we do not know if our bound in Theorem  \ref{thm_rank_k_covariance_approximation_new} is tight for every input matrix $M$, we do verify that it is tight for every $k \leq d$ and every value of the gap ratio $\frac{\sigma_k}{ \sigma_k- \sigma_{k+1}}$, up to factors of $(\log d)^{\log \log d}$ hidden in the $\tilde{O}$ notation (see Appendix \ref{appendix_tightness} for details).
An interesting open problem is whether complex Gaussian noise is necessary to achieve our bounds in Theorems \ref{thm_utility} and \ref{thm_rank_k_covariance_approximation_new}, or whether our analysis can be extended to real Gaussian noise.

\subsection{Eigenvalue gaps under complex Gaussian perturbations }\label{sec_technical_results}
One of the key steps in this paper is to reduce the proof of Theorem \ref{thm_rank_k_covariance_approximation_new} to estimating gaps between eigenvalues of the matrix $M+G+G^*$ where $G$ is a random matrix with i.i.d. complex Gaussian entries.
This reduction is non-trivial and is explained in Section \ref{sec_technical_overview}.
The random matrix $G+G^\ast$ is referred to as the Gaussian Unitary Ensemble (GUE) when $G$ has i.i.d. complex Gaussian entries, and as the Gaussian Orthogonal Ensemble (GOE) when G has i.i.d. real Gaussian entries.

\begin{theorem}[\bf Eigenvalue gaps of  Gaussian Unitary Ensemble (GUE) and Gaussian Orthogonal Ensemble (GOE)]\label{thm:eigenvalue_gap}
Let $M$ be a complex Hermitian matrix (or a real symmetric $M \in \mathbb{R}^{d \times d}$).
Let $A := M + G + G^\ast$ where $G$ is a matrix with i.i.d. complex (or real) standard Gaussian entries, and denote by $\eta_1,\ldots, \eta_d$ the eigenvalues of $A$.
Then
$$
    \mathbb{P}\left(\eta_i - \eta_{i+1} \leq s \frac{1}{\mathfrak{b}\sqrt{d}}\right) \leq s^{\beta +1} + \frac{1}{d^{1000}}$$ for all $s>0$, and for all $1\leq i < d$,
where $\beta=2$ for the complex Hermitian case (and $\beta=1$ for the real-symmetric case), and $\mathfrak{b} = (\log d)^{ L\log \log d}$ and $L$ is a universal constant.    
\end{theorem}
\noindent
The proof of Theorem \ref{thm:eigenvalue_gap} is presented in Section \ref{section_GUE_proof} and an overview appears in Section \ref{sec_technical_gaps}.
We note that the term $\frac{1}{d^{1000}}$ in Theorem \ref{thm:eigenvalue_gap}  can be replaced by $\frac{1}{d^C}$ for any universal constant $C>0$.
Thus, Theorem \ref{thm:eigenvalue_gap} says that for any $s> d^{-C}$ (where $C$ can be taken to be any large universal constant), the probability that any gap $\eta_i- \eta_{i+1}$ of a Hermitian matrix $M$ perturbed by a GUE random matrix is less than or equal to $\tilde{O}\left(\frac{s}{\sqrt{d}}\right)$ is $O(s^3)$.
The $s^3$ dependence is important to our analysis of the Frobenius-distance utility in Theorem  \ref{thm_rank_k_covariance_approximation_new},  where we wish to bound the time-average of the second moment of the inverse gaps $\mathbb{E}\left[\frac{1}{(\gamma_i(t) - \gamma_j(t))^2}\right]$.
Theorem \ref{thm:eigenvalue_gap} allows us to bound this term by $O(d)$.
We use it to bound the (squared) expected Frobenius utility, $\mathbb{E}[\|\hat{M}_k -  M_k \|_F^2] \leq O(kd)$, thus implying the bound in Theorem \ref{thm_rank_k_covariance_approximation_new}.

For the special case when $M=0$, the matrix $A$ is a GUE or GOE random matrix depending on whether we add complex or real Gaussian noise.   
 The distribution of the gaps of the GUE/GOE random matrices in the limit as $d \rightarrow \infty$ was studied, e.g., in \cite{dyson1963random, tao2013asymptotic, arous2013extreme}, and was also studied non-asymptotically in e.g. \cite{nguyen2017random}.
However, to the best of our knowledge, we are not aware of a previous (non-asymptotic in $d$) lower bound on the gaps of the complex Hermitian GUE random matrices which scales as small as $O(s^3)$ (or  $O(s^2)$ for the real symmetric GOE).
For instance, \cite{nguyen2017random}, which studies eigenvalue gaps of Wigner random matrices with sub-Gaussian tails--a more general class of random matrices which includes as a special case the GUE/GOE random matrices--show a bound of $\mathbb{P}\left(\eta_i- \eta_{i+1} \leq \frac{s}{\sqrt{d}}\right) \leq O(s^2)$ for the eigenvalues $\eta$ of the complex Hermitian GUE (or $\mathbb{P}\left(\eta_i- \eta_{i+1} \leq \frac{s}{\sqrt{d}}\right)  \leq O(s)$ in the case of the real symmetric GOE) for any  $s> d^{-C}$ where $C>0$ is a universal constant (Corollary 2.2 in \cite{nguyen2017random}, which they can extend to the complex case).
On the other hand, we note that \cite{nguyen2017random} focus on matrix universality results that apply to a larger class of  random matrices than the GUE/GOE random matrices,
and that our bound includes additional factors of $(\log d)^{\log \log d}$ hidden in the $\tilde{O}$ notation.
We believe the decay rates of $1- O(s^3)$ for the GUE (and $1- O(s^2)$ for the GOE) in our eigenvalue gap bounds are tight, see Section \ref{sec_technical_overview}.

Finally, we note that Theorem \ref{thm:eigenvalue_gap} may be of independent interest to subareas of mathematics, physics, and numerical analysis where the eigenvalue gaps of the GOE or GUE random matrices arise.
 There is a long line of work which studies the statistics of the eigenvalues of GOE/GUE random matrices, including their gap statistics  \cite{ben2013extreme, feng2018large, feng2019small, forrester2018functional, giraud2022probing, bui2018gaps}, (and, more generally, the gap statistics of Wigner random matrices \cite{tao2013asymptotic, nguyen2017random}).
 The eigenvalue gap statistics of the GOE/GUE random matrix have applications to numerous areas of mathematics, including, e.g. analytic number theory where the local statistics of the zeros of the Riemann zeta function are conjectured to follow the distribution of the GUE eigenvalues \cite{montgomery1973pair, rudnick1996zeros, blomer2017small}.
   They also have applications to quantum physics, where, starting with Wigner who used the eigenvalue statistics of the GOE to model the distribution of large atomic nuclei \cite{Wigner_1956_report}, the local statistics of the energy level of chaotic quantum Hamiltonians are conjectured to follow the eigenvalue statistics of GOE or GUE random matrices  (see e.g., \cite{dyson1962statistical, bohigas1986spectral, guhr1998random, valko2014random, cotler2017chaos}). 
    Moreover, eigenvalue gap bounds for matrices perturbed by random noise have been used to bound the convergence rate of randomized numerical linear algebra algorithms (see e.g.  \cite{kulkarni2020random, peng2021solving, meyer2024unreasonable}).
\section{Preliminaries}

\subsection{Brownian motion and It\^o calculus}

In this section, we give preliminaries on Brownian motion and Stochastic calculus (also referred to as It\^o calculus).
A Brownian motion $W(t)$ is a continuous process that has stationary
independent Gaussian increments  (see e.g., \cite{morters2010brownian}).
In a multi-dimensional Brownian motion,  each coordinate is an independent and identical Brownian motion. 
The filtration $\mathcal{F}_t$  generated by  $W(t)$ is defined as $\sigma \left(\cup_{s \leq t} \sigma(W(s))\right)$, where $\sigma(\Omega)$ is the $\sigma$-algebra generated by $\Omega$.
$W(t)$ is a martingale with respect to $\mathcal{F}_t$.

\begin{definition}[\bf It\^o Integral]
Let $W(t)$ be a Brownian motion for $t \geq 0$, let $\mathcal{F}_t$ be the filtration generated by $W(t)$, and let $z(t):  \mathcal{F}_t \rightarrow \mathbb{R}$ be a stochastic process adapted to $\mathcal{F}_t$.
The It\^o integral is defined as 
$$\int_0^T z(t) \mathrm{d}W(t) := \lim_{\omega \rightarrow 0} \sum_{i=1}^{\frac{T}{\omega}} z(i\omega)\times[W((i+1)\omega) -W(i\omega)].$$   
\end{definition}

\noindent
The following lemma generalizes the chain rule of deterministic derivatives to stochastic derivatives. It allows one to compute the derivative of a function $f(X(t))$ of a stochastic process $X(t)$.
We state It\^o's Lemma in its integral form:
\begin{lemma}[\bf It\^o's Lemma, integral form with no drift; Theorem 3.7.1 of \cite{lawler2010stochastic}] \label{lemma_ito_lemma_new}
Let $f:  \mathbb{R}^n \rightarrow \mathbb{R}$ be any twice-differentiable function.
Let $W(t) \in \mathbb{R}^n$  be a Brownian motion, and let $X(t)  \in \mathbb{R}^n$ be an It\^o diffusion process with mean zero defined by the following stochastic differential equation:
\begin{equation}
\mathrm{d}X_j(t) = \sum_{i=1}^d R_{i j }(t) \mathrm{d}W_i(t)
\end{equation}
for some It\^o diffusion $R(t) \in \mathbb{R}^{n \times n}$ adapted to the filtration generated by the Brownian motion $W(t)$.
Then for any $T\geq 0$,
\begin{eqnarray*}
  f(X(T)) -f(X(0)) &=& \int_0^T \sum_{i=1}^n \sum_{\ell=1}^n \left(\frac{\partial}{\partial X_\ell} f(X(t))\right) R_{i \ell}(t) \mathrm{d}W_i(t)\\
  & &  +  \qquad \frac{1}{2} \int_0^T \sum_{i=1}^n \sum_{j=1}^n \sum_{\ell=1}^n \left(\frac{\partial^2}{\partial X_{j} \partial X_\ell} f(X(t))\right) R_{i j }(t) R_{i \ell}(t) \mathrm{d}t.
\end{eqnarray*}
\end{lemma}
\noindent
We note that the above version of It\^o's Lemma ( Lemma \ref{lemma_ito_lemma_new}) is given for real-valued variables.
When we apply It\^o's Lemma to complex matrix-valued stochastic processes, we will separate the real and imaginary parts of the It\^o integral and apply It\^o's Lemma separately to each part.

\begin{definition}[\bf Strong solution to SDE; Definition 5.3.1 in \cite{karatzas1991brownian}]
Given a standard Brownian motion $W_t$ on $\mathbb{R}^d$, and any $\mu: \mathbb{R}^d \rightarrow \mathbb{R}^d$ and $R: \mathbb{R}^d \rightarrow \mathbb{R}^{d \times d}$, a strong solution to the stochastic differential equation (SDE) $$\mathrm{d} X_t = \mu(X_t) \mathrm{d}t +R(X_t) \mathrm{d}W_t$$ with initial condition $x \in \mathbb{R}^d$ is a stochastic process $X_t$ adapted to $W_t$ with continuous paths such that, almost surely,
\begin{equation}
X_t = x + \int_0^t \mu(X_s) \mathrm{d} s + \int_0^t R(X_s) \mathrm{d} W_s,
\end{equation}
for all $t \geq 0$.
\end{definition}
\noindent
In particular, we note that a strong solution $X_t$ is adapted to a {\em particular} Brownian motion $W_t$.
 In other words, $X_t$ is probabilistically coupled to the Brownian motion $W_t$.
This concept will allow us to compare the solution of two SDEs by adapting them (coupling them) to the same Brownian motion $W_t$.

\color{black}

\subsection{Dyson Brownian motion}\label{sec_DBM}

 Let $W(t) \in  \mathbb{C}^{d \times d}$ be a matrix where the real part (and complex part) of each entry is an independent standard Brownian motion with distribution $N(0, tI_d)$ at time $t$, and let $B(t) := W(t) + W(t)^\ast$.
Define the Hermitian-matrix valued stochastic process $\Phi(t)$ as follows:
\begin{equation} \label{eq_DBM_matrix}
    \Phi(t):= M + B(t) \qquad  \forall t\geq 0.
\end{equation}
At every time $t>0$, the eigenvalues $\gamma_1(t), \ldots, \gamma_d(t)$ of $\Phi(t)$ are real-valued and  distinct w.p. $1$, and \eqref{eq_DBM_matrix} induces a stochastic process on the eigenvalues and eigenvectors.
The  evolution of the eigenvalues can be expressed by the following stochastic differential equations  (SDE) \cite{dyson1962brownian}: 
\begin{equation} \label{eq_DBM_eigenvalues}
   \mathrm{d} \gamma_i(t) = \mathrm{d}B_{i i}(t) +  \beta \sum_{j \neq i} \frac{1}{\gamma_i(t) - \gamma_j(t)} \mathrm{d}t \qquad \qquad \forall i \in [d], t > 0,
\end{equation}
where the parameter $\beta=2$ for the complex case ($\beta=1$ for the real matrix Brownian motion) (Figure \ref{fig_DBM}).

\begin{figure}
    \centering
    \includegraphics[width=0.4\textwidth]{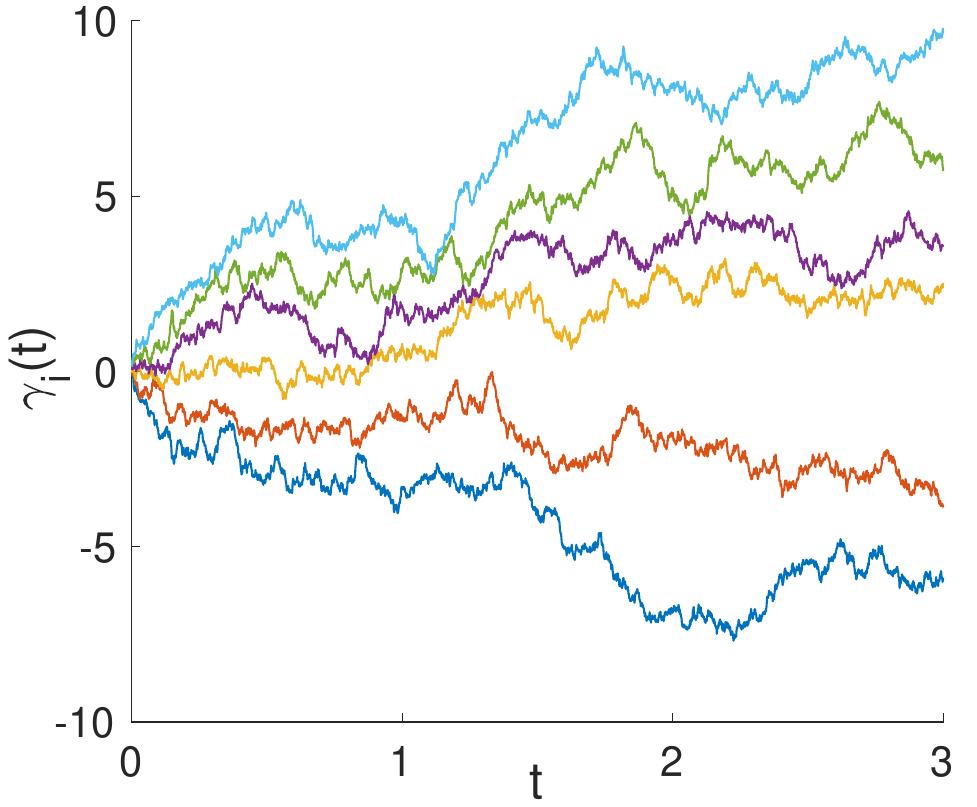}
    \includegraphics[width=0.4\textwidth]{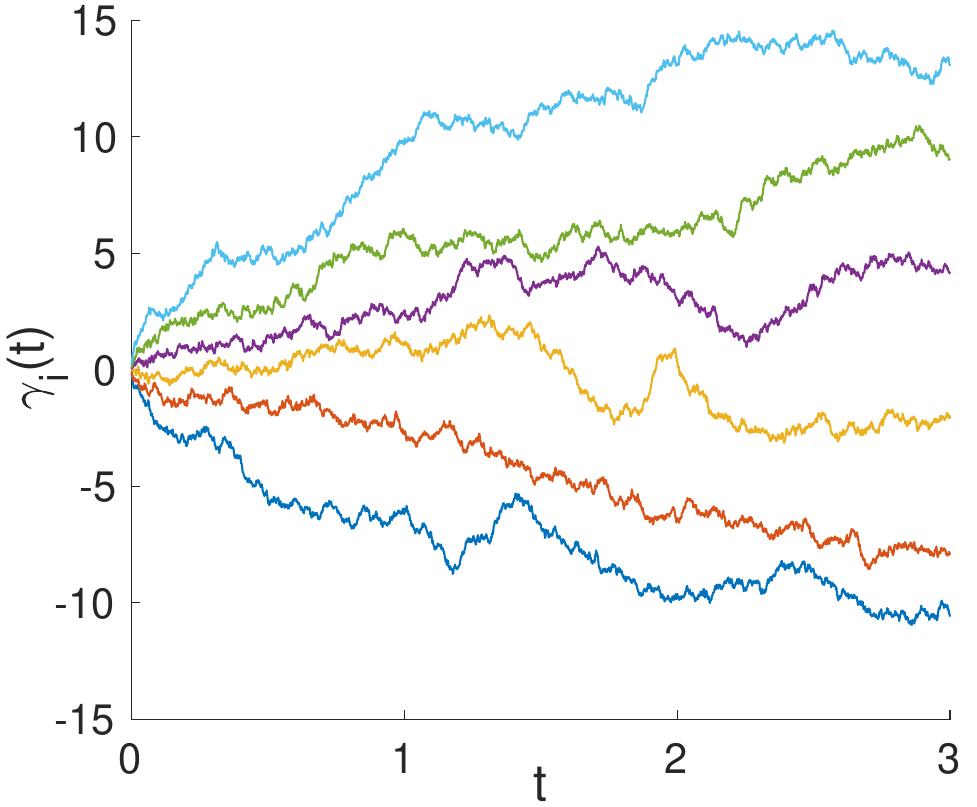}
    \vspace{-2mm}
    \caption{One run of a simulation of the eigenvalues $\gamma_1(t) \geq \cdots \geq \gamma_d(t)$ of Dyson Brownian, in the real case (left) and the complex case (right) with initial condition $\gamma_1(0) = \cdots = \gamma_d = 0$, for $d=6$.
    In the complex case, eigenvalue repulsion is stronger and the gaps between the eigenvalues are not as small as in the real case.
    }
    \vspace{-3mm}
    \label{fig_DBM}
\end{figure}

The corresponding eigenvector process $u_1(t), \ldots, u_d(t)$, referred to as the Dyson vector flow, is also a ``diffusion'' and,  conditional on the eigenvalue process \eqref{eq_DBM_eigenvalues}, is given by the following SDEs:
\begin{equation}\label{eq_DBM_eigenvectors}
  \mathrm{d}u_i(t) = \sum_{j \neq i} \frac{\mathrm{d}B_{ij}(t)}{\gamma_i(t) - \gamma_j(t)}u_j(t) - \frac{\beta}{2}\sum_{j \neq i} \frac{\mathrm{d}t}{(\gamma_i(t)- \gamma_j(t))^2}u_i(t) \qquad \qquad \forall i \in [d], t > 0.
\end{equation}

\paragraph{Properties of Dyson Brownian motion.}
Let $\mathcal{O}(d)$ denote the space of $d \times d$ real orthogonal matrices, and $\mathcal{U}(d)$ the space of $d \times d$ complex unitary matrices. 
The following lemma, which guarantees the existence and uniqueness of solutions to the eigenvalue \eqref{eq_DBM_eigenvalues} and eigenvector SDE's \eqref{eq_DBM_eigenvectors}, is known -- see Theorem 2.3(a) in \cite{bourgade2017eigenvector} and Lemma 4.3.3 in \cite{anderson2010introduction} for solutions of just the eigenvalue process for any $\beta \geq 1$.
While the solutions are random processes, the outcome of these solutions can be shown to be unique when coupled with the underlying Brownian motion processes driving the SDE.
Such a coupling is referred to as a ``strong solution'' to the SDE (see e.g. \cite{lawler2010stochastic}).
 In the following, we define
 \begin{equation}\label{eq:WeylChamber}
 \mathcal{W}_d := \{(x_1, \ldots, x_d) \in \mathbb{R}^d:  x_1 \geq \cdots \geq x_d\}.
 \end{equation}
\begin{lemma}[\bf Existence and uniqueness of solutions to Dyson Brownian motion] \label{lemma_strong}
Consider any $T \geq T_0 \geq 0$ and $\beta \in \{1,2\}$.
Let $\{\gamma(t)\}_{t \in [0,T_0]} \subseteq \mathcal{W}_d$ be a continuous initial path for \eqref{eq_DBM_eigenvalues} and let $\{u(t)\}_{t \in [0,T_0]} \subseteq \mathcal{U}(d)$ if $\beta=2$ (or  $\{u(t)\}_{t \in [0,T_0]} \subseteq \mathcal{O}(d)$  if $\beta =1$) be a continuous initial path for \eqref{eq_DBM_eigenvectors}.
Then there exists a unique strong solution for the system of SDEs \eqref{eq_DBM_eigenvalues} on all of $[0,T]$.
Moreover, there exists a unique strong solution on all of $[0,T]$ for the system of SDEs comprising \eqref{eq_DBM_eigenvalues} and \eqref{eq_DBM_eigenvectors}.

\end{lemma}
\noindent
In particular (by the definition of strong solution) the existence of strong solutions implies that the paths of Dyson Brownian motion are almost surely continuous on $[0,\infty)$.
This fact will be useful in proving our gap comparison theorem for coupled solutions of Dyson Brownian motions (Lemma \ref{lemma_gap_comparison}).
The following result shows that the paths of Dyson Brownian motion are continuous with respect to their initial conditions:
\begin{lemma}[\bf Continuity w.r.t. initial condition; Proposition 4.3.5 in \cite{anderson2010introduction}]\label{lemma_continuity} 
Let $\gamma$ be a strong solution to \eqref{eq_DBM_eigenvalues} for any initial condition $\gamma(0)\in \mathcal{W}_d$. 
Then, at any time $t\geq 0$, $\gamma(t)$ is a continuous function of the initial condition $\gamma(0)$.
\end{lemma}
\noindent
The following lemma is known; see Theorem 1.1 in \cite{inukai2006collision} and also \cite{rogers1993interacting}.
\begin{lemma}[\bf Non-collision of Dyson Brownian motion for $\beta \geq 1$] \label{lemma_DBM_collision}
Let $\gamma$ be a solution to \eqref{eq_DBM_eigenvalues} with any initial condition $\gamma(0) \in \mathcal{W}_d$.
Let $\tau := \inf \{t>0: \gamma_i(t) = \gamma_j(t) \textrm{ for some } i\neq j\}$ be the first positive time any of the particles in  $\gamma(t)$ collide.
Then if $\beta \geq 1$, $\mathbb{P}(\tau <\infty) = 0$.
\end{lemma}

\subsection{Matrix inequalities}

The following  lemmas will help us bound the gaps in the eigenvalues of Hermitian matrices perturbed by a (Gaussian) Hermitian random matrix: 

\begin{lemma} [\bf Theorem 4.4.5 of \cite{vershynin2018high}, special case\footnote{The theorem is stated for sub-Gaussian entries in terms of a constant $C$; this constant is $C=2$ in the special case where the entries are $N(0,1)$ Gaussian.}] \label{lemma_concentration} 
Let $W \in \mathbb{R}^{d \times d}$ with i.i.d. $N(0,1)$ entries. Then 
\begin{equation*}
    \mathbb{P}(\|W\|_2 > 2\sqrt{d} +s) <  2e^{-s^2}
\end{equation*}
 for any $s>0$.
\end{lemma}
\noindent
Note that Lemma \ref{lemma_concentration} also applies (up to a constant factor) to complex Gaussian matrices $W_1 + \mathfrak{i}W_2$ where $W_1, W_2$ have i.i.d. real $N(0,1)$ entries, since $\|W_1 + \mathfrak{i}W_2\|_2 \leq \|W_1\|_2 + \|W_2\|_2$.

\begin{lemma}[\bf Weyl's Inequality \cite{bhatia2013matrix}]\label{lemma_weyl}
If $A,B \in \mathbb{C}^{d\times d}$ are two Hermitian matrices, and denoting the $i$'th-largest eigenvalue of any Hermitian matrix $M$ by $\sigma_i(M)$, we have
\begin{equation*}
 \sigma_i(A) + \sigma_d(B) \leq  \sigma_i(A + B) \leq  \sigma_i(A) + \sigma_1(B).
\end{equation*}
\end{lemma}

\begin{lemma}[\bf Spectral norm bound] \label{lemma_spectral_martingale_b}
 For some universal constant $C$, and every $T>0$, we have,
      $$ \mathbb{P}\left(\sup_{t \in [0,T]}\|B(t)\|_2 > \sqrt{T}(\sqrt{d} + \alpha)\right) \leq e^{-C \alpha^2} \qquad \forall \alpha >0.$$
\end{lemma}
\noindent
The proof of Lemma \ref{lemma_spectral_martingale_b} is standard and given in Appendix \ref{sec_proof_of_lemma_spectral_martingale_b}.

\subsection{Davis-Kahan Sin-Theta theorem}

The following lemma gives a deterministic bound on the change to the subspace spanned by the top-$k$ eigenvectors of a Hermitian matrix when it is perturbed by the addition of another Hermitian matrix.
Let A and $\hat{A}$ be two Hermitian matrices with eigenvalue decompositions
\begin{equation}\label{eq_eigenvalue_decomposition1}
    A = U \Lambda U^\ast = (U_1, U_2) \left({\begin{array}{cc}
   \Lambda_1 &  \\
    &  \Lambda_2 \\
  \end{array}}  \right) \left({\begin{array}{c}
   U_1^\ast \\
      U_2^\ast \\
  \end{array}}  \right)
\end{equation}

\begin{equation}\label{eq_eigenvalue_decomposition2}
    \hat{A} = \hat{U} \hat{\Lambda} \hat{U}^\ast = (\hat{U}_1, \hat{U}_2) \left({\begin{array}{cc}
   \hat{\Lambda}_1 &  \\
    &   \hat{\Lambda}_2 \\
  \end{array}}  \right) \left({\begin{array}{c}
    \hat{U}_1^\ast \\
       \hat{U}_2^\ast \\
  \end{array}}  \right).
\end{equation}

\begin{lemma}[\bf sin-$\Theta$ Theorem \cite{davis1970rotation}] \label{lemma_SinTheta}
Let $A, \hat{A}$ be two Hermitian matrices with eigenvalue decompositions given in \eqref{eq_eigenvalue_decomposition1} and \eqref{eq_eigenvalue_decomposition2}.
Suppose that there are $\alpha > \beta>0$ and $\Delta>0$ such that the spectrum of $\Lambda_1$ is contained in the interval $[\alpha, \beta]$ and the spectrum of $\hat{\Lambda}_2$ lies entirely outside of the interval $(\alpha -\Delta, \beta + \Delta)$.
Then
\newcommand{\vertiii}[1]{{\left\vert\kern-0.25ex\left\vert\kern-0.25ex\left\vert #1 
    \right\vert\kern-0.25ex\right\vert\kern-0.25ex\right\vert}}
\begin{equation*}
    \vertiii{U_1 U_1^\ast - \hat{U}_1 \hat{U}_1^\ast}  \leq \frac{\vertiii{\hat{A}- A}}{\Delta},
\end{equation*}
where $\vertiii{\cdot}$ denotes the operator or Frobenius norm (or, more generally, any unitarily invariant norm).
\end{lemma}

\subsection{Probability formulas} 

The following Proposition is well-known (see e.g. \cite{corwin2024lower}):
\begin{proposition}[\bf Layer-cake formula]\label{lemma_layer_cake}
Let  $p \geq 1$ and let $\zeta$ be a non-negative random variable.
Then
\begin{equation*}
\mathbb{E}[\zeta^p] = p\int_{0}^\infty s^{p-1}\mathbb{P}(\zeta > s) \mathrm{d}s.
\end{equation*}
\end{proposition}

\color{black}

\section{Overview of proofs}\label{sec_technical_overview}

We bound the Frobenius-distance utility for the covariance approximation problem $\|\hat{M}_k - M_k\|_F$,  
where $\hat{M}:= M + G + G^\ast$ and $G$ is a matrix of i.i.d. standard complex Gaussians (Theorem \ref{thm_rank_k_covariance_approximation_new}).
Here $\hat{V}_k$ and $V_k$ denote the matrices whose columns are the top-$k$ eigenvectors of $M$, $\hat{M}$ respectively.  
For simplicity, we assume $T=1$ in this section.

The privacy guarantee in Theorem \ref{thm_utility} follows directly from prior works on the (real) Gaussian mechanism (see Section \ref{sec_proof_utility_privacy} for details).

\subsection{Deterministic perturbation bounds}\label{sec_deterministic_perturbation_bounds}

Any bound on the utility $\|\hat{M}_k - M_k\|_F = \|\hat{V} \hat{\Sigma}_k \hat{V}^\ast - V \Sigma_k V^\ast\|_F$ must (at the very least) also bound the distance $\|\hat{V}_k \hat{V}_k^\ast - V_k V_k^\ast\|_F$ between the projection matrices onto the subspace $V_k$ and $\hat{V}_k$ spanned by the top-$k$ eigenvectors of $M$ and $\hat{M}$.
This is because 
$$\|\hat{V} \hat{\Sigma}_k \hat{V}^\ast - V \Sigma_k V^\ast\|_F \geq \Omega(\|\hat{V} \Sigma_k \hat{V}^\ast - V \Sigma_k V^\ast\|_F) \geq \Omega(\sigma_{k} \cdot \|\hat{V}_k \hat{V}_k^\ast - V_k V_k^\ast\|_F).$$
Thus, one approach to bounding $\|\hat{V} \hat{\Sigma}_k \hat{V}^\ast - V \Sigma_k V^\ast\|_F$ is to first apply deterministic perturbation bounds on $\|\hat{V}_k \hat{V}_k^\ast - V_k V_k^\ast\|_F$, such as those of the Davis-Kahan theorem \cite{davis1970rotation} restated here in Equation \eqref{eq_DK} (see also Lemma \ref{eq_t7} for a more general version of this theorem).
 Plugging in the high-probability bound  $\|E\|_2 = O(\sqrt{d})$ (e.g., from Lemma \ref{lemma_spectral_martingale_b}), and using the fact that $\|\hat{V}_k \hat{V}_k^\ast - V_k V_k^\ast\|_F \leq \sqrt{k} \|\hat{V}_k \hat{V}_k^\ast - V_k V_k^\ast\|_2$, gives $\|\hat{V}_k \hat{V}_k^\ast - V_k V_k^\ast\|_F \leq \frac{\sqrt{k} \sqrt{d}}{\sigma_k - \sigma_{k+1}}$ with high probability.

To obtain bounds for the utility $\|\hat{V} \hat{\Sigma}_k \hat{V}^\ast - V \Sigma_k V^\ast\|_F \leq  \|\hat{V} \Sigma_k \hat{V}^\ast - V \Sigma_k V^\ast\|_F + \|\hat{V} \hat{\Sigma}_k \hat{V}^\ast - \hat{V} \Sigma_k \hat{V}^\ast\|_F$ of the covariance matrix approximation, one can decompose 
\begin{equation}\label{eq_n78}
V \Sigma_k V^\ast = \sum_{i=1}^{k-1} (\sigma_i- \sigma_{i+1})V_i V_i^\ast + \sigma_k V_k V_k^\ast,
\end{equation}
and apply the Davis-Kahan theorem to each projection matrix $V_i V_i^\ast$ (see Appendix \ref{sec_challenges} for details, and additional discussion on deterministic approaches):
\begin{eqnarray}\label{eq_DK_1}
    & &\!\!\!\!\!\!\!\!\!\!\!\!\!\!\! \|\hat{V} \Sigma_k \hat{V}^\ast - V \Sigma_k V^\ast\|_F 
    \stackrel{\textrm{Eq. } \eqref{eq_n78}}{=} \left\|\sum_{i=1}^{k-1} (\sigma_i- \sigma_{i+1})\hat{V}_i \hat{V}_i^\ast + \sigma_k \hat{V}_k \hat{V}_k^\ast - \left(\sum_{i=1}^{k-1} (\sigma_i- \sigma_{i+1})V_i V_i^\ast + \sigma_k V_k V_k^\ast\right) \right\|_F \nonumber\\ 
    &=&   \left\|\sum_{i=1}^{k-1} (\sigma_i - \sigma_{i+1}) (\hat{V}_i \hat{V}_i^\ast - V_i V_i^\ast) + \sigma_k (\hat{V}_k \hat{V}_k^\ast - V_k V_k^\ast)\right\|_F \nonumber\\
    &\leq & \sum_{i=1}^{k-1} (\sigma_i - \sigma_{i+1}) \|\hat{V}_i \hat{V}_i^\ast - V_i V_i^\ast\|_F + \sigma_k \|\hat{V}_k \hat{V}_k^\ast - V_k V_k^\ast\|_F \nonumber\\
    & = & O\left(k^{1.5} \sqrt{d} + \frac{\sigma_k}{\sigma_k - \sigma_{k+1}} \sqrt{k} \sqrt{d}\right).
\end{eqnarray}
Unfortunately, when $E$ is a Hermitian Gaussian random matrix, this bound is not tight up to a factor of $k$.
Roughly, this is because, while the  Davis-Kahan theorem used to bound each term $\|\hat{V}_i \hat{V}_i^\ast - V_i V_i^\ast\|_F$ is tight for worst-case $E$, it is not tight when $E$ is a Gaussian random matrix.
Moreover, \eqref{eq_DK_1} bounds the Frobenius norm by adding up $k$ separate perturbation bounds, one for each projection matrix $V_i V_i^\ast$, while making worst-case assumptions on the cross-terms $\mathrm{tr} \left((\hat{V}_i \hat{V}_i^\ast - V_i V_i^\ast)   (\hat{V}_j \hat{V}_j^\ast - V_j V_j^\ast)\right)$ for $i \neq j$ which may not hold when $E$ is a random matrix.

\subsection{Bounding the utility of the Gaussian mechanism with Dyson Brownian motion}

As a first step to obtaining a tighter utility bound, we would ideally like to add up the Frobenius norm of the summands $(\sigma_i - \sigma_{i+1}) (\hat{V}_i \hat{V}_i^\ast - V_i V_i^\ast)$ in \eqref{eq_DK_1} as a sum-of-squares rather than as a simple sum, in order to decrease the r.h.s. by a factor of $\sqrt{k}$.
However, to do so we would need to bound the cross-terms $\mathrm{tr} \left((\hat{V}_i \hat{V}_i^\ast - V_i V_i^\ast)   (\hat{V}_j \hat{V}_j^\ast - V_j V_j^\ast)\right)$ for $i \neq j$.
To bound each of these cross-terms we need to carefully track the interactions between the eigenvectors in the subspaces $\mathcal{V}_i$ and $\mathcal{V}_j$ as the noise $E$ is added to the input matrix $M$.

We handle these interaction terms by viewing the addition of noise as a continuous-time Hermitian-matrix valued diffusion 
\begin{equation}\label{eq_n120}
\Phi(t) = M + B(t),
\end{equation}
whose eigenvalues $\gamma_i(t)$ and eigenvectors $u_i(t)$, $i\in [d]$, evolve over time.
Here, $B(t) := W(t) + W(t)^\ast$, where $W(t)$ is a $d \times d$ matrix where the real part (and complex part) of each entry is an independent standard Brownian motion with distribution $N(0, tI_d)$ at time $t$.
The key motivation for this approach is that, as the derivative $\mathrm{d}\Phi(t) = \mathrm{d}B(t)$ of this matrix-valued diffusion is independent of $\Phi(\tau)$ at all previous times $\tau \leq t$ (and thus independent of the eigenvalues  $\gamma_i(\tau)$ and eigenvectors $u_i(\tau)$ for $\tau \leq t$), it allows us to ``add up'' the infinitesimal perturbation to the utility at each time $t$ as an independent term without the need to handle higher-order (in $t$) interaction terms between the eigenvectors which would arise if one were to express these interaction terms using deterministic perturbation theory approach.

\subsubsection{Deriving an SDE for the utility of low-rank approximation} \label{sec_SDE_utility_overview}

We use the evolution equations \eqref{eq_DBM_eigenvectors} for the eigenvectors $u_i(t)$ to track the utility over time.
Let $\Phi(t) = U(t) \Gamma(t) U(t)^\ast$  be a spectral decomposition of the Hermitian matrix $\Phi(t)$ at every time $t$ where $\Gamma(t)$ is a diagonal matrix of eigenvalues at time $t$ and $U(t)$ a unitary matrix of eigenvectors.
We now define the rank-$k$ matrix $\Theta(t)$ to be the Hermitian matrix with fixed eigenvalues $\lambda_1 \geq \cdots  \geq\lambda_d$, where $\lambda_i = \gamma_i(0)$ for $i\leq k$ and  $\lambda_i =0$ for $i>k$,   and with eigenvectors $U(t)$:  
 $  \Theta(t):= U(t) \Lambda U(t)^\ast$ for all $t \in [0,T],$
 where $\Lambda :=\mathrm{diag}(\lambda_1, \ldots, \lambda_d)$.

$\Theta(t)$ is itself a Hermitian matrix-valued diffusion. 
  Decomposing $\Theta(t) = \sum_{i=1}^d \lambda_i u_i(t)u_i^\ast(t)$ allows us to use the SDEs for the eigenvalue  \eqref{eq_DBM_eigenvalues}  and eigenvector evolution \eqref{eq_DBM_eigenvectors}, together with Ito’s Lemma (Lemma \ref{lemma_ito_lemma_new}; the ``chain rule'' of stochastic calculus), to compute the Ito derivative for $\Theta(T)$,
 \begin{eqnarray} \label{eq_Ito_derivative_Theta}
        \mathrm{d}\Theta(t)    &=&    \frac{1}{2}\sum_{i=1}^{d} \sum_{j \neq i} \frac{\lambda_i - \lambda_j}{\gamma_i(t)-\gamma_j(t)}(u_i(t) u_j^\ast(t)\mathrm{d}B_{ij}(t) + u_j(t) u_i^\ast(t)\mathrm{d}B_{ij}^\ast(t)) \nonumber\\
      &+& \sum_{i=1}^{d} \sum_{j\neq i} \frac{\lambda_i - \lambda_j}{(\gamma_i(t)-\gamma_j(t))^2} u_i(t) u_i^\ast(t) \mathrm{d}t.
\end{eqnarray}

\subsubsection{Integrating the SDE to upper bound the utility}

We integrate the SDE \eqref{eq_Ito_derivative_Theta} to get an expression for the expected 
utility:
\begin{eqnarray} \label{eq_t7}
  \mathbb{E}\left[\left\|\Theta(T) -  \Theta(0)\right \|_F^2 \right] &=& \frac{1}{2} \mathbb{E}\left[\left\| \int_0^{T}\sum_{i=1}^{d} \sum_{j \neq i} (\lambda_i - \lambda_j) \frac{\mathrm{d}B_{ij}(t)}{\gamma_i(t)-\gamma_j(t)}(u_i(t) u_j^\ast(t) + u_j(t) u_i^\ast(t))\right\|_F^2 \right] \nonumber\\
    &  &  +  \quad\mathbb{E}\left[\left\| \int_0^{T}\sum_{i=1}^{d} \sum_{j\neq i} (\lambda_i - \lambda_j) \frac{\mathrm{d}t}{(\gamma_i(t) - \gamma_j(t))^2} u_i(t) u_i^\ast(t) \right\|_F^2 \right].
    \end{eqnarray}
    The idea is that roughly speaking, each differential term $\frac{\mathrm{d}B_{ij}(t)}{\gamma_i(t)-\gamma_j(t)}(u_i(t) u_j^\ast(t) + u_j(t) u_i^\ast(t))$ adds noise to the matrix independently of the other terms at every time $t$ since the stochastic derivatives of the Brownian motions, $\mathrm{d}B_{ij}(t)$, are independent for every $i,j,t$ and independent of the $u_i(s)$ for all current and past times $s \leq t$.
    This allows the contribution of each of these terms to the (squared) Frobenius norm of the first term on the r.h.s. to add up as a sum of squares.
     Integrating \eqref{eq_t7} via Ito's Lemma (restated in our preliminaries as Lemma \ref{lemma_ito_lemma_new}), 
    we obtain an expression for the utility as a sum of squares of the ratios of the eigenvalue gaps:
    \begin{equation}\label{eq_t7_2} 
\mathbb{E}\left[\left\|\Theta(T) -  \Theta(0)\right \|_F^2 \right] =   \sum_{i=1}^{d} \int_0^{T} \mathbb{E}\left[ \sum_{j \neq i}  \frac{(\lambda_i - \lambda_j)^2}{(\gamma_i(t)-\gamma_j(t))^2} \right]
     +    T \mathbb{E}\left[\left(\sum_{j\neq i} \frac{\lambda_i - \lambda_j}{(\gamma_i(t)-\gamma_j(t))^2}\right)^2   \right]\mathrm{d}t.
\end{equation}
 We note that, since $\Phi(t)$ is a complex-valued diffusion, we apply Ito's lemma separately to the real and imaginary parts of $\Phi(t)$ when deriving \eqref{eq_t7_2} (see Remark \ref{remark_Ito_complex} for details).

\begin{remark}[\bf Ito's lemma on complex-valued processes]\label{remark_Ito_complex}
When applied to complex differentiable (i.e., holomorphic) functions, some of the terms in Ito’s lemma vanish (see e.g. Theorem 2.2.9 in \cite{Complex_Ito_notes}, which gives the real-valued version of Ito’s lemma applied to complex functions with $\mathbb{C}$ identified as $\mathbb{R}^2$, and the cancelations which arise when it is applied to complex analytic functions that are provided in the discussions following that theorem).
However, the Frobenius norm utility function we bound is not {\em complex} differentiable.
This is because, by the Cauchy-Riemannn equations, any real-valued complex differentiable function must be everywhere constant.
For this reason, we integrate the stochastic process for the Frobenius norm utility by applying the real-valued version of Ito's lemma to its real and imaginary components.
\end{remark}

\subsubsection{Bounding the eigenvalue gaps with Weyl's inequality}

As a first attempt to bound the gap terms $\gamma_i(t) - \gamma_j(t)$ in \eqref{eq_t7_2} for all $i,j \leq k$, $i \neq j$, we use Weyl's inequality (restated here as Lemma \ref{lemma_weyl}), a deterministic eigenvalue perturbation bound which says that $\gamma_i(t) - \gamma_j(t) \geq \gamma_i(0) - \gamma_j(0) - \|B(t)\|_2$ for all $t$.
However, since $\|B(t)\|_2 = \Theta(\sqrt{d})$ w.h.p. for all $t\in [0,T]$, for Weyl's inequality to imply a non-trivial bound on $\gamma_i(t) - \gamma_j(t)$ for all $i,j \leq k$, $i \neq j$, we must require that all gaps in the top-$k$  eigenvalues of $M$ satisfy  $\gamma_i(0) - \gamma_{i+1}(0) = \sigma_i - \sigma_{i+1} \geq \Omega(\sqrt{d})$ for every $i\leq k$.
 Under this assumption, we obtain a bound of  $\gamma_i(t) - \gamma_j(t) \geq \Omega(\gamma_i(0) - \gamma_j(0))$ w.h.p. for every $i,j \leq k$, $i \neq j$ and $t\in [0,T]$.
 Plugging in this eigenvalue gap bound into \eqref{eq_t7_2} and simplifying, we would get that $$ \mathbb{E}\left[\left\|\hat{M}_k - M_k\right \|_F^2 \right] \approx \mathbb{E}\left[\left\|\Theta(T) -  \Theta(0)\right \|_F^2 \right] \leq \tilde{O}\left(\sqrt{k}\sqrt{d} \frac{\sigma_k}{\sigma_k-\sigma_{k+1}}\right)$$ under the assumption that  $\sigma_i - \sigma_{i+1} \geq \Omega(\sqrt{d})$ for every $i\leq k$.

\subsection{From initial eigengaps to  bounds on the eigengaps of Dyson Brownian motion}

To bound the utility of the Gaussian mechanism 
without any assumptions on the initial eigenvalue gaps $\sigma_i - \sigma_{i+1}$ for $i\neq k$, we would like to prove bounds on the gaps $\gamma_i(t) - \gamma_j(t)$ which hold even when initial gaps $\gamma_i(0) - \gamma_j(0) = \sigma_i- \sigma_{i+1}$ may not be $\Omega(\sqrt{d})$.
Unfortunately, since $\|B(t)\|_2 \geq \Omega(\sqrt{d})$ w.h.p. for $t = \Omega(1)$, we cannot rely on deterministic eigenvalue bounds such as Weyl's inequality, as this would not give any bound on $\gamma_i(t) - \gamma_{i+1}(t)$ unless $\sigma_i- \sigma_{i+1} \geq \Omega(\sqrt{d})$. 
To bypass this difficulty we would ideally like to obtain probabilistic lower bounds on the eigenvalue gaps $\gamma_i(t)- \gamma_{i+1}(t)$ which hold for any initial conditions on the top-$k$ eigengaps of $\gamma(0)$.

\subsubsection{Widening the eigengaps by adding complex Gaussian noise}
To see what bounds we might hope to show, note that if $\gamma(0) = 0$ then $\gamma(t)$ has the same joint distribution as the eigenvalues $\eta_1,\ldots, \eta_d$ of the rescaled GOE (GUE) matrix $\sqrt{t} (G+G^\ast)$ where $G$ is a matrix of i.i.d. real (complex) Gaussians.
This joint distribution is given by the following formula \cite{dyson1963random, ginibre1965statistical},
\begin{equation}\label{eq_joint_density}
    f(\eta_1,\ldots, \eta_d) = \frac{1}{R_\beta} \prod_{i<j} |\eta_i - \eta_j|^\beta e^{-\frac{1}{2} \sum_{i=1}^d \eta_i^2}, 
\end{equation}
where $R_\beta := \int \prod_{i<j} | \eta_i - \eta_j|^\beta e^{-\frac{1}{2} \sum_{i=1}^d \eta_i^2} \mathrm{d} \eta_1 \cdots \mathrm{d} \eta_d $ is a normalization constant.

From the repulsion factor $|\eta_i- \eta_{i+1}|^\beta$ in the joint distribution of the eigenvalues \eqref{eq_joint_density}, (and noting that the average eigenvalue gap of the standard GOE/GUE matrix $(G+G^\ast)$ is $\Theta\left(\frac{1}{\sqrt{d}}\right)$ w.h.p. since $\|G+G^\ast\|_2 = \Omega(\sqrt{d})$), roughly speaking one might expect that the GOE/GUE eigenvalue gaps satisfy $$  \mathbb{P}\left(\eta_i- \eta_{i+1} \leq \frac{s}{\sqrt{d}}\right) = O\left(\int_0^s z^\beta \mathrm{d}z\right) = O\left(s^{\beta +1}\right)$$ for all $s \geq 0$, where $\beta = 1$ in the real case and $\beta = 2$ in the complex case.
Assuming we can obtain such a bound, we would like to apply these bounds to bound the expectations of the terms on the r.h.s. of  \eqref{eq_t7_2}.  
The terms on the r.h.s. of  \eqref{eq_t7_2} with the smallest denominator, and therefore the most challenging to bound, are the terms $\mathbb{E}\left[\frac{(\lambda_i- \lambda_{i+1})^2}{(\gamma_i(t) - \gamma_{i+1}(t))^4}\right]$.
Assuming for the moment that we are able to show that
\begin{equation}\label{eq_conjectured_gap_bounds}
     \mathbb{P}\left(\gamma_i(t)- \gamma_{i+1}(t) \leq s\frac{\sqrt{t}}{\sqrt{d}}\right) \leq s^{\beta +1},
    \end{equation}
   then we would have the following bound for terms with denominators of order $r$:
\begin{equation}\label{eq_t8}
 \mathbb{E}\left[\frac{1}{(\gamma_i(t) - \gamma_{i+1}(t))^r}\right] 
 = \int_0^{\infty} \mathbb{P}\left(\gamma_i(t)- \gamma_{i+1}(t) \leq s^{-\frac{1}{r}}\right) \mathrm{d}s
 \leq  \left(\frac{d}{t}\right)^{\frac{r}{2}}\int_0^\infty s^{-\frac{1}{r}{(\beta+1)}}\mathrm{d}s.
\end{equation}
For the terms of order $r=2$, the r.h.s. of \eqref{eq_t8} is  $\int_0^\infty s^{-\frac{1}{2}{(\beta+1)}}\mathrm{d}s= \infty$ in the real case where $\beta=1$.
To bypass this problem, we observe that when the Gaussian noise is complex the integral on the r.h.s. of \eqref{eq_t8} becomes  $\int_0^\infty s^{-\frac{1}{2}{(\beta+1)}}\mathrm{d}s= O(1)$ since $\beta=2$ in the complex case.
Thus, while in the real case, one expects the gaps to be small enough that their inverse second moment $\mathbb{E}\left[\frac{1}{(\gamma_i(t) - \gamma_j(t))^2}\right]$ is infinite, in the complex case the repulsion between eigenvalues allows the gaps to be large enough that the inverse second moment is finite.
This motivates replacing the real Gaussian perturbation in the Gaussian mechanism with Complex-valued Gaussian noise (Algorithm \ref{alg_quaternion_Gaussian}).

\subsubsection{An SDE for a rank-$k$ matrix diffusion with dynamically changing eigenvalues to track the utility under small initial eigengaps.}
Unfortunately, for the highest-order terms, of order $r=4$, the r.h.s. of \eqref{eq_t8} is  $\int_0^\infty s^{-\frac{1}{4}{(\beta+1)}}\mathrm{d}s= \infty$ even in the complex case where $\beta = 2$.
To get around this problem we replace the fixed eigenvalues  $\lambda_i = \sigma_i$ for $i \leq k$, of the rank-$k$ stochastic process $\Theta(t)$, with eigenvalues $\lambda_i(t)$ which change dynamically over time where at each time $t \geq 0$ we set $\lambda_i(t) = \gamma_i(t)$ for $i \leq k$ and $\lambda_i(t) = 0$ for $i>k$, in the hope that this will lead to cancellations in the highest-order terms.
This gives us a new rank-$k$ stochastic process $\Psi(t):= U(t) \Lambda(t) U(t)^\ast$ with dynamically changing eigenvalues $\Lambda(t) :=\mathrm{diag}(\lambda_1(t), \ldots, \lambda_d(t))$.
Since $\Psi(T) = \hat{M}_k$ and $\Psi(0) = M_k$, our goal is to bound $\|\hat{M}_k - M_k \|_F = \|\Psi(T) - \Psi(0)\|_F$.
Roughly speaking, this would lead to cancellations in the terms on the r.h.s. of \eqref{eq_t7_2} at every time $t \geq 0$:  the second-order terms would be reduced to constant terms $$  \frac{(\lambda_i(t) - \lambda_j(t))^2}{(\gamma_i(t)-\gamma_j(t))^2} = \frac{(\gamma_i(t) - \gamma_j(t))^2}{(\gamma_i(t)-\gamma_j(t))^2} = 1$$ for $i\neq j$ $i,j \leq k$,
and fourth-order terms would be reduced to second-order terms, e.g.,  
 $$  \frac{(\lambda_i(t)- \lambda_{i+1}(t))^2}{(\gamma_i(t) - \gamma_{i+1}(t))^4} = \frac{(\gamma_i(t)- \gamma_{i+1}(t))^2}{(\gamma_i(t) - \gamma_{i+1}(t))^4} = \frac{1}{(\gamma_i(t)- \gamma_{i+1}(t))^2}$$ for $i< k$.
This would allow us to obtain a finite bound for the expectation on the r.h.s. of \eqref{eq_t7_2}.

Towards this end, we first use the equations for the evolution of the eigenvalues \eqref{eq_DBM_eigenvalues} and eigenvectors \eqref{eq_DBM_eigenvectors} of Dyson Brownian motion to derive an SDE for our new rank-$k$ process $\Psi(t)$ (Lemma \ref{Lemma_projection_differntial} and \eqref{eq_ito_derivative}):
\begin{equation}\label{eq_t9}
  \mathrm{d}\Psi(t)  = \sum_{i=1}^d \lambda_i(t) \mathrm{d}(u_i(t) u_i^\ast(t))) + (\mathrm{d}\lambda_i(t)) (u_i(t) u_i^\ast(t)) + \mathrm{d}\lambda_i(t) \mathrm{d}( u_i(t) u_i^\ast(t)),
\end{equation}
where  
$$  \mathrm{d}(u_i(t) u_i^\ast(t))
     =  \sum_{j \neq i} \frac{u_i(t) u_j^\ast(t)\mathrm{d}B_{ij}(t)  + u_j(t) u_i^\ast(t)\mathrm{d}B_{ij}^\ast(t)}{\gamma_i(t) - \gamma_j(t)}
    - \frac{(u_i(t) u_i^\ast(t) - u_j(t)u_j^\ast(t))\mathrm{d}t}{(\gamma_i(t)- \gamma_j(t))^2},$$ and where $\mathrm{d}\lambda_i(t)$ is given by \eqref{eq_DBM_eigenvalues}.
    The last term  $\mathrm{d}\lambda_i(t) \mathrm{d}( u_i(t) u_i^\ast(t))$ in \eqref{eq_t9} vanishes as it consists only of higher-order differential terms.
Applying It\^o's lemma to compute the integral $\left\|\int_0^T \mathrm{d} \Psi(t)\right\|_F^2$ for the change in the (squared) Frobenius  distance, we get (Lemma \ref{Lemma_integral} and \eqref{eq_ito_integral_1} in the Proof of Theorem \ref{thm_rank_k_covariance_approximation_new}),
\begin{eqnarray}
& & \!\!\!\!\!\!\!\!\!\!\!\!\!\!\!\!\!\!\!\!\!\!\!\!\!\!\!\!\!\!\!\!\!\!\!\! \mathbb{E}[\| \Psi(T) - \Psi(0)\|_F^2] = \left\|\int_0^T \mathrm{d} \Psi(t)\right\|_F^2 \nonumber \\
& \leq & \int_{0}^{T}   \mathbb{E}\left[ \sum_{i=1}^{d}  \sum_{j \neq i}  \frac{(\lambda_i(t) - \lambda_j(t))^2}{(\gamma_i(t) - \gamma_j(t))^2} \mathrm{d}t \right] +   T \int_{0}^{T}\mathbb{E}\left[\sum_{i=1}^{d}\left(\sum_{j\neq i} \frac{\lambda_i(t) - \lambda_j(t)}{(\gamma_i(t) - \gamma_j(t))^2}\right)^2   \right]\mathrm{d}t \nonumber\\ 
& & + \quad \int_{0}^T \sum_{i=1}^k\mathbb{E}\left[  \left(\sum_{j \neq i} \frac{1}{\gamma_i(t) - \gamma_j(t)}\right)^2 \right] \mathrm{d}t.  
     \end{eqnarray}
Plugging in our choice of $\lambda_i(t)$, we get (Equations \eqref{eq_u1} and \eqref{eq_a4} in the proof of Theorem \ref{thm_rank_k_covariance_approximation_new}),
\begin{eqnarray}\label{eq_t9_2}
 \left\|\int_0^T \mathrm{d} \Psi(t)\right\|_F^2 &\leq &  \sum_{i=1}^{k}\int_{0}^{T}   \mathbb{E}\left[ \left( k +  \sum_{j >k}  \frac{(\gamma_i(t))^2}{(\gamma_i(t)- \gamma_j(t))^2} \right)\right]\nonumber \\
     &+&    T \mathbb{E}\bigg[\left(\sum_{j \neq i: j\leq k} \frac{1}{\gamma_i(t)- \gamma_j(t)}\right)^2   + \bigg(\sum_{j> k} \frac{\gamma_i(t)}{(\gamma_i(t)- \gamma_j(t))^2}\bigg)^2 \bigg] \nonumber \\
   & &  + \quad \mathbb{E}\left[  \left(\sum_{j \neq i} \frac{1}{\gamma_i(t) - \gamma_j(t)}\right)^2 \right] \mathrm{d}t.
    \end{eqnarray}
\noindent If we can prove the conjectured  gap bounds \eqref{eq_conjectured_gap_bounds}, we will have from \eqref{eq_t8} that $$ \mathbb{E}\left[\frac{1}{(\gamma_i(t) - \gamma_j(t))^2}\right] \leq \frac{d}{t (i-j)^2}$$ for all $i\neq j$ and, more generally, that  
$$ \mathbb{E}\left[\frac{1}{(\gamma_i(t) - \gamma_j(t))(\gamma_\ell(t) - \gamma_r(t))}\right] \leq \frac{d}{t \min((i-j)^2, (\ell-r)^2)}$$ for all $i\neq j$, $\ell\neq r$.
Moreover, if we assume a bound only on the $k$'th eigenvalue gap of $M$, $\sigma_k - \sigma_{k+1} \geq \Omega(\sqrt{d})$  (without assuming any bounds on the other eigenvalue gaps of $M$), we have by Weyl's inequality that $\gamma_k(t) - \gamma_{k+1}(t) \geq \sigma_k- \sigma_{k+1} - \|B(t)\|_2 = \Omega(\sigma_i- \sigma_{i+1})$.
Plugging these conjectured probabilistic bounds, together with the worst-case Weyl inequality bounds for the $k$'th gap $\gamma_k(t) - \gamma_{k+1}(t) \geq \Omega(\sigma_i- \sigma_{i+1})$, into \eqref{eq_t9_2} gives (Equation \eqref{eq_u3} in the proof of Theorem \ref{thm_rank_k_covariance_approximation_new}),
$$\mathbb{E}\left[\left\|\hat{M}_k -  M_k\right \|_F^2\right]= \left\|\int_0^T \mathrm{d} \Psi(t)\right\|_F^2 \leq \tilde{O}\left(kd \frac{\sigma_k^2}{(\sigma_k - \sigma_{k+1})^2}\right).$$

\subsection{Bounding the eigenvalue gaps of Dyson Brownian motion}\label{sec_technical_gaps}

To complete the proof of Theorem \ref{thm_rank_k_covariance_approximation_new},  we still need to show the conjectured bounds in \eqref{eq_conjectured_gap_bounds} (or at least show a close approximation to these bounds).
We do this by proving Lemmas \ref{lemma_gap_comparison} and \ref{lemma_GUE_gaps}, and present an overview of their proofs in this section.
 We start by recalling a few important ideas and results from random matrix theory.

\subsubsection{Useful ideas from random matrix theory} Starting with \cite{dyson1963random, ginibre1965statistical}, many works have made use of the intuition that the eigenvalues of a random matrix tend to repel each other, and can be interpreted in the context of statistical mechanics as a many-body system of charged particles undergoing a Brownian motion.
These particles repel each other with an ``electrical force'' arising from a potential that decays logarithmically with the distance between pairs of particles (see e.g. \cite{rodriguez2014calibration}). 
The dynamics of these particles are described by the eigenvalue evolution equations \eqref{eq_DBM_eigenvalues} discovered by \cite{dyson1963random}, where the diffusion term $\mathrm{d} \gamma_i(t) = \mathrm{d}B_{i i}(t)$ describes the random component of each particle's motion and the terms $\frac{\beta}{\gamma_i(t) - \gamma_j(t)}$ describe the repulsion between particles; the parameter $\beta$ can be interpreted either as the strength of the electrical force, or equivalently, as the (inverse) temperature of the system.

\cite{dyson1963random} showed that from the evolution equations \eqref{eq_DBM_eigenvalues} one can obtain the joint distribution of the eigenvalues of Dyson Brownian motion at equilibrium \eqref{eq_joint_density}.
If one initializes the matrix Brownian motion with all eigenvalues at $0$, at every time $t$ the matrix Brownian motion is in equilibrium (after scaling by $\frac{1}{\sqrt{t}}$) and equal in distribution to a GOE or GUE matrix scaled by $\sqrt{t}$, and thus \eqref{eq_joint_density} gives an explicit formula for the joint distribution of the eigenvalues of the GOE random matrix (for the real case $\beta = 1$) and GUE random matrix (for the complex case $\beta = 2$).

In the limit as $\beta\rightarrow \infty$ (with appropriate rescaling), the temperature of the system can be thought of as going to zero, and the solution to the evolution equations \eqref{eq_DBM_eigenvalues} converges to a deterministic solution with particles ``frozen'' at  $\gamma_i(t) = \sqrt{t}\omega_i$ with probability 1 for some $\omega_1,\ldots, \omega_d \in \mathbb{R}$.
It has long been observed \cite{ginibre1965statistical, girko1985circular} that the gaps $\omega_i - \omega_{i+1}$ between these particles is, roughly, $\frac{1}{\sqrt{d}}$ in the ``bulk'' of the spectrum (i.e., the set of eigenvalues with index $cd <i< d- cd$ for any small constant $c$), while the particles have larger gaps near the edge of the spectrum.

More recently, \cite{erdHos2012rigidity} showed (restated here as Lemma \ref{lemma_rigidity}) that with high probability, the eigenvalues $\eta$ of the GOE/GUE random matrices are ``rigid'' in the sense that each eigenvalue $\eta_i$ falls within a small distance $\tilde{O}(\min(i, d-i+1)^{-\frac{1}{3}} d^{-\frac{1}{6}})$ of the zero-temperature eigenvalue $\omega_i$, where $\min(i, d-i+1)^{-\frac{1}{3}} d^{-\frac{1}{6}}$ is the average eigengap size in the region of the spectrum
  near $\omega_i$:
\begin{equation}\label{eq_ridgidity}
 |\eta_i - \omega_i| \leq O( \min(i, d-i+1)^{-\frac{1}{3}} d^{-\frac{1}{6}} \log(d)^{\log \log d}), \qquad \forall i \in [d], \quad \textrm{ w.h.p.}
\end{equation}

%
\subsubsection{Our results on eigenvalue gaps of Dyson Brownian motion (Overview of the proof of Theorem \ref{thm:eigenvalue_gap})} \label{overview_eigenvalue_gaps}

\paragraph{Reducing the problem of bounding the eigenvalue gaps from any initial condition to the zero initial condition.}

To bound the eigenvalue gaps of Dyson Brownian motion from any initial condition $\gamma(0)$ (Theorem \ref{thm:eigenvalue_gap}), we would like to make use of the closed-form expression \eqref{eq_joint_density}, which gives the joint density for the eigenvalues of Dyson Brownian motion initialized at $\gamma(0)=0$.
Unfortunately, in the real case, to the best of our knowledge, we are not aware of a closed-form expression for the joint density of the eigenvalues of Dyson Brownian motion for general initial conditions $\gamma(0)$.  
Moreover, while a joint density formula (Proposition 1.1 of \cite{johansson2001universality}) is available in the complex case for general initial conditions $\gamma(0)$, this formula is more difficult to work with as it includes additional determinantal terms not present in the joint density formula \eqref{eq_joint_density}  for the special case when $\gamma(0) = 0$.

To overcome these difficulties, we first show, in the following lemma, that one can reduce the task of bounding the eigenvalue gaps of Dyson Brownian motion from any initial condition, to the problem of bounding the gaps of a Dyson Brownian motion initialized at the $0$ vector.
\begin{lemma}[\bf Eigenvalue-gap comparison Lemma]\label{lemma_gap_comparison}
\quad Let $\beta \geq 1$, and let $\xi(t) = (\xi_1(t), \ldots, \xi_d(t))$ \, \, and \, \, $\gamma(t) =$ \\ $(\gamma_1(t), \ldots, \gamma_d(t))$ be two solutions of \eqref{eq_DBM_eigenvalues} (with parameter $\beta$) coupled to the same underlying Brownian motion $B(t)$, starting respectively from initial conditions $\xi(0), \gamma(0)$. 
Assume that $\xi_i(0) - \xi_{i+1}(0)  \leq \gamma_i(0) - \gamma_{i+1}(0)$ for all $1\leq i < d$.
Then, with probability $1$, $\xi_i(t) - \xi_{i+1}(t)  \leq \gamma_i(t) - \gamma_{i+1}(t)$ for all  $t>0$ and all $1\leq i < d$.
\end{lemma}
\noindent
 We give an overview of the proof of Lemma \ref{lemma_gap_comparison} below; the full proof appears in Section \ref{sec_gap_comparison_proof}.
Note that \cite{anderson2010introduction} show a different eigenvalue comparison theorem (their Lemma 4.3.6) which says that if $\xi$ and $\gamma$ are two coupled Dyson Brownian motions with initial conditions satisfying $\xi_i(0) \leq \gamma_i(0)$ for all $i \in [d]$, then with probability $1$,  $\xi_i(t) \leq \gamma_i(t)$ at every $t \geq 0$.  However, this does not imply the gaps of $\gamma_i(t)$ are at least as large as the corresponding gaps of $\xi_i(t)$ since we could have that $\gamma_i(t) - \gamma_{i+1}(t) < \xi_i(t)- \xi_{i+1}(t)$ even if $\xi_i(t) \leq \gamma_i(t)$ for all $i$; see also \cite{erdHos2011universality, landon2017convergence,lee2016bulk} for results about the eigenvalues of Dyson Brownian motion and their gaps from non-zero initial conditions.

To prove Lemma \ref{lemma_gap_comparison}, we must show that whenever the initial gaps of $\gamma(0)$ are greater than or equal to the corresponding initial gaps of $\xi(0)$, $\gamma_i(0) - \gamma_{i+1}(0) \geq \xi_i(0) - \xi_{i+1}(0)$,  with probability $1$ the gaps of $\gamma(t)$ remain greater than or equal to the gaps of the coupled process $\xi(t)$ at every time $t \geq 0$.
 The idea behind the proof of Lemma \ref{lemma_gap_comparison} is to consider the net ``electrostatic pressure'' on each gap $\gamma_i(t) - \gamma_{i+1}(t)$-- that is, the difference between the sum of the forces from the eigenvalues $\gamma_j(t)$ for $j \notin\{i, i+1\}$ pushing on the gap $\gamma_i(t) - \gamma_{i+1}(t)$ from the outside to compress it, and the force from the repulsion between the eigenvalues $\gamma_i(t)$ and $\gamma_{i+1}(t)$ pushing to expand the gap.
More formally, this net pressure is $\mathrm{d} \left(\gamma_i(t) - \gamma_{i+1}(t)\right) = \mathrm{d}\gamma_i(t) -\mathrm{d} \gamma_{i+1}(t)$ and, thus, we can compute it using \eqref{eq_DBM_eigenvalues}:
\begin{eqnarray} \label{eq_pressure}
     \mathrm{d} \gamma_i(t) - \mathrm{d} \gamma_{i+1} (t) =  \mathrm{d}B_{i, i}(t) +   \sum_{j \neq i} \frac{\beta \mathrm{d}t}{\gamma_i(t) - \gamma_j(t)}   -  \left(\mathrm{d}B_{i+1, i+1}(t) + \! \!  \sum_{j \neq i+1} \frac{\beta \mathrm{d}t}{\gamma_{i+1}(t) - \gamma_j(t)}  \right).
\end{eqnarray}
Ideally, we would like to show that at any time where all the gaps of $\gamma(t)$ are at least as large as all the gaps of $\xi(t)$, 
we have $\mathrm{d} \gamma_i(t) - \mathrm{d} \gamma_{i+1}(t)\geq \mathrm{d} \xi_i(t) - \mathrm{d} \xi_{i+1}(t)$.
This in turn {\em would} imply that the gaps of $\gamma(t)$ expand faster (or contract slower) than the corresponding gaps of $\xi(t)$, and hence that the gaps of  $\gamma(t)$ remain larger than those of $\xi(t)$ at every time $t \geq 0$.
Unfortunately, the opposite may be true: if the eigenvalue gap $\gamma_i(t) - \gamma_{i+1}(t)$ is much larger than the gap $\xi_i(t) - \xi_{i+1}(t)$ then the repulsion between $\gamma_i(t)$ and $\gamma_{i+1}(t)$  pushing to expand the gap $\gamma_i(t) - \gamma_{i+1}(t)$ is much smaller than the repulsion pushing to expand the gap $\xi_i(t) - \xi_{i+1}(t)$.

To solve this problem, we prove Lemma  \ref{lemma_gap_comparison} by a contradiction argument.
Towards this end, we first define 
$\tau := \inf\{t\geq 0: \xi_i(t) - \xi_{i+1}(t) > \gamma_i(t) - \gamma_{i+1}(t) \textrm{ for some } i \in [d]  \}$  to be the first time where for some $i$,  the size of the $i$'th gap  $\xi_i(t) - \xi_{i+1}(t)$ becomes larger than the $i$'th gap  $\gamma_i(t) - \gamma_{i+1}(t)$ of $\gamma(t)$.
We assume (falsely), that $\tau < \infty$ and show that this leads to a contradiction.

Since the initial gaps of $\gamma(0)$ are at least as large as those of $\xi(0)$, and since the trajectories $\gamma(t)$ and $\xi(t)$ are continuous w.p.\ $1$, by the intermediate value theorem there must be an $i\in[d]$ such that 
\begin{equation}\label{eq_t10}
     \gamma_i(\tau) - \gamma_{i+1}(\tau) = \xi_i(\tau) - \xi_{i+1}(\tau),
\end{equation}
and the other gaps at time $\tau$ satisfy
   $\gamma_j(\tau) - \gamma_{j+1}(\tau) \geq \xi_j(\tau) - \xi_{j+1}(\tau)$ for $j \in [d]$. 
Plugging \eqref{eq_t10} into \eqref{eq_pressure}, we
obtain the difference in net electrostatic pressure on the $i$'th gap of $\gamma$ and $\xi$ at time $\tau$:
\begin{eqnarray}
      (\mathrm{d} \gamma_i(\tau) - \mathrm{d} \gamma_{i+1} (\tau)) -   (\mathrm{d} \xi_i(\tau) - \mathrm{d} \xi_{i+1} (\tau)) 
     =  \sum_{j \neq i, i+1} \frac{\beta \mathrm{d}\tau}{\gamma_i(\tau) - \gamma_j(\tau)}  -  \frac{\beta  \mathrm{d}\tau  }{\xi_i(\tau) - \xi_j(\tau)}
     \geq 0. \label{eq_t12}
\end{eqnarray}
The Brownian motion terms $\mathrm{d}B$ from \eqref{eq_pressure} 
cancel as we have coupled the processes $\gamma$ and $\xi$ by setting their underlying Brownian motions $B(t)$ to be equal. 
The terms $\frac{1}{\gamma_i(\tau) - \gamma_{i+1}(\tau)}$ and  $\frac{1}{\xi_i(\tau) - \xi_{i+1}(\tau)}$ arising from  \eqref{eq_pressure} which describe repulsion between the $i$'th and $i+1$'th eigenvalues cancel 
by \eqref{eq_t10}.
Thus, we are only left with the forces from the other eigenvalues pushing to compress the $i$'th gap of $\xi(\tau)$ and $\gamma(\tau)$ from the outside, which 
 allows us to then show that since 
the gaps of $\gamma(\tau)$ are at least as large as the corresponding gaps of $\xi(\tau)$ at time $\tau$, the r.h.s. of \eqref{eq_t12} is greater than or equal to $0$ (Proposition \ref{sec_gap_comparison_proof}).

Next, we would like to show that \eqref{eq_t12} implies that the $i$'th gap of $\xi$ does {\em not} become larger than the $i$'th gap of $\gamma$ at time $\tau$, leading to a contradiction.
Unfortunately, \eqref{eq_t12} is not sufficient to show this, since, if $(\mathrm{d} \gamma_i(\tau) - \mathrm{d} \gamma_{i+1} (\tau)) -   (\mathrm{d} \xi_i(\tau) - \mathrm{d} \xi_{i+1} (\tau)) = 0$ we might have that the {\em second} derivative of the gaps of $\gamma$,  $(\mathrm{d}^2 \gamma_i(\tau) - \mathrm{d}^2 \gamma_{i+1} (\tau))$ is strictly smaller than the second derivative of the gaps of $\xi$,  $(\mathrm{d}^2 \xi_i(\tau) - \mathrm{d}^2 \xi_{i+1} (\tau))$. 
 To overcome this problem, we observe that, since 
 the gaps of $\gamma(\tau)$ are at least the corresponding gaps of $\xi(\tau)$ at time $\tau$, the only way the r.h.s. of \eqref{eq_t12} could be $0$ is if all the gaps of $\gamma(\tau)$ are equal to the corresponding gaps of $\xi(\tau)$. 
 In this case, the gaps would be equal at {\em every} time $t$ since solutions of Dyson Brownian motion are unique w.r.t. the underlying Brownian motion which defines our coupling  (see e.g. \cite{anderson2010introduction}, restated as Lemma \ref{lemma_strong}).
Thus, without loss of generality, we may assume that there is at least one $j$ such that the $j$'th gap of $\gamma(\tau)$ is strictly greater than the $j$'th gap of $\xi(\tau)$.
This in turn implies the r.h.s. of \eqref{eq_t12} is {\em strictly} greater than $0$, and hence the $i$'th gap of $\gamma$ becomes strictly {\em larger} than the $i$'th gap of $\xi$ in an open neighborhood of the time $\tau$.
This contradicts the definition of 
$\tau$, and hence by contradiction, we have $\tau = \infty$, and therefore the gaps of $\gamma(t)$ are greater than or equal to the corresponding gaps of $\xi(t)$ at every time $t \geq 0$.

\paragraph{Bounding the eigenvalue gaps of the GUE/GOE random matrix.}

Roughly speaking, to complete the proof of Theorem \ref{thm:eigenvalue_gap} we must show the conjectured lower bound of $\mathbb{P}\left(\eta_i - \eta_{i+1} \leq \frac{1}{\sqrt{d}} s\right)$ $\leq s^{\beta +1}$ for any $i$ and $s \geq 0$,  when $\eta_1,\ldots, \eta_d$ are the eigenvalues of the GOE/GUE random matrix.
The proof for the complex Hermitian GUE  $(\beta =2)$ case and the real symmetric GOE $(\beta = 1)$ case are nearly identical.
We first show how to complete the proof for the complex case, then show how to modify the proof for the real case.

As a first approach, we would ideally like to integrate the formula for the joint eigenvalue density $f(\eta)$  \eqref{eq_joint_density} over the set 
\begin{equation}\label{eq_n182}
A(s):= \left\{\eta \in \mathcal{W}_d: \eta_i - \eta_{i+1} \leq \frac{1}{\sqrt{d}} s \right\},
\end{equation}
 where $\mathcal{W}_d$ was defined in \eqref{eq:WeylChamber}.
 This gives
\begin{eqnarray}\label{eq_t14}
   \mathbb{P}\left(\eta_i - \eta_{i+1} \leq \frac{1}{\sqrt{d}} s\right) =  \int_{A(s)}   f(\eta) \mathrm{d} \eta   \stackrel{\textrm{Eq. } \eqref{eq_joint_density}}{=} \frac{1}{R_2} \int_{A(s)}   \prod_{\ell<j} | \eta_{\ell} - \eta_j|^2 e^{-\frac{1}{2} \sum_{\ell=1}^d \eta_{\ell}^2} \mathrm{d} \eta.
\end{eqnarray}
Unfortunately, we do not know of a closed-form expression for the $d$-dimensional integral \eqref{eq_t14}.

To get around this problem, suppose that we can somehow find a map $\phi: \mathcal{W}_d \rightarrow \mathcal{W}_d$ such that the following holds for every $\eta \in A(s)$, 
\begin{itemize}
\item     the term $|\eta[i] - \eta[i+1]|$ in the formula \eqref{eq_joint_density} for the joint eigenvalue density $f(\eta)$ satisfies
\begin{equation}\label{eq_n122}
|\phi(\eta)[i] - \phi(\eta)[i+1]| \geq \frac{1}{s}|\eta_i - \eta_{i+1}|
\end{equation}
\item    all other terms in the formula for $f(\eta)$ remain unchanged when applying $\phi$ to $\eta$, that is,  
\begin{equation}\label{eq_n123}
|\phi(\eta_j) - \phi(\eta_\ell)| = |\eta_j - \eta_\ell| \qquad \qquad \forall  (j, \ell) \neq (i, i+1),
\end{equation}
and
\begin{equation}\label{eq_n124}
 e^{-\frac{1}{2} \sum_{\ell=1}^d \phi(\eta_{\ell})^2} = e^{-\frac{1}{2} \sum_{\ell=1}^d \eta_{\ell}^2}.
\end{equation}
\end{itemize}
If we can construct a function $\phi$ satsfying \eqref{eq_n122}, \eqref{eq_n123}, and \eqref{eq_n124} then by \eqref{eq_joint_density}   for every $\eta \in A(s)$ we would have that 
\begin{eqnarray}\label{eq_n126}
f(\phi(\eta)) &\stackrel{\textrm{Eq. } \eqref{eq_joint_density}}{=}&  \frac{1}{R_2}\prod_{\ell<j} | \phi(\eta_{\ell}) - \phi(\eta_j)|^2 e^{-\frac{1}{2} \sum_{\ell=1}^d \phi(\eta_{\ell})^2} \nonumber\\
&\stackrel{\textrm{Eq. } \eqref{eq_n122}, \eqref{eq_n123}, \eqref{eq_n124}}{\geq}&  \frac{1}{s^2}   \frac{1}{R_2} \prod_{\ell<j} | \eta_{\ell} - \eta_j|^2 e^{-\frac{1}{2} \sum_{\ell=1}^d \eta_{\ell}^2}  \nonumber\\
&\stackrel{\textrm{Eq. } \eqref{eq_joint_density}}{=}& \frac{1}{s^2}f(\eta).
\end{eqnarray}

\noindent
Moreover, roughly speaking, one might hope that, since $\phi$ expands one of the eigenvalue gaps by $\frac{1}{s}$  (Inequality \eqref{eq_n122}) and leaves all the other gaps unchanged (Equation \eqref{eq_n123}), the map $\phi$ would be invertible and the Jacobian determinant of such a map would satisfy 
\begin{equation}\label{eq_n125}
\mathrm{det}(J_\phi(\eta)) \geq \frac{1}{s}
\end{equation}
 for all $\eta \in A(s)$.
This in turn would imply that the r.h.s. of \eqref{eq_t14} would satisfy
\begin{eqnarray}\label{eq_t15}
   \mathbb{P}\left(\eta_i - \eta_{i+1} \leq \frac{1}{\sqrt{d}} s\right) \!
  \stackrel{\textrm{Eq.}\, \eqref{eq_t14}}{=}  \! s^3 \int_{A(s)}   f(\eta)  \frac{1}{s^3} \mathrm{d} \eta 
 \stackrel{\textrm{Eq.}\, \eqref{eq_n126}, \eqref{eq_n125}}{ \leq} \!\! s^3 \int_{A(s)}  f(\eta)   \frac{f(\phi(\eta))}{f(\eta)}  \mathrm{det}(J_\phi(\eta))  \mathrm{d} \eta
  \leq s^3.
  \end{eqnarray}
The last step holds since $\phi$ is injective and $f$ is a probability density, implying the integral is at most $1$.
\color{black}

Unfortunately, one can easily see that there does not exist a map $\phi$ which expands the $i$'th gap term $|\eta_i-\eta_{i+1}|$ in the joint eigenvalue density  \eqref{eq_joint_density} by $\frac{1}{s}$ (condition \eqref{eq_n122}), but leaves all other terms unchanged (conditions \eqref{eq_n123} and \eqref{eq_n124}).
This is because, to expand $|\eta_i-\eta_{i+1}|$ but leave the other gap terms unchanged, one would, e.g., have to translate the other eigenvalues  $\eta_j$ for $j \leq i$ aside by an amount  $(\frac{1}{s}-1)|\eta_i-\eta_{i+1}|$.
To circumvent this problem, we instead consider a different map $\phi : \mathcal{W}_d \rightarrow \mathcal{W}_d$ which, roughly speaking, expands the eigenvalue gap $\eta_i-\eta_{i+1}$ by a factor of $\frac{1}{s}$, leaves all other gaps unchanged, and translates the eigenvalues of $\eta_j$ for $j \leq i$ to the left by an amount $\frac{1}{s}(\eta_i-\eta_{i+1})$ to make room for the expanded eigenvalue gap (see equations \eqref{eq_phi1}-\eqref{eq_phi3} for the full definition of $\phi$).
Since  $\eta_i = \Theta(\sqrt{d})$ w.h.p., when e.g. $|\eta_i-\eta_{i+1}| \geq \Theta(\frac{1}{\sqrt{d}})$ this would decrease the exponential term in the joint density by a factor of 
$$ e^{-\frac{1}{2} \sum_{j=1}^i (\phi(\eta)[i] - \eta_i)^2} \approx e^{-\frac{1}{2} \sum_{j=1}^i (\phi(\eta)[i] - \eta_i) \sqrt{d} } \geq  e^{-\frac{1}{2} \sum_{j=1}^i \frac{\sqrt{d}}{\sqrt{d}}} = e^i.$$
For $i = O(1)$, this is not an issue as then one has $e^i = O(1)$ and, hence,   $$ \frac{f(\phi(\eta))}{f(\eta)} \geq \Omega\left(\frac{1}{s^2}\right)$$ (see Lemma \ref{lemma_density_ratio_edge}).
Roughly speaking this fact, together with a bound on the Jacobian determinant of $\phi$ (Lemma \ref{prop_Jacobian_phi} which says  $\mathrm{det}(J_\phi(\eta)) \geq \frac{1}{s}$) and since $\phi$ is injective (Proposition \ref{prop_map_phi}), allows us to use the above map $\phi$ to show that \eqref{eq_t15} holds whenever the $i$'th eigenvalue gap is near the edge of the spectrum ($i \leq \tilde{O}(1)$).

To bound  $\eta_i-\eta_{i+1}$ for $i \geq \tilde{\Omega}(1)$, which are not near the edge of the spectrum, 
we will use the rigidity property of the GUE eigenvalues  \eqref{eq_ridgidity} 
(\cite{erdHos2012rigidity}; restated here as Lemma \ref{lemma_rigidity}).
Roughly, this rigidity property says that none of the eigenvalues $\eta_i$ fall more than a distance $\mathfrak{b} = O(\log(d)^{\log \log d}) = \tilde{O}(1)$ from their ``zero-temperature'' locations $\omega_i$.
Hence, $\eta_{j} \in [a,b]$ for all $i - \mathfrak{b} \leq j \leq i + \mathfrak{b}$,  where $a:=\eta_{i-\mathfrak{b}} \geq \omega_i - \mathfrak{b}^2\sqrt{d}$ and $b:= \eta_{i+\mathfrak{b}} \leq \omega_i + \mathfrak{b}^2\sqrt{d}$ w.h.p.

To apply this rigidity property, we define a new map  $g : \mathcal{W}_d \rightarrow \mathcal{W}_d$ where $g(\eta)$ leaves all eigenvalues $\eta_j$ outside $[a,b]$ fixed, 
 and $g(\eta)$ expands the $i$'th eigengap by a factor of $\frac{1}{s}$:   $g(\eta)[i] - g(\eta)[i+1] \geq \frac{1}{s}(\eta_i - \eta_{i+1})$. 
To ``make room'' for the expansion of the $i$'th gap without changing the locations of the eigenvalues outside $[a,b]$, it shrinks the eigengaps inside $[a,b]$ by a factor of $1-\alpha$ where $\alpha := (\frac{1}{s}-1)\frac{\eta_i - \eta_{i+1}}{b-a} \leq \mathfrak{b}^{-3}$ whenever $\eta \in A(\frac{s}{ \mathfrak{b}})$ because $\eta_i - \eta_{i+1} \leq s\frac{1}{ \mathfrak{b} \sqrt{d}}$  if $\eta \in A(\frac{s}{ \mathfrak{b}})$
(See \eqref{eq_g3}-\eqref{eq_g1} for the definition of $g$).
Thus, roughly, for all $\eta \in A(\frac{s}{\mathfrak{b}})$, 
\begin{eqnarray}\label{eq_t16}
 \frac{f(g(\eta))}{f(\eta)} \!= \!\prod_{j \neq \ell} \frac{|g(\eta)[j] - g(\eta)[\ell]|^2}{|\eta_j - \eta_{\ell}|^2}  e^{-\frac{1}{2} \sum_{j= i - \mathfrak{b}}^{i +\mathfrak{b}} \eta_j^2- g(\eta)[j]^2} 
    \geq \frac{1}{s^2}  (1-\alpha)^{2\mathfrak{b}^2}  e^{-\frac{1}{2} \sum_{j= i - \mathfrak{b}}^{i +\mathfrak{b}} \frac{1}{\mathfrak{b}}} 
    \geq \frac{1}{s^2}.
\end{eqnarray}
The first inequality holds since the product has $O(\mathfrak{b}^2)$ ``repulsion'' terms $\frac{|g(\eta)[j] - g(\eta)[\ell]|^2}{|\eta_j - \eta_{\ell}|^2} \geq (1-\alpha)^2$  where  $\ell, j \in [i-\mathfrak{b}, i+\mathfrak{b}]$ and one term $\frac{|g(\eta)[i] - g(\eta)[i+1]|^2}{|\eta_i - \eta_{i+1}|^2} \geq \frac{1}{s^2}$.
Replacing $g$ with $\phi$ and $A(s)$ with $A(\frac{s}{\mathfrak{b}})$ in $\eqref{eq_t15}$, and plugging in \eqref{eq_t16}  
we get, roughly, that for all $s \geq 0$ and all $i \in [d]$,
\begin{equation}
     \mathbb{P}\left(\eta_i - \eta_{i+1} \leq \frac{1}{\mathfrak{b} \sqrt{d}} s \right) = \int_{A(\frac{s}{\mathfrak{b}})} f(\eta) \mathrm{d} \eta  \leq s^3 \int_{A(\frac{s}{\mathfrak{b}})}  f(\eta) \times  \frac{f(g(\eta))}{f(\eta)}  \mathrm{det}(J_g(\eta))  \mathrm{d} \eta \leq s^3.
\end{equation}
This completes the proof overview of Lemma \ref{lemma_GUE_gaps} for the complex case.

\paragraph{Extending the proof of Theorem \ref{thm:eigenvalue_gap} from the complex case to the real case.}
 We note that many results in the random matrix literature rely on explicit determinantal formulas that are only available for complex-valued random matrices (see e.g. \cite{ratnarajah2004eigenvalues, johansson2005random, leake2020polynomial}).
 {For the special case of complex Hermitian matrices ($\beta =2$), it is possible to simplify the proof of our eigenvalue gap bounds (Theorem \ref{thm:eigenvalue_gap}) by viewing the eigenvalues of complex Dyson Brownian motion as a determinantal point process.}
However, our proofs avoid determinantal methods to allow our results to generalize to the real case.
 Indeed, the proof of Theorem \ref{thm:eigenvalue_gap} (which we state for the complex case) can be extended to the real case with minor modifications.
The main difference is that, for the real case, the repulsion term  $|\eta_{\ell} - \eta_j|^2$ in the joint eigenvalue density for the GUE random matrix \eqref{eq_t14} is replaced with $|\eta_{\ell} - \eta_j|^1$ for the GOE.
This changes the $s^3$ terms in \eqref{eq_t15} into $s^2$ terms, and the $\frac{1}{s^2}$ terms in \eqref{eq_t16} into a $\frac{1}{s}$.
Thus, for the real ($\beta =1$) case we get a $s^{\beta+1} = s^2$ term on the r.h.s. of Theorem \ref{thm:eigenvalue_gap}, in place of the term $s^{\beta+1} = s^3$ which appears in the complex ($\beta = 2$) version of Theorem \ref{thm:eigenvalue_gap}.

\section{Subspace recovery}\label{sec:subspace}

In the rank-$k$ subspace recovery problem, given a $d \times d$ covariance matrix $M$ with eigenvalues $\sigma_1 \geq \cdots \geq \sigma_d \geq 0$, the goal is to find a rank-$k$ projection matrix $H$ (corresponding to a rank-$k$ subspace) that minimizes the Frobenius distance to the projection matrix onto the subspace spanned by the top-$k$ eigenvectors of $M$.  
It is well-known (see \cite{bhatia2013matrix}) that the solution to this problem is the matrix $V_k V_k^\top$, where $V_k$ is the $d \times k$ matrix whose columns are the top-$k$ eigenvectors of $M$.
In the private version of this problem, the goal is to output a corresponding approximation $\hat{H}$ to $V_k V_k^\top$ that, in addition, satisfies the $(\eps,\delta)$-DP constraint.

For the subspace recovery problem, 
 \cite{dwork2014analyze} analyze a version of the Gaussian mechanism of \cite{dwork2006our}, where one perturbs the entries of $M$ by adding a symmetric matrix $E$ with i.i.d. Gaussian entries $N(0,\nfrac{\sqrt{\log\frac{1}{\delta}}}{\eps})$, to obtain an $(\eps, \delta)$-differentially private mechanism which outputs a perturbed matrix $\hat{M} = M+E$.
They then post-process this matrix $\hat{M}$ to obtain a rank-$k$ projection matrix which projects onto the subspace spanned by the top-$k$ eigenvectors of $\hat{M}$.
For this mechanism, \cite{dwork2014analyze} prove a Frobenius-distance bound of  $\|\hat{H}- H\|_F \leq \tilde{O}\left(\frac{\sqrt{kd}}{(\sigma_k - \sigma_{k+1})}\right)$  whenever $\sigma_k - \sigma_{k+1} > \tilde{\Omega}(\sqrt{d})$ (implied by their Theorem 6, which is stated for the spectral norm).

Using similar techniques to the proof of Theorem \ref{thm_rank_k_covariance_approximation_new}, one can obtain the following bound for the rank-$k$ subspace recovery problem.
 \begin{theorem}[\bf Frobenius  bound for Private Subspace Recovery] \label{thm_rank_k_subspace}
Suppose we are given $k>0$, $T>0$, and a Hermitian matrix  $M \in \mathbb{C}^{d \times d}$ (or a real symmetric matrix $M \in \mathbb{R}^{d \times d}$) with eigenvalues  $\sigma_1 \geq \cdots \geq \sigma_d \geq 0$.   
Let $\hat{M} := M + \sqrt{T}[(W_1 + \mathfrak{i}W_2) + (W_1 + \mathfrak{i}W_2)^\ast]$ (or, in the real case, $\hat{M} := M + \sqrt{T}(W_1 + W_1^\ast)$), where $W_1, W_2 \in \mathbb{R}^{d \times d}$ have entries which are independent $N(0,1)$ random variables.
 Denote, by  $\sigma_1 \geq \cdots \geq \sigma_d$ the eigenvalues of $M$, and by $V_k$ and  $\hat{V}_k$ the matrices whose columns are the top-$k$ eigenvectors of $M$ and $\hat{M}$ respectively. 
Suppose that $M$ satisfies Assumption \ref{assumption_gap} $(M, k, T)$.  
Then we have
$$ 
\sqrt{\mathbb{E}\left[\left\|\hat{V}_k\hat{V}_k^\ast -  V_k V_k^\ast \right \|_F^2\right]} \leq \tilde{O}\left(\sqrt{\sum_{i=1}^k \sum_{j= k+1}^d \frac{1}{(\sigma_i - \sigma_j)^2}}\right)\cdot \sqrt{T}.
$$
\end{theorem}
\noindent
The proof of Theorem \ref{thm_rank_k_subspace} is simpler than the proof of Theorem \ref{thm_rank_k_covariance_approximation_new}, and can handle either real-valued or complex-valued Gaussian perturbations.
The main difference is that we analyze a projection-matrix-valued rank-$k$ diffusion with eigenvalues $\lambda_i =1$ for $i<k$
 and $\lambda_i =0$ for $i>k$.
 As all the gaps $\lambda_i-\lambda_{i+1}$ between consecutive eigenvalues of this matrix diffusion, aside from the $k$'th gap, are equal to $0$, many of the terms on the r.h.s. of the expression \eqref{eq_t7_2} for the utility cancel.
 The remaining terms can be bounded simply via Weyl's inequality. 
We give a detailed outline of the proof in Appendix \ref{Subpsace_recovery_outline}.

Theorem \ref{thm_rank_k_subspace} immediately implies existence of an $(\epsilon,\delta)$-differentially private mechanism which, given an input matrix $M \in \mathbb{R}^{d \times d}$, outputs a rank-$k$ projection matrix $P$ satisfying the utility bound $$\mathbb{E}[\|P-  V_k V_k^\ast  \|_F] \leq \tilde{O}\left(\sqrt{\sum_{i=1}^k \sum_{j= k+1}^d \frac{1}{(\sigma_i - \sigma_j)^2}}\right)\cdot \frac{\log^{\frac{1}{2}} \frac{1}{\delta}}{\epsilon}.$$
 This mechanism outputs the matrix $P = \hat{V}_k\hat{V}_k^\ast$ defined in the statement of Theorem \ref{thm_rank_k_subspace} (for $T = \frac{\log \nicefrac{1}{\delta}}{\epsilon^2}$), and is guaranteed to be $(\epsilon,\delta)$-differentially private from the privacy guarantees given in prior works on the Gaussian mechanism (see e.g. \cite{dwork2014analyze}).

For matrices $M$ satisfying  $\sigma_k - \sigma_{k+1} \geq 4\sqrt{Td}$, Theorem \ref{thm_rank_k_subspace} recovers (in expectation) the bound on the Frobenius norm given in Theorem 6 of \cite{dwork2014analyze} (which they derive from the worst-case perturbation bound of \cite{davis1970rotation}, restated here as Inequality \eqref{eq_DK}), which states that $\|\hat{V}_k\hat{V}_k^\ast - V_k V_k^\ast \|_F \leq O\left(\frac{\sqrt{k}\sqrt{d}}{\sigma_k - \sigma_{k+1}} \sqrt{T}\right)$ w.h.p. 
Moreover, for many input matrices $M$, Theorem \ref{thm_rank_k_subspace} implies stronger bounds than those implied by \cite{dwork2014analyze, davis1970rotation}.
 For instance, if the eigenvalues of $M$ also satisfy $\sigma_i - \sigma_{i+1} \geq \Omega(\sigma_k - \sigma_{k+1})$ for all $i < k$ (or, more generally, if we have $\sigma_i - \sigma_{k+1} \geq \Omega((i-k)(\sigma_k - \sigma_{k+1}))$ for all $i<k$), then the bound in our Theorem \ref{thm_rank_k_subspace} implies $\mathbb{E}[\|\hat{V}_k\hat{V}_k^\ast - V_k V_k^\ast \|_F] \leq O\left(\frac{\sqrt{d}}{\sigma_k- \sigma_{k+1}} \sqrt{T}\right)$, improving on the bound implied by \cite{davis1970rotation} and \cite{dwork2014analyze} by a factor of $\sqrt{k}$.
  As another example, if $\sigma_i - \sigma_{i+1} \geq \Omega\left(\frac{\sigma_k - \sigma_{k+1}}{\sqrt{i-k}}\right)$ for all $i > k$,  then Theorem \ref{thm_rank_k_subspace} implies $\mathbb{E}[\|\hat{V}_k\hat{V}_k^\ast - V_k V_k^\ast \|_F] \leq O\left(\frac{\sqrt{k} \log^{\nicefrac{1}{2}}(d)}{\sigma_k- \sigma_{k+1}} \sqrt{T}\right)$, improving on the bound implied by \cite{davis1970rotation} and \cite{dwork2014analyze} by a factor of $\tilde{O}(\sqrt{d})$.

More specifically, in the $(\epsilon, \delta)$-differential privacy application considered in \cite{dwork2014analyze}, \(M = A^\top A\) where \(A\) is an \(n \times d\) data matrix where each of the \(n\) rows is “clipped” such that it has norm at most 1. Thus, for Theorem \ref{thm_rank_k_subspace} to hold in this setting, it is necessary (but not sufficient) for the data matrix to have at least \(n \geq \tilde{\Omega}\left(k \sqrt{d} \frac{\log^{\nicefrac{1}{2}} \nicefrac{1}{\delta}}{\epsilon}\right) \) rows in order for \(M\) to satisfy Assumption \ref{assumption_gap} $(M, k, T)$ with $T = \frac{\log \nicefrac{1}{\delta}}{\epsilon^2}$, which requires that \(\sigma_k - \sigma_{k+1} > \tilde{\Omega}\left(\sqrt{d} \frac{\log^{\nicefrac{1}{2}} \nicefrac{1}{\delta}}{\epsilon}\right)\).
    In many privacy applications, where $n$ is the number of datapoints and $d$ may be the number of features in a dataset, $n$ is oftentimes larger than $d^{\frac{3}{2}}$ (see e.g. the discussion in \cite{brown2021covariance}),
    in which case one has $n \geq d^{\frac{3}{2}}\geq k \sqrt{d}$ for any $k \leq d$. 
Theorem \ref{thm_rank_k_subspace} improves over \cite{dwork2014analyze} when the eigenvalues of \(M\) satisfy \(\sqrt{\sum_{i=1}^k \sum_{j= k+1}^d  \frac{1}{(\sigma_i - \sigma_j)^2}} < \tilde{O}\left(\frac{\sqrt{k}\sqrt{d}}{\sigma_k - \sigma_{k+1}}  \frac{\log^{\nicefrac{1}{2}} \nicefrac{1}{\delta}}{\epsilon}\right)\).
     The magnitude of the improvement is by a factor of \(\frac{\sigma_k - \sigma_{k+1}}{\sqrt{k}\sqrt{d}} \sqrt{\sum_{i=1}^k \sum_{j= k+1}^d  \frac{1}{(\sigma_i - \sigma_j)^2}}\). 
     In the aforementioned setting of  matrices \(M\) with eigenvalues satisfying \(\sigma_i - \sigma_{i+1} \geq \Omega(\sigma_k- \sigma_{k+1})\) for all \(i<k\), the improvement is by a factor of \(\sqrt{k}\). 
    If one has \(\sigma_i - \sigma_{i+1} \geq \frac{\sigma_k - \sigma_{k+1}}{\sqrt{i-k}}\) for all \(i > k\), the improvement is by a factor of \(\sqrt{d}\).
      As one concrete example, for matrices $M$ with spectrum $\sigma_i = (d - i) \times c \sqrt{d}$ for $i \in [d]$, where $c = \tilde{\Theta}\left(\frac{\log^{\nicefrac{1}{2}} \nicefrac{1}{\delta}}{\epsilon}\right)$, our result improves (in expectation) by a factor of $\sqrt{d}$ over the bound in  \cite{dwork2014analyze}.
\color{black}

 Finally, recall that \cite{o2018random} provide eigenvector perturbation bounds for matrices $\hat{M} := M+E$ in the special case when the input matrix $M$ is a deterministic low-rank matrix of rank $r \geq k$ and the matrix $E$ is a random matrix.
 If one directly applies the bound in their Theorem 18 to the setting when $E$ is a Hermitian Gaussian random matrix, one obtains a bound on the quantity $\|\hat{V}_k \hat{V}_k^\ast - V_k V_k^\ast\|_F$.
Theorem \ref{thm_rank_k_subspace} improves (in expectation) on the resulting bound by a factor of $\sqrt{k}$ whenever e.g. $\sigma_{k}-\sigma_{k+1} = \Omega(\sigma_k)$ and $\sigma_k \geq \Omega(\sqrt{d})$.

\begin{remark}
In addition to the metric  $\|\hat{H}- H\|_F$,  \cite{dwork2014analyze} also provide bounds for the problem of recovering a subspace in the inner product metric $\langle M, \hat{H}-H \rangle$ under $(\epsilon, \delta)$-differential privacy (see also \cite{chaudhuri2012near, dwork2014analyze, gilad2017smooth,amin2019differentially}).
 While a bound on the Frobenius norm distance $\| \hat{H}-H \|_F \leq b$  implies an upper bound on the inner product metric $\langle M, H - \hat{H} \rangle \leq 2\| M_{k}\|_F \cdot b$ (by the Cauchy-Schwarz inequality), an upper bound on the inner product metric does not (in general) imply any upper bound on the Frobenius-norm distance $\| \hat{H} -H \|_F$.
The bounds in \cite{dwork2014analyze} are tight with respect to this metric for matrices $M$ with worst-case spectrum $\sigma_1 \geq \cdots \geq \sigma_d$.
It is an interesting open problem whether one can improve on these bounds for matrices $M$ with specific spectral profiles $\sigma_1 \geq \cdots \geq \sigma_d$.
\end{remark}
\color{black}

     \begin{remark}[\bf Spectral norm bounds]\label{rem_spectral_norm}

We leave as an open problem whether one can extend our bounds on the Frobenius norm utility $\|\hat{M}_k -  M_k \|_F$ in Theorem \ref{thm_rank_k_covariance_approximation_new}, to a bound on the spectral norm utility $\|\hat{M}_k -  M_k \|_2$ that is tighter than the trivial bound $\|\hat{M}_k -  M_k \|_2 \leq \|\hat{M}_k -  M_k \|_F$ (and similarly, whether one can extend our bounds in Theorem \ref{thm_rank_k_subspace} on the Frobenius norm utility for subspace recovery $\|\hat{V}_k\hat{V}_k^\ast - V_k V_k^\ast \|_F$ to a bound on the spectral norm utility  $\|\hat{V}_k\hat{V}_k^\ast - V_k V_k^\ast \|_2$).
 Recall that, to bound the Frobenius norm utility in Theorem \ref{thm_rank_k_covariance_approximation_new}, we use Ito's lemma (Lemma \ref{lemma_ito_lemma_new}) to compute an expression for the stochastic derivative of the (squared) Frobenius norm utility as a function of the eigenvalue gaps of Dyson Brownian motion.
 We then use our eigenvalue gap bounds for Dyson Brownian motion (Theorem \ref{thm:eigenvalue_gap}) to bound this stochastic derivative.
A key fact which allows us to bound this stochastic derivative is that the (squared) Frobenius norm of a matrix is a differentiable function of its entries (and of its eigenvalues), with second partial derivatives of magnitude $O(1)$.
 The main challenge in extending our techniques to the spectral norm is that the spectral norm of a matrix is a non-differentiable function of its entries (and is also a non-differentiable function of its eigenvalues), and its first and second derivatives have singularities at points where eigenvalue gaps of the matrix vanish.
 One possible approach to extending our utility bounds to the spectral norm may be to use high-probability eigenvalue gap bounds such as those in Theorem \ref{thm:eigenvalue_gap} to show that the first and second derivatives of the spectral norm of a matrix undergoing Dyson Brownian motion are small ``on average'' over time.
 \end{remark}
 \color{black}

\section{Differentially private rank-$k$ approximation: Proof of Theorem \ref{thm_utility}} \label{sec_proof_utility_privacy}

\begin{proof}[Proof of Theorem \ref{thm_utility}] \,

\paragraph{Privacy.}

The real Gaussian mechanism, $M + \sqrt{T}(W_1 + W_1^\top)$, where $W_1$ is a matrix with i.i.d.  $N(0,1)$ entries, was studied in \cite{dwork2014analyze} and shown to be $(\epsilon, \delta)$-differentially private for $T = \frac{2\log\frac{1.25}{\delta}}{\eps^2}$.
Our Algorithm \ref{alg_quaternion_Gaussian} is $(\epsilon, \delta)$-differentially private since it is a post-processing of the real Gaussian mechanism.
This is because any post-processing of an $(\varepsilon, \delta)$-differentially private mechanism (which does not have access to the original input matrix $M$) is guaranteed to be  $(\varepsilon, \delta)$-differentially private (see e.g. \cite{dwork2006our},  \cite{dwork2014algorithmic}).
To see why Algorithm \ref{alg_quaternion_Gaussian}  is a post-processing of the real Gaussian mechanism, observe that

\begin{eqnarray*}
    \hat{M} &=& M + G\\
    &=& M +  W + W^\ast\\
    &=& M +  (W_1 +  W_2 \mathfrak{i}) + (W_1 +  W_2 \mathfrak{i})^\ast\\
    &=& M + W_1 + W_1^\top + [W_2 \mathfrak{i} + (W_2 \mathfrak{i})^\ast].
\end{eqnarray*}

\paragraph{Utility of complex matrix $\hat{M}_k$ implies Utility of real matrix $Y$.}
Let $M = V \Sigma V^\top$ be a diagonalization of the real symmetric input matrix $M$ with eigenvalues $\sigma_1\geq \cdots \geq \sigma_d \geq 0$.
Let $M_k = V \Sigma_k V^\top$ be a (non-private) rank-$k$ approximation of $M$, where $\Sigma_k = \mathrm{diag}(\sigma_1,\ldots, \sigma_k, 0, \ldots,0)$.
Suppose we can show an upper bound on $\|\hat{M}_k -  M_k\|_F$, where $\hat{M}_k$ is the complex matrix in Algorithm \ref{alg_quaternion_Gaussian}.

Let $$\mathfrak{R}_k:= \{A \in \mathbb{R}^{d\times d}: \mathrm{rank}(A) \leq k\}$$ denote the set of real $d \times d$ rank-$k$ matrices.
 Since $Y = \mathrm{Real}(\hat{V} \hat{\Sigma}_k \hat{V}^\ast)$, we have that $Y \in \mathrm{argmin}_{Z \in \mathfrak{R}_k} \{\|\hat{M}_k-Z\|_F\}$.
This is because $\mathrm{Real}(\hat{V} \hat{\Sigma}_k \hat{V}^\ast)$ is a matrix of rank at most $k$ and the real and imaginary parts of $\hat{V} \hat{\Sigma}_k \hat{V}^\ast$ are orthogonal to each other in the Frobenius inner product. 
Thus, since $Y \in \mathrm{argmin}_{Z \in \mathfrak{R}_k} \{\|\hat{M}_k-Z\|_F\}$ and $M_k \in  \mathfrak{R}_k$ is also in the set of real-valued rank-$k$ matrices, we have that
\begin{equation*}
\|\hat{M}_k-Y\|_F \leq \|\hat{M}_k -  M_k\|_F.
\end{equation*}

\noindent
Therefore, we have
\begin{equation}\label{eq_o2}
\|Y - M_k\|_F \leq \|\hat{M}_k-Y\|_F + \|\hat{M}_k -  M_k\|_F \leq 2  \|\hat{M}_k -  M_k\|_F.
\end{equation}
Plugging in our bound for $\sqrt{\mathbb{E}[\|\hat{M}_k -  M_k\|_F^2]}$ from  Theorem \ref{thm_rank_k_covariance_approximation_new} into \eqref{eq_o2}, we get that
$$\sqrt{\mathbb{E}[\|M_k - Y\|_F^2]}  \leq \tilde{O}\left(\sqrt{kd} \frac{\sigma_k}{\sigma_k - \sigma_{k+1}} \times  \frac{\sqrt{\log\frac{1}{\delta}}}{\eps}\right).$$
\end{proof}

\section{Structure of the proofs of Theorems \ref{thm_utility}, \ref{thm_rank_k_covariance_approximation_new}, and \ref{thm:eigenvalue_gap}}

A diagram of the structure of the proof of Theorem \ref{thm_rank_k_covariance_approximation_new} (and its corollary, Theorem \ref{thm_utility}) is given in Figure \ref{fig_proof_diagram_1}; this diagram takes as input Theorem \ref{thm:eigenvalue_gap}.
 For a diagram of the structure of the proof of Theorem \ref{thm:eigenvalue_gap}, see Figure \ref{fig_proof_diagram_2}. 
The proofs of the different theorems, lemmas, and propositions used to prove Theorem \ref{thm_rank_k_covariance_approximation_new}, \ref{thm_utility}, and \ref{thm:eigenvalue_gap} are given in the following order:
\begin{enumerate}
\item Lemma \ref{lemma_spectral_martingale}
\item Lemma \ref{lemma_utility_rare_event}
\item Proposition \ref{lemma_gap_concentration}
\item Proposition \ref{lemma_t0}
\item Lemma \ref{Lemma_projection_differntial}
\item Lemma \ref{Lemma_integral}
\item Completing the proof of Theorem \ref{thm_rank_k_covariance_approximation_new}
\end{enumerate}
\noindent
The proof of the above theorems, lemmas, and propositions take as input Theorem \ref{thm:eigenvalue_gap}, and related lemmas and a corollary which follow from Theorem \ref{thm:eigenvalue_gap}.
These results, and propositions and lemmas used to prove these results, are  proved in the following order
\begin{enumerate}
\item[(8)] Proposition \ref{prop_stochastic_derivative_comparison}
\item[(9)] Lemma \ref{lemma_gap_comparison}
\item[(10)] Lemma \ref{lemma_bad_event}
\item[(11)] Proposition \ref{prop_sum_nonindependent}
\item[(12)] Corollary \ref{lemma_gaps_any_start}
\end{enumerate}
\noindent
In particular, Lemma \ref{lemma_gap_comparison} reduces the task of proving Theorem \ref{thm:eigenvalue_gap} to proving Lemma \ref{lemma_GUE_gaps}, which is a special case of Theorem \ref{thm:eigenvalue_gap} where the initial matrix is $M = 0$.
The intermediate results towards the proof of Lemma \ref{lemma_GUE_gaps} are proved in the following order:
\begin{enumerate}
\item[(13)] Proposition \ref{prop_classical}
\item[(14)] Proposition \ref{prop_n1}
\item[(15)] Proposition \ref{prop_map}
\item[(16)] Lemma \ref{prop_Jacobian}
\item[(17)] Lemma \ref{lemma_mean_field}
\item[(18)] Lemma \ref{lemma_density_ratio}
\item[(19)] Proposition \ref{prop_map_phi}
\item[(20)] Lemma \ref{prop_Jacobian_phi}
\item[(21)] Lemma \ref{lemma_density_ratio_edge}
\item[(22)] Lemma \ref{lemma_GUE_gaps}
\end{enumerate}
\noindent
Finally, the proof of Lemma \ref{lemma_spectral_martingale_b}, which we use to prove Lemmas \ref{lemma_spectral_martingale} and \ref{lemma_bad_event}, is deferred to Appendix \ref{sec_proof_of_lemma_spectral_martingale_b} as it is standard.
 A list of key notations used in the proofs is given in Appendix \ref{sec_glossary_of_notation}.

 In each proof, we give explanations for why the different steps, and the different lines in each block of equations or inequalities, hold.
 For the steps or equation/inequality lines that are evident, we do not provide an explanation.
  For equation or inequality lines that hold as a consequence of another equation, Theorem, etc., we reference that equation, Theorem, etc. above the equality or inequality sign.
 In blocks of equations and inequalities with multiple lines, we have selectively numbered lines that require additional explanation. 
 Depending on the context, the numbering on the last line refers to the l.h.s. of the first line and the r.h.s. of the last line, or may refer to just the last line itself.
 If there is at least one inequality in a block of equations, then the whole equation is an inequality.

\begin{figure}
  \resizebox{\textwidth}{!}{%
\begin{tikzpicture}[
roundnode/.style={circle, draw=green!60, fill=green!5, very thick, minimum size=7mm},
squarednode/.style={rectangle, draw=black!60, rounded corners=.2cm, fill=black!0, very thick, minimum size=5mm},
]
\node[squarednode, text width=10cm,align=center]      (Privacy_Theorem)                              {\textbf{Theorem \ref{thm_utility}:} Private low-rank covariance approximation };
\node[squarednode, text width=10cm,align=center]      (Main_Theorem)  [below=of Privacy_Theorem]             {\textbf{Theorem \ref{thm_rank_k_covariance_approximation_new}:} Frobenius bound for complex Gaussian perturbations};
\node[squarednode,text width=2.7cm,align=center]      (Lemma_7_2)       [below=of Main_Theorem] {\textbf{Lemma \ref{lemma_utility_rare_event}:}  Frobenius bound under ``bad'' rare event.};
\node[squarednode, text width=3cm,align=center]      (Lemma_7_1)       [below=of Lemma_7_2, xshift=0cm] {\textbf{Lemma \ref{lemma_spectral_martingale}:} Probability of ``bad'' rare event occurring };
\node[squarednode, text width=3cm,align=center]      (Lemma_8_3)       [below=of Lemma_7_1, xshift=0cm] {\textbf{Lemma \ref{lemma_bad_event}:} Showing gaps are uniformly bounded below over time with high probability};
\node[squarednode, text width=3cm,align=center]      (Theorem_2_3)       [below=of Lemma_8_3, yshift=-0.5cm] {\textbf{Theorem \ref{thm:eigenvalue_gap}:} Eigenvalue gaps of GUE/GOE from any initial condition };


\node[squarednode,text width=4cm,align=center]      (Corollary_8_5)       [right=of Lemma_7_1, xshift=-0.8cm] {\textbf{Corollary \ref{lemma_gaps_any_start}:} Bound on gaps between non-neighboring eigenvalues};

\node[squarednode,draw=black!30, dashed, fill=gray!10, text width=2.5cm,align=center]      (Proposition_8_4)       [below=of Corollary_8_5, right=of Lemma_8_3, below=of Corollary_8_5] {\color{gray} \textbf{Proposition \ref{prop_sum_nonindependent}:} Concentration bound for sums of random variables};

\node[squarednode, draw=black!30, dashed, fill=gray!10, text width=3cm,align=center]      (Lemma_3_8)       [left=of Theorem_2_3] {\color{gray} \textbf{Lemma \ref{lemma_spectral_martingale_b}:} Spectral norm bound};


\node[squarednode, text width=2cm,align=center]      (Lemma_7_6)       [left=of Lemma_7_2,  yshift=-0.5cm, xshift=0.4cm] {\textbf{Lemma \ref{Lemma_integral}:}  Frobenius norm expression as integral of inverse eigengaps};

\node[squarednode, text width=2cm,align=center]      (Lemma_7_5)       [left=of Lemma_7_6, xshift=0.7cm] {\textbf{Lemma \ref{Lemma_projection_differntial}:} It\^o derivative $\mathrm{d} (u_i(t) u_j^\ast(t))$};

\node[squarednode,draw=black!30, dashed, fill=gray!10, text width=2.2cm,align=center]      (Proposition_7_4)       [left=of Lemma_7_5, xshift=0.7cm] {\color{gray} \textbf{Prop. \ref{lemma_t0}:} Crude ``jump start'' bound};

\node[squarednode,draw=black!30, dashed, fill=gray!10, text width=2.2cm,align=center]      (Lemma_7_3)       [left=of Proposition_7_4, xshift=0.7cm] {\color{gray} \textbf{Prop. \ref{lemma_gap_concentration}:} ``Worst-case'' eigenvalue gap bound};

\draw[<-, thick] (Privacy_Theorem.south) -- (Main_Theorem.north);
\draw[<-, thick] (Main_Theorem.south) -- (Lemma_7_2.north);
\draw[<-, thick] (Lemma_7_2.south) -- (Lemma_7_1.north);
\draw[<-, thick] (Lemma_7_1.south) -- (Lemma_8_3.north);
\draw[<-, thick] (Lemma_8_3.south) -- (Theorem_2_3.north);

\draw[->, thick] (Corollary_8_5.north) -- ++(0,4cm) -| (Main_Theorem.south);

\draw[->, thick] (Theorem_2_3.east) -- ++(5cm,0) |- (Corollary_8_5.east);
\draw[<-, thick] (Corollary_8_5.south)--(Proposition_8_4.north);


\draw[->, thick] (Lemma_3_8.north) |- (Lemma_8_3.west);

\draw[->, thick] (Lemma_3_8.north) |-(Lemma_7_1.west);


\draw[->, thick] (Lemma_7_6.north) -- ++(0,0.4cm) -|   (Main_Theorem.south);

\draw[<-, thick] (Main_Theorem.west)-| (Lemma_7_5.north);

\draw[<-, thick] (Main_Theorem.west)-| (Proposition_7_4.north);

\draw[->, thick]  (Lemma_7_3.north) |-  (Main_Theorem.west);

\draw[<-, thick](Lemma_7_6.west) -| (Lemma_7_5.east);

\end{tikzpicture}}
\caption{A diagram showing the structure of the proof of Theorems \ref{thm_rank_k_covariance_approximation_new} and \ref{thm_utility}. Lower-level lemmas and propositions are denoted by subdued dashed boxes. (See Figure \ref{fig_proof_diagram_2} for a diagram of the structure of the proof of Theorem \ref{thm:eigenvalue_gap}.)}\label{fig_proof_diagram_1}
\end{figure}
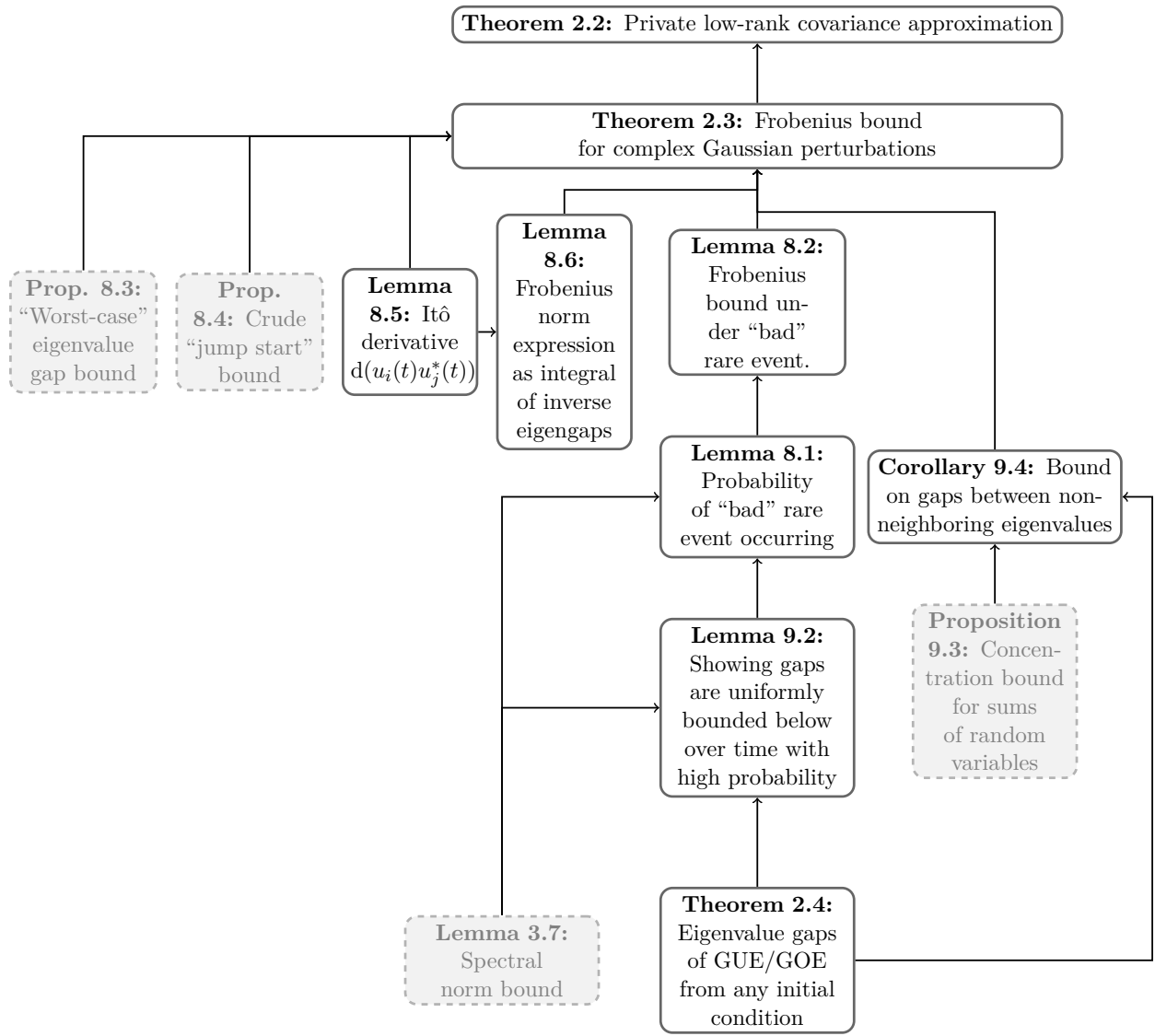

\begin{figure}
  \resizebox{0.9\textwidth}{!}{%
\begin{tikzpicture}[
roundnode/.style={circle, draw=green!60, fill=green!5, very thick, minimum size=7mm},
squarednode/.style={rectangle, draw=black!60, rounded corners=.2cm, fill=black!0, very thick, minimum size=5mm},
]

\node[squarednode, text width=3cm,align=center]      (Theorem_2_3)     {\textbf{Theorem \ref{thm:eigenvalue_gap}:} Eigenvalue gaps of GUE/GOE from any initial condition};

\node[squarednode, text width=3cm,align=center]      (Lemma_8_6)       [below=of Theorem_2_3, yshift=-2cm] {\textbf{Lemma \ref{lemma_GUE_gaps}:} Eigengaps of GUE/GOE  initialized at $0$};

\node[squarednode, text width=3.7cm,align=center]      (Lemma_4_1)       [left=of Theorem_2_3, yshift=-1cm] {\textbf{Lemma \ref{lemma_gap_comparison}:} Eigenvalue gap comparison Lemma };

\node[squarednode,draw=black!30, dashed, fill=gray!10, text width=3.7cm,align=center]      (Proposition_8_1)       [below=of Lemma_4_1, yshift=0.6cm] {\color{gray}\textbf{Proposition \ref{prop_stochastic_derivative_comparison}}: Comparison bound for time-derivative of eigenvalue gaps.};

\node[squarednode, text width=3cm,align=center]      (Lemma_8_13)       [below=of Lemma_8_6, xshift=0cm, yshift=-1cm] {\textbf{Lemma \ref{lemma_density_ratio}:} Bounding how much $g$ changes  eigenvalues' joint density (bulk case)};

\node[squarednode, text width=3cm,align=center]      (Lemma_8_12)       [below=of Lemma_8_13, yshift=-0.5cm] {\textbf{Lemma \ref{lemma_mean_field}:} Mean-field approximation for far-away eigenvalues};

\node[squarednode,draw=black!30, dashed, fill=gray!10, text width=3cm,align=center]      (Proposition_8_9)       [below=of Lemma_8_12] {\color{gray} \textbf{Proposition \ref{prop_n1}:} Implies existence of the spectrum mapping $g$ satisfying \eqref{eq_b3}-\eqref{eq_g1} };

\node[squarednode, draw=black!30, dashed, fill=gray!10, text width=3cm,align=center]      (Proposition_8_7)       [below=of Proposition_8_9] {\color{gray} \textbf{Proposition \ref{prop_classical}:} Position of ``classical eigenvalues''};

\node[squarednode,draw=black!30, dashed, fill=gray!10, text width=2.5cm,align=center]      (Proposition_8_10)       [right=of Lemma_8_12, xshift=-0.5cm] {\color{gray} \textbf{Prop. \ref{prop_map}:} Cardinality of pre-image, and Lipschitz properties, of $g$ };

\node[squarednode, text width=4cm,align=center]      (Proposition_8_11)       [right=of Lemma_8_13,  xshift=2cm, yshift=0cm] {\textbf{Lemma \ref{prop_Jacobian}:} Bounding the Jacobian determinant of $g$};

\node[squarednode,draw=black!30, fill=blue!10, text width=2.5cm,align=center]      (Proposition_8_15)       [left=of Lemma_8_13, xshift=0cm] {\textbf{Lemma \ref{prop_Jacobian_phi}:} Bounding the Jacobian determinant of $\phi$};

\node[squarednode,draw=black!30, fill=blue!10, text width=2.5cm,align=center]      (Lemma_8_17)       [left=of Proposition_8_15, xshift=-1.7cm, yshift=0cm] {\textbf{Lemma \ref{lemma_density_ratio_edge}:} Bounding how much the map $\phi$ changes eigenvalues joint density (edge case)};

\node[squarednode,draw=black!30, dashed, fill=blue!10, text width=2.5cm,align=center]      (Proposition_8_14)       [below=of Lemma_8_17, xshift=1.8cm, yshift=0cm] {\color{gray} \textbf{Proposition \ref{prop_map_phi}:} Injectivity and Lipschitz properties of $\phi$};

\node[squarednode,draw=black!30, dashed, fill=blue!10, text width=3cm,align=center]      (Proposition_9_13)       [below=of Proposition_8_14] {\color{gray} \textbf{Proposition \ref{prop_n1_edge}:} Implies existence of the spectrum mapping $\phi$ satisfying \eqref{eq_phi1}-\eqref{eq_phi3} };

\draw[<-, thick] (Theorem_2_3.south) -- (Lemma_8_6.north);

\draw[<-, thick] (Theorem_2_3.west) -| (Lemma_4_1.north);

\draw[<-, thick] (Lemma_4_1.south) -- (Proposition_8_1.north);

\draw[<-, thick] (Lemma_8_6.south) -- (Lemma_8_13.north);

\draw[<-, thick] (Lemma_8_13.south) -- (Lemma_8_12.north);

\draw[->, thick] (Proposition_8_9.north) -- ++(0,0.4cm) -|  (Lemma_8_12.south);

\draw[->, thick] (Proposition_8_7.north) -- (Proposition_8_9.south);

\draw[->, thick] (Proposition_8_9.north) -- ++(0,0.2cm) -|   (Proposition_8_10.south);

\draw[->, thick]   (Proposition_8_9.north)  -- ++(0,0.2cm) -|  (Proposition_8_11.south);

\draw[->, thick]  (Proposition_8_7.north) -- ++(0,0.4cm) -| (Proposition_8_11.south)  ;

\draw[thick, postaction={decorate}, decoration={markings, 
        mark=between positions 0.4 and 1 step 3cm with {\arrow{>}}}]  (Proposition_8_10.north)  -- ++(0,0.4cm) -| (Proposition_8_11.south) ;

\draw[thick, postaction={decorate}, decoration={markings, 
        mark=between positions 0.4 and 2 step 4cm with {\arrow{>}}}] (Proposition_8_10.north) -- ++(0,0.4cm)  -|(Lemma_8_13.south);

\draw[thick, postaction={decorate}, decoration={markings, 
        mark=between positions 0.2 and 1 step 3.6cm with {\arrow{>}}}] (Proposition_8_10.north) -- ++(0,4.5cm) -| (Lemma_8_6.south);

\draw[->, thick]  (Proposition_8_10.west) --(Lemma_8_12.east);

\draw[->,thick]
(Proposition_8_11.north)  -- ++(0,1.5cm) -| (Lemma_8_6.south);

\draw[->, thick] (Proposition_8_14.north) -- ++(0,0.4cm) -| (Lemma_8_17.south)  ;

\draw[->, thick]  (Proposition_8_14.north)  -- ++(0,0.4cm) -- ++(1.5cm,0) |-  (Lemma_8_6.west);

\draw[->, thick]  (Proposition_8_15.north) |- (Lemma_8_6.west);

\draw[->, thick]  (Lemma_8_17.north)   |- (Lemma_8_6.west) ;

\draw[->, thick]  (Proposition_9_13.north)   |- (Proposition_8_14.south) ;

\draw[->, thick]  (Proposition_8_7.west)   -| (Proposition_9_13.south) ;

\draw[->, thick]  (Proposition_8_7.west)   -| (Lemma_8_17.south) ;

\draw[->, thick]  (Proposition_9_13.east)   -| (Proposition_8_15.south) ;

\draw[->, thick]  (Proposition_9_13.east)   -- ++(.2cm,0) |-  (Lemma_8_6.west) ;

\end{tikzpicture}}
\caption{A diagram of the  proof of Theorem \ref{thm:eigenvalue_gap}. 
Lemmas and Propositions below Lemma \ref{lemma_GUE_gaps} are separated into results dealing with the ``bulk'' of the eigenvalue spectrum, and analogous (but slightly simpler to prove) results dealing with the ``edge'' of the eigenvalue spectrum (denoted by blue boxes). 
 Throughout the diagram, lower-level propositions are denoted by subdued dashed boxes.
}\label{fig_proof_diagram_2}
\end{figure}
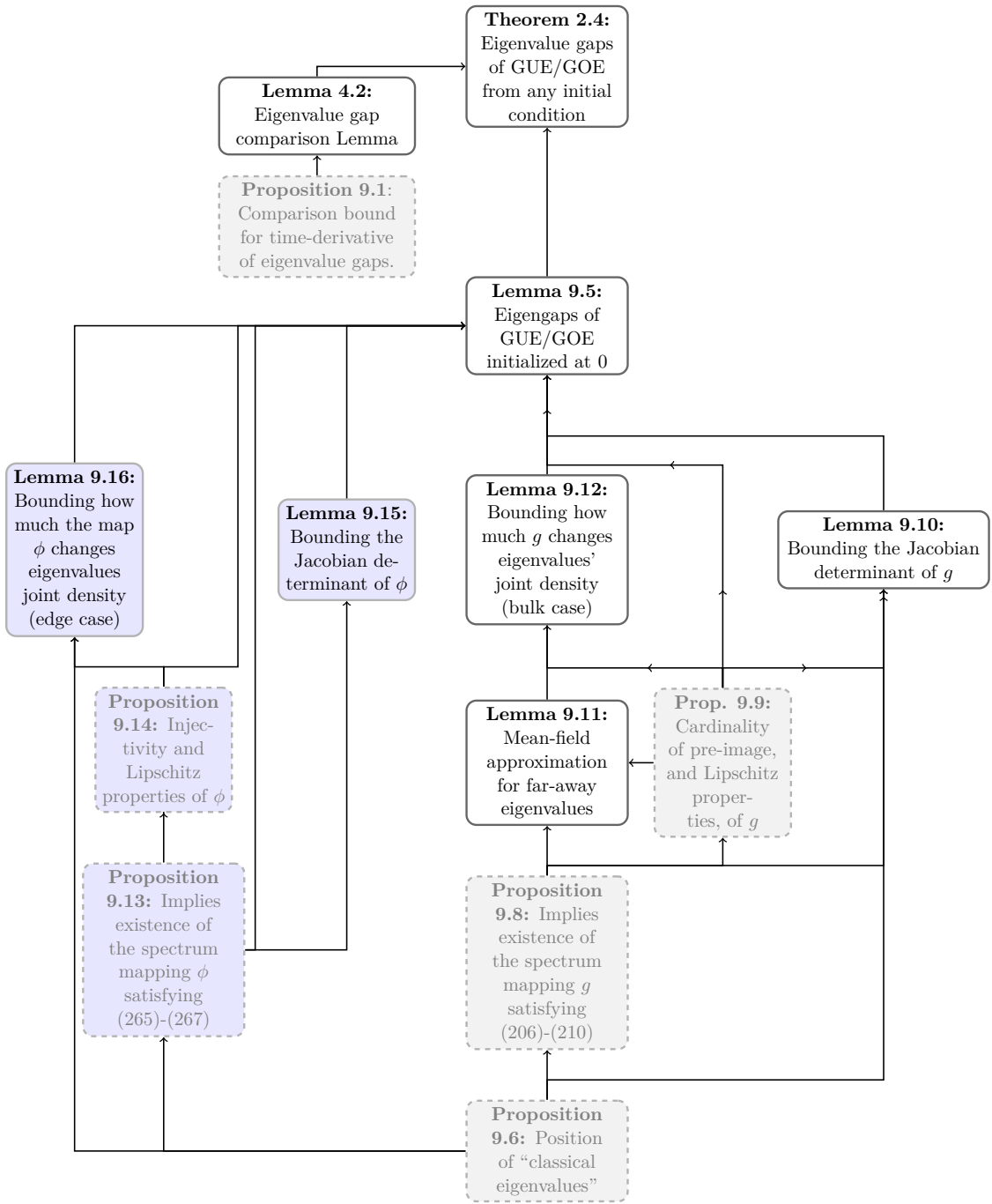

\color{black}

\section{Complex Gaussian perturbations: Proof of Theorem \ref{thm_rank_k_covariance_approximation_new}}\label{sec_utility}

\subsection{Defining the stochastic process on the space of rank-$k$ matrices}
Recall that, to bound the Frobenius norm utility in Theorem \ref{thm_rank_k_covariance_approximation_new}, we view the addition of Gaussian noise as a continuous-time Hermitian-matrix valued diffusion 
\begin{equation}\label{eq_n95}
\Phi(t) = M + B(t),
\end{equation}
whose eigenvalues $\gamma_i(t)$ and eigenvectors $u_i(t)$, $i\in [d]$, evolve over time.
Here, $B(t) = W(t) + W(t)^\ast$, where $W(t)$ is a $d \times d$ matrix where the real part (and complex part) of each entry is an independent standard Brownian motion with distribution $N(0, tI_d)$ at time $t$.
We will use the evolution equations \eqref{eq_DBM_eigenvectors} for the eigenvectors $u_i(t)$ to track the utility over time.

Towards this end, 
%
%
 at every time $t \geq 0$, let
 \begin{equation}\label{eq_n92}
     \Phi(t) = U(t) \Gamma(t) U(t)^\ast
 \end{equation}
  be a spectral decomposition of the symmetric matrix $\Phi(t)$,  where $\Gamma(t)$ is a diagonal matrix with diagonal entries $\gamma_1(t) \geq \cdots \geq \gamma_d(t)$ that are the eigenvalues of $\Phi(t)$, and $U(t) = [u_1(t), \ldots, u_d(t)]$ is a $d\times d$ unitary matrix whose columns $u_1(t), \ldots, u_d(t)$ are an orthonormal basis of eigenvectors of $\Phi(t)$.

To track the utility of the rank-$k$ approximation in Theorem \ref{thm_rank_k_covariance_approximation_new}, we define a rank-$k$ matrix-valued stochastic process $\Psi(t)$.
%
%
 At every time $t \geq 0$, define $\Psi(t)$ to be the symmetric matrix with any eigenvalues $\lambda_1(t) \geq \cdots \geq \lambda_d(t)$, where $\Lambda(t) :=\mathrm{diag}(\lambda_1(t), \ldots, \lambda_d(t))$,   and with eigenvectors given by the columns of $U(t)$: 
 \begin{equation}\label{eq_n96}
       \Psi(t):= U(t) \Lambda(t) U(t)^\ast \qquad\forall t \in [0,T].
    \end{equation} 
In the following proof, for all $t \geq 0$, we fix  
\begin{equation}\label{eq_n45}
    \lambda_i(t) = \begin{cases}
        \gamma_i(t) \qquad \textrm{ for } i\leq k,\\
        0 \qquad \quad \, \, \, \textrm{ for } i>k.%
    \end{cases}
\end{equation}

\subsection{Preliminary results}\label{sec_preliminary_results}

Before we begin the main part of the proof of Theorem \ref{thm_rank_k_covariance_approximation_new} (Section \ref{sec_proof_of_rank_k_covariance_approximation_new}), we first provide preliminary results which bound the Frobenius norm utility in the rare ``bad'' event when one or more eigenvalue gaps are unusually small.
In Lemma \ref{lemma_spectral_martingale}, we show that as a consequence of Theorem \ref{thm:eigenvalue_gap} (or, more specifically, Lemma  \ref{lemma_bad_event} which we will derive from  Theorem \ref{thm:eigenvalue_gap}), this ``bad'' event happens with very low probability.

Towards this end, for every $\alpha>0$, define the ``bad'' event $\hat{E}_\alpha$ as follows 
\begin{eqnarray}\label{eq_n77}
\hat{E}_\alpha &:=& \left\{\sup_{t \in [0,T]}\|B(t)\|_2 > 4\sqrt{T}(\sqrt{d} + \alpha) \right\} \cup \left\{\sup_{t \in [0,t_0]}\|B(t)\|_2 > 4\sqrt{t_0}(\sqrt{d} + \alpha) \right\} \nonumber\\
& & \qquad \qquad \qquad \cup \left\{\inf_{t_0 \leq t \leq T, 1\leq i < d  }\gamma_i(t) - \gamma_{i+1}(t) \leq \frac{1}{d^{10}} \frac{\sqrt{t}}{\mathfrak{b}\sqrt{d}} \right\}. 
\end{eqnarray}
In the following, we set 
\begin{equation}\label{eq_alpha}
\alpha = 20\log^{\frac{1}{2}}(\sigma_1 d (T+1))
\end{equation}
and
\begin{equation}\label{eq_t0}
 t_0 = \frac{1}{(kd)^{10} + k\alpha^2 +\sigma_1^2}.
 \end{equation}
The following lemma shows that  $\hat{E}_\alpha$ occurs with very low probability:
\begin{lemma}[\bf Probability of ``bad'' event occurring] \label{lemma_spectral_martingale}
 For every $T>0$ and every $\alpha \geq 20\log^{\frac{1}{2}}(\sigma_1 d (T+1))$, we have,
      $ \mathbb{P}\left(\hat{E}_\alpha\right) \leq 4\sqrt{\pi} e^{-\frac{1}{8}\alpha^2}  + \frac{T}{d^{600}}.$
\end{lemma}
\begin{proof}
\begin{eqnarray*}
    \mathbb{P}\left(\hat{E}_\alpha\right) \!\!\!\!\!\! &\stackrel{\textrm{Eq. } \eqref{eq_n77}}{\leq} & \!\!\!\! \mathbb{P}\left(\sup_{t \in [0,T]}\|B(t)\|_2 > 4\sqrt{T}(\sqrt{d} + \alpha)\right) +  \mathbb{P}\left(\sup_{t \in [0,t_0]}\|B(t)\|_2 > 4\sqrt{t_0}(\sqrt{d} + \alpha)\right)\\
    & & + \quad \mathbb{P}\left(\inf_{t_0 \leq t \leq T, 1\leq i < d  }\gamma_i(t) - \gamma_{i+1}(t) \leq \frac{1}{d^{10}} \frac{\sqrt{t}}{\mathfrak{b}\sqrt{d}}\right)\\
     &\stackrel{\textrm{Lem. } \ref{lemma_bad_event}}{\leq}  & \!\!\!\! \mathbb{P}\left(\sup_{t \in [0,T]}\|B(t)\|_2 > 4\sqrt{T}(\sqrt{d} + \alpha)\right) +  \mathbb{P}\left(\sup_{t \in [0,t_0]}\|B(t)\|_2 > 4\sqrt{t_0}(\sqrt{d} + \alpha)\right)  + \frac{T}{d^{600}}\\
    & \stackrel{\textrm{ Lem. \ref{lemma_spectral_martingale_b}}}{\leq} & 4\sqrt{\pi} e^{-\frac{1}{8}\alpha^2}  + \frac{T}{d^{600}}.
\end{eqnarray*}
\end{proof}
\noindent
The following lemma bounds the amount which the ``bad'' event $\hat{E}_\alpha$ contributes to the expected utility, and reduces the problem of proving Theorem \ref{thm_rank_k_covariance_approximation_new} to the problem of bounding the expected utility when the ``bad'' event $\hat{E}_\alpha$ does not occur.
%

\begin{lemma}\label{lemma_utility_rare_event}
If $\alpha \geq 20\log^{\frac{1}{2}}(\sigma_1 d (T+1))$, then we have
\begin{equation*}
    \mathbb{E}[\|\Psi(T) - \Psi(0)\|_F^2] \leq 4\mathbb{E}[\|\Psi(T) - \Psi(0)\|_F^2 \times \mathbbm{1}_{\hat{E}_\alpha^c}] +dT.
    \end{equation*}
\end{lemma}
\begin{proof}

\begin{equation}\label{eq_u12}
    \mathbb{E}[\|\Psi(T) - \Psi(0)\|_F^2] \leq 4\mathbb{E}[\|\Psi(T) - \Psi(0)\|_F^2 \times \mathbbm{1}_{\hat{E}_\alpha^c}] + 4\mathbb{E}[\|\Psi(T) - \Psi(0)\|_F^2 \times \mathbbm{1}_{\hat{E}_\alpha}]. 
\end{equation}

\begin{eqnarray}
    \|\Psi(T) - \Psi(0)\|_F &\stackrel{\textrm{Eq. } \eqref{eq_n96},\, \, \eqref{eq_n45}}{=}& \| U(T) \Gamma_k(T) U(T)^\ast - U(0) \Gamma_k(0) U(0)^\ast\|_F \nonumber\\
&\stackrel{\textrm{Tri. Ineq.}}{\leq}  & \| U(T) \Gamma_k(T) U(T)^\ast\|_F + \|U(0) \Gamma_k(0) U(0)^\ast\|_F \nonumber\\
  &= & \sqrt{\sum_{i=1}^k \gamma_i^2(T)} + \sqrt{\sum_{i=1}^k \gamma_i^2(0)} \label{eq_n93}\\
  & \leq & \sqrt{\sum_{i=1}^d \gamma_i^2(T)} + \sqrt{\sum_{i=1}^d \gamma_i^2(0)} \nonumber\\
    &=& \| U(T) \Gamma(T) U(T)^\ast\|_F + \|U(0) \Gamma(0) U(0)^\ast\|_F \label{eq_n90} \label{eq_n94}\\
     &\stackrel{\textrm{Eq. } \eqref{eq_n92}}{=}&  \|\Phi(T)\|_F+ \|\Phi(0)\|_F \nonumber\\
    &\stackrel{\textrm{Eq. } \eqref{eq_n95}}{=}& \|M + B(T)\|_F + \|M\|_F \nonumber\\
    &\stackrel{\textrm{Tri. Ineq.}}{\leq}  & 2\|M\|_F +\|B(T)\|_F, \label{eq_n91}
\end{eqnarray}
where \eqref{eq_n93} and \eqref{eq_n94} hold since, for any $t\geq 0$, $\Gamma(t)$ is defined in \eqref{eq_n92} as a diagonal matrix with diagonal entries $\gamma_1(t) \geq \cdots \geq \gamma_d(t)$, and since the squared Frobenius norm of any Hermitian matrix is equal to the sum of squares of its eigenvalues.

Therefore,
\begin{eqnarray}
    \mathbb{E}[\|\Psi(T) - \Psi(0)\|_F^2 \times \mathbbm{1}_{\hat{E}_\alpha}] 
    &\stackrel{\textrm{Eq. } \eqref{eq_n91}}{\leq} & \mathbb{E}[ (2\|M\|_F +\|B(T)\|_F)^2\times \mathbbm{1}_{\hat{E}_\alpha}]\nonumber\\
    & = & \mathbb{E}[ (4\|M\|_F^2 +4 \|M\|_F \|B(T)\|_F + \|B(T)\|_F^2)\times \mathbbm{1}_{\hat{E}_\alpha}]\nonumber\\
        &\leq & \mathbb{E}[ (8\|M\|_F^2 + 4\|B(T)\|_F^2)\times \mathbbm{1}_{\hat{E}_\alpha}] \label{eq_n97}\\
    &= & 8\|M\|_F^2\times \mathbb{P}(\hat{E}_\alpha)+ 4\mathbb{E}[\|B(T)\|_F^2\times \mathbbm{1}_{\hat{E}_\alpha}] \nonumber\\
      &\leq & 8\|M\|_F^2\times \mathbb{P}(\hat{E}_\alpha)+ 4\sqrt{d}\mathbb{E}[\|B(T)\|_2^2\times \mathbbm{1}_{\hat{E}_\alpha}]\nonumber\\
      &\stackrel{\textrm{Prop. } \ref{lemma_layer_cake}}{=}&  8\|M\|_F^2\times \mathbb{P}(\hat{E}_\alpha)+ 4\sqrt{d} \int_{16T(\sqrt{d} + \alpha)^2}^\infty\mathbb{P}(\|B(T)\|_2^2 >s)\mathrm{d}s \nonumber\\
            &=& 8\|M\|_F^2\times \mathbb{P}(\hat{E}_\alpha)+ 4\sqrt{d} \int_{16T(\sqrt{d} + \alpha)^2}^\infty\mathbb{P}(T\|W\|_2^2 >s)\mathrm{d}s \label{eq_n98}\\
             &=& 8\|M\|_F^2\times \mathbb{P}(\hat{E}_\alpha)+ 4\sqrt{d} \int_{16T(\sqrt{d} + \alpha)^2}^\infty\mathbb{P}\left(\|W\|_2 >\frac{\sqrt{s}}{\sqrt{T}}\right)\mathrm{d}s \nonumber\\
                      &\stackrel{\textrm{Lem. \ref{lemma_concentration}}}{\leq} & 8\|M\|_F^2\times \mathbb{P}(\hat{E}_\alpha)+ 4\sqrt{d} \int_{16T(\sqrt{d} + \alpha)^2}^\infty 2e^{-\left(\frac{\sqrt{s}}{\sqrt{T}} -2\sqrt{d}\right)^2}\mathrm{d}s \nonumber\\
                       &=& 8\|M\|_F^2\times \mathbb{P}(\hat{E}_\alpha)+ 4\sqrt{d} \int_{16T(\sqrt{d} + \alpha)^2}^\infty 2e^{-\left(\frac{s}{T} -2\sqrt{d}\frac{\sqrt{s}}{\sqrt{T}} +4d\right)}\mathrm{d}s\nonumber\\
                            &\leq &  8\|M\|_F^2\times \mathbb{P}(\hat{E}_\alpha)+ 4\sqrt{d} \int_{16T(\sqrt{d} + \alpha)^2}^\infty 2e^{-\left(\frac{s}{2T} +4d\right)}\mathrm{d}s \label{eq_n99}\\
                              &= & 8\|M\|_F^2\times \mathbb{P}(\hat{E}_\alpha)+ 4\sqrt{d} e^{-4d} \int_{16T(\sqrt{d} + \alpha)^2}^\infty 2e^{-\frac{s}{2T}}\mathrm{d}s\nonumber\\
                                 &= & 8\|M\|_F^2\times \mathbb{P}(\hat{E}_\alpha)+ 4\sqrt{d} e^{-4d} 4T e^{-\frac{16T(\sqrt{d} + \alpha)^2}{2T}}\nonumber\\
                               &= & 8\|M\|_F^2\times \mathbb{P}(\hat{E}_\alpha)+ 4\sqrt{d} e^{-4d} 4T e^{-8(\sqrt{d} + \alpha)^2}\nonumber\\
                               &\leq & 8\|M\|_F^2\times \mathbb{P}(\hat{E}_\alpha)+ T \label{eq_n100}\\ 
                               &\stackrel{\textrm{Lem. \ref{lemma_spectral_martingale}}}{\leq} &  8\|M\|_F^2\times \left(4\sqrt{\pi} e^{-\frac{1}{8}\alpha^2}+ \frac{T}{d^{600}}\right)+ T\nonumber\\
                            &\leq & 8 d \sigma_1^2 \times \left(4\sqrt{\pi} e^{-\frac{1}{8}\alpha^2}+ \frac{T}{d^{600}}\right) +T \nonumber\\
                            &\leq & \frac{1}{8}d T + \frac{T}{d^{200}} +T, \label{eq_n101}\\
                            &\leq & \frac{1}{4}d T \label{eq_u13}
\end{eqnarray}
where $W$ is a matrix with i.i.d. $N(0,1)$ entries, and $Y\sim(0,\frac{1}{2})$.
{\eqref{eq_n97} holds since for any $a,b \geq 0$ we have that either $ab \leq a^2$ or $ab \leq b^2$.
\eqref{eq_n98} holds since the random matrix $B(T)$ is equal in distribution to $\sqrt{T}W$. \eqref{eq_n99} holds since $\frac{s}{2T} \geq 2 \sqrt{d} \frac{\sqrt{s}}{\sqrt{T}}$ for every $s \in [16T(\sqrt{d} + \alpha)^2, \infty)$. \eqref{eq_n100} holds since $\alpha \geq 20\log^{\frac{1}{2}}(\sigma_1 d (T+1))$. }
 \eqref{eq_n101} holds because $\alpha \geq 20\log^{\frac{1}{2}}(\sigma_1 d (T+1))$ and $\sigma_1^2 \leq d^{100}$.

Plugging \eqref{eq_u13} into \eqref{eq_u12} completes the proof.
\end{proof}
\noindent
The following proposition will be useful in bounding the eigenvalue gaps $\gamma_i(t) - \gamma_j(t)$ for $i\leq k < j$.

\begin{proposition}[\bf ``Worst-case'' eigenvalue gap bound]\label{lemma_gap_concentration}
Whenever  $\gamma_i(0) - \gamma_{i+1}(0) \geq 8 \sqrt{T} \sqrt{d}$ for every $i \in S$ and $T>0$ and some subset $S \subset [d-1]$, we have that for any $\alpha>0$, 
\begin{equation*}
    \bigcup_{i\in S} \left\{\inf_{t \in [0,T]} \gamma_i(t) - \gamma_{i+1}(t) <  \frac{1}{2}\left((\gamma_i(0) - \gamma_{i+1}(0)) - \alpha\right) \right\} \subseteq \hat{E}_\alpha.
    \end{equation*}
\end{proposition}
\begin{proof}
Since, at every time $t$, $\Phi(t)= M + B(t)$ and  $\gamma_1(t) \geq \cdots \geq \gamma_d(t)$ are the eigenvalues of $\Phi(t)$, Weyl's Inequality (Lemma \ref{lemma_weyl}) implies that
\begin{equation}\label{eq_gap1b}
    \gamma_i(t) - \gamma_{i+1}(t) \geq \gamma_i(0) - \gamma_{i+1}(0) - \|B(t)\|_2, \qquad \forall t\in[0,T], i \in [d].
\end{equation}
 {and, hence, that} 
 \begin{eqnarray} \label{eq_gap1}
   \inf_{t \in [0,T]} \gamma_i(t) - \gamma_{i+1}(t)  &\stackrel{\textrm{Eq. \eqref{eq_gap1b}}}{\geq} & \inf_{t \in [0,T]}\left(\gamma_i(0) - \gamma_{i+1}(0) - \|B(t)\|_2\right) \nonumber\\
   & = &  \gamma_i(0) - \gamma_{i+1}(0) -  \sup_{t \in [0,T]} \|B(t)\|_2.
\end{eqnarray}
Therefore, \eqref{eq_gap1} implies that 

    \begin{eqnarray}
& & \!\!\!\!\!\!\!\!\!\!\!\!\!\! \!\!\! \!\!\!\!\!\!\!\!\!\!\!\!\!\!\!\!\!\!  \!\!\! \!\!\!\!\!\!\!\!\!\!\!\!\!\!\!\!\!\!  \bigcup_{i\in S} \left\{\inf_{t \in [0,T]} \gamma_i(t) - \gamma_{i+1}(t) <  \frac{1}{2}\left((\gamma_i(0) - \gamma_{i+1}(0)) - \alpha\right) \right\} \nonumber \\    
            &\stackrel{\textrm{Eq. \eqref{eq_gap1}}}{\subseteq} & \bigcup_{i\in S} \left\{ \gamma_i(0) - \gamma_{i+1}(0) - \sup_{t \in [0,T]} \|B(t)\|_2  <  \frac{1}{2}((\gamma_i(0) - \gamma_{i+1}(0)) - \alpha) \right\}  \nonumber\\
         &=&  \bigcup_{i\in S} \left\{ \sup_{t \in [0,T]} \|B(t)\|_2  >  \frac{1}{2}((\gamma_i(0) - \gamma_{i+1}(0)) + \alpha) \right\} \nonumber\\
                  &\subseteq &  \bigcup_{i\in S} \left\{ \sup_{t \in [0,T]} \|B(t)\|_2  >   2\sqrt{T} \sqrt{d} + \frac{1}{2} \alpha \right\} \label{eq_n160}\\
                  &=& \left\{\sup_{t \in [0,T]} \|B(t)\|_2  >   2\sqrt{T} \sqrt{d} + \frac{1}{2} \alpha \right\}\nonumber\\ 
                  & =& \hat{E}_\alpha, \nonumber
    \end{eqnarray}
where \eqref{eq_n160} holds since the statement of Proposition \ref{lemma_gap_concentration} assumes that $\gamma_i(0) - \gamma_{i+1}(0) \geq 8 \sqrt{T} \sqrt{d}$.
\end{proof}
\noindent
The following proposition provides a crude bound on the Frobenius distance over the very short time interval $[0,t_0]$, which we will use to ``jump-start'' our more sophisticated bound on the much longer interval $[t_0, T]$:
\begin{proposition}\label{lemma_t0}
Suppose that  $\sigma_k - \sigma_{k+1} \geq 4T \sqrt{d} + 2\alpha$.
Then for every $0 \leq t_0 < 1$ we have
\begin{equation*}
    \|\Psi(t_0) - \Psi(0)\|_F \times \mathbbm{1}_{\hat{E}_\alpha^c} \leq \sqrt{t_0}\left(2 \sqrt{k} (\sqrt{d} + \alpha)  + 8\sigma_1\right)
\end{equation*}
with probability $1$.
\end{proposition}

\begin{proof}
At every time $t\geq 0$, let $U_k(t)$ denote the $d\times k$ matrix consisting of the first $k$ columns of $U(t)$.
Further, let $\Gamma_k(t)$ denote the $k\times k$ matrix consisting of the first $k$ rows and columns of $\Gamma(t)$.

\begin{eqnarray} \label{eq_u4}
    \|\Psi(t_0) - \Psi(0)\|_F &=& \|U_k(t_0) \Gamma_k(t_0) U_k(t_0)^\ast - U_k(0) \Gamma_k(0) U_k(0)^\ast\|_F \nonumber\\
        &\stackrel{\textrm{Tri.\ Ineq.}}{\leq} & \|U_k(t_0) \Gamma_k(t_0) U_k(t_0)^\ast -  U_k(t_0) \Gamma_k(0) U_k(t_0)^\ast\|_F \nonumber \\
        & & + \quad \|U_k(t_0) \Gamma_k(0) U_k(t_0)^\ast - U_k(t_0) \Gamma_k(0) U_k(0)^\ast\|_F \nonumber\\
        & & + \quad \|U_k(t_0) \Gamma_k(0) U_k(0)^\ast - U_k(0) \Gamma_k(0) U_k(0)^\ast \|_F \nonumber\\
                &\leq & \|U_k(t_0)\|_2^2 \times  \|\Gamma_k(t_0) -\Gamma_k(0)\|_F \nonumber\\
                & & + \quad \|U_k(t_0)\|_2 \times \|\Gamma_k(0)\|_2 \times \| U_k(t_0)^\ast - U_k(0)^\ast\|_F \nonumber\\
        & & + \quad \|U_k(t_0) -U_k(0)\|_F \times \|\Gamma_k(0)\|_2 \times \|U_k(0)^\ast \|_2 \label{eq_n141}\\
        &=& \|\Gamma_k(t_0) -\Gamma_k(0)\|_F  +  \sigma_1  \| U_k(t_0)^\ast - U_k(0)^\ast\|_F \nonumber\\
        & & + \sigma_1 \|U_k(t_0) -U_k(0)\|_F \label{eq_n140}\\
                &=& \|\Gamma_k(t_0) -\Gamma_k(0)\|_F  + 2\sigma_1 \|U_k(t_0) -U_k(0)\|_F,
        \end{eqnarray}
     where \eqref{eq_n141} holds since $\|AB\|_F \leq \|A\|_2 \|B\|_F$ for any two matrices $A,B$. 
        \eqref{eq_n140} holds since $\|U_k(t)\|_2 = 1$ for all $t\geq 0$, and since $\|\Gamma_k(0)\|_2 = \sigma_1$ since $\Gamma_k(0) = M$.

By Lemma \ref{lemma_SinTheta}, we have
\begin{eqnarray}\label{eq_u5}
 \|U_k(t_0)U_k^\ast(t_0) -U_k(0)U_k^\ast(0)\|_F  &\stackrel{\textrm{Lem. \ref{lemma_SinTheta}}}{\leq}  &\frac{\|\Phi(t_0) - \Phi(0)\|_F}{\gamma_k(0) - \gamma_{k+1}(t_0)}\nonumber\\
 &=&  \frac{\|B(t_0)\|_F}{\gamma_k(0) - \gamma_{k+1}(t_0)}.
\end{eqnarray}
By Weyl's Inequality (Lemma \ref{lemma_weyl}), we have that, whenever the event  $\hat{E}^c_\alpha$ occurs,
\begin{eqnarray}\label{eq_u6}
\gamma_{k+1}(t_0) &\stackrel{\textrm{Lem. \ref{lemma_weyl}}}{\leq} & \gamma_{k+1}(0) + \|B(t_0)\|_2\nonumber\\
&\leq & \gamma_{k+1}(0) + 2\sqrt{t_0}(\sqrt{d} + \alpha) \label{eq_n142}\\
&= & \sigma_{k+1} + 2\sqrt{t_0}(\sqrt{d} + \alpha),
\end{eqnarray}
where \eqref{eq_n142} is by the definition of the event $\hat{E}^c_\alpha$ in \eqref{eq_n77}.
Thus, \eqref{eq_u6} implies that
\begin{eqnarray}
\gamma_k(0)-\gamma_{k+1}(t_0) &\stackrel{\textrm{Eq. \ref{eq_u6}}}{\geq} & \gamma_k(0)- \sigma_{k+1} - 2\sqrt{t_0}(\sqrt{d} + \alpha)\nonumber\\
&= & \sigma_k - \sigma_{k+1} - 2\sqrt{t_0}(\sqrt{d} + \alpha)\nonumber\\
&\geq & \frac{1}{2}(\sigma_k - \sigma_{k+1}), \label{eq_u7}
\end{eqnarray}
where \eqref{eq_u7} holds because $\sigma_k - \sigma_{k+1} \geq 4T \sqrt{d} + 2\alpha$ and $T \geq 1> t_0$.
Thus, plugging \eqref{eq_u7} into \eqref{eq_u5}, we have that whenever the event  $\hat{E}^c_\alpha$ occurs,
\begin{eqnarray}
 \|U_k(t_0)U_k^\ast(t_0) -U_k(0)U_k^\ast(0)\|_F  &\leq  & \frac{2\|B(t_0)\|_F}{\sigma_k - \sigma_{k+1}}\\
 &\leq & \frac{4\sqrt{t_0}(\sqrt{d} + \alpha)}{\sigma_k - \sigma_{k+1}}, \label{eq_u8}
\end{eqnarray}
where \eqref{eq_u8} is by the definition of the event $\hat{E}^c_\alpha$ in \eqref{eq_n77}.
We also have (by, e.g., Inequality (27) in \cite{mangoubi2022private}) that
\begin{equation}\label{eq_u9}
    \|U_k(t_0) -U_k(0)\|_F \leq \|U_k(t_0)U_k^\ast(t_0) -U_k(0)U_k^\ast(0)\|_F.
\end{equation}
Therefore, plugging in \eqref{eq_u9} into \eqref{eq_u8}, we get that, whenever the event  $\hat{E}^c_\alpha$ occurs,
\begin{equation}\label{eq_u10}
    \|U_k(t_0) -U_k(0)\|_F \leq \frac{4\sqrt{t_0}(\sqrt{d} + \alpha)}{\sigma_k - \sigma_{k+1}}.
\end{equation}
Plugging in \eqref{eq_u10} into \eqref{eq_u4} we get
\begin{eqnarray}
    \|\Psi(t_0) - \Psi(0)\|_F \times \mathbbm{1}_{\hat{E}_\alpha^c} \!\!\!\!\!\!\!\!\!\!\!\!\!\!\! &\stackrel{\textrm{Eq. \eqref{eq_u4}}}{\leq}  & \!\!\!\!\!\!\!\!\!\!\!\!\!\!\! \|\Gamma_k(t_0) -\Gamma_k(0)\|_F \times\mathbbm{1}_{\hat{E}_\alpha^c}  + 2\sigma_1 \|U_k(t_0) -U_k(0)\|_F\times \mathbbm{1}_{\hat{E}_\alpha^c} \nonumber\\
    &\stackrel{\textrm{Eq. \eqref{eq_u10}}}{\leq} &\|\Gamma_k(t_0) -\Gamma_k(0)\|_F \times\mathbbm{1}_{\hat{E}_\alpha^c}  + 2\sigma_1 \frac{4\sqrt{t_0}(\sqrt{d} + \alpha)}{\sigma_k - \sigma_{k+1}} \nonumber\\
    &\leq  & \sqrt{k}\|\Gamma_k(t_0) -\Gamma_k(0)\|_2 \times\mathbbm{1}_{\hat{E}_\alpha^c}  + 2\sigma_1 \frac{4\sqrt{t_0}(\sqrt{d} + \alpha)}{\sigma_k - \sigma_{k+1}} \nonumber\\
    &\stackrel{\textrm{Weyl's Ineq. } (\textrm{Lem. \ref{lemma_weyl})}}{\leq} &  \sqrt{k} \times \|B(t_0)\|_2 \times\mathbbm{1}_{\hat{E}_\alpha^c}  + 2\sigma_1 \frac{4\sqrt{t_0}(\sqrt{d} + \alpha)}{\sigma_k - \sigma_{k+1}} \nonumber\\
                &\leq & \sqrt{k} \times  2\sqrt{t_0}(\sqrt{d} + \alpha)  + 2\sigma_1 \frac{4\sqrt{t_0}(\sqrt{d} + \alpha)}{\sigma_k - \sigma_{k+1}} \label{eq_n177}\\
                &\leq & \sqrt{k} \times  2\sqrt{t_0}(\sqrt{d} + \alpha)  + 2\sigma_1 \frac{4\sqrt{t_0}(\sqrt{d} + \alpha)}{\sqrt{d} + \alpha} \label{eq_n178}\\
                 &\leq  & 2\sqrt{t_0} \sqrt{k} (\sqrt{d} + \alpha)  + 8\sigma_1 \sqrt{t_0} \nonumber\\
                 &= &  \sqrt{t_0}(2 \sqrt{k} (\sqrt{d} + \alpha)  + 8\sigma_1), \label{eq_u11}
        \end{eqnarray}
          \eqref{eq_n177} is by the definition of the event $\hat{E}^c_\alpha$ in \eqref{eq_n77}, and \eqref{eq_n178} holds by our assumption that $\sigma_k - \sigma_{k+1} \geq  4T \sqrt{d} + 2\alpha$ and since $T\geq 1$.
\end{proof}

\subsection{Proof of Theorem \ref{thm_rank_k_covariance_approximation_new}} \label{sec_proof_of_rank_k_covariance_approximation_new}

In this section, we compute the stochastic (Ito) derivative of the rank-k stochastic process $\Psi(t)$ (which was defined in \eqref{eq_n96}).
We then apply Ito's lemma to express the utility $\|\Psi(T) -\Psi(0)\|_F$ a stochastic integral.
 This stochastic integral is a function of the eigenvalue gaps of Dyson Brownian motion,
 and we apply the eigenvalue gap bounds of Corollary \ref{lemma_gaps_any_start} of Theorem \ref{thm:eigenvalue_gap} to bound this integral.

The following lemma computes the stochastic derivative of the projection matrices $u_i(t) u_i^\ast(t)$ onto the eigenvectors $u_i(t)$ of $\Psi(t)$, which we will then use to compute the derivative of $\Psi(t) = \sum_{i=1}^d \lambda_i(t) u_i(t) u_i^\ast(t)$ in the proof of theorem \ref{thm_rank_k_covariance_approximation_new}.
%

\begin{lemma}[\bf  It\^o derivative $\mathrm{d} (u_i(t) u_i^\ast(t))$]\label{Lemma_projection_differntial}
For all $t \in [0,T]$, 
\begin{eqnarray*}
    \mathrm{d}(u_i(t) u_i^\ast(t)) 
     &=& \sum_{j \neq i} \frac{1}{\gamma_i(t) - \gamma_j(t)}(u_i(t) u_j^\ast(t)\mathrm{d}B_{ij}(t) + u_j(t) u_i^\ast(t)\mathrm{d}B_{ij}^\ast(t)) \nonumber \\
    & &- \sum_{j\neq i} \frac{\mathrm{d}t}{(\gamma_i(t)- \gamma_j(t))^2} (u_i(t) u_i^\ast(t) - u_j(t)u_j^\ast(t)).
\end{eqnarray*}
\end{lemma}

\begin{proof}
To compute the stochastic Ito derivative $\mathrm{d}(u_i(t) u_i^\ast(t))$ we apply the Dyson Brownian motion equations \eqref{eq_DBM_eigenvectors}.
For any $t \in [0,T]$, we have
\begin{eqnarray}\label{eq_eq_derivative1}
  \mathrm{d}(u_i(t) u_i^\ast(t))
 & =& (u_i(t) + \mathrm{d}u_i(t))(u_i(t) + \mathrm{d}u_i(t))^\ast - u_i(t)u_i^\ast(t) \nonumber\\
  &\stackrel{\textrm{Eq. }  \eqref{eq_DBM_eigenvectors}}{=} &\left(u_i(t)+ \sum_{j \neq i} \frac{\mathrm{d}B_{ij}(t)}{\gamma_i(t) - \gamma_j(t)}u_j(t) - \sum_{j \neq i} \frac{\mathrm{d}t}{(\gamma_i(t)- \gamma_j(t))^2}u_i(t) \right) \nonumber\\
  & &  \times \quad \left(u_i(t) + \sum_{j \neq i} \frac{\mathrm{d}B_{ij}(t)}{\gamma_i(t) - \gamma_j(t)}u_j(t) - \sum_{j \neq i} \frac{\mathrm{d}t}{(\gamma_i(t)- \gamma_j(t))^2}u_i(t) \right)^\ast - u_i(t)u_i^\ast(t) \nonumber\\
  &= & u_i(t) u_i^\ast(t) + \sum_{j \neq i} \frac{1}{\gamma_i(t) - \gamma_j(t)}(u_i(t) u_j^\ast(t)\mathrm{d}B_{ij}(t) + u_j(t) u_i^\ast(t)\mathrm{d}B_{ij}^\ast(t))\nonumber\\
  &\qquad -&   \sum_{j\neq i} \frac{\mathrm{d}t}{(\gamma_i(t)- \gamma_j(t))^2} u_i(t) u_i^\ast(t) +  \sum_{j \neq i}  \sum_{\ell \neq i} \frac{\mathrm{d}B_{ij}(t)\mathrm{d}B_{i\ell}^\ast(t)}{(\gamma_i(t) - \gamma_j(t)) (\gamma_i(t) - \gamma_\ell(t))} u_j(t)u_{\ell}^\ast(t) \nonumber\\
  & \qquad -&  \varphi_1(t)\varphi_2^\ast(t)  -\varphi_2(t)\varphi_1^\ast(t) + -\varphi_2(t)\varphi_2^\ast(t) - u_i(t)u_i^\ast(t),
  \end{eqnarray}
  where we define $\varphi_1(t):= \sum_{j \neq i} \frac{\mathrm{d}B_{ij}(t)}{\gamma_i(t) - \gamma_j(t)}u_j(t)$ and $\varphi_2(t):=\sum_{j \neq i} \frac{\mathrm{d}t}{(\gamma_i(t)- \gamma_j(t))^2}u_i(t)$.
The terms $\varphi_1(t) \varphi_2^\ast(t)$ and $\varphi_2(t) \varphi_1^\ast(t)$ have differentials $O(\mathrm{d}B_{ij} \mathrm{d}t)$, and $\varphi_2(t) \varphi_2^\ast(t)$ has differentials $O(\mathrm{d}t^2)$; thus, all three terms vanish in the stochastic derivative by Ito's Lemma \ref{lemma_ito_lemma_new} (applied separately to the real and imaginary parts of these terms).
Therefore, \eqref{eq_eq_derivative1} implies that the stochastic derivative $\mathrm{d}(u_i(t) u_i^\ast(t))$ satisfies
  \begin{eqnarray}
  \mathrm{d}(u_i(t) u_i^\ast(t)) &\stackrel{\textrm{Eq. }  \eqref{eq_eq_derivative1}}{=}&  \sum_{j \neq i} \frac{1}{\gamma_i(t) - \gamma_j(t)}(u_i(t) u_j^\ast(t) \mathrm{d}B_{ij}(t) + u_j(t) u_i^\ast(t) \mathrm{d}B_{ij}^\ast(t))\nonumber \\
  & & \quad - \sum_{j\neq i} \frac{\mathrm{d}t}{(\gamma_i(t)- \gamma_j(t))^2} u_i(t) u_i^\ast(t)\nonumber\\
  & & \quad + \sum_{j \neq i}  \sum_{\ell \neq i} \frac{\mathrm{d}B_{ij}(t)\mathrm{d}B_{i\ell}^\ast(t)}{(\gamma_i(t) - \gamma_j(t)) (\gamma_i(t) - \gamma_\ell(t))} u_j(t)u_{\ell}^\ast(t)\nonumber\\
    &=&  \sum_{j \neq i} \frac{1}{\gamma_i(t) - \gamma_j(t)}(u_i(t) u_j^\ast(t)\mathrm{d}B_{ij}(t) + u_j(t) u_i^\ast(t)\mathrm{d}B_{ij}^\ast(t))\nonumber\\
  & & \quad - \sum_{j\neq i} \frac{\mathrm{d}t}{(\gamma_i(t)- \gamma_j(t))^2} u_i(t) u_i^\ast(t) + \sum_{j \neq i} \frac{\mathrm{d}B_{ij}(t) \mathrm{d}B_{ij}^\ast(t)}{(\gamma_i(t) - \gamma_j(t))^2} u_j(t)u_j^\ast(t) \label{eq_n146}\\
      &=&  \sum_{j \neq i} \frac{1}{\gamma_i(t) - \gamma_j(t)}(u_i(t) u_j^\ast(t)\mathrm{d}B_{ij}(t) + u_j(t) u_i^\ast(t)\mathrm{d}B_{ij}^\ast(t))\nonumber\\
  & & \quad - \sum_{j\neq i} \frac{\mathrm{d}t}{(\gamma_i(t)- \gamma_j(t))^2} u_i(t) u_i^\ast(t) + \sum_{j \neq i} \frac{\mathrm{d}t}{(\gamma_i(t) - \gamma_j(t))^2} u_j(t)u_j^\ast(t), \label{eq_eq_derivative2}
\end{eqnarray}
where \eqref{eq_n146} holds since all terms $\mathrm{d}B_{ij}(t)\mathrm{d}B_{i\ell}^\ast(t)$ with $j \neq \ell$ in the sum $$\sum_{j \neq i}  \sum_{\ell \neq i} \frac{\mathrm{d}B_{ij}(t)\mathrm{d}B_{i\ell}^\ast(t)}{(\gamma_i(t) - \gamma_j(t)) (\gamma_i(t) - \gamma_\ell(t))} u_j(t)u_{\ell}^\ast(t)$$  vanish by Ito's Lemma \ref{lemma_ito_lemma_new} since they have mean 0 and are $O(\mathrm{d}B_{ij}(t) \mathrm{d}B_{i\ell}^\ast(t))$; we are therefore left only with the terms $j = \ell$ in the sum which have differential terms $\mathrm{d}B_{ij}(t)\mathrm{d}B_{ij}^\ast(t)$ which have mean $\mathrm{d}t$ plus higher-order terms which vanish by Ito's Lemma \ref{lemma_ito_lemma_new}. 
Therefore  \eqref{eq_eq_derivative2} implies that
\begin{eqnarray*}
    \mathrm{d}(u_i(t) u_i^\ast(t)) 
     &=& \sum_{j \neq i} \frac{1}{\gamma_i(t) - \gamma_j(t)}(u_i(t) u_j^\ast(t)\mathrm{d}B_{ij}(t) + u_j(t) u_i^\ast(t)\mathrm{d}B_{ij}^\ast(t))\nonumber \\
    & & \quad - \sum_{j\neq i} \frac{\mathrm{d}t}{(\gamma_i(t)- \gamma_j(t))^2} (u_i(t) u_i^\ast(t) - u_j(t)u_j^\ast(t)).
\end{eqnarray*}
\end{proof}

\noindent
 In the proof of Theorem \ref{thm_rank_k_covariance_approximation_new}, we will show that
  $\mathrm{d} \Psi(t) = \sum_{i=1}^d \lambda_i(t) \mathrm{d}(u_i(t) u_i^\ast(t))) + (\mathrm{d}\lambda_i(t)) (u_i(t) u_i^\ast(t))$,
and use this expression to bound the stochastic integral $\| \Psi(T) - \Psi(t_0) \|_F^2 = \left \| \int_{t_0}^T \mathrm{d} \Psi(t) \right \|_F^2$.
Towards this end, we first apply Lemma \ref{Lemma_projection_differntial} to bound the component of this stochastic integral arising from the term $\sum_{i=1}^d \lambda_i(t) \mathrm{d}(u_i(t) u_i^\ast(t)))$ in the above expression for $\mathrm{d} \Psi(t)$.

\begin{lemma} \label{Lemma_integral}
For any $T\geq t_0 \geq 0$,

\begin{eqnarray*}
\mathbb{E}\left[\left\|\int_{t_0}^T \sum_{i=1}^d   \lambda_i(t) \mathrm{d}( u_i(t) u_i^\ast(t))\right\|_F^2 \times \mathbbm{1}_{\hat{E}_\alpha^c} \right] &\leq& 
 32\int_{t_0}^{T}   \mathbb{E}\left[ \sum_{i=1}^{d}  \sum_{j \neq i}  \frac{(\lambda_i(t) - \lambda_j(t))^2}{(\gamma_i(t) -  \gamma_j(t))^2} \times \mathbbm{1}_{\hat{E}_\alpha^c} \right]  \mathrm{d}t \nonumber\\ 
     & &  + \quad   2 T \int_{t_0}^{T}\mathbb{E}\left[\sum_{i=1}^{d}\left(\sum_{j\neq i} \frac{\lambda_i(t) - \lambda_j(t)}{(\gamma_i(t) -  \gamma_j(t))^2}\right)^2 \times \mathbbm{1}_{\hat{E}_\alpha^c}  \right]\mathrm{d}t.
\end{eqnarray*}

\end{lemma}

\begin{proof}

\begin{eqnarray}
 \sum_{i=1}^d   \lambda_i(t) \mathrm{d}( u_i(t) u_i^\ast(t)) &\stackrel{\textrm{Lem. } \ref{Lemma_projection_differntial}}{=}&
 \sum_{i=1}^d     \sum_{j \neq i} \frac{\lambda_i(t)}{\gamma_i(t) - \gamma_j(t)}(u_i(t) u_j^\ast(t)\mathrm{d}B_{ij}(t) + u_j(t) u_i^\ast(t)\mathrm{d}B_{ij}^\ast(t)) \nonumber\\
& & \qquad \qquad  - \sum_{i=1}^d   \sum_{j\neq i} \frac{ \lambda_i(t) }{(\gamma_i(t)- \gamma_j(t))^2} (u_i(t) u_i^\ast(t) - u_j(t)u_j^\ast(t)) \mathrm{d}t \nonumber \\
 &=& \frac{1}{2} \sum_{i=1}^d     \sum_{j \neq i} \frac{\lambda_i(t)- \lambda_j(t)}{\gamma_i(t) - \gamma_j(t)}(u_i(t) u_j^\ast(t)\mathrm{d}B_{ij}(t) + u_j(t) u_i^\ast(t)\mathrm{d}B_{ij}^\ast(t)) \nonumber\\
& &  - \sum_{i=1}^d   \sum_{j\neq i} \frac{ \lambda_i(t) - \lambda_j(t) }{(\gamma_i(t)- \gamma_j(t))^2} u_i(t) u_i^\ast(t) \mathrm{d}t  \label{eq_n102},
\end{eqnarray}
where \eqref{eq_n102} holds since $B_{ij}(t) = B_{ji}(t)^\ast$ for all $i,j$ and all $t \geq 0$ because $B(t) = W(t) + W(t)^\ast$ is Hermitian.

Therefore,
\begin{eqnarray} \label{eq_int_1}
& & \left\|\int_{t_0}^T \sum_{i=1}^d   \lambda_i(t) \mathrm{d}( u_i(t) u_i^\ast(t))\right\|_F^2 \nonumber\\ 
 &\stackrel{\textrm{Eq. } \eqref{eq_n102}}{=}& \bigg \| \frac{1}{2} \int_{t_0}^T     \sum_{i=1}^d     \sum_{j \neq i} \frac{\lambda_i(t)- \lambda_j(t)}{\gamma_i(t) - \gamma_j(t)}(u_i(t) u_j^\ast(t)\mathrm{d}B_{ij}(t) + u_j(t) u_i^\ast(t)\mathrm{d}B_{ij}^\ast(t))   \nonumber\\
& & \qquad \qquad  -  \int_{t_0}^T \sum_{i=1}^d   \sum_{j\neq i} \frac{ \lambda_i(t) - \lambda_j(t) }{(\gamma_i(t)- \gamma_j(t))^2} u_i(t) u_i^\ast(t) \mathrm{d}t   \bigg\|_F^2\nonumber \\
&\stackrel{\textrm{Tri. Ineq. }}{\leq}&  2 \left \| \int_{t_0}^T     \sum_{i=1}^d     \sum_{j \neq i} \frac{\lambda_i(t)- \lambda_j(t)}{\gamma_i(t) - \gamma_j(t)}(u_i(t) u_j^\ast(t)\mathrm{d}B_{ij}(t) + u_j(t) u_i^\ast(t)\mathrm{d}B_{ij}^\ast(t))   \right \|_F^2 \nonumber\\
& & \qquad \qquad  +  \left \|  \int_{t_0}^T \sum_{i=1}^d   \sum_{j\neq i} \frac{ \lambda_i(t) - \lambda_j(t) }{(\gamma_i(t)- \gamma_j(t))^2} u_i(t) u_i^\ast(t) \mathrm{d}t   \right\|_F^2.
\end{eqnarray}

\noindent
The first term on the r.h.s. of \eqref{eq_int_1} (inside its Frobenius norm) is a ``diffusion'' term--that is, the integral has mean $0$ and Brownian motion differentials $\mathrm{d}B_{ij}(t)$ inside the integral.
The second term on the r.h.s. (inside its Frobenius norm) is a ``drift'' term-- that is, the integral has non-zero mean and deterministic differentials $\mathrm{d}t$ inside the integral.
We bound the diffusion and drift terms separately.

\paragraph{Bounding the diffusion term.}

We first use It\^o's Lemma (Lemma \ref{lemma_ito_lemma_new}) to bound the diffusion term in \eqref{eq_int_1}.
The idea is to apply Ito's Lemma separately to the real and complex parts of the integrand.
Define \begin{equation}\label{eq_n3}
    X(t):=  \int_{t_0}^{t}\sum_{i=1}^{d} \sum_{j \neq i} \frac{\lambda_i(t)- \lambda_j(t)}{\gamma_i(t) - \gamma_j(t)}(u_i(s) u_j^\ast(s)\mathrm{d}B_{ij}(s) + u_j(s) u_i^\ast(s)\mathrm{d}B_{ij}^\ast(s))
    \end{equation}
    for all $t \geq 0.$
Then
\begin{equation}\label{eq_n2}
\mathrm{d}X_{\ell r}(t) =  \sum_{i=1}^{d} \sum_{j \neq i} R_{(\ell r) (i j)}(t) \mathrm{d}B_{(ij)}(t) + Q_{(\ell r) (i j)}(t) \mathrm{d}B_{(ij)}^\ast(t) \qquad \qquad \forall \ell,r \in[d],\,\,  t \geq 0,
\end{equation}
where for all $t \geq0$ we define 
\begin{eqnarray}\label{eq_n5}
R_{(\ell r) (i j)}(t) :=& \left(\frac{\lambda_i(t)- \lambda_j(t)}{\gamma_i(t) - \gamma_j(t)}u_i(t) u_j^\ast(t) \right)[\ell, r] & \quad \forall \ell, r, i, j \in [d], \textrm{ s.t. } i \neq j,\nonumber\\
Q_{(\ell r) (i j)}(t) :=& \left(\frac{\lambda_i(t)- \lambda_j(t)}{\gamma_i(t) - \gamma_j(t)} u_j(t) u_i^\ast(t) \right)[\ell, r] & \quad \forall \ell, r, i, j \in [d], \textrm{ s.t. } i \neq j,\nonumber\\
R_{(\ell r) (i j)}(t) =& Q_{(\ell r) (i j)}(t) := 0 & \quad \forall \ell, r, i, j \in [d], \textrm{ s.t. } i = j,
\end{eqnarray}
and where we denote by either $H_{\ell r}$ or $H[\ell, r]$  the $(\ell, r)$'th entry of any matrix $H$.
Thus, by separating \eqref{eq_n2} into real and imaginary components, we have
\begin{eqnarray}\label{eq_t1}
\mathrm{d}X_{\ell r}(t) &\stackrel{\textrm{Eq. } \eqref{eq_n2}}{=}& \sum_{i=1}^{d} \sum_{j \neq i} R_{(\ell r) (i j)}(t) \mathrm{d}B_{(ij)}(t) + Q_{(\ell r) (i j)}(t) \mathrm{d}B_{(ij)}^\ast(t) \qquad \qquad \forall t \geq 0,\nonumber\\
&=& \sum_{i=1}^{d} \sum_{j \neq i} [\mathcal{R}(R_{(\ell r) (i j)}(t)) + \mathfrak{i}  \mathcal{I}(R_{(\ell r) (i j)}(t))] \times  [\mathcal{R}(\mathrm{d}B_{(ij)}(t)) + \mathfrak{i} \mathcal{I}(\mathrm{d}B_{(ij)}(t))]\nonumber\\
& & + \quad  \sum_{i=1}^{d} \sum_{j \neq i} [\mathcal{R}(Q_{(\ell r) (i j)}(t)) + \mathfrak{i}  \mathcal{I}(Q_{(\ell r) (i j)}(t))]\times  [\mathcal{R}(\mathrm{d}B_{(ij)}^\ast(t)) + \mathfrak{i} \mathcal{I}(\mathrm{d}B_{(ij)}^\ast(t))]\nonumber\\
&=& \sum_{i=1}^{d} \sum_{j \neq i} \mathcal{R}(R_{(\ell r) (i j)}(t))\mathcal{R}(\mathrm{d}B_{(ij)}(t))  + \mathfrak{i}  \mathcal{I}(R_{(\ell r) (i j)}(t)) \mathcal{R}(\mathrm{d}B_{(ij)}(t))\nonumber\\
& & \quad + \mathfrak{i}\mathcal{R}(R_{(\ell r) (i j)}(t)) \mathcal{I}(\mathrm{d}B_{(ij)}(t)) - \mathcal{I}(R_{(\ell r) (i j)}(t)) \mathcal{I}(\mathrm{d}B_{(ij)}(t))\nonumber\\
& & + \quad \sum_{i=1}^{d} \sum_{j \neq i} \mathcal{R}(Q_{(\ell r) (i j)}(t))\mathcal{R}(\mathrm{d}B_{(ij)}^\ast(t))  + \mathfrak{i}  \mathcal{I}(Q_{(\ell r) (i j)}(t)) \mathcal{R}(\mathrm{d}B_{(ij)}^\ast(t))\nonumber\\
& & + \quad \mathfrak{i}\mathcal{R}(Q_{(\ell r) (i j)}(t)) \mathcal{I}(\mathrm{d}B_{(ij)}^\ast(t)) - \mathcal{I}(Q_{(\ell r) (i j)}(t)) \mathcal{I}(\mathrm{d}B_{(ij)}^\ast(t)).
\end{eqnarray}

\noindent
Our goal is to bound $\mathbb{E}[\|X(T) - X(t_0)\|_F^2]$. 
Towards this end, let $f: \mathbb{R}^{d \times d}: \rightarrow \mathbb{R}$ be the function which takes as input a $d \times d$ matrix and outputs the square of its Frobenius norm: $f(Y):= \|Y\|_F^2 = \sum_{i=1}^d \sum_{j=1}^d Y_{ij}^2$ for every $Y \in \mathbb{R}^{d\times d}$.
Then
\begin{equation}\label{eq_int_5}
\frac{\partial^2 }{\partial Y_{ij} \partial Y_{\alpha \beta}}f(Y) =
\begin{cases} 
      2 & \textrm{if }  \, \, \, (i,j)= (\alpha, \beta), \\
      0 & \textrm{otherwise}.
   \end{cases}
   \end{equation}

\noindent
Then, denoting by $e_{\ell, r}$ the matrix with $1$ in the $(\ell,r)$'th entry and zeros everywhere else, by \eqref{eq_t1} we have
\begin{eqnarray}\label{eq_t2}
    \|X(T) - X(t_0)  \|_F^2 &=&  \left\|\sum_{\ell, r} \int_{t_0}^T \mathrm{d}X_{\ell r}(t)e_{\ell, r}  \right\|_F^2 \nonumber\\
    &\stackrel{\textrm{Eq. }  \eqref{eq_t1}}{\leq}&  \left\| \int_{t_0}^T \sum_{\ell, r} \sum_{i=1}^{d} \sum_{j \neq i}\mathcal{R}(R_{(\ell r) (i j)}(t))\mathcal{R}(\mathrm{d}B_{(ij)}(t)) e_{\ell, r} \right\|_F^2  \nonumber\\
& & + \quad   \left\|\int_{t_0}^T \sum_{\ell, r} \sum_{i=1}^{d} \sum_{j \neq i} \mathcal{I}(R_{(\ell r) (i j)}(t)) \mathcal{R}(\mathrm{d}B_{(ij)}(t)) e_{\ell, r} \right\|_F^2\nonumber\\
& & + \quad \left\|\int_{t_0}^T \sum_{\ell, r} \sum_{i=1}^{d} \sum_{j \neq i} \mathcal{R}(R_{(\ell r) (i j)}(t)) \mathcal{I}(\mathrm{d}B_{(ij)}(t))e_{\ell, r} \right\|_F^2 \nonumber\\
& & + \quad \left\|\int_{t_0}^T \sum_{\ell, r} \sum_{i=1}^{d} \sum_{j \neq i} \mathcal{I}(R_{(\ell r) (i j)}(t)) \mathcal{I}(\mathrm{d}B_{(ij)}(t)) e_{\ell, r} \right\|_F^2\nonumber\\
& & + \quad  \left\|\int_{t_0}^T \sum_{\ell, r} \sum_{i=1}^{d} \sum_{j \neq i}\mathcal{R}(Q_{(\ell r) (i j)}(t))\mathcal{R}(\mathrm{d}B_{(ij)}^\ast(t))
e_{\ell, r} \right\|_F^2 \nonumber\\
& & + \quad \left\|\int_{t_0}^T \sum_{\ell, r} \sum_{i=1}^{d} \sum_{j \neq i} \mathcal{I}(Q_{(\ell r) (i j)}(t)) \mathcal{R}(\mathrm{d}B_{(ij)}^\ast(t))e_{\ell, r} \right\|_F^2\nonumber\\
& &  + \quad \left\|\int_{t_0}^T \sum_{\ell, r =1}^d \sum_{i=1}^{d} \sum_{j \neq i}\mathcal{R}(Q_{(\ell r) (i j)}(t)) \mathcal{I}(\mathrm{d}B_{(ij)}^\ast(t))e_{\ell, r} \right\|_F^2 \nonumber\\
& & + \quad \left\| \int_{t_0}^T \sum_{\ell, r} \sum_{i=1}^{d} \sum_{j \neq i}\mathcal{I}(Q_{(\ell r) (i j)}(t)) \mathcal{I}(\mathrm{d}B_{(ij)}^\ast(t))e_{\ell, r} \right\|_F^2.
\end{eqnarray}
Since all of the terms on the r.h.s. of \eqref{eq_t2} are entirely real or imaginary for all $t\geq 0$, we can apply  It\^o's Lemma (Lemma \ref{lemma_ito_lemma_new}) individually to each of these terms.
The proof to bound each of these eight terms is identical (if we replace $\mathcal{R}$ with $\mathcal{I}$, $R$ with $Q$, and/or $\mathrm{d}B_{(ij)}(t)$ with $\mathrm{d}B_{(ij)}^\ast(t)$)), since $\mathcal{R}(\mathrm{d}B_{(ij)}(t))$, $\mathcal{R}(\mathrm{d}B_{(ij)}^\ast(t)),$  $\mathcal{I}(\mathrm{d}B_{(ij)}(t))$, $\mathcal{I}(\mathrm{d}B_{(ij)}^\ast(t))$ are equal in distribution.
Thus, without loss of generality, we only present the proof of how to bound the term\\ $\left\| \int_{t_0}^T \sum_{\ell, r} \sum_{i=1}^{d} \sum_{j \neq i}\mathcal{R}(R_{(\ell r) (i j)}(t))\mathcal{R}(\mathrm{d}B_{(ij)}(t))e_{\ell, r}  \right\|_F^2$.
Towards this end, define
\begin{equation}\label{eq_n4}
    Y(t):=  \int_{t_0}^t \sum_{\ell, r} \sum_{i=1}^{d} \sum_{j \neq i}\mathcal{R}(R_{(\ell r) (i j)}(t))\mathcal{R}(\mathrm{d}B_{(ij)}(t)) e_{\ell, r}\qquad \forall t\geq 0.
    \end{equation}
 Then we have,
\begin{eqnarray}\label{eq_int_2b}
       & &  \mathbb{E}\left [\left\| \int_{t_0}^T \sum_{\ell, r} \sum_{i=1}^{d} \sum_{j \neq i}\mathcal{R}(R_{(\ell r) (i j)}(t))\mathcal{R}(\mathrm{d}B_{(ij)}(t)) e_{\ell, r} \right\|_F^2  \times \mathbbm{1}_{\hat{E}_\alpha^c}\right]  \nonumber\\
   & & \stackrel{\textrm{Eq. } \eqref{eq_n4}}{=} \mathbb{E}[(f(Y(T)) - f(Y(t_0)))  \times \mathbbm{1}_{\hat{E}_\alpha^c}] \nonumber\\
     & & \stackrel{\textrm{It\^o's Lem. (Lem. \ref{lemma_ito_lemma_new})}}{=} \mathbb{E}\left [\frac{1}{2} \int_{t_0}^t \sum_{\ell, r} \sum_{\alpha, \beta} \left(\frac{\partial}{ \partial Y_{\alpha \beta}} f(Y(t))\right) \mathcal{R}(R_{(\ell r) (\alpha \beta)}(t)) \mathcal{R}(\mathrm{d}B_{\ell r}(t))  \times \mathbbm{1}_{\hat{E}_\alpha^c} \right] \nonumber\\
     & & + \quad \mathbb{E}\left [\frac{1}{2} \int_{t_0}^t \sum_{\ell, r} \sum_{i,j} \sum_{\alpha, \beta} \left(\frac{\partial^2}{\partial Y_{ij} \partial Y_{\alpha \beta}} f(Y(t))\right) \mathcal{R}(R_{(\ell r) (i j)}(t))  \mathcal{R}(R_{(\ell r) (\alpha \beta)}(t)) \mathrm{d}t  \times \mathbbm{1}_{\hat{E}_\alpha^c} \right] \nonumber\\
        &=& 0 \, \, +  \, \, \mathbb{E}\left [\frac{1}{2} \int_{t_0}^t \sum_{\ell, r} \sum_{i,j} \sum_{\alpha, \beta} \left(\frac{\partial^2}{\partial Y_{ij} \partial Y_{\alpha \beta}} f(Y(t))\right) \mathcal{R}(R_{(\ell r) (i j)}(t)) \mathcal{R}(R_{(\ell r) (\alpha \beta)}(t)) \mathrm{d}t  \times \mathbbm{1}_{\hat{E}_\alpha^c}\right], \qquad
        \end{eqnarray}
        where \eqref{eq_int_2b} holds since $$\mathbb{E}\left[\int_{t_0}^T  \left(\frac{\partial}{ \partial Y_{\alpha \beta}} f(Y(t))\right) \mathcal{R}(R_{(\ell r) (\alpha \beta)}(t)) \mathcal{R}(\mathrm{d}B_{\ell r}(t))  \times \mathbbm{1}_{\hat{E}_\alpha^c} \right] = 0,$$
        for each $\ell, r, \alpha, \beta \in [d]$ because $\mathrm{d}B_{\ell r}(s)$ is independent of both $Y(t)$ and $R(t)$ for all $s \geq t$ and the Brownian motion increments $\mathrm{d}B_{\alpha \beta}(s)$ satisfy $\mathbb{E}[\int_t^{\tau} \mathrm{d}B_{\alpha \beta}(s)] = \mathbb{E}[B_{\alpha \beta}(\tau) - B_{\alpha \beta}(t)]= 0$ for any $\tau \geq t$.
        Thus, plugging \eqref{eq_int_5} into \eqref{eq_int_2b}, we have 
   \begin{eqnarray}
        & & \!\!\!\!\!\!\!\! \! \!\!\!\!\!\!\!\!\!\!\!\!\!\!\!\! \! \!\!\!\!\!\!\!\!\!\!\!\!\!\!\!\! \! \!\!\!\!\!\!\!\!\!\!\!\!\!\!\!\mathbb{E}\left [\left\| \int_{t_0}^T\sum_{\ell, r} \sum_{i=1}^{d} \sum_{j \neq i}\mathcal{R}(R_{(\ell r) (i j)}(t))\mathcal{R}(\mathrm{d}B_{(ij)}(t)) e_{\ell, r}\right\|_F^2 \times \mathbbm{1}_{\hat{E}_\alpha^c} \right]  \nonumber\\
   &\stackrel{\textrm{Eq. } \eqref{eq_int_5}, \eqref{eq_int_2b}}{=}& \mathbb{E}\left [\frac{1}{2} \int_{t_0}^t \sum_{\ell, r} \sum_{i=1}^{d} \sum_{j \neq i} 2  [\mathcal{R}(R_{(\ell r) (i j)}(t))]^2 \mathrm{d}t  \times \mathbbm{1}_{\hat{E}_\alpha^c} \right]\nonumber\\
                   &\stackrel{\textrm{Eq. } \eqref{eq_n5}}{=}& 
                   \mathbb{E}\left [ \int_{t_0}^t \sum_{\ell, r} \sum_{i=1}^{d} \sum_{j \neq i}\left(\left( \frac{\lambda_i(t)- \lambda_j(t)}{\gamma_i(t) - \gamma_j(t)} \mathcal{R}\left(u_i(t) u_j^\ast(t)\right)\right)[\ell, r]\right)^2 \mathrm{d}t  \times \mathbbm{1}_{\hat{E}_\alpha^c} \right]\nonumber\\
            &=& \mathbb{E}\left [ \int_{t_0}^t   \left \|\sum_{i=1}^{d} \sum_{j \neq i}\frac{\lambda_i(t)- \lambda_j(t)}{\gamma_i(t) - \gamma_j(t)} \mathcal{R}\left(u_i(t) u_j^\ast(t)\right)\right\|_F^2\mathrm{d}t  \times \mathbbm{1}_{\hat{E}_\alpha^c} \right]\label{eq_n168}\\
            &=& \mathbb{E}\left [ \int_{t_0}^t   \left \|\mathcal{R}\left(\sum_{i=1}^{d} \sum_{j \neq i}\frac{\lambda_i(t)- \lambda_j(t)}{\gamma_i(t) - \gamma_j(t)} u_i(t) u_j^\ast(t)\right)\right\|_F^2\mathrm{d}t  \times \mathbbm{1}_{\hat{E}_\alpha^c} \right]\nonumber\\
            &\leq& \mathbb{E}\left [ \int_{t_0}^t   \left \|\sum_{i=1}^{d} \sum_{j \neq i}\frac{\lambda_i(t)- \lambda_j(t)}{\gamma_i(t) - \gamma_j(t)} u_i(t) u_j^\ast(t)\right\|_F^2\mathrm{d}t  \times \mathbbm{1}_{\hat{E}_\alpha^c} \right] \label{eq_n169}\\
            &=& \mathbb{E}\left [ \int_{t_0}^t   \sum_{i=1}^{d} \sum_{j > i}\left \|\frac{\lambda_i(t)- \lambda_j(t)}{\gamma_i(t) - \gamma_j(t)} (u_i(t) u_j^\ast(t) + u_j(t) u_i^\ast(t))\right\|_F^2\mathrm{d}t  \times \mathbbm{1}_{\hat{E}_\alpha^c} \right] \label{eq_n170}\\
            &=& \mathbb{E}\left [ \int_{t_0}^t   \sum_{i=1}^{d} \sum_{j > i}2\left(\frac{\lambda_i(t)- \lambda_j(t)}{\gamma_i(t) - \gamma_j(t)}\right)^2\mathrm{d}t  \times \mathbbm{1}_{\hat{E}_\alpha^c} \right] \label{eq_n171}\\
                             &=& \int_{t_0}^{T}   \mathbb{E}\left[ \sum_{i=1}^{d}  \sum_{j \neq i}  \frac{(\lambda_i(t) - \lambda_j(t))^2}{(\gamma_i(t) - \gamma_j(t))^2} \mathrm{d}t \times \mathbbm{1}_{\hat{E}_\alpha^c} \right]. \label{eq_int_2b2}
          \end{eqnarray}
\noindent
\eqref{eq_n168} holds by the definition of the Frobenius norm.
\eqref{eq_n169} holds since for any matrix $A$ we have $\|\mathcal{R}(A)\|_F \leq \|A\|_F$.
 \eqref{eq_n170} holds because $\langle u_i(t) u_j^\ast(t),  u_\ell(t) u_h^\ast(t) \rangle = 0$ for all $(i,j) \neq (\ell,h)$ and all $t \geq 0$, since $u_1(t),\ldots, u_d(t)$ are (complex) orthogonal as they are eigenvectors of a Hermitian matrix.
 \eqref{eq_n171} holds because  $\|u_i(t) u_j^\ast(t) + u_j(t) u_i^\ast(t)\|_F^2 = 2$ for all $i,j$ and all $t \geq 0$.

Thus, plugging \eqref{eq_int_2b2} into \eqref{eq_t2} (and recalling that, from the discussion after \eqref{eq_t2}, the bound we derive in \eqref{eq_int_2b2} holds without loss of generality for all eight terms in \eqref{eq_t2}), we have that
\begin{eqnarray}\label{eq_int_2}
    & &  \!\!\!\!\!\!\!\!\!\!\!\!\!\!\!\!\!\!\! \! \!\!\!\!\!\!\!\!\!\!\!\!\!\!\!\! \! \!\!\!\!\!\!\!\!\!\!\!\!\!\!\!\! \! \!\!\!\!\!\!\!\!\!\!\!\!\!\!\! \mathbb{E}\left[ \left\|\int_{t_0}^{t}\sum_{i=1}^{d} \sum_{j \neq i} \frac{\lambda_i(t)- \lambda_j(t)}{\gamma_i(t) - \gamma_j(t)}(u_i(s) u_j^\ast(s)\mathrm{d}B_{ij}(s) + u_j(s) u_i^\ast(s)\mathrm{d}B_{ij}^\ast(s))\right\|_F^2 \times \mathbbm{1}_{\hat{E}_\alpha^c} \right]\nonumber\\
    &\stackrel{\textrm{Eq. } \eqref{eq_n3}}{=}&\mathbb{E}\left[ \|X(T) - X(t_0)\|_F^2 \right ]\nonumber\\
    &\stackrel{\textrm{Eq. } \eqref{eq_t2}, \eqref{eq_int_2b2}}{\leq}& 32\int_{t_0}^{T}   \mathbb{E}\left[ \sum_{i=1}^{d}  \sum_{j \neq i}  \frac{(\lambda_i(t) - \lambda_j(t))^2}{(\gamma_i(t) - \gamma_j(t))^2} \mathrm{d}t \times \mathbbm{1}_{\hat{E}_\alpha^c} \right].
    \end{eqnarray}

\paragraph{Bounding the drift term.}

To bound the drift term in \eqref{eq_int_1}, we use the Cauchy-Schwarz inequality:
\begin{eqnarray}
     & & \!\!\!\!\!\!\!\!\!\!\!\!\! \! \!  \! \! \!\!\!\!\!\!\!\!\!\!\!\!\!\!\!\! \! \! \! \! \!\!\!\!\!\!\!\!\!\!\!\!\!\!\!\! \! \!\!\!\!\!\!\!\!\!\!\!\!\!\!\!\! \! \!\!\!\!\!\!\!\!\!\!\!\!\!\!\! \left\|\int_{t_0}^{T}\sum_{i=1}^{d} \sum_{j\neq i}  \frac{\lambda_i(t) - \lambda_j(t)}{(\gamma_i(t) - \gamma_j(t))^2} u_i(t) u_i^\ast(t) \mathrm{d}t \right\|_F^2\nonumber \\
     & =  &    \left\|\int_{t_0}^{T}\sum_{i=1}^{d} \sum_{j\neq i}  \frac{\lambda_i(t) - \lambda_j(t)}{(\gamma_i(t) - \gamma_j(t))^2} u_i(t) u_i^\ast(t) \times 1 \mathrm{d}t \right\|_F^2 \nonumber\\ 
     &     \stackrel{\textrm{Cauchy-Schwarz Ineq.}}{\leq}&     \int_{t_0}^{T}\left\|\sum_{i=1}^{d} \sum_{j\neq i}  \frac{\lambda_i(t) - \lambda_j(t)}{(\gamma_i(t) - \gamma_j(t))^2} u_i(t) u_i^\ast(t)\right\|_F^2 \mathrm{d}t\times \int_{t_0}^{T} 1^2 \mathrm{d}t \label{eq_n127}\\
     &=&   T \int_{t_0}^{T}\left\|\sum_{i=1}^{d} \sum_{j\neq i} \frac{\lambda_i(t) - \lambda_j(t)}{(\gamma_i(t) - \gamma_j(t))^2} u_i(t) u_i^\ast(t) \right\|_F^2 \mathrm{d}t \nonumber \\
          &=&   T \int_{t_0}^{T}\sum_{i=1}^{d} \left\|\sum_{j\neq i} \frac{\lambda_i(t) - \lambda_j(t)}{(\gamma_i(t) - \gamma_j(t))^2} u_i(t) u_i^\ast(t) \right\|_F^2 \mathrm{d}t \label{eq_n128}\\
                    &=&   T \int_{t_0}^{T}\sum_{i=1}^{d} \left\|\left(\sum_{j\neq i} \frac{\lambda_i(t) - \lambda_j(t)}{(\gamma_i(t) - \gamma_j(t))^2}\right) u_i(t) u_i^\ast(t) \right\|_F^2 \mathrm{d}t \nonumber\\
                                        &=&   T \int_{t_0}^{T}\sum_{i=1}^{d}\left(\sum_{j\neq i} \frac{\lambda_i(t) - \lambda_j(t)}{(\gamma_i(t) - \gamma_j(t))^2}\right)^2 \left\| u_i(t) u_i^\ast(t) \right\|_F^2 \mathrm{d}t \nonumber\\
                                                                                &=&   T \int_{t_0}^{T}\sum_{i=1}^{d}\left(\sum_{j\neq i} \frac{\lambda_i(t) - \lambda_j(t)}{(\gamma_i(t) - \gamma_j(t))^2}\right)^2 \times 1 \mathrm{d}t, \label{eq_int_3}
\end{eqnarray}
where \eqref{eq_n127} is by the Cauchy-Schwarz Inequality for integrals (applied to each entry of the matrix-valued integral).
\eqref{eq_n128} holds since $\langle u_i(t) u_i^\ast(t) , u_j(t) u_j^\ast(t) \rangle = 0$ for all $i \neq j$.
\eqref{eq_int_3} holds since $\| u_i(t) u_i^\ast(t) \|_F^2=1$ for all $t \geq 0$.
Therefore, taking the expectation on both sides of \eqref{eq_int_1}, and plugging \eqref{eq_int_2} and  \eqref{eq_int_3} into  \eqref{eq_int_1}, we have
\begin{eqnarray} \label{eq_int_4}
\mathbb{E}\left[\left\|\int_{t_0}^T \sum_{i=1}^d   \lambda_i(t) \mathrm{d}( u_i(t) u_i^\ast(t))\right\|_F^2   \times \mathbbm{1}_{\hat{E}_\alpha^c} \right] &\leq&  32\int_{t_0}^{T}   \mathbb{E}\left[ \sum_{i=1}^{d}  \sum_{j \neq i}  \frac{(\lambda_i(t) - \lambda_j(t))^2}{(\gamma_i(t) - \gamma_j(t))^2}  \times \mathbbm{1}_{\hat{E}_\alpha^c} \right]\mathrm{d}t \nonumber\\
  &+&   T \int_{t_0}^{T}\mathbb{E}\left[\sum_{i=1}^{d}\left(\sum_{j\neq i} \frac{\lambda_i(t) - \lambda_j(t)}{(\gamma_i(t) - \gamma_j(t))^2}\right)^2 \times \mathbbm{1}_{\hat{E}_\alpha^c} \right] \mathrm{d}t .
\end{eqnarray}
\end{proof}

\noindent
To prove Theorem \ref{thm_rank_k_covariance_approximation_new}, we apply Ito’s Lemma (Lemma \ref{lemma_ito_lemma_new}) to the function $f(X) := \| X \|_F^2$ to obtain an expression for the utility $\| \Psi(T) - \Psi(0) \|_F^2$ as a stochastic integral.
 We then plug in Lemmas \ref{Lemma_projection_differntial} and \ref{Lemma_integral} into this expression, and apply our high-probability bounds on the eigenvalue gaps of Dyson Brownian motion (Corollary \ref{lemma_gaps_any_start} of Theorem \ref{thm:eigenvalue_gap}, which we prove in Section \ref{section_GUE_proof}) to bound the expected utility.

\begin{proof}[Proof of theorem \ref{thm_rank_k_covariance_approximation_new}]
In the following, we set $t_0 := \frac{1}{(kd)^{10} + k\alpha^2 +\sigma_1^2}$.
We first compute the Ito derivative of $\Psi(t) := \sum_{i=1}^d \lambda_i(t) u_i(t) u_i^\ast(t)$:
\begin{eqnarray}
    \mathrm{d} \Psi(t) &=& \sum_{i=1}^d(\lambda_i(t) + \mathrm{d}\lambda_i(t)) (u_i(t) u_i^\ast(t) + \mathrm{d}(u_i(t) u_i^\ast(t))) - \lambda_i(t) u_i(t) u_i^\ast(t) \label{eq_n162} \\
    &=&\sum_{i=1}^d \lambda_i(t) \mathrm{d}(u_i(t) u_i^\ast(t))) + (\mathrm{d}\lambda_i(t)) (u_i(t) u_i^\ast(t)) + \mathrm{d}\lambda_i(t) \mathrm{d}( u_i(t) u_i^\ast(t)), \label{eq_ito_derivative_b}
\end{eqnarray}
where \eqref{eq_n162} holds due to the product rule of stochastic calculus.

From Lemma \ref{Lemma_projection_differntial}, we have that, for all $t \in [0,T]$, 
\begin{eqnarray}\label{eq_duu}
\mathrm{d}(u_i(t) u_i^\ast(t))&=&        \sum_{j \neq i} \frac{1}{\gamma_i(t) - \gamma_j(t)}(u_i(t) u_j^\ast(t)\mathrm{d}B_{ij}(t) + u_j(t) u_i^\ast(t)\mathrm{d}B_{ij}^\ast(t)) \nonumber \\ & & + \quad \sum_{j \neq i} \frac{\mathrm{d}t}{(\gamma_i(t)- \gamma_j(t))^2} (u_i(t) u_i^\ast(t) - u_j(t)u_j^\ast(t)).
\end{eqnarray}
By definition \eqref{eq_n45}, for all $i\leq k$,  $\lambda_i(t) = \gamma_i(t)$ for all $t\geq 0$.
Thus, for all $i \leq k$, we have that 
\begin{eqnarray}
& & \!\!\!\!\!\!\!\! \! \!\! \! \!\!\!\!\!\!\!\!\!\!\!\!\!\!\!\! \! \! \mathrm{d}\lambda_i(t) \mathrm{d}( u_i(t) u_i^\ast(t))\nonumber\\
&\stackrel{\textrm{Eq. } \eqref{eq_n45}}{=}& \mathrm{d}\gamma_i(t) \mathrm{d}( u_i(t) u_i^\ast(t)) \label{eq_dlambda_du.1} \\
&\stackrel{\textrm{Eq. \eqref{eq_duu}}}{=}& \mathrm{d}\gamma_i(t) \bigg[\sum_{j \neq i} \frac{1}{\gamma_i(t) - \gamma_j(t)}(u_i(t) u_j^\ast(t)\mathrm{d}B_{ij}(t) + u_j(t) u_i^\ast(t)\mathrm{d}B_{ij}^\ast(t))\nonumber\\
& & \quad \quad  \quad + \quad \sum_{j \neq i} \frac{\mathrm{d}t}{(\gamma_i(t)- \gamma_j(t))^2} (u_i(t) u_i^\ast(t) - u_j(t)u_j^\ast(t))\bigg]\nonumber\\
&=& \sum_{j \neq i} \frac{\mathrm{d}\gamma_i(t)}{\gamma_i(t) - \gamma_j(t)}(u_i(t) u_j^\ast(t)\mathrm{d}B_{ij}(t) + u_j(t) u_i^\ast(t) \mathrm{d}B_{ij}^\ast(t) )\nonumber\\
& &  + \quad \sum_{j \neq i} \frac{\mathrm{d}\gamma_i(t)\mathrm{d}t}{(\gamma_i(t)- \gamma_j(t))^2} (u_i(t) u_i^\ast(t) - u_j(t)u_j^\ast(t))\nonumber\\
&\stackrel{\textrm{Eq. } \eqref{eq_DBM_eigenvalues}}{=}& \sum_{j \neq i} \left(\mathrm{d}B_{i i}(t) +   \sum_{j \neq i} \frac{1}{\gamma_i(t) - \gamma_j(t)} \mathrm{d}t \right)\frac{1}{\gamma_i(t) - \gamma_j(t)}(u_i(t) u_j^\ast(t)\mathrm{d}B_{ij}(t) + u_j(t) u_i^\ast(t)\mathrm{d}B_{ij}^\ast(t))\nonumber\\
& & + \quad \sum_{j \neq i} \left(\mathrm{d}B_{i i}(t) +   \sum_{j \neq i} \frac{1}{\gamma_i(t) - \gamma_j(t)} \mathrm{d}t\right) \frac{\mathrm{d}t}{(\gamma_i(t)- \gamma_j(t))^2} (u_i(t) u_i^\ast(t) - u_j(t)u_j^\ast(t))\nonumber\\
&=& \sum_{j \neq i} \left(\mathrm{d}B_{i i}(t)\mathrm{d}B_{ij}(t) +  2 \sum_{j \neq i} \frac{1}{\gamma_i(t) - \gamma_j(t)} \mathrm{d}t \, \mathrm{d}B_{ij}(t) \right)\frac{1}{\gamma_i(t) - \gamma_j(t)}u_i(t) u_j^\ast(t)\nonumber\\
& & + \quad \sum_{j \neq i} \left(\mathrm{d}B_{i i}(t)\mathrm{d}B_{ij}^\ast(t) +  2 \sum_{j \neq i} \frac{1}{\gamma_i(t) - \gamma_j(t)} \mathrm{d}t \, \mathrm{d}B_{ij}^\ast(t) \right)\frac{1}{\gamma_i(t) - \gamma_j(t)} u_j(t) u_i^\ast(t)\nonumber\\
& & + \quad  \sum_{j \neq i} \left(\mathrm{d}B_{i i}(t) \mathrm{d}t +  2 \sum_{j \neq i} \frac{1}{\gamma_i(t) - \gamma_j(t)} (\mathrm{d}t)^2\right) \frac{1}{(\gamma_i(t)- \gamma_j(t))^2} (u_i(t) u_i^\ast(t) - u_j(t)u_j^\ast(t))\nonumber\\
& =  & 0 \label{eq_dlambda_du},
\end{eqnarray}
 where \eqref{eq_dlambda_du} holds since, for all $i,j \in [d]$, the Ito differentials $\mathrm{d}B_{i i}(t)\mathrm{d}B_{ij}(t)$ and $\mathrm{d}B_{i i}(t)\mathrm{d}B_{ij}^\ast(t)$ vanish because $\mathrm{d}B_{i i}(t)$ and $\mathrm{d}B_{ij}(t)$ are uncorrelated with mean zero, and the Ito differentials $\mathrm{d}B_{i i}(t) \mathrm{d}t$ and $(\mathrm{d}t)^2$ vanish because they are higher-order terms.
Therefore, plugging in \eqref{eq_dlambda_du} into \eqref{eq_ito_derivative_b}, we have that
\begin{eqnarray}
    \mathrm{d} \Psi(t)  &\stackrel{\textrm{Eq. \eqref{eq_ito_derivative_b}}}{=}& \sum_{i=1}^d \lambda_i(t) \mathrm{d}(u_i(t) u_i^\ast(t))) + (\mathrm{d}\lambda_i(t)) (u_i(t) u_i^\ast(t)) + \mathrm{d}\lambda_i(t) \mathrm{d}( u_i(t) u_i^\ast(t)) \nonumber\\
    &=&\sum_{i=1}^k \lambda_i(t) \mathrm{d}(u_i(t) u_i^\ast(t))) + (\mathrm{d}\lambda_i(t)) (u_i(t) u_i^\ast(t)) + \mathrm{d}\lambda_i(t) \mathrm{d}( u_i(t) u_i^\ast(t)) \label{eq_n161}\\
     &\stackrel{\textrm{Eq. \eqref{eq_dlambda_du}}}{=}&\sum_{i=1}^k \lambda_i(t) \mathrm{d}(u_i(t) u_i^\ast(t))) + (\mathrm{d}\lambda_i(t)) (u_i(t) u_i^\ast(t)), \label{eq_ito_derivative}
\end{eqnarray}
where \eqref{eq_n161} holds since  $\lambda_i(t) = 0$ for all $i >k$ and all $t\geq 0$.
Therefore, we have 
\newpage
\begin{eqnarray}  \label{eq_ito_integral_1}
& &\!\!\!\!\!\!\!\!\!\!\!\!\!\!\! \! \!\!\!\!\!\!\! \! \!\! \! \!\!\!\!\!\!\!\!\!\!\!\!\!\!\!\! \! \!\! \! \!\!\!\!\!\!\!\!\!\!\!\!\!\!\!\! \! \!\mathbb{E}\left[\left\|\Psi (T) -  \Psi(t_0)\right \|_F^2 \times \mathbbm{1}_{\hat{E}_\alpha^c} \right] \nonumber\\
&=&\mathbb{E}\left[\left\|\int_{t_0}^T \mathrm{d}\Psi(t)\right \|_F^2 \times \mathbbm{1}_{\hat{E}_\alpha^c} \right] \nonumber\\
&\stackrel{\textrm{Eq. } \eqref{eq_ito_derivative}}{=}&  \mathbb{E}\left[\left\|\int_{t_0}^T \sum_{i=1}^d   \lambda_i(t) \mathrm{d}( u_i(t) u_i^\ast(t)) + (\mathrm{d} \lambda_i(t)) u_i(t) u_i^\ast(t) \right\|_F^2 \times \mathbbm{1}_{\hat{E}_\alpha^c} \right] \nonumber\\
&\stackrel{\textrm{Tri. Ineq.}}{\leq} & \mathbb{E}\left[\left\|\int_{t_0}^T \sum_{i=1}^d   \lambda_i(t) \mathrm{d}( u_i(t) u_i^\ast(t))\right\|_F^2 \times \mathbbm{1}_{\hat{E}_\alpha^c} \right]\nonumber \\
& & + \quad  \mathbb{E}\left[ \left\|\int_{t_0}^T \sum_{i=1}^d  (\mathrm{d} \lambda_i(t)) u_i(t) u_i^\ast(t) \right\|_F^2 \times \mathbbm{1}_{\hat{E}_\alpha^c}\right]  \nonumber\\
&\stackrel{\textrm{Lem. \ref{Lemma_integral}}}{\leq}  & 32 \int_{t_0}^{T}  \mathbb{E}\left[ \sum_{i=1}^{d}  \sum_{j \neq i}  \frac{(\lambda_i(t) - \lambda_j(t))^2}{(\gamma_i(t)- \gamma_j(t))^2} \times \mathbbm{1}_{\hat{E}_\alpha^c}\right]\mathrm{d}t \nonumber \\
    & &  +  \quad    T \int_{t_0}^{T}\mathbb{E}\left[\sum_{i=1}^{d}\left(\sum_{j\neq i} \frac{\lambda_i(t) - \lambda_j(t)}{(\gamma_i(t)- \gamma_j(t))^2}\right)^2 \times \mathbbm{1}_{\hat{E}_\alpha^c}\right] \mathrm{d}t \nonumber\\
    & & + \quad \mathbb{E}\left[\left\|\int_{t_0}^T \sum_{i=1}^d (\mathrm{d} \lambda_i(t))  u_i(t) u_i^\ast(t) \right\|_F^2 \times \mathbbm{1}_{\hat{E}_\alpha^c} \right],
\end{eqnarray}
Plugging in $\lambda_i(t) = \gamma_i(t)$ for $i \leq k$ and $\lambda_i(t) = 0$  for $i>k$ into \eqref{eq_ito_integral_1}, we have

 \begin{eqnarray}
& &\!\!\!\!\!\!\!\!\! \! \!\! \! \!\!\!\!\!\!\!\!\!\!\!\!\!\!\!\! \! \! \mathbb{E}\left[\left\|\Psi (T) -  \Psi(t_0)\right \|_F^2 \times \mathbbm{1}_{\hat{E}_\alpha^c}\right] \nonumber\\
&\stackrel{\textrm{Eq. } \eqref{eq_ito_integral_1}}{\leq}& 32\int_{t_0}^{T}   \mathbb{E}\left[ \sum_{i=1}^{d}  \sum_{j \neq i}  \frac{(\lambda_i(t) - \lambda_j(t))^2}{(\gamma_i(t)- \gamma_j(t))^2} \times \mathbbm{1}_{\hat{E}_\alpha^c}\right]\mathrm{d}t \nonumber \\
    & &  + \quad     T \int_{t_0}^{T}\mathbb{E}\left[\sum_{i=1}^{d}\left(\sum_{j\neq i} \frac{\lambda_i(t) - \lambda_j(t)}{(\gamma_i(t)- \gamma_j(t))^2}\right)^2 \times \mathbbm{1}_{\hat{E}_\alpha^c}\right] \mathrm{d}t \nonumber\\
    & &  + \mathbb{E}\left[\left\|\int_{t_0}^T \sum_{i=1}^d (\mathrm{d} \lambda_i(t))  u_i(t) u_i^\ast(t)\right\|_F^2  \times \mathbbm{1}_{\hat{E}_\alpha^c}\right] \nonumber\\
&=&  32\int_{t_0}^{T}   \mathbb{E}\left[ \sum_{i=1}^{k} \left( k +  \sum_{j >k}  \frac{(\gamma_i(t))^2}{(\gamma_i(t)- \gamma_j(t))^2} \right)  \times \mathbbm{1}_{\hat{E}_\alpha^c}\right]\mathrm{d}t \nonumber \\
    & & + \quad     T \int_{t_0}^{T}\mathbb{E}\left[\sum_{i=1}^{k}\left(\sum_{j \neq i: j\leq k} \frac{1}{\gamma_i(t)- \gamma_j(t)} + \sum_{j> k} \frac{\gamma_i(t)}{(\gamma_i(t)- \gamma_j(t))^2}\right)^2  \times \mathbbm{1}_{\hat{E}_\alpha^c}\right] \mathrm{d}t \nonumber\\
    & & + \quad  \mathbb{E}\left[\left\|\int_{t_0}^T \sum_{i=1}^d (\mathrm{d} \lambda_i(t))  u_i(t) u_i^\ast(t)\right\|_F^2  \times \mathbbm{1}_{\hat{E}_\alpha^c}\right]   \label{eq_n163}\\
    &\leq &  32\int_{t_0}^{T}   \mathbb{E}\left[ \sum_{i=1}^{k} \left( k +  \sum_{j >k}  \frac{(\gamma_i(t))^2}{(\gamma_i(t)- \gamma_j(t))^2} \right)  \times \mathbbm{1}_{\hat{E}_\alpha^c}\right]\mathrm{d}t \nonumber \\
    & &  + \quad    4T \int_{t_0}^{T}\sum_{i=1}^{k}\mathbb{E}\left[\left(\sum_{j \neq i: j\leq k} \frac{1}{\gamma_i(t)- \gamma_j(t)}\right)^2  \times \mathbbm{1}_{\hat{E}_\alpha^c} +  \left(\sum_{j> k} \frac{\gamma_i(t)}{(\gamma_i(t)- \gamma_j(t))^2}\right)^2  \times \mathbbm{1}_{\hat{E}_\alpha^c}\right] \mathrm{d}t\nonumber\\
    & & + \quad  \mathbb{E}\left[\left\|\int_{t_0}^T \sum_{i=1}^d (\mathrm{d} \lambda_i(t))  u_i(t) u_i^\ast(t)\right\|_F^2  \times \mathbbm{1}_{\hat{E}_\alpha^c}\right], \label{eq_u1}
\end{eqnarray}
where \eqref{eq_n163} is obtained by plugging in  $\lambda_i(t) = \gamma_i(t)$ for $i \leq k$ and $\lambda_i(t) = 0$  for $i>k$.

\paragraph{Bounding the second moment of the inverse gaps.}
By Corollary \ref{lemma_gaps_any_start} we have that for every $1\leq i < j \leq d$,
\begin{equation}\label{eq_v1}
    \mathbb{P}\left(\left\{ \gamma_i(t) - \gamma_{j}(t) \leq (j-i) \times s \frac{\sqrt{t}}{\mathfrak{b}\sqrt{d}} \right\} \cap \hat{E}_\alpha^c\right) \leq  s^3 \qquad \forall s>0, t>0.
\end{equation}
Thus, for $t \leq T$,
\begin{eqnarray} \label{eq_second_inverse_moment}
    & &\!\!\! \! \!\! \! \!\!\!\!\!\!\!\!\!\!\!\!\!\!\!\! \! \!\! \! \!\!\!\!\!\!\!\!\!\!\!\!\!\!\!\! \! \! \mathbb{E}\left[\frac{1}{(\gamma_i(t) - \gamma_{j}(t))^2}  \times \mathbbm{1}_{\hat{E}_\alpha^c}\right]\nonumber\\
     &\leq & \mathbb{E}\left[\frac{1}{(\gamma_i(t) - \gamma_{j}(t))^2} \times \mathbbm{1}\left\{\gamma_i(t) - \gamma_{j}(t) \leq (j-i) \times \frac{\sqrt{t}}{ \mathfrak{b}\sqrt{d}}\right\}  \times \mathbbm{1}_{\hat{E}_\alpha^c}\right]\nonumber\\
     & & + \quad  \mathbb{E}\left[\frac{1}{(\gamma_i(t) - \gamma_{j}(t))^2} \times \mathbbm{1}\left\{\gamma_i(t) - \gamma_{j}(t) > (j-i) \times \frac{\sqrt{t}}{ \mathfrak{b}\sqrt{d}}\right\}  \times \mathbbm{1}_{\hat{E}_\alpha^c}\right] \nonumber\\
    &\leq & \mathbb{E}\left[\frac{1}{(\gamma_i(t) - \gamma_{j}(t))^2} \times \mathbbm{1}\left\{\gamma_i(t) - \gamma_{j}(t) \leq (j-i) \times \frac{\sqrt{t}}{ \mathfrak{b}\sqrt{d}}\right\}  \times \mathbbm{1}_{\hat{E}_\alpha^c}\right] + \frac{\mathfrak{b}^2 d}{(j-i)^2 t} \label{eq_n147}\\
    &\stackrel{\textrm{Prop. } \ref{lemma_layer_cake}}{=}& \int_{\frac{\mathfrak{b}^2d}{(j-i)^2 t}}^{\infty}   \mathbb{P}\left(\left\{\frac{1}{(\gamma_i(t) - \gamma_{i+1}(t))^2} \geq s \right\} \cap \hat{E}_\alpha^c \right) \mathrm{d}s + \frac{\mathfrak{b}^2 d}{(j-i)^2 t} \label{eq_n148}\\
    &=& \int_{\frac{\mathfrak{b}^2 d}{(j-i)^2 t}}^{\infty}   \mathbb{P}\left(\left\{(\gamma_i(t) - \gamma_{i+1}(t))^2 \leq s^{-1}  \right\} \cap \hat{E}_\alpha^c \right)  \mathrm{d}s + \frac{\mathfrak{b}^2 d}{(j-i)^2 t} \nonumber\\
    &=& \int_{\frac{\mathfrak{b}^2 d}{(j-i)^2 t}}^{\infty}   \mathbb{P}\left(\left\{\gamma_i(t) - \gamma_{i+1}(t) \leq s^{-\frac{1}{2}}\right\} \cap \hat{E}_\alpha^c\right) \mathrm{d}s + \frac{\mathfrak{b}^2 d}{(j-i)^2 t}\nonumber\\
    &\stackrel{\textrm{Eq. \eqref{eq_v1}}}{\leq}& \int_{\frac{\mathfrak{b}^2 d}{(j-i)^2 t}}^{\infty}  \left(\frac{\mathfrak{b}^2 d}{(j-i)^2 t}\right)^{\frac{3}{2}} s^{-\frac{3}{2}} \mathrm{d}s + \frac{\mathfrak{b}^2 d}{(j-i)^2 t} \nonumber\\
    & = & -\frac{3}{2} \left(\frac{\mathfrak{b}^2 d}{(j-i)^2 t}\right)^{\frac{3}{2}}  s^{-\frac{1}{2}} \bigg|_{\frac{\mathfrak{b}^2 d}{(j-i)^2 t}}^{\infty} + \frac{\mathfrak{b}^2 d}{(j-i)^2 t}\nonumber\\
      & = & \frac{3}{2} \frac{\mathfrak{b}^2 d}{(j-i)^2 t} + \frac{\mathfrak{b}^2 d}{(j-i)^2 t}\nonumber\\
        & \leq & 3 \frac{\mathfrak{b}^2 d}{(j-i)^2 t},
\end{eqnarray}
{where \eqref{eq_n147} holds since $\mathbb{E}\left[\frac{1}{(\gamma_i(t) - \gamma_{j}(t))^2} \times \mathbbm{1}\left\{\gamma_i(t) - \gamma_{j}(t) > (j-i) \times \frac{\sqrt{t}}{ \mathfrak{b}\sqrt{d}}\right\}  \times \mathbbm{1}_{\hat{E}_\alpha^c}\right] \leq \frac{1}{\left((j-i) \times \frac{\sqrt{t}}{ \mathfrak{b}\sqrt{d}}\right)^2}$.
\eqref{eq_n148} holds by the layer cake formula (Proposition \ref{lemma_layer_cake}).}
Thus, for any $t_0>0$,
\begin{eqnarray*}
    \mathbb{E}\left[\int_{t_0}^T \frac{1}{(\gamma_i(t) - \gamma_{j}(t))^2} \mathrm{d}t  \times \mathbbm{1}_{\hat{E}_\alpha^c}\right] &=&\int_{t_0}^T \mathbb{E}\left[ \frac{1}{(\gamma_i(t) - \gamma_{j}(t))^2}  \times \mathbbm{1}_{\hat{E}_\alpha^c}\right]  \mathrm{d}t\\
     &\stackrel{\textrm{Eq. } \eqref{eq_second_inverse_moment}}{\leq}&  3 \int_{t_0}^T  \frac{\mathfrak{b}^2 d}{(j-i)^2 t} \mathrm{d}t\\
     &=&  3 \frac{\mathfrak{b}^2 d}{(j-i)^2} \log(t)|_{t_0}^T\\
     &=& 3 \frac{\mathfrak{b}^2 d}{(j-i)^2}\times (\log(T) - \log(t_0)).
    \end{eqnarray*}

\paragraph{Bounding the term $\mathbb{E}\left[\left\|\int_0^T \sum_{i=1}^d (\mathrm{d} \lambda_i(t))  u_i(t) u_i^\ast(t)\right\|_F^2  \times \mathbbm{1}_{\hat{E}_\alpha^c}\right]$.}

 For $i > k$, $\mathrm{d} \lambda_i(t) = 0$.
 For $i \leq k$, we have $\lambda_i(t) = \gamma_i(t)$ and thus,
\begin{align}\label{eq_n6}
  (\mathrm{d} \lambda_i(t))  u_i(t) u_i^\ast(t) =   (\mathrm{d} \gamma_i(t))  u_i(t) u_i^\ast(t) \stackrel{\textrm{Eq. }\eqref{eq_DBM_eigenvalues}}{=} \left(\mathrm{d}B_{i i}(t) +  2 \sum_{j \neq i} \frac{1}{\gamma_i(t) - \gamma_j(t)} \mathrm{d}t \right) u_i(t) u_i^\ast(t) \qquad  \forall i \leq k,
\end{align}
where the second equality is by the SDE which governs the evolution of the eigenvalues of Dyson Brownian motion \eqref{eq_DBM_eigenvalues}.

 Thus,
 \begin{eqnarray}
    & &   \!\!\!\!\!\!\!\!\!\!\!\!\! \! \!\! \! \!\!\!\!\!\!\!\!\!\!\!\!\!\!\!\! \! \!\! \! \!\!\!\!\!\!\!\!\!\!\!\!\!\!\!\! \! \!\mathbb{E}\left[\left\|\int_{t_0}^T \sum_{i=1}^d (\mathrm{d} \lambda_i(t))  u_i(t) u_i^\ast(t)\right\|_F^2  \times \mathbbm{1}_{\hat{E}_\alpha^c} \right]\nonumber\\
       &\stackrel{\textrm{Eq. } \eqref{eq_n6}}{=}&\mathbb{E}\left[\left\|\int_{t_0}^T \sum_{i=1}^k \left(\mathrm{d}B_{i i}(t) +  2 \sum_{j \neq i} \frac{1}{\gamma_i(t) - \gamma_j(t)} \mathrm{d}t \right) u_i(t) u_i^\ast(t)\right\|_F^2 \times \mathbbm{1}_{\hat{E}_\alpha^c} \right] \label{eq_n164}\\
     &\stackrel{\textrm{Tri. Ineq.}}{\leq} & 3\mathbb{E}\left[\left\|\int_{t_0}^T \sum_{i=1}^k  u_i(t) u_i^\ast(t) \mathrm{d}B_{i i}(t) \right\|_F^2 \times \mathbbm{1}_{\hat{E}_\alpha^c} \right] \nonumber\\
     & & + \quad  6 \mathbb{E}\left[\left\|\int_{t_0}^T \sum_{i=1}^k \sum_{j \neq i} \frac{1}{\gamma_i(t) - \gamma_j(t)}  u_i(t) u_i^\ast(t) \mathrm{d}t\right\|_F^2 \times \mathbbm{1}_{\hat{E}_\alpha^c} \right], \label{eq_a1}
 \end{eqnarray}
where \eqref{eq_n164} holds by Equation \eqref{eq_n6} and since $\lambda_i(t) = 0$ for all $i> k$ and all $t \geq 0$. 

    To bound the first term on the r.h.s. of \eqref{eq_a1}, we will apply Ito's lemma for real-valued functions (Lemma \ref{lemma_ito_lemma_new}).
     Towards this end, we first note that
\begin{eqnarray}
& & \mathbb{E}\left[\left\|\int_{t_0}^T \sum_{i=1}^d  u_i(t) u_i^\ast(t) \mathrm{d}B_{i i}(t) \right\|_F^2 \times \mathbbm{1}_{\hat{E}_\alpha^c} \right]
         \leq   \mathbb{E}\left[\left\|\int_{t_0}^T \sum_{i=1}^d  u_i(t) u_i^\ast(t) \mathrm{d}B_{i i}(t) \right\|_F^2 \right] \nonumber\\
         && \qquad =   \mathbb{E}\left[\left\| \mathcal{R}\left(\int_{t_0}^T \sum_{i=1}^d  u_i(t) u_i^\ast(t) \mathrm{d}B_{i i}(t) \right) \right\|_F^2 \right]
        + \mathbb{E}\left[\left\| \mathcal{I}\left(\int_{t_0}^T \sum_{i=1}^d  u_i(t) u_i^\ast(t) \mathrm{d}B_{i i}(t) \right) \right\|_F^2 \right] \label{eq_n83}\\
                & & \qquad =   \mathbb{E}\left[\left\| \int_{t_0}^T \sum_{i=1}^d  \mathcal{R}\left(u_i(t) u_i^\ast(t)\right) \mathrm{d}B_{i i}(t)  \right\|_F^2 \right] + \mathbb{E}\left[\left\| \int_{t_0}^T \sum_{i=1}^d  \mathcal{I}\left(u_i(t) u_i^\ast(t)\right) \mathrm{d}B_{i i}(t)  \right\|_F^2 \right], \label{eq_n82}
         \end{eqnarray}
         where Equation \eqref{eq_n83} holds since $\|A\|_F^2 = \|\mathcal{R}(A)\|_F^2 + \|\mathcal{I}(A)\|_F^2$ for any $A \in \mathbb{C}^{d \times d}$. 
         Equation \eqref{eq_n82} holds because $\mathrm{d}B_{i i}(t)$ is real-valued since, by definition, $B(t) = W(t) + W(t)^\ast$ for all $ t \geq 0$.

We first show how to bound the real term on the r.h.s. of \eqref{eq_n82}; as the derivation for the bound on the imaginary term on the r.h.s. of \eqref{eq_n82} is identical to that of the real term if we replace $\mathcal{R}(\cdot)$ with $\mathcal{I}(\cdot)$, we omit the proof for the imaginary term.

      Towards this end, define for all $t \geq t_0$,
    \begin{equation}\label{eq_n81}
        \mathcal{X}(t) := \int_{t_0}^t \sum_{i=1}^d \mathcal{R}(u_i(s) u_i^\ast(s)) \mathrm{d}B_{i i}(s).
            \end{equation}
            Then for all $t \geq t_0$,
                \begin{equation}\label{eq_n84}
       \mathrm{d} \mathcal{X}(t) = \sum_{i=1}^d  \mathcal{R}(u_i(t) u_i^\ast(t)) \mathrm{d}B_{i i}(t).
            \end{equation}
    Plugging in $f(Y) = \|Y\|_F^2 = \sum_{i=1}^d \sum_{j=1}^d Y_{ij}^2$ into Ito's Lemma (Lemma \ref{lemma_ito_lemma_new}),
     we have that
    \begin{eqnarray}
        & & \mathbb{E}\left[\left\| \int_{t_0}^T \sum_{i=1}^d  \mathcal{R}\left(u_i(t) u_i^\ast(t)\right) \mathrm{d}B_{i i}(t)  \right\|_F^2 \right] \nonumber\\
         & & \qquad \stackrel{\textrm{Eq. } \eqref{eq_n81}}{=} \mathbb{E}[f(\mathcal{X}(T))] \nonumber\\
                  & & \qquad  \stackrel{\textrm{Lem. \ref{lemma_ito_lemma_new}}, \, \, \textrm{Eq.} \, \, \eqref{eq_n84}}{=}                  \mathbb{E}\left[ \int_{t_0}^T  \sum_{i=1}^d  \sum_{\alpha, \beta \in [d]} \frac{\partial}{\partial \mathcal{X}_{\alpha \beta}} f(\mathcal{X}(t))\times \mathcal{R} (u_i(t) u_i^\ast(t))[\alpha, \beta] \times \mathrm{d}B_{i i}(t) \right] \nonumber\\
 \nonumber\\
                  & & \qquad \qquad \qquad \qquad \qquad +\mathbb{E}\left[ \frac{1}{2} \int_{t_0}^T \sum_{i=1}^d  \sum_{\alpha, \beta \in [d]} \sum_{\ell, r \in [d]} \frac{\partial^2}{\partial \mathcal{X}_{\alpha \beta} \partial \mathcal{X}_{\ell r} } f(\mathcal{X}(t)) \mathrm{d}t \right ] \nonumber\\
                  & & \qquad = 0 +  \frac{1}{2} \mathbb{E}\left[\int_{t_0}^T \sum_{i=1}^d  \sum_{\alpha, \beta \in [d]} \sum_{\ell, r \in [d]} \frac{\partial^2}{\partial \mathcal{X}_{\alpha \beta} \partial \mathcal{X}_{\ell r} } f(\mathcal{X}(t)) \mathrm{d}t \right ] \label{eq_n80}\\
         & & \qquad \stackrel{\textrm{Eq. } \eqref{eq_int_5}}{=}   \frac{1}{2} \mathbb{E}\left[\int_{t_0}^T \sum_{i=1}^d  \sum_{\alpha = 1}^d \sum_{\beta = 1}^d 2 \times (\mathcal{R}(u_i(t) u_i^\ast(t))[\alpha, \beta])^2  \mathrm{d}t \right ] \nonumber\\
         & & \qquad = \mathbb{E}\left[\sum_{i=1}^d \int_{t_0}^T \mathcal{R}(\|u_i(t) u_i^\ast(t) \|_F^2) \mathrm{d}t \right],  \label{eq_n86} 
    \end{eqnarray}
    where \eqref{eq_n80} holds since $\mathrm{d}B_{i i}(t)$ is independent of $\mathcal{X}(t)$ for all $t\geq 0$ by \eqref{eq_n81}, and since $\mathrm{d}B_{i i}(t)$ is independent of $u_i(t)$ for all $t \geq 0$ and all $i \in [d]$.
   
   Moreover, if we replace $\mathcal{R}(\cdot)$ with $\mathcal{I}$ in \eqref{eq_n81}, \eqref{eq_n84} and \eqref{eq_n86}, we get that
   \begin{equation}\label{eq_n85}
 \mathbb{E}\left[\left\| \int_{t_0}^T \sum_{i=1}^d  \mathcal{I}\left(u_i(t) u_i^\ast(t)\right) \mathrm{d}B_{i i}(t)  \right\|_F^2 \right]  = \mathbb{E}\left[ \sum_{i=1}^d \int_{t_0}^T \mathcal{I}(\|u_i(t) u_i^\ast(t) \|_F^2) \mathrm{d}t \right ].
   \end{equation}   
  Thus, we have
   \begin{eqnarray}
& & \mathbb{E}\left[\left\|\int_{t_0}^T \sum_{i=1}^d  u_i(t) u_i^\ast(t) \mathrm{d}B_{i i}(t) \right\|_F^2 \times \mathbbm{1}_{\hat{E}_\alpha^c} \right] \nonumber\\
                & \stackrel{\textrm{Eq. } \eqref{eq_n82}}{\leq}&    \mathbb{E}\left[\left\| \int_{t_0}^T \sum_{i=1}^d  \mathcal{R}\left(u_i(t) u_i^\ast(t)\right) \mathrm{d}B_{i i}(t)  \right\|_F^2 \right] + \mathbb{E}\left[\left\| \int_{t_0}^T \sum_{i=1}^d  \mathcal{I}\left(u_i(t) u_i^\ast(t)\right) \mathrm{d}B_{i i}(t)  \right\|_F^2 \right] \nonumber \\
&\stackrel{\textrm{Eq. } \eqref{eq_n86}, \, \, \eqref{eq_n85}}{=}&  \mathbb{E}\left[ \sum_{i=1}^d \int_{t_0}^T \mathcal{R}(\|u_i(t) u_i^\ast(t) \|_F^2) \mathrm{d}t \right ]  + \mathbb{E}\left[2\sum_{i=1}^d \int_{t_0}^T \mathcal{I}(\|u_i(t) u_i^\ast(t) \|_F^2) \mathrm{d}t \right ] \nonumber\\
 &=& \mathbb{E}\left[ \sum_{i=1}^d \int_{t_0}^T \|u_i(t) u_i^\ast(t) \|_F^2 \mathrm{d}t \right ]  \label{eq_n87}\\
  &=& \mathbb{E}\left[ \sum_{i=1}^d \int_{t_0}^T 1 \mathrm{d}t \right ]  \label{eq_n88}\\
 &=&  (T-t_0) d,  \label{eq_a2}
   \end{eqnarray}
    where \eqref{eq_n87} holds since $\|A\|_F^2 = \|\mathcal{R}(A)\|_F^2 + \|\mathcal{I}(A)\|_F^2$ for any $A \in \mathbb{C}^{d \times d}$. 
    \eqref{eq_n88} holds since $\|u_i(t)\| = 1$ for all $t \geq 0$ and all $i \in [d]$ because $u_i(t)$ is an eigenvector.

To bound the second term on the r.h.s. of \eqref{eq_a1}, we have
\begin{eqnarray}\label{eq_a3}
    & &  \!\!\!\!\!\!\!\! \! \!\! \! \!\!\!\!\!\!\!\!\!\!\!\!\!\!\!\! \! \!\! \! \!\!\!\!\!\!\!\!\!\!\!\!\!\!\!\! \! \!\mathbb{E}\left[\left\|\int_{t_0}^T \sum_{i=1}^k \sum_{j \neq i} \frac{1}{\gamma_i(t) - \gamma_j(t)}  u_i(t) u_i^\ast(t) \mathrm{d}t\right\|_F^2 \times \mathbbm{1}_{\hat{E}_\alpha^c} \right]\nonumber\\
    &\leq  &   \mathbb{E}\left[\int_{t_0}^T \left\|\sum_{i=1}^k \sum_{j \neq i} \frac{1}{\gamma_i(t) - \gamma_j(t)}  u_i(t) u_i^\ast(t) \right\|_F^2 \mathrm{d}t \times \int_{t_0}^T 1^2 \mathrm{d}t \times \mathbbm{1}_{\hat{E}_\alpha^c} \right] \label{eq_n172}\\
       &= &   (T-t_0)\mathbb{E}\left[\int_{t_0}^T \left\|\sum_{i=1}^k \sum_{j \neq i} \frac{1}{\gamma_i(t) - \gamma_j(t)}  u_i(t) u_i^\ast(t) \right\|_F^2 \mathrm{d}t \times \mathbbm{1}_{\hat{E}_\alpha^c} \right]\nonumber\\
             &= &   (T-t_0)\mathbb{E}\left[\int_{t_0}^T \sum_{i=1}^k \left\| \sum_{j \neq i} \frac{1}{\gamma_i(t) - \gamma_j(t)}  u_i(t) u_i^\ast(t) \right\|_F^2 \mathrm{d}t \times \mathbbm{1}_{\hat{E}_\alpha^c} \right]\label{eq_n173}\\
                          &=&    (T-t_0)\mathbb{E}\left[\int_{t_0}^T \sum_{i=1}^k  \left(\sum_{j \neq i} \frac{1}{\gamma_i(t) - \gamma_j(t)}\right)^2 \left\| u_i(t) u_i^\ast(t) \right\|_F^2 \mathrm{d}t \times \mathbbm{1}_{\hat{E}_\alpha^c} \right],\nonumber\\
                            &=&    (T-t_0)\int_{t_0}^T \sum_{i=1}^k\mathbb{E}\left[  \left(\sum_{j \neq i} \frac{1}{\gamma_i(t) - \gamma_j(t)}\right)^2 \times \mathbbm{1}_{\hat{E}_\alpha^c} \right] \mathrm{d}t,
\end{eqnarray}
where \eqref{eq_n172} is by the Cauchy-Schwarz inequality,
and \eqref{eq_n173} holds since $\langle u_i(t) u_i^\ast(t), \, \, u_\ell(t) u_\ell^\ast(t) \rangle = 0$ for all $i \neq \ell$.
Therefore, plugging in \eqref{eq_a2} and \eqref{eq_a3} into \eqref{eq_a1}, we have
 \begin{align}\label{eq_a4}
    &\mathbb{E}\left[\left\|\int_{t_0}^T \sum_{i=1}^d (\mathrm{d} \lambda_i(t))  u_i(t) u_i^\ast(t)\right\|_F^2 \times \mathbbm{1}_{\hat{E}_\alpha^c} \right]\nonumber \\
    & \qquad \qquad \leq 3(T-t_0) d +  6 (T-t_0)\int_{t_0}^T \sum_{i=1}^k\mathbb{E}\left[  \left(\sum_{j \neq i} \frac{1}{\gamma_i(t) - \gamma_j(t)}\right)^2 \times \mathbbm{1}_{\hat{E}_\alpha^c} \right] \mathrm{d}t.
 \end{align}

\paragraph{Bounding the term $\mathbb{E}\left[  \left(\sum_{j \neq i} \frac{1}{\gamma_i(t) - \gamma_j(t)}\right)^2 \times \mathbbm{1}_{\hat{E}_\alpha^c} \right]$.}

Consider any subset $S \subseteq \{1,\ldots,d\}$.
Then

\begin{eqnarray}\label{eq_v2}
    & &   \!\! \! \!\!\!\!\!\!\!\!\!\!\!\!\!\!\!\! \! \!\! \! \!\!\!\!\!\!\!\!\!\!\!\!\!\!\!\! \! \! \mathbb{E}\left[  \left(\sum_{j \in S, j \neq i} \frac{1}{\gamma_i(t) - \gamma_j(t)}\right)^2 \times \mathbbm{1}_{\hat{E}_\alpha^c} \right] \nonumber\\
    &=&   \mathbb{E}\left[  \sum_{j \in S, j \neq i} \, \, \sum_{\ell \in S, \ell \neq i} \frac{1}{(\gamma_i(t) - \gamma_j(t))(\gamma_i(t) - \gamma_\ell(t))}  \times \mathbbm{1}_{\hat{E}_\alpha^c} \right] \nonumber\\
    & = &  \sum_{j \in S, j \neq i} \, \, \sum_{\ell \in S, \ell \neq i}   \mathbb{E}\left[ \frac{1}{(j-i)(
    \ell - i)\frac{\gamma_i(t) - \gamma_j(t)}{j-i}\times \frac{\gamma_i(t) - \gamma_\ell(t)}{\ell-i}}  \times \mathbbm{1}_{\hat{E}_\alpha^c} \right] \nonumber\\
     & =&   \sum_{j \in S, j \neq i} \, \, \sum_{\ell \in S, \ell \neq i}  \frac{1}{(j-i)(
    \ell - i)}\mathbb{E}\left[ \frac{1}{\frac{\gamma_i(t) - \gamma_j(t)}{j-i} \times \frac{\gamma_i(t) - \gamma_\ell(t)}{\ell-i}}  \times \mathbbm{1}_{\hat{E}_\alpha^c} \right] \nonumber\\
    &\leq & \sum_{j \in S, j \neq i} \, \, \sum_{\ell \in S, \ell \neq i}   \frac{1}{|(j-i)(
    \ell - i)|} \mathbb{E}\left[ \frac{1}{\left(\frac{\gamma_i(t) - \gamma_j(t)}{j-i}\right)^2} \times \mathbbm{1}_{\hat{E}_\alpha^c}  +  \frac{1}{ \left(\frac{\gamma_i(t) - \gamma_\ell(t)}{\ell-i}\right)^2}  \times \mathbbm{1}_{\hat{E}_\alpha^c} \right] \label{eq_n149}\\
    &\stackrel{\textrm{Eq. \eqref{eq_second_inverse_moment}}}{\leq} &
     \sum_{j \in S, j \neq i} \, \, \sum_{\ell \in S, \ell \neq i}   \frac{1}{|(j-i)(
    \ell - i)|}\left(3 \mathfrak{b}^2\frac{ d}{t} + 3 \mathfrak{b}^2\frac{ d}{t}\right) \nonumber\\
    &= & 6 \mathfrak{b}^2 \frac{d}{t}  \sum_{j \in S, j \neq i} \, \, \sum_{\ell \in S, \ell \neq i}  \frac{1}{|(j-i)(
    \ell - i)|} \nonumber\\
    &=& 6 \mathfrak{b}^2 \frac{d}{t} \sum_{j \in S, j \neq i} \frac{1}{|j-i|} \sum_{\ell \in S, \ell \neq i}  \frac{1}{|
    \ell - i|} \nonumber\\
        &= & 6 \mathfrak{b}^2 \frac{d}{t} \left( \sum_{j \in S, j \neq i} \frac{1}{|j-i|}\right)^2 \nonumber\\
    &\leq  & 6 \mathfrak{b}^2 \frac{d}{t} \log^2 (d+1)  \label{eq_n150}\\
    &\leq &  12 \mathfrak{b}^2 \frac{d}{t} \log^2 d,
\end{eqnarray}
{where \eqref{eq_n149} holds since $ab \leq (a^2 +b^2)$ for any numbers $a,b \in \mathbb{R}$.
 \eqref{eq_n150} holds since $\sum_{j \in S, j \neq i} \frac{1}{|j-i|} \leq \sum_{j= 1}^d \frac{1}{j} \leq \int_1^{d+1} \frac{1}{s}\mathrm{d}s = \log(d+1)$.}  

\paragraph{Completing the proof.}

 \begin{eqnarray}
& & \! \! \!\!\!\!\!\!\!\!\!\!\!\!\!\!\!\! \! \!\! \! \!\!\!\!\!\!\!\!\!\!\!\!\!\!\!\! \mathbb{E}\left[\left\|\Psi (T) -  \Psi(t_0)\right \|_F^2 \times   \mathbbm{1}_{\hat{E}_\alpha^c}\right]\nonumber\\ 
    &\stackrel{\textrm{Eq.\ \eqref{eq_u1}, \eqref{eq_a4}}}{\leq}& 32\int_{t_0}^{T}   \mathbb{E}\left[ \sum_{i=1}^{k} \left( k +  \sum_{j >k}  \frac{(\gamma_i(t))^2}{(\gamma_i(t) - \gamma_j(t))^2} \right)  \times \mathbbm{1}_{\hat{E}_\alpha^c}\right]\mathrm{d}t \nonumber \\
    & &  +    4T \int_{t_0}^{T}\sum_{i=1}^{k}\mathbb{E}\left[\left(\sum_{j \neq i: j\leq k} \frac{1}{\gamma_i(t) - \gamma_j(t)}\right)^2 \times \mathbbm{1}_{\hat{E}_\alpha^c} +  \left(\sum_{j> k} \frac{\gamma_i(t)}{(\gamma_i(t) - \gamma_j(t))^2}\right)^2  \times \mathbbm{1}_{\hat{E}_\alpha^c}\right] \mathrm{d}t\nonumber\\
    & & +  3(T-t_0) d +  6 (T-t_0)\int_{t_0}^T \sum_{i=1}^k\mathbb{E}\left[  \left(\sum_{j \neq i} \frac{1}{\gamma_i(t) - \gamma_j(t)}\right)^2 \times \mathbbm{1}_{\hat{E}_\alpha^c} \right] \mathrm{d}t \nonumber\\
    &\stackrel{\textrm{Eq. \eqref{eq_v2}}}{\leq}& 
    32\int_{t_0}^{T}   \mathbb{E}\left[ \sum_{i=1}^{k} \left( k +  \sum_{j >k}  \frac{(\gamma_i(t))^2}{(\gamma_i(t) - \gamma_j(t))^2} \right)  \times \mathbbm{1}_{\hat{E}_\alpha^c}\right]\mathrm{d}t \nonumber \\
    & & +  \quad    4T \int_{t_0}^{T}\sum_{i=1}^{k}12 \mathfrak{b}^2 \frac{d}{t} \log^2 d  +  \mathbb{E}\left[\left(\sum_{j> k} \frac{\gamma_i(t)}{(\gamma_i(t) - \gamma_j(t))^2}\right)^2  \times \mathbbm{1}_{\hat{E}_\alpha^c}\right] \mathrm{d}t\nonumber\\
    & & + \quad 3(T-t_0) d +  6 (T-t_0)\int_{t_0}^T \sum_{i=1}^k 12 \mathfrak{b}^2  \frac{d}{t} \log^2 d  \nonumber\\
    &\leq &     32\int_{t_0}^{T}   \mathbb{E}\left[ \sum_{i=1}^{k} \left( k +  \sum_{j >k}  16 \frac{\sigma_k^2}{(\sigma_k-\sigma_{k+1})^2} \right)  \right]\mathrm{d}t \label{eq_n129} \\
    & & + \quad     48 \mathfrak{b}^2 T kd (\log T- \log t_0) \log^2 d  +  \int_{t_0}^{T}  \mathbb{E}\left[\left(\sum_{j> k} 16 \frac{\sigma_k}{(\sigma_k-\sigma_{k+1}) \sqrt{T} \sqrt{d}}\right)^2 \right] \mathrm{d}t\nonumber\\
    & & + \quad  3(T-t_0) d +  6 (T-t_0) 12 \mathfrak{b}^2  kd (\log^2 d) (\log(T)-\log(t_0)) \nonumber\\
     &\stackrel{\textrm{Lem. } \ref{lemma_gap_concentration}}{\leq}& 32T\left( k^2 +   16 kd \frac{\sigma_k^2}{(\sigma_k-\sigma_{k+1})^2} \right)\nonumber \\
    & & + \quad     48 \mathfrak{b}^2 T kd (\log T- \log t_0) \log^2 d  +  \frac{16^2}{T} d \frac{\sigma_k^2}{(\sigma_k-\sigma_{k+1})^2}(T-t_0) \nonumber\\
    & & + \quad 3(T-t_0) d +  6 (T-t_0) 12 \mathfrak{b}^2  kd (\log^2 d) (\log(T)-\log(t_0))\nonumber\\
    &\leq &\frac{1}{2}10^4 \mathfrak{b}^2  kd T  \frac{\sigma_k^2}{(\sigma_k -\sigma_{k+1})^2} (\log^3 d) \log(\sigma_1 + T),   \label{eq_u2}
\end{eqnarray}
where \eqref{eq_u2} holds because $-\log t_0 \leq 20 \log(d)$ since  
$t_0 = \frac{1}{(kd)^{10} + k\alpha^2 +\sigma_1^2} \geq \frac{1}{(kd)^{10} + 400k\log(\sigma_1 + T) +\sigma_1^2} $.
Moreover, \eqref{eq_n129} holds because Proposition \ref{lemma_gap_concentration} implies that, since  $\sigma_k - \sigma_{k+1} \geq \sqrt{T} \sqrt{d} + 40\log^{\frac{1}{2}}(\sigma_1 + T)$, whenever $\hat{E}_\alpha^c$ occurs, we have 
\begin{align*}
    \gamma_k(t) - \gamma_{k+1}(t) \geq  \frac{1}{2}((\sigma_k - \sigma_{k+1}) - \alpha) \geq \frac{1}{4}((\sigma_k - \sigma_{k+1}) - \alpha) 
    &\geq \frac{1}{4}\sqrt{T} \sqrt{d} \qquad \forall t \geq 0,
    \end{align*}
    because $\alpha = 20\log^{\frac{1}{2}}(\sigma_1 d (T+1))$.
    
    Therefore, plugging in \eqref{eq_u2} into Lemma \ref{lemma_utility_rare_event}, we have that
 \begin{eqnarray}
    & & \!\!\!\!\!\!\!\! \! \! \! \! \!\!\!\!\!\!\!\!\!\!\!\!\!\!\!\! \! \! \! \! \!\!\!\!\!\!\!\!\!\!\!\!\!\!\!\! \! \! \mathbb{E}[\|\Psi(T) -  \Psi(0)\|_F^2] \nonumber\\
    &\stackrel{\textrm{Lem. \ref{lemma_utility_rare_event}}}{\leq} & 4\mathbb{E}[\|\Psi(T) - \Psi(0)\|_F^2 \times \mathbbm{1}_{\hat{E}_\alpha^c}] +dT \nonumber\\   
    &\stackrel{\textrm{Tri. Ineq.}}{\leq} & 16\mathbb{E}[\|\Psi(T) - \Psi(t_0)\|_F^2 \times \mathbbm{1}_{\hat{E}_\alpha^c}] +  16\mathbb{E}[\|\Psi(t_0) - \Psi(0)\|_F^2 \times \mathbbm{1}_{\hat{E}_\alpha^c}] + dT  \nonumber\\   
    &\stackrel{\textrm{Eq. \eqref{eq_u2}}}{\leq} & \frac{1}{4} 10^6 \mathfrak{b}^2  kd T  \frac{\sigma_k^2}{(\sigma_k-\sigma_{k+1})^2} (\log^3 d) + 16\mathbb{E}[\|\Psi(t_0) - \Psi(0)\|_F^2 \times \mathbbm{1}_{\hat{E}_\alpha^c}]  +dT \nonumber\\
    &\leq & \frac{1}{2} 10^6 \mathfrak{b}^2  kd T  \frac{\sigma_k^2}{(\sigma_k-\sigma_{k+1})^2} (\log^3 d) + \mathbb{E}[\|\Psi(t_0) - \Psi(0)\|_F^2 \times \mathbbm{1}_{\hat{E}_\alpha^c}] \nonumber\\
       &\stackrel{\textrm{Prop. \ref{lemma_t0}}}{\leq} & \frac{1}{2}  10^6 \mathfrak{b}^2  kd T  \frac{\sigma_k^2}{(\sigma_k-\sigma_{k+1})^2} (\log^3 d) + 40^2 \label{eq_n151}\\
       &\leq & 10^6 \mathfrak{b}^2  kd T  \frac{\sigma_k^2}{(\sigma_k-\sigma_{k+1})^2} (\log^3 d) \log(\sigma_1 + T),  \label{eq_u3}
\end{eqnarray}
where \eqref{eq_n151} holds by Proposition \ref{lemma_t0} since $t_0 = \frac{1}{(kd)^{10} + k\alpha^2 +\sigma_1^2}$.
\end{proof}

\section{Eigenvalue gaps of Gaussian Unitary Ensemble: Proof of Theorem \ref{thm:eigenvalue_gap}}\label{section_GUE_proof}

\subsection{Eigenvalue gap comparison result: Proof of Lemma \ref{lemma_gap_comparison}} \label{sec_gap_comparison_proof}

Before proving our high-probability bound on the eigenvalue gaps of Dyson Brownian motion (Theorem \ref{thm:eigenvalue_gap}), we first prove Lemma \ref{lemma_gap_comparison}.
This lemma reduces the task of bounding the eigenvalue gaps of Dyson Brownian motion from any initial condition, to the problem of bounding the gaps of a Dyson Brownian motion initialized at the $0$ vector.

The following proposition, which compares the size of the eigenvalue gaps of two solutions to the stochastic differential equations \eqref{eq_DBM_eigenvalues} of Dyson Brownian motion, and the stochastic derivative of their gaps, will be useful in proving Lemma \ref{lemma_gap_comparison}.
\begin{proposition}\label{prop_stochastic_derivative_comparison}
Consider any strong solutions $\gamma, \xi$ to \eqref{eq_DBM_eigenvalues}, for any $\beta \geq 1$. 
Suppose that for some $i \in [d]$ and at some time $t \geq 0$,
\begin{equation} \label{eq_z3}
\gamma_i(t) - \gamma_{i+1}(t) = \xi_i(t) - \xi_{i+1}(t) >0
\end{equation}
and
\begin{equation}  \label{eq_z4}
\gamma_j(t) - \gamma_{j+1}(t) \geq \xi_j(t) - \xi_{j+1}(t) > 0 \qquad \forall j \in [d-1].
\end{equation}
Then
\begin{equation}  \label{eq_z5}
\mathrm{d} \gamma_i(t) - \mathrm{d}  \gamma_{i+1}(t) \geq \mathrm{d} \xi_i(t) - \mathrm{d}  \xi_{i+1}(t).
\end{equation}

\end{proposition}

\begin{proof}
First, note that for any numbers $b > c >0$ and all $a >0$ we have that
\begin{equation} \label{eq_z1}
\frac{1}{a+b} - \frac{1}{b} > \frac{1}{a+c} - \frac{1}{c}.
\end{equation}

\paragraph{Bounding the repulsion forces when  $j > i+1$.}
For any $j > i+1$ we have that by \eqref{eq_z1} (setting $a=\gamma_i(t) - \gamma_{i+1}(t)$,  $b=\gamma_{i+1}(t) - \gamma_j(t)$, and $c=\xi_{i+1}(t) - \xi_j(t)$,  and noting that \eqref{eq_z4} implies that $b \geq c >0$ since $j > i+1$),
\begin{align} \label{eq_z2}
&\frac{1}{\gamma_i(t) - \gamma_{i+1}(t) +(\gamma_{i+1}(t) - \gamma_j(t))} - \frac{1}{\gamma_{i+1}(t) - \gamma_j(t)} \nonumber\\
& \qquad \qquad \qquad  \qquad \qquad \geq \frac{1}{\gamma_i(t) - \gamma_{i+1}(t) +(\xi_{i+1}(t) - \xi_j(t))} - \frac{1}{\xi_{i+1}(t) - \xi_j(t)}.
\end{align}
Plugging in \eqref{eq_z3} into \eqref{eq_z2}, we get that
\begin{align} \label{eq_z6}
&\frac{1}{\gamma_i(t) - \gamma_{i+1}(t) +(\gamma_{i+1}(t) - \gamma_j(t))} - \frac{1}{\gamma_{i+1}(t) - \gamma_j(t)}\nonumber\\
& \qquad \qquad \qquad  \qquad \qquad \geq \frac{1}{\xi_i(t) - \xi_{i+1}(t) +(\xi_{i+1}(t) - \xi_j(t))} - \frac{1}{\xi_{i+1}(t) - \xi_j(t)}.
\end{align}
Simplifying \eqref{eq_z6}, we get
\begin{equation} \label{eq_z7}
\frac{1}{\gamma_i(t) - \gamma_j(t)} - \frac{1}{\gamma_{i+1}(t) - \gamma_j(t)} \geq \frac{1}{\xi_i(t) - \xi_j(t)} - \frac{1}{\xi_{i+1}(t) - \xi_j(t)} \qquad \qquad \forall j > i+1.
\end{equation}

\paragraph{Bounding the repulsion forces when $j < i$.}  Next, consider any $j < i$.
Then by \eqref{eq_z1}  (setting $a = \gamma_i(t) - \gamma_{i+1}(t)$, $b=\gamma_{j}(t) - \gamma_i(t)$, and $c=\xi_{j}(t) - \xi_i(t)$ into \eqref{eq_z1}, and noting that \eqref{eq_z4} implies that $b \geq c >0$ since $j < i$), we get
\begin{align} \label{eq_z8}
\frac{1}{\gamma_{j}(t) - \gamma_i(t) + (\gamma_{i}(t) - \gamma_{i+1}(t))} - \frac{1}{\gamma_{j}(t) - \gamma_i(t)}\geq \frac{1}{\xi_{j}(t) - \xi_i(t) + (\gamma_{i}(t) - \gamma_{i+1}(t))} - \frac{1}{\xi_{j}(t) - \xi_i(t)}.
\end{align}
Then plugging in \eqref{eq_z3} into \eqref{eq_z8}, we have
\begin{align} \label{eq_z9}
\frac{1}{\gamma_{j}(t) - \gamma_i(t) + (\gamma_{i}(t) - \gamma_{i+1}(t))} - \frac{1}{\gamma_{j}(t) - \gamma_i(t)}\geq \frac{1}{\xi_{j}(t) - \xi_i(t) + (\xi_{i}(t) - \xi_{i+1}(t))} - \frac{1}{\xi_{j}(t) - \xi_i(t)}.
\end{align}
Simplifying \eqref{eq_z9}, we get
\begin{equation} \label{eq_z10}
\frac{1}{\gamma_i(t) - \gamma_{j}(t)} -\frac{1}{\gamma_{i+1}(t) - \gamma_{j}(t)}  \geq   \frac{1}{\xi_i(t)- \xi_{j}(t)} - \frac{1}{\xi_{i+1}(t)- \xi_{j}(t)} \qquad \qquad \forall j < i.
\end{equation}
Therefore,  \eqref{eq_z7} and \eqref{eq_z10} together imply that
\begin{equation} \label{eq_z11}
\frac{1}{\gamma_i(t) - \gamma_{j}(t)} -\frac{1}{\gamma_{i+1}(t) - \gamma_{j}(t)}  \geq   \frac{1}{\xi_i(t)- \xi_{j}(t)} - \frac{1}{\xi_{i+1}(t)- \xi_{j}(t)} \qquad \qquad \forall j \in [d] \backslash \{i, i+1\}.
\end{equation}

\paragraph{Bounding the gap derivative.}
By \eqref{eq_DBM_eigenvalues} and \eqref{eq_z11} we have that
\begin{eqnarray*}
& & \!\!\!\! \!\!\!\! \!\!\!\! \!\!\!\!  \!\!\!\! \!\!\!\! \!\!\!\! \!\!\!\!  \!\!\!\! \!\!\!\! \mathrm{d} \gamma_i(t) - \mathrm{d} \gamma_{i+1} (t) \nonumber\\
&\stackrel{\textrm{Eq. \eqref{eq_DBM_eigenvalues}}}{=}&  \left(\mathrm{d}B_{i, i}(t) +  \beta \sum_{j \neq i} \frac{1}{\gamma_i(t) - \gamma_j(t)} \mathrm{d}t \right) -  \left(\mathrm{d}B_{i+1, i+1}(t) +  \beta \sum_{j \neq i+1} \frac{1}{\gamma_{i+1}(t) - \gamma_j(t)} \mathrm{d}t \right)\\
&=& \mathrm{d}B_{i, i}(t)  - \mathrm{d}B_{i+1, i+1}(t)  + \beta  \mathrm{d}t  \sum_{j \in [d] \backslash \{i, i+1\}}   \frac{1}{\gamma_i(t) - \gamma_j(t)} -  \frac{1}{\gamma_{i+1}(t) - \gamma_j(t)}\\
&\stackrel{\textrm{Eq. \eqref{eq_z11}}}{\geq} &  \mathrm{d}B_{i, i}(t)  - \mathrm{d}B_{i+1, i+1}(t)  + \beta  \mathrm{d}t  \sum_{j \in [d] \backslash \{i, i+1\}}  \frac{1}{\xi_i(t)- \xi_{j}(t)} - \frac{1}{\xi_{i+1}(t)- \xi_{j}(t)}\\
&=&  \left(\mathrm{d}B_{i, i}(t) +  \beta \sum_{j \neq i} \frac{1}{\xi_i(t) - \xi_j(t)} \mathrm{d}t \right) -  \left(\mathrm{d}B_{i+1, i+1}(t) +  \beta \sum_{j \neq i+1} \frac{1}{\xi_{i+1}(t) - \xi_j(t)} \mathrm{d}t \right)\\
&\stackrel{\textrm{Eq. \eqref{eq_DBM_eigenvalues}}}{=}&  \mathrm{d} \xi_i(t) - \mathrm{d} \xi_{i+1} (t).
\end{eqnarray*}
This proves \eqref{eq_z5} and completes the proof of the proposition.
\end{proof}

\begin{proof}[Proof of Lemma \ref{lemma_gap_comparison}]
First, we note that since by Lemma \ref{lemma_continuity} at every time $t \geq 0$ the strong solution $\gamma(t)$ is a continuous function of the initial conditions $\gamma(0)$, without loss of generality we may assume that that the initial eigenvalue gaps of $\gamma$ are {\em strictly} greater than the corresponding eigenvalue gaps of $\xi$:
\begin{equation}\label{eq_initial_gaps}
    \xi_i(0) - \xi_{i+1}(0)  < \gamma_i(0) - \gamma_{i+1}(0) \qquad \qquad 1\leq i < d.
\end{equation}
We prove Lemma \ref{lemma_gap_comparison} by contradiction.
Let $\tau := \inf\{t\geq 0: \xi_i(t) - \xi_{i+1}(t) > \gamma_i(t) - \gamma_{i+1}(t) \textrm{ for some } i \in [d]  \}$ be the first time where the size of the $i$'th gaps ``cross'' for some $i \in [d]$ (in other words $\tau$ is the first time when the conclusion of Lemma \ref{lemma_gap_comparison} fails to hold).

\paragraph{Assumption towards a contradiction.} Suppose (towards a contradiction) that $\tau < \infty$.
By the definition of strong solutions, strong solutions to stochastic differential equations are almost surely continuous on $[0,\infty)$, and hence we have that both $\gamma(t)$ and $\xi(t)$ are almost surely continuous at every $t \in [0,\infty)$.
Therefore, since $\tau< \infty$, by the intermediate value theorem, we must have that, for some $i\in [d-1]$ the $i$'th gap of $\xi$ and the $i$'th gap of $\gamma$ are equal at the time $\tau$, and that at this time $\tau$ all the other gaps of $\gamma$ are  at least as large as the corresponding gaps of $\xi$:
\begin{eqnarray}
      \gamma_i(\tau) - \gamma_{i+1}(\tau) &=& \xi_i(\tau) - \xi_{i+1}(\tau),\label{eq_w1}\\ 
\nonumber\\
      \gamma_j(\tau) - \gamma_{j+1}(\tau) &\geq  & \xi_j(\tau) - \xi_{j+1}(\tau) \qquad \forall j \in[d-1]. \label{eq_w2} 
\end{eqnarray}
Moreover, by Lemma \ref{lemma_DBM_collision} we have that, almost surely, the particles $\gamma_1(t), \ldots, \gamma_d(t)$ of the Dyson Brownian motion $\gamma(t)$ do not collide with each other on all of $(0,\infty)$.
The same non-collision property holds for the particles $\xi_1(t), \ldots, \xi_d(t)$ of $\xi(t)$.
In other words, we have that, almost surely, 
\begin{eqnarray} \label{eq_w3}
\gamma_j(t) - \gamma_{j+1}(t) &>&0 \qquad \forall t \in (0,\infty), \,\, j \in [d-1], \nonumber\\
\xi_j(t) - \xi_{j+1}(t) &>& 0 \qquad \forall t \in (0,\infty), \,\, j \in [d-1].
\end{eqnarray}
Therefore, plugging in \eqref{eq_w1}, \eqref{eq_w2} and \eqref{eq_w3} into Proposition \ref{prop_stochastic_derivative_comparison}, we have that
\begin{equation}  \label{eq_w4}
\mathrm{d} \gamma_i(\tau) - \mathrm{d}  \gamma_{i+1}(\tau) \geq \mathrm{d} \xi_i(\tau) - \mathrm{d}  \xi_{i+1}(\tau).
\end{equation}
Next, we consider two cases: when $\mathrm{d} \gamma_i(\tau) - \mathrm{d}  \gamma_{i+1}(\tau) > \mathrm{d} \xi_i(\tau) - \mathrm{d}  \xi_{i+1}(\tau)$, and when $\mathrm{d} \gamma_i(\tau) - \mathrm{d}  \gamma_{i+1}(\tau) = \mathrm{d} \xi_i(\tau) - \mathrm{d}  \xi_{i+1}(\tau)$.

\paragraph{Case 1,  $\mathrm{d}  \gamma_i(\tau) - \mathrm{d}  \gamma_{i+1}(\tau) > \mathrm{d} \xi_i(\tau) - \mathrm{d}  \xi_{i+1}(\tau)$.}
For any $w \in \mathcal{W}_d$ (where $\mathcal{W}_d$ was defined in \eqref{eq:WeylChamber}), define the ``drift'' function
\begin{equation}\label{eq_n7}
    \mu_i(w) := \beta\sum_{j \neq i} \frac{1}{w_i - w_j}.
\end{equation}
Then we have that
\begin{eqnarray}\label{eq_w5}
& & \!\!\! \!\!\!\!  \!\!\!\! \!\!\!\!\!\!\!\! \!\!\!\!  \!\!\!\! \!\!\!\!  (\mathrm{d} \gamma_i(t) - \mathrm{d} \gamma_{i+1} (t))- (\mathrm{d} \xi_i(t) - \mathrm{d} \xi_{i+1} (t))\\ &\stackrel{\textrm{Eq. \eqref{eq_DBM_eigenvalues}}}{=}&  \left[\left(\mathrm{d}B_{i, i}(t) +  \beta \sum_{j \neq i} \frac{1}{\gamma_i(t) - \gamma_j(t)} \mathrm{d}t \right)
-  \left(\mathrm{d}B_{i+1, i+1}(t) +  \beta \sum_{j \neq i+1} \frac{1}{\gamma_{i+1}(t) - \gamma_j(t)} \mathrm{d}t \right) \right] \nonumber\\
& &   - \quad \left[\left(\mathrm{d}B_{i, i}(t) +  \beta \sum_{j \neq i} \frac{1}{\xi_i(t) - \xi_j(t)} \mathrm{d}t \right) -  \left(\mathrm{d}B_{i+1, i+1}(t) +  \beta \sum_{j \neq i+1} \frac{1}{\xi_{i+1}(t) - \xi_j(t)} \mathrm{d}t \right) \right] \nonumber\\
&\stackrel{\textrm{Eq. } \eqref{eq_n7}}{=}& \mu_i(\gamma(t)) - \mu_{i+1}(\gamma(t)) - (\mu_i(\xi(t)) - \mu_{i+1}(\xi(t))) \qquad  \qquad\forall t \geq 0.\nonumber
\end{eqnarray}
Since, in Case 1, 
$\mathrm{d}  \gamma_i(\tau) - \mathrm{d}  \gamma_{i+1}(\tau) > \mathrm{d} \xi_i(\tau) - \mathrm{d}  \xi_{i+1}(\tau)$,
we have by \eqref{eq_w5} that
\begin{align}\label{eq_w6}
&\mu_i(\gamma(\tau)) - \mu_{i+1}(\gamma(\tau)) - (\mu_i(\xi(\tau)) - \mu_{i+1}(\xi(\tau))) \nonumber\\
& \qquad \qquad \qquad  \qquad \qquad \stackrel{\textrm{Eq. } \eqref{eq_w5}}{=} (\mathrm{d} \gamma_i(\tau) - \mathrm{d} \gamma_{i+1} (\tau))- (\mathrm{d} \xi_i(\tau) - \mathrm{d} \xi_{i+1} (\tau)) >0.
\end{align}
From \eqref{eq_w3}, we have that, almost surely, all the gaps of $\gamma(t)$ and $\xi(t)$ are strictly greater than zero at every time $t \in (0,\infty)$.
Therefore, since $\gamma(t)$ and $\xi(t)$ are almost surely continuous on all $t\in [0,\infty)$,
we must have that, almost surely, $\mu(\gamma(t))$ and $\mu(\xi(t))$ are also continuous on all  $t\in(0,\infty)$. 
 Therefore, since $\mu(\gamma(t))$ and $\mu(\xi(t))$ are almost surely continuous on $(0,\infty)$, by \eqref{eq_w6}  we must have that there exists some open interval $\mathcal{I} \subset (0,\infty)$ containing $\tau$ such that
\begin{align}\label{eq_w7}
 &(\mathrm{d} \gamma_i(t) - \mathrm{d} \gamma_{i+1} (t))- (\mathrm{d} \xi_i(t) - \mathrm{d} \xi_{i+1} (t)) \nonumber\\
& \qquad \qquad \qquad  \qquad \qquad
 \stackrel{\textrm{Eq. } \eqref{eq_w5}}{=} \mu_i(\gamma(\tau)) - \mu_{i+1}(\gamma(\tau)) - (\mu_i(\xi(\tau)) - \mu_{i+1}(\xi(\tau))) \stackrel{\textrm{Eq. } \eqref{eq_w6}}{>} 0,  
\end{align}
for all $t \in \mathcal{I}$. 
Consider any $t \in \mathcal{I}$ such that $t > \tau$. 
 Then
\begin{eqnarray}
& & \!\!\!\!  \!\!\!\! \!\!\!\!\!\!\!\! \!\!\!\!  \!\!\!\! \!\!\!\!\!\!\!\! \!\!\!\!  \!\!\!\! \!\!\!\! (\gamma_i(t) -\gamma_{i+1} (t)) - (\xi_i(t) -  \xi_{i+1} (t)) \nonumber\\
    &\stackrel{\textrm{Eq. \eqref{eq_w2}}}{\geq}& [(\gamma_i(t) -\gamma_{i+1} (t)) - (\xi_i(t) -  \xi_{i+1} (t))] -   [(\gamma_i(\tau) -  \gamma_{i+1} (\tau)) - (\xi_i(\tau) -  \xi_{i+1} (\tau))] \nonumber\\
    &=&  \int_{\tau}^{t}  \bigg((\mathrm{d} \gamma_i(s) - \mathrm{d} \gamma_{i+1} (s))- (\mathrm{d} \xi_i(s) - \mathrm{d} \xi_{i+1} (s))\bigg) \mathrm{d}s \nonumber\\
    &\stackrel{\textrm{Eq. \eqref{eq_w7}}}{>}& 0, \label{eq_w8}
\end{eqnarray}
where \eqref{eq_w8} holds by Inequality \eqref{eq_w7} since $\tau < t$, and since  $[\tau, t] \subseteq \mathcal{I}$ because $\mathcal{I}$ is an interval containing both $\tau$ and $t$. 
Therefore Inequality \eqref{eq_w8}, together with the fact that $\mathcal{I}$ is an open interval containing $\tau$, implies that there exits some $\tau' \in \mathcal{I}$ where $\tau'>\tau$ such that 
\begin{equation}\label{eq_n8}
    \gamma_i(t) -\gamma_{i+1} (t) > \xi_i(t) -  \xi_{i+1} (t) \qquad \forall \tau<t <\tau'.
\end{equation}
Therefore \eqref{eq_n8} implies that  $\inf\{t\geq 0: \xi_i(t) - \xi_{i+1}(t) > \gamma_i(t) - \gamma_{i+1}(t)\} \geq \tau' >\tau$ and hence that $\tau \neq \inf\{t\geq 0: \xi_i(t) - \xi_{i+1}(t) > \gamma_i(t) - \gamma_{i+1}(t)\}$ for any $i \in [d]$.
This contradicts the definition of $\tau$.
Therefore, by contradiction our assumption that $\tau < \infty$ is false.

\paragraph{Case 2, $\mathrm{d} \gamma_i(\tau) - \mathrm{d}  \gamma_{i+1}(\tau) = \mathrm{d} \xi_i(\tau) - \mathrm{d}  \xi_{i+1}(\tau)$.}

Consider the system of stochastic differential equations for the process $\gamma_i(t) - \gamma_{i+1} (t)$:
\begin{align}\label{eq_w9}
&\mathrm{d} \gamma_i(t) - \mathrm{d} \gamma_{i+1} (t) 
\nonumber\\
& \qquad 
\stackrel{\textrm{Eq. \eqref{eq_DBM_eigenvalues}}}{=}  \left(\mathrm{d}B_{i, i}(t) +  \beta \sum_{j \neq i} \frac{1}{\gamma_i(t) - \gamma_j(t)} \mathrm{d}t \right) -  \left(\mathrm{d}B_{i+1, i+1}(t) +  \beta \sum_{j \neq i+1} \frac{1}{\gamma_{i+1}(t) - \gamma_j(t)} \mathrm{d}t \right)
\end{align}
for all $i \in [d]$, and the system of stochastic differential equations for the process $\xi_i(t) - \xi_{i+1} (t)$:
\begin{align}\label{eq_w10}
 &\mathrm{d} \xi_i(t) - \mathrm{d} \xi_{i+1} (t)
 \nonumber\\
& \qquad 
 \stackrel{\textrm{Eq. \eqref{eq_DBM_eigenvalues}}}{=}
\left(\mathrm{d}B_{i, i}(t) +  \beta \sum_{j \neq i} \frac{1}{\xi_i(t) - \xi_j(t)} \mathrm{d}t \right)  -  \left(\mathrm{d}B_{i+1, i+1}(t) +  \beta \sum_{j \neq i+1} \frac{1}{\xi_{i+1}(t) - \xi_j(t)} \mathrm{d}t \right)
\end{align}
for all $i \in [d]$. 

Then we have that
\begin{eqnarray}
0 &=& (\mathrm{d} \gamma_i(\tau) - \mathrm{d} \gamma_{i+1} (\tau))- (\mathrm{d} \xi_i(\tau) - \mathrm{d} \xi_{i+1} (\tau))\label{eq_n167}\\
&\stackrel{\textrm{Eq.\ \eqref{eq_w9},\eqref{eq_w10}}}{=}&  \left[\left(\beta \sum_{j \neq i} \frac{1}{\gamma_i(\tau) - \gamma_j(\tau)} \mathrm{d}t \right)
-  \left(\beta \sum_{j \neq i+1} \frac{1}{\gamma_{i+1}(\tau) - \gamma_j(\tau)} \mathrm{d}t \right) \right] \nonumber\\
& &  - \left[\left(\beta \sum_{j \neq i} \frac{1}{\xi_i(\tau) - \xi_j(\tau)} \mathrm{d}t \right) -  \left(\beta \sum_{j \neq i+1} \frac{1}{\xi_{i+1}(\tau) - \xi_j(\tau)} \mathrm{d}t \right) \right], \label{eq_w11}
\end{eqnarray}
 where \eqref{eq_n167} holds since, in Case 2, $\mathrm{d} \gamma_i(\tau) - \mathrm{d}  \gamma_{i+1}(\tau) = \mathrm{d} \xi_i(\tau) - \mathrm{d}  \xi_{i+1}(\tau)$.
Plugging \eqref{eq_w1} 
into \eqref{eq_w11}, we have that
\begin{eqnarray}\label{eq_n10}
& & 0 
=  \left[\left(\beta \sum_{j \in [d] \backslash \{i,i+1\}} \frac{1}{\gamma_i(\tau) - \gamma_j(\tau)} \mathrm{d}t \right)
-  \left(\beta \sum_{j \in [d] \backslash \{i,i+1\}} \frac{1}{\gamma_{i+1}(\tau) - \gamma_j(\tau)} \mathrm{d}t \right) \right] \nonumber\\
& & - \quad \left[\left(\beta \sum_{j \in [d] \backslash \{i,i+1\}} \frac{1}{\xi_i(\tau) - \xi_j(\tau)} \mathrm{d}t \right) -  \left(\beta \sum_{j \in [d] \backslash \{i,i+1\}} \frac{1}{\xi_{i+1}(\tau) - \xi_j(\tau)} \mathrm{d}t \right) \right], \nonumber\\
&=&  \beta \mathrm{d}t \sum_{j \in [d] \backslash \{i,i+1\}} \left[\frac{1}{\gamma_i(\tau) - \gamma_j(\tau)}  -   \frac{1}{\gamma_{i+1}(\tau) - \gamma_j(\tau)} \right] - \left[ \frac{1}{\xi_i(\tau) - \xi_j(\tau)} -  \frac{1}{\xi_{i+1}(\tau) - \xi_j(\tau)} \right].  \qquad \quad
\end{eqnarray}
Moreover, we also have  
 from \eqref{eq_w2} that $ \gamma_j(\tau) - \gamma_{j+1}(\tau) \geq \xi_j(\tau) - \xi_{j+1}(\tau)$ for all $j \in [d-1]$.
Thus, by Inequality \eqref{eq_z1} we have that
\begin{equation}\label{eq_n11}
\left[\frac{1}{\gamma_i(\tau) - \gamma_j(\tau)}
-  \frac{1}{\gamma_{i+1}(\tau) - \gamma_j(\tau)} \right] - \left[ \frac{1}{\xi_i(\tau) - \xi_j(\tau)} -  \frac{1}{\xi_{i+1}(\tau) - \xi_j(\tau)} \right] \geq 0 \qquad \forall j \in [d] \backslash \{i,i+1\},
\end{equation}
and moreover that if   $\gamma_\ell(\tau) - \gamma_{\ell+1}(\tau) > \xi_\ell(\tau) - \xi_{\ell+1}(\tau)$ for any $\ell \in \{i+1,i+2, \ldots, d-1\}$, 
\begin{equation}\label{eq_n12}
\left[\frac{1}{\gamma_i(\tau) - \gamma_{\ell+1}(\tau)}
-  \frac{1}{\gamma_{i+1}(\tau) - \gamma_{\ell+1}(\tau)} \right] - \left[ \frac{1}{\xi_i(\tau) - \xi_{\ell+1}(\tau)} -  \frac{1}{\xi_{i+1}(\tau) - \xi_{\ell+1}(\tau)} \right] >0,
\end{equation}
and moreover that if   $\gamma_\ell(\tau) - \gamma_{\ell+1}(\tau) > \xi_\ell(\tau) - \xi_{\ell+1}(\tau)$ for any $\ell \in \{1,2, \ldots, i-1\}$,
\begin{equation}\label{eq_n13}
\left[\frac{1}{\gamma_i(\tau) - \gamma_{\ell}(\tau)}
-  \frac{1}{\gamma_{i+1}(\tau) - \gamma_{\ell}(\tau)} \right] - \left[ \frac{1}{\xi_i(\tau) - \xi_{\ell}(\tau)} -  \frac{1}{\xi_{i+1}(\tau) - \xi_{\ell}(\tau)} \right] >0.
\end{equation}
Thus, by \eqref{eq_n11}, \eqref{eq_n12}, and \eqref{eq_n13}, the only way for the r.h.s. of \eqref{eq_w11} to be equal to zero is if we have
\begin{equation}\label{eq_n9}
    \xi_j(\tau) -\xi_{j+1}(\tau) = \gamma_j(\tau) -\gamma_{j+1}(\tau) \qquad  \qquad \forall j \in [d-1],
    \end{equation}
   since we also have that $\xi_i(\tau) -\xi_{i+1}(\tau) = \gamma_i(\tau) -\gamma_{i+1}(\tau)$  
 by \eqref{eq_w1}.

Moreover, by Lemma \ref{lemma_strong}, for any initial conditions $\gamma(\tau)$ and $\xi(\tau)$, the processes $\gamma$ and $\xi$ have unique strong solutions  on $(0,\infty)$.
Therefore, since the stochastic differential equations \eqref{eq_DBM_eigenvalues} for $\gamma$ and $\xi$ are invariant to spatial translations, we must have by \eqref{eq_n9} that
\begin{equation}\label{eq_w12}
    \xi_i(t) -\xi_{i+1}(t) = \gamma_i(t) -\gamma_{i+1}(t) \qquad \forall t \geq \tau, i \in [d].
\end{equation}
By \eqref{eq_w12}, we have that  $\tau = \inf\{t\geq 0: \xi_i(t) - \xi_{i+1}(t) > \gamma_i(t) - \gamma_{i+1}(t) \textrm{ for some } i \in [d]  \} = \infty$.
This contradicts our assumption that $\tau < \infty$.
Therefore, by contradiction our assumption that $\tau < \infty$ is false.\\
To summarize, we have now shown, in both Case 1 and Case 2, that our assumption that $\tau< \infty$ is false.
Thus, by contradiction, we have that $\tau= \infty$ and hence that $\xi_i(t) - \xi_{i+1}(t)  \leq \gamma_i(t) - \gamma_{i+1}(t)$ for all  $t>0$ and all $1\leq i < d$.

\end{proof}

\subsubsection{Showing gaps are uniformly bounded below over time with high probability}
The following lemma (Lemma \ref{lemma_bad_event}), which we have used above to prove Lemma \ref{lemma_spectral_martingale}, shows that the high-probability bounds on the eigenvalue gaps of Dyson Brownian motion of Theorem \ref{thm:eigenvalue_gap} hold {\em uniformly} of the time interval $[0,T]$.

\begin{lemma}\label{lemma_bad_event}
Let $\gamma(t) = (\gamma_1(t), \ldots, \gamma_d(t))$ be a strong solution to \eqref{eq_DBM_eigenvalues} starting from any initial $\gamma(0) \in \mathcal{W}_d$. 
Then for any $t_0 \geq \frac{1}{d^{40}}$ and any $T>0$ we have
\begin{equation} \label{eq_e1}
    \mathbb{P}\left(\inf_{t_0 \leq t \leq T, \, \, 1\leq i < d  }\gamma_i(t) - \gamma_{i+1}(t) \leq \frac{1}{d^{10}} \frac{\sqrt{t}}{\mathfrak{b}\sqrt{d}}\right) \leq \frac{T}{d^{600}},
\end{equation}
for any $d \geq N_0$ where $N_0$ is a universal constant.

\end{lemma}

\begin{proof}
By Weyl's Inequality (Lemma \ref{lemma_weyl}), we have that for any $z \geq t_0$,
\begin{eqnarray}
& &\!\!\!\!\!\!\!\!\!\!  \!\!\! \!\!\!\!\!\!\!\!\!\!\!\!\!\!\!\!\!\! \mathbb{P}\left( \inf_{1\leq i < d } \gamma_i(t) - \gamma_{i+1}(t) \leq \frac{1}{d^{10}} \frac{\sqrt{t}}{\mathfrak{b}\sqrt{d}}   \quad \textrm{ for some }  t \in \left[z, z + \frac{1}{d^{200}}\right] \right) \nonumber\\
&=& \mathbb{P}\left( \gamma_i(t) - \gamma_{i+1}(t) \leq \frac{1}{d^{10}} \frac{\sqrt{t}}{\mathfrak{b}\sqrt{d}}   \qquad \textrm{ for some }  t \in \left[z, z + \frac{1}{d^{200}}\right], i \in[d-1] \right) \nonumber\\
 & \stackrel{\textrm{Lem. \ref{lemma_weyl}}}{\leq}& \mathbb{P}\left( \gamma_i(z) - \gamma_{i+1}(z) \leq \frac{1}{d^{10}} \frac{\sqrt{t}}{\mathfrak{b}\sqrt{d}} +  2 \|B(t)\|_2   \, \, \textrm{ for some }  t \in \left[z, z + \frac{1}{d^{200}}\right], i \in[d-1] \right) \nonumber\\
  &\leq &\mathbb{P}\left( \gamma_i(z) - \gamma_{i+1}(z) \leq \frac{1}{d^{10}} \frac{\sqrt{t}}{\mathfrak{b}\sqrt{d}} +  4 \frac{1}{d^{200}} \sqrt{d}   \qquad \textrm{ for some }  t \in \left[z, z + \frac{1}{d^{200}}\right], i \in[d-1]  \right)\nonumber\\
   & & + \quad  \mathbb{P}\left(\sup_{t \in [0, \frac{1}{d^{200}}]} \|B(t)\|_2  > 2 \frac{1}{d^{200}} \sqrt{d} \right)\nonumber\\
   &  \stackrel{\textrm{Lem.  \ref{lemma_spectral_martingale_b}}}{\leq}&  \mathbb{P}\left( \gamma_i(z) - \gamma_{i+1}(z) \leq \frac{1}{d^{10}} \frac{\sqrt{t}}{\mathfrak{b}\sqrt{d}} +  4 \frac{1}{d^{200}} \sqrt{d}     \quad \textrm{ for some }   t \in \left[z, z + \frac{1}{d^{200}}\right], i \in[d-1]   \right) \nonumber\\
   & & + \quad  \frac{1}{d^{1000}} \label{eq_n174}\\
   & \leq & \mathbb{P}\left( \gamma_i(z) - \gamma_{i+1}(z) \leq \frac{2}{d^{10}} \frac{\sqrt{z}}{\mathfrak{b}\sqrt{d}}       \qquad \textrm{ for some }     i \in[d-1]    \right)+  \frac{1}{d^{1000}}\nonumber\\
    & =  & \mathbb{P} \left(\bigcup_{i=1}^{d-1} \left\{ \gamma_i(z) - \gamma_{i+1}(z) \leq \frac{2}{d^{10}} \frac{\sqrt{z}}{\mathfrak{b}\sqrt{d}} \right\}  \right)+  \frac{1}{d^{1000}}\nonumber\\
      & \leq & \sum_{i=1}^{d-1} \mathbb{P}\left( \gamma_i(z) - \gamma_{i+1}(z) \leq \frac{2}{d^{10}} \frac{\sqrt{z}}{\mathfrak{b}\sqrt{d}}   \right)+  \frac{1}{d^{1000}}\label{eq_n175}\\
   &  \stackrel{\textrm{Th.  \ref{thm:eigenvalue_gap}}}{\leq}&  \sum_{i=1}^{d-1}  \left(\frac{2}{d^{10}}\right)^{3} + \frac{1}{d^{1000}} \label{eq_n176}\\
   &\leq & \frac{1}{d^{997}}, \label{eq_e2}
 \end{eqnarray}
 where \eqref{eq_n174} holds by Lemma \ref{lemma_spectral_martingale_b} whenever $d \geq N_0$ for some sufficiently large universal constant $N_0$.
\eqref{eq_n175} follows from a union bound. 
 \eqref{eq_n176} holds by Theorem \ref{thm:eigenvalue_gap}, since the solution $\gamma(z) = (\gamma_1(z), \ldots, \gamma_d(z))$ to the eigenvalue evolution equations \eqref{eq_DBM_eigenvalues} at time $z$ with initial condition $\gamma(0) = (\gamma_1(0), \ldots, \gamma_d(0))$ have the same joint distribution as the eigenvalues of a random matrix $\sqrt{z}(\frac{1}{\sqrt{z}}\mathrm{diag}\left(\gamma_1(0), \ldots, \gamma_d(0)) + (G+G^\ast)\right)$ where $G$ has i.i.d. complex standard Gaussian entries.
Thus, we have,
\begin{eqnarray}
& & \!\!\!\!\!\!\! \!\!\!\!\!\!\! \!\!\!\!\!\!\!  \mathbb{P}\left(\inf_{t_0 \leq t \leq T, \, \, 1\leq i < d  }\gamma_i(t) - \gamma_{i+1}(t) \leq \frac{1}{d^{10}} \frac{\sqrt{t}}{\mathfrak{b}\sqrt{d}}\right) \nonumber\\
& = & \mathbb{P}\left(\inf_{1\leq i < d} \gamma_i(t) - \gamma_{i+1}(t) \leq \frac{1}{d^{10}} \frac{\sqrt{t}}{\mathfrak{b}\sqrt{d}}      \, \,  \textrm{ for some }   t\in [t_0, T] \right) \nonumber\\
&= & \mathbb{P}\left( \bigcup_{z\in [t_0, T] \cap \frac{1}{d^{200}} \mathbb{Z}} \left \{ \inf_{1\leq i < d  }\gamma_i(t) - \gamma_{i+1}(t) \leq \frac{1}{d^{10}} \frac{\sqrt{t}}{\mathfrak{b}\sqrt{d}}   \, \,   \textrm{ for some }   t \in [z, z + \frac{1}{d^{200}}] \cap [t_0, T] \right \} \right) \nonumber\\
&\leq & \mathbb{P}\left( \bigcup_{z\in [t_0, T] \cap \frac{1}{d^{200}} \mathbb{Z}} \left \{ \inf_{1\leq i < d }\gamma_i(t) - \gamma_{i+1}(t) \leq \frac{1}{d^{10}} \frac{\sqrt{t}}{\mathfrak{b}\sqrt{d}}   \, \,   \textrm{ for some }   t \in [z, z + \frac{1}{d^{200}}] \right \} \right) \nonumber\\
&\leq &\sum_{z\in [t_0, T] \cap \frac{1}{d^{200}}\mathbb{Z}}  \mathbb{P}\left( \inf_{1\leq i < d }\gamma_i(t) - \gamma_{i+1}(t) \leq \frac{1}{d^{10}} \frac{\sqrt{t}}{\mathfrak{b}\sqrt{d}}   \, \,   \textrm{ for some }   t \in [z, z + \frac{1}{d^{200}}] \right) \label{eq_n152}\\
&\stackrel{\textrm{Eq. } \eqref{eq_e2}}{\leq}& \sum_{z\in [t_0, T] \cap \frac{1}{d^{200}}\mathbb{Z}}  \frac{1}{d^{997}} \nonumber\\
  & \leq & d^{200} T \times  \frac{1}{d^{997}} \nonumber\\
   &\leq & \frac{T}{d^{600}}, \nonumber
\end{eqnarray}
where \eqref{eq_n152} is a union bound. 
\end{proof}

\subsubsection{Gaps between not necessarily neighboring eigenvalues} \label{section_non_neighboring_gaps}

Recall that Theorem \ref{thm:eigenvalue_gap} provides a high-probability bound on the gaps between {\em neighboring} eigenvalues of Dyson Brownian motion.
In this section, we extend the high-probability bounds of Theorem \ref{thm:eigenvalue_gap} to gaps between not-necessarily neighboring eigenvalues $\gamma_i(t) - \gamma_j(t)$  for {\em any} $i<j$ (Corollary \ref{lemma_gaps_any_start}).

The following concentration bound will be helpful in proving Corollary \ref{lemma_gaps_any_start}.
\begin{proposition}\label{prop_sum_nonindependent}
Let $F: \mathbb{R} \rightarrow \mathbb{R}$ be a nondecreasing function, and let $r \in \mathbb{N}$.
 Suppose that $X_1, \ldots, X_r$  are (not necessarily independent) non-negative random variables satisfying $\mathbb{P}(X_i \leq s) \leq F(s)$ for all $i \in [r]$ and all $s \geq 0$.
Then
\begin{equation} \label{eq_sum_nonindependent}
    \mathbb{P}\left(\sum_{i=1}^r X_i \leq \frac{1}{2}rs\right) \leq 2F(s) \qquad \forall s \geq 0. 
\end{equation}
\end{proposition}

\begin{proof}
Let $s \geq 0$. 
 Let $E$ be the ``bad'' event that $|\{i \in [r]: X_i \leq s\}| \geq \frac{r}{2}$.
Choose $J$ uniformly at random from $\{1,\ldots, r\}$, independent of the outcomes of the random variables $X_1,\cdots, X_r$.
Then
\begin{equation}\label{eq_n14}
\mathbb{P}(J \in \{i \in [r]: X_i \leq s\} | E) \geq \frac{1}{2}.
\end{equation}
Therefore,
\begin{align}\label{eq_n16}
    \mathbb{P}(X_J \leq s) =  \mathbb{P}(X_J \leq s| E) \times \mathbb{P}(E) = \mathbb{P}(J \in \{i \in [r]: X_i \leq s\} | E)  \times \mathbb{P}(E) \stackrel{\textrm{Eq. } \eqref{eq_n14}}{\geq} \frac{1}{2} \mathbb{P}(E).
\end{align}
Moreover, since $X_1,\ldots, X_r$ are non-negative random variables, we have that 
\begin{equation}\label{eq_n15}
   \left \{\sum_{i=1}^r X_i \leq \frac{1}{2}rs\right\} \subseteq E.
\end{equation}
Therefore,
\begin{equation*}
      \mathbb{P}\left(\sum_{i=1}^r X_i \leq \frac{1}{2}rs\right) \stackrel{\textrm{Eq. } \eqref{eq_n15}}{\leq}  \mathbb{P}(E) \stackrel{\textrm{Eq. } \eqref{eq_n16}}{\leq}  2  \mathbb{P}(X_J \leq s) \leq 2F(s),
\end{equation*}
where the last inequality holds since $\mathbb{P}(X_i \leq s) \leq F(s)$  for all $i \in [r]$ and since the random variable $J$ is sampled independently from the outcomes of $X_1,\cdots, X_r$.
\end{proof}

\noindent
\begin{corollary}[\bf Gaps between not-necessarily neighboring  eigenvalues]\label{lemma_gaps_any_start}
\qquad  \qquad \qquad  \qquad \qquad \qquad \qquad . 
 Let $\gamma(t) = (\gamma_1(t), \ldots, \gamma_d(t))$ be a strong solution  of \eqref{eq_DBM_eigenvalues} starting from any initial $\gamma(0) \in \mathcal{W}_d$.
Then for every $t>0$, every $i,j \in [d]$ where $i<j$, and every $\alpha>0$,
\begin{equation*}
    \mathbb{P}\left(\left\{ \gamma_i(t) - \gamma_{j}(t) \leq (j-i) \times s \frac{\sqrt{t}}{\mathfrak{b}\sqrt{d}} \right\} \cap \hat{E}_\alpha^c\right) \leq  s^3 \qquad \forall s>0, t>0.
\end{equation*}
\end{corollary}

\begin{proof}

Since the solution $\gamma(t) = (\gamma_1(t), \ldots, \gamma_d(t))$ to the eigenvalue evolution equations \eqref{eq_DBM_eigenvalues} at time $t$ with initial condition $\gamma(0) = (\gamma_1(0), \ldots, \gamma_d(0))$ has the same joint distribution as the eigenvalues of a random matrix\\ $\sqrt{t}(\frac{1}{\sqrt{t}}\mathrm{diag}\left(\gamma_1(0), \ldots, \gamma_d(0)) + (G+G^\ast)\right)$ where $G$ has i.i.d. complex standard Gaussian entries, by Theorem \ref{thm:eigenvalue_gap} we have that
\begin{equation}\label{eq_e3}
    \mathbb{P}\left(\left \{\gamma_i(t) - \gamma_{i+1}(t) \leq s \frac{\sqrt{t}}{\mathfrak{b}\sqrt{d}}\right\} \cap \hat{E}_\alpha^c\right) \leq 2s^{3} \qquad \forall s>0, \forall 1\leq i < d.
\end{equation}
Define $X_\ell := \gamma_{i+\ell}(t) - \gamma_{i+\ell+1}(t) $ for all $\ell \in \{0,1, \ldots, j-i-1\}$.
Then plugging \eqref{eq_e3} into Proposition \ref{prop_sum_nonindependent}, we have that
\begin{eqnarray*}
   &\mathbb{P}\left(\left\{ \gamma_i(t) - \gamma_{j}(t) \leq (j-i) \times s \frac{\sqrt{t}}{2\mathfrak{b}\sqrt{d}} \right \} \cap\hat{E}_\alpha^c\right)\nonumber
   =\mathbb{P}\left(\left\{ \sum_{\ell=1}^{j-i} X_j \leq \frac{1}{2}(j-i) \times s \frac{\sqrt{t}}{\mathfrak{b}\sqrt{d}} \right \} \cap \hat{E}_\alpha^c \right)\\
    & \stackrel{\textrm{Prop.  \ref{prop_sum_nonindependent}, \textrm{ Eq. } \eqref{eq_e3}}}{\leq}  4 s^{3}.
\end{eqnarray*}
Redefining $\mathfrak{b}$ to be $4$ times the original value of $\mathfrak{b}$ completes the proof.
\end{proof}

\subsection{Proof of Gaussian Unitary Ensemble eigenvalue gap bound for zero initial condition}

 Lemma \ref{lemma_gap_comparison} reduces the task of proving Theorem \ref{thm:eigenvalue_gap} to the following special case of Theorem \ref{thm:eigenvalue_gap} where the initial matrix $M = 0$:

\begin{lemma}[\bf Eigenvalue gaps of  Gaussian Unitary Ensemble (GUE) and Gaussian Orthogonal Ensemble (GOE), from zero initial condition]\label{lemma_GUE_gaps}
Let $A := G + G^\ast$ where $G$ is a matrix with i.i.d. complex (or real) standard Gaussian entries, and denote by $\eta_1,\ldots, \eta_d$ the eigenvalues of $A$.
Then
$$
    \mathbb{P}\left(\eta_i - \eta_{i+1} \leq s \frac{1}{\mathfrak{b}\sqrt{d}}\right) \leq s^{\beta +1} + \frac{1}{d^{1000}}$$ for all $s>0$, and for all $1\leq i < d$,
where $\beta=2$ for the complex Hermitian case (and $\beta=1$ for the real-symmetric case), and $\mathfrak{b} = (\log d)^{ L\log \log d}$ and $L$ is a universal constant.
\end{lemma}

\noindent Plugging in Lemma \ref{lemma_GUE_gaps} into Lemma \ref{lemma_gap_comparison} completes the proof of Theorem \ref{thm:eigenvalue_gap}.
 For simplicity of exposition, we give the proof of Lemma \ref{lemma_GUE_gaps} in this section for the complex-Hermitian GUE case (the proof for the real-symmetric GOE case follows with minor modifications; see the last paragraph of Section \ref{overview_eigenvalue_gaps} for details).

\subsubsection{Eigenvalue ridgidity}

Denote by $\eta_1,\ldots, \eta_d$ the eigenvalues of the GUE (or GOE) random matrix-- that is the matrix $G+G^\ast$ where each entry of $G$ is an independent standard complex (or real) Gaussian.
To bound the eigenvalue gaps $\eta_i-\eta_{i+1}$ for $i \geq \tilde{\Omega}(1)$, which are not near the edge of the spectrum, 
we will use the fact that
the eigenvalue gaps of the GUE/GOE satisfy a rigidity property (\cite{erdHos2012rigidity}; restated here as Lemma \ref{lemma_rigidity}).
Roughly, for every $i \in [d]$ the $i$'th eigenvalue $\eta_i$ does not deviate by more than $\mathrm{polylog}(d)$ times the average gap size $\eta_i - \eta_{i+1}.
$
More formally, for every $i\in [d]$ we define the ``classical'' eigenvalue location $\omega_i$ to be the number such that
\begin{equation}\label{eq_a7}
      d\int_{\frac{\omega_i}{\sqrt{d}}}^{\infty}\rho(x) \mathrm{d}x = i - 1,
\end{equation}
where $\rho(x):= \frac{1}{2\pi}\sqrt{\max(4 - x^2, \,\, 0)}$ is the semi-circle law.
For convenience, we also define $\omega_{d+1} := -2 \sqrt{d}$ (that way, the locations of the  $\omega_{d+1} \leq \omega_d \leq \cdots \leq \omega_1$  are symmetric about $0$).

 The following proposition, which provides upper and lower bounds on the classical eigenvalue locations $\omega_i$, will be useful when applying the eigenvalue rigidity property of Lemma \ref{lemma_rigidity}.
\begin{proposition}\label{prop_classical}
The classical eigenvalues $\omega_i$ satisfy
\begin{equation}\label{eq_a10}
2\sqrt{d}-  \frac{9}{2}d^{-\frac{1}{6}} (i-1)^{\frac{2}{3}}   \leq \omega_i \leq 2\sqrt{d}- d^{-\frac{1}{6}} (i-1)^{\frac{2}{3}} \qquad \forall \ 1\leq i \leq \frac{d}{2},
\end{equation}
\begin{equation}\label{eq_a11}
    d^{-\frac{1}{6}} (d-i+1)^{\frac{2}{3}}  -2\sqrt{d}\leq \omega_{i} \leq  \frac{9}{2}d^{-\frac{1}{6}} (d-i+1)^{\frac{2}{3}} -2\sqrt{d} \qquad \forall \ \frac{d}{2} \leq i \leq d.
\end{equation}
Moreover, their gaps satisfy
\begin{equation}\label{eq_a12}
d^{-\frac{1}{6}}\min(i, d-i+1)^{-\frac{1}{3}}\leq \omega_i - \omega_{i+1} \leq 2\pi d^{-\frac{1}{6}} \min(i, d-i+1)^{-\frac{1}{3}} \qquad \forall \ 1\leq i \leq d,
\end{equation}

\end{proposition}
\begin{proof}
Since $\rho(x)= \frac{1}{2\pi}\sqrt{4 - x^2}$ for all $x \in [-2,2]$, we have that 
\begin{equation}\label{eq_a13}
    \frac{1}{2\pi} \sqrt{x+2} \leq \rho(x) \leq \frac{1}{2\pi} 2\sqrt{x+2}  \qquad \qquad  \forall x\in[-2,1].
\end{equation}
Furthermore, since $\rho(x)$ is symmetric about $0$, \eqref{eq_a13} implies that
\begin{equation}\label{eq_a13b}
    \frac{1}{2\pi}\sqrt{2-x} \leq \rho(x) \leq \frac{1}{2\pi} 2\sqrt{2-x} \qquad \qquad \forall x\in[-1,2].
\end{equation}
Moreover, \eqref{eq_a13} also implies that
\begin{equation}\label{eq_n17}
\int_{-2}^x \rho(s) \mathrm{d} s
\stackrel{\textrm{Eq. } \eqref{eq_a13}}{\geq} \int_{-2}^x    \frac{\sqrt{2}}{2\pi} \sqrt{s+2} \mathrm{d} s = \frac{\sqrt{2}}{2\pi}  \frac{2}{3}(x+2)^{\frac{3}{2}}   \qquad \qquad  \forall x\in[-2,1],
\end{equation}
and that
\begin{equation}\label{eq_n18}
\int_{-2}^x \rho(s) \mathrm{d} s
\stackrel{\textrm{Eq. } \eqref{eq_a13}}{\leq} 2\int_{-2}^x    \frac{1}{2\pi} \sqrt{s+2} \mathrm{d} s =  \frac{1}{2\pi}  \frac{4}{3}(x+2)^{\frac{3}{2}}   \qquad \qquad  \forall x\in[-2,1].
\end{equation}
    Thus we have
      \begin{equation}\label{eq_a5}
\int_{-2}^{d^{-\frac{2}{3}} i^{\frac{2}{3}} -2} \rho(x) \mathrm{d} x \stackrel{\textrm{Eq. } \eqref{eq_n18}}{\leq} \frac{1}{2\pi} \frac{4}{3}( d^{-\frac{2}{3}} i^{\frac{2}{3}})^{\frac{3}{2}} \leq \frac{i}{d} \qquad \qquad \forall \ 1\leq i \leq d,
\end{equation}
and
      \begin{equation} \label{eq_a6}
\int_{-2}^{\frac{9}{2}d^{-\frac{2}{3}} i^{\frac{2}{3}} -2} \rho(x) \mathrm{d} x \stackrel{\textrm{Eq. } \eqref{eq_n17}}{\geq} \frac{1}{2\pi} \frac{2}{3}\left(\frac{9}{2}d^{-\frac{2}{3}} i^{\frac{2}{3}}\right)^{\frac{3}{2}} \geq \frac{i}{d} \qquad \qquad  \forall \ 1\leq i \leq \frac{d}{2}.
\end{equation}
Since $\rho(x)$ is nonnegative, $\int_{-2}^x \rho(s) \mathrm{d} s$ is nondecreasing in $x$.
Therefore, from \eqref{eq_a5} and \eqref{eq_a6}, we have by the definition of $\omega_i$ (Equation \eqref{eq_a7}) that

\begin{equation}\label{eq_a8}
d^{-\frac{2}{3}} (d-i+1)^{\frac{2}{3}}  -2 \stackrel{\textrm{Eq. } \eqref{eq_a5}}{\leq} \frac{\omega_{i}}{\sqrt{d}} \stackrel{\textrm{Eq. } \eqref{eq_a6}}{\leq} \frac{9}{2}d^{-\frac{2}{3}} (d-i+1)^{\frac{2}{3}} -2 \qquad \forall \, \,  \frac{d}{2} \leq i \leq d ,
\end{equation}
which proves \eqref{eq_a11}.
Moreover, since the density $\rho(x)$ is symmetric about $0$, \eqref{eq_a8} implies that
\begin{equation} \label{eq_a8b}
2-  \frac{9}{2}d^{-\frac{2}{3}} (i-1)^{\frac{2}{3}}   \leq \frac{\omega_i}{\sqrt{d}} \leq 2- d^{-\frac{2}{3}} (i-1)^{\frac{2}{3}}\qquad \qquad \forall \, \,   1\leq i \leq \frac{d}{2},
\end{equation}
which proves \eqref{eq_a10}.
Moreover, since $\rho(x)$ is nonincreasing  on [0,2] we  also have that for all $2\leq i \leq \frac{d}{2} +1$,
\begin{equation} \label{eq_a14}
    \frac{\omega_i}{\sqrt{d}} - \frac{\omega_{i+1}}{\sqrt{d}} \leq \frac{1}{d\times \rho(\omega_{i})} \stackrel{\textrm{Eq. \eqref{eq_a13b}}}{\leq} \frac{1}{d     \frac{1}{2\pi} \sqrt{2-\omega_{i}}} \stackrel{\textrm{Eq. \eqref{eq_a8b}}}{\leq}  \frac{2  \pi}{d\sqrt{   d^{-\frac{2}{3}} (i-1)^{\frac{2}{3}}}} \leq 2 \pi d^{-\frac{2}{3}}i^{-\frac{1}{3}}.
\end{equation}
and that, for all $1\leq i \leq \frac{d}{2} +1$,
\begin{equation} \label{eq_a15}
    \frac{\omega_i}{\sqrt{d}} - \frac{\omega_{i+1}}{\sqrt{d}} \geq \frac{1}{d\times \rho(\omega_{i+1})} \stackrel{\textrm{Eq. \eqref{eq_a13b}}}{\geq} \frac{1}{2d     \frac{1}{2\pi} \sqrt{2-\omega_{i+1}}} \stackrel{\textrm{Eq. \eqref{eq_a8b}}}{\geq}  \frac{1}{2d     \frac{1}{2\pi} \sqrt{     \frac{9}{2}d^{-\frac{2}{3}} i^{\frac{2}{3}}}} \geq \frac{\pi}{\sqrt{ \frac{9}{2}}} d^{-\frac{2}{3}}i^{-\frac{1}{3}}.
\end{equation}
Therefore,
\begin{equation}\label{eq_a9}
\frac{\pi}{\sqrt{ \frac{9}{2}}} d^{-\frac{1}{6}}i^{-\frac{1}{3}} \stackrel{\textrm{Eq. \eqref{eq_a15}}}{\leq} \omega_i - \omega_{i+1} \stackrel{\textrm{Eq. \eqref{eq_a14}}}{\leq} 2 \pi d^{-\frac{1}{6}}i^{-\frac{1}{3}} \qquad \forall 2\leq i \leq \frac{d}{2} +1.
\end{equation}
Moreover, plugging in $i=2$ to \eqref{eq_a10} and the fact that $\omega_1 = 2\sqrt{d}$, we have that
\begin{equation}\label{eq_a10b}
d^{-\frac{1}{6}}   \leq \omega_1 -  \omega_2 \leq 3d^{-\frac{1}{6}}.
\end{equation}
Therefore \eqref{eq_a9} and \eqref{eq_a10b} together imply that,
\begin{equation}\label{eq_a9b}
\frac{\pi}{\sqrt{ \frac{9}{2}}} d^{-\frac{1}{6}}i^{-\frac{1}{3}} \leq \omega_i - \omega_{i+1} \leq 2 \pi d^{-\frac{1}{6}}i^{-\frac{1}{3}} \qquad \forall 1\leq i \leq \frac{d}{2} +1.
\end{equation}
Finally,  since the density $\rho(x)$ is symmetric about $0$, \eqref{eq_a9b} implies that
\begin{equation*}
 \frac{\pi}{\sqrt{ \frac{9}{2}}} d^{-\frac{1}{6}}\min(i, d-i+1)^{-\frac{1}{3}}\leq \omega_i - \omega_{i+1} \leq 2\pi d^{-\frac{1}{6}} \min(i, d-i+1)^{-\frac{1}{3}} \qquad \forall 1\leq i \leq d,
\end{equation*}
which proves \eqref{eq_a12}.
\end{proof}

\begin{lemma}[\bf Eigenvalue rigidity of GUE/GOE (Theorem 2.2 of \cite{erdHos2012rigidity})]\label{lemma_rigidity}
There exist universal constants $C\geq 1$ and $c_1, c_2, N_0 >0$ such that for every $L \in \left[c_1, \frac{\log(10 d)}{10 (\log\log d)^2}\right]$ and every $d \geq N_0$,
\begin{equation*}
    \mathbb{P}\left(\exists j \in [d] : |\eta_j - \omega_j| \geq (\log d)^{ L\log \log d} \min(j, d-j+1)^{-\frac{1}{3}} d^{-\frac{1}{6}} )\right) \leq C \exp[ -(\log d)^{c_2 L \log \log d} ].
\end{equation*}
\end{lemma}

\subsubsection{Bounding the eigenvalue gaps of the GUE matrix}

In this section, we prove high-probability bounds for the eigenvalue gaps of the GUE random matrix (Lemma \ref{lemma_GUE_gaps}).

\paragraph{Step 1.}  Define the ``eigenvalue rigidity'' event $E$ and show that it holds with high probability (Use Lemma \ref{lemma_rigidity}).
Set $L:= \max( \frac{2}{c_2} \log \log(C), c_1, 1)$;  thus, $L$ is a universal constant.
Define the event $E$ as follows:
\begin{equation}\label{eq_n19}
E:= \left \{\eta \in \mathcal{W}_d : |\eta_j - \omega_j| < (\log d)^{ L\log \log d} \min(j, d-j+1)^{-\frac{1}{3}} d^{-\frac{1}{6}} )  \,\,\,\, \forall j \in [d] \right\},
\end{equation}
where $\mathcal{W}_d$ was defined in \eqref{eq:WeylChamber}.
Then (observing that one can replace the universal constant $N_0$ in Lemma \ref{lemma_rigidity} with another universal constant such that Lemma \ref{lemma_rigidity} holds with values of $N_0, C, c_2, c_1$ such that $\max( \frac{2}{c_2} \log \log(C), c_1) \leq \frac{\log(10 d)}{10 (\log\log d)^2}$ and $N_0 \geq e^{4}$), we have by Lemma \ref{lemma_rigidity} that
\begin{equation} \label{eq_rigidity_1}
    \mathbb{P}(E^c) \leq  C \exp[ -(\log d)^{c_2 L \log \log d}]\leq  \exp[ -(\log d)^{2 \log \log d}] \leq  \exp[ -(\log d)^{2}] \leq  d^{ -\log d} \leq \frac{1}{d^{1000}},
    \end{equation}
    for all $d \geq N_0$, where $N_0$ is a  universal constant.
Define 
\begin{equation}\label{eq_n37}
\mathfrak{b} := 10^6(\log d)^{ L\log \log d}.
\end{equation}
Further, define $\omega_j=\eta_{j} = + \infty$ for all $j < d$ and  $\omega_j=\eta_{j} = - \infty$ for all $j>d$. 

\paragraph{Step 2.} 
Show a preliminary lower bound on the gaps between non-neighboring eigenvalues whose indices are $\tilde{\Omega}(1)$ apart, which holds whenever the event $E$ occurs (Proposition \ref{prop_n1}).

Towards this end, define
 \begin{eqnarray}\label{eq_n28}
 j_{\mathrm{min}}&:=& \max(i-\mathfrak{b}^2, 1), \nonumber\\
 j_{\mathrm{max}}&:=& \min(i+\mathfrak{b}^2, d).
  \end{eqnarray}
Define the following quantities:

\begin{eqnarray}\label{eq_n26}
    a_{\mathrm{min}}&:= &\omega_{j_{\mathrm{max}}} - \frac{1}{30} \mathfrak{b}^2 d^{-\frac{1}{6}} \min(i, d-i)^{-\frac{1}{3}}\nonumber\\
a_{\mathrm{max}}&:=& \omega_{j_{\mathrm{max}}} + \frac{1}{30} \mathfrak{b}^2 d^{-\frac{1}{6}} \min(i, d-i)^{-\frac{1}{3}}\nonumber\\
b_{\mathrm{min}}&:= & \omega_{j_{\mathrm{min}}} - \frac{1}{30} \mathfrak{b}^2 d^{-\frac{1}{6}} \min(i, d-i)^{-\frac{1}{3}}\nonumber\\
b_{\mathrm{max}}&:=& \omega_{j_{\mathrm{min}}} + \frac{1}{30} \mathfrak{b}^2 d^{-\frac{1}{6}} \min(i, d-i)^{-\frac{1}{3}}
\end{eqnarray}

\begin{proposition}\label{prop_n1}
Suppose that $\eta \in E$. 
Then for all $\mathfrak{b}^2 < i < d-\mathfrak{b}^2$ we have
\begin{equation}\label{eq_c2}
    \eta_{j_{\mathrm{min}}} - \eta_{j_{\mathrm{max}}} \geq    \frac{29}{30} \mathfrak{b}^2 d^{-\frac{1}{6}} \min(i, d-i)^{-\frac{1}{3}} \geq  \frac{29}{30} \mathfrak{b}^2 \frac{1}{\sqrt{d}}.
\end{equation}
Moreover, we also have that 
$ \eta_{j_{\mathrm{max}}}  \in[a_{\mathrm{min}}, a_{\mathrm{max}}]$ and  $ \eta_{j_{\mathrm{min}}}  \in[b_{\mathrm{min}}, b_{\mathrm{max}}]$.

\end{proposition}
\begin{proof}
Without loss of generality, we may assume that 
$i \leq \frac{1}{2}d$,  since the GUE matrices $G$ and $-G$ have the same distribution and hence the joint eigenvalue distribution of the GUE is symmetric about $0$.
If $E$ occurs, then by the definition of the event $E$ \eqref{eq_n19}, we have
\begin{eqnarray}
 \eta_{j_{\mathrm{min}}} - \eta_{j_{\mathrm{max}}} 
  & \stackrel{\textrm{Eq. } \eqref{eq_n28}}{=} &
    \eta_{i- \mathfrak{b}^2} - \eta_{i+ \mathfrak{b}^2} \nonumber\\
    &\stackrel{\textrm{Eq. } \eqref{eq_n19}}{\geq} &    \omega_{i- \mathfrak{b}^2} - \omega_{i+ \mathfrak{b}^2} - 2\mathfrak{b} (i-\mathfrak{b}^2)^{-\frac{1}{3}} d^{-\frac{1}{6}} \nonumber\\
    & \stackrel{\textrm{Prop. } \ref{prop_classical}}{\geq}& \mathfrak{b}^2 \times  d^{-\frac{1}{6}} i^{-\frac{1}{3}} - 2\mathfrak{b} (i-\mathfrak{b}^2)^{-\frac{1}{3}} d^{-\frac{1}{6}} \nonumber\\
        & \geq &\mathfrak{b}^2 \times  d^{-\frac{1}{6}} i^{-\frac{1}{3}} - 2\mathfrak{b} \left(\frac{i}{2\mathfrak{b}^2}\right)^{-\frac{1}{3}} d^{-\frac{1}{6}} \label{eq_n130}\\
         & \geq &\mathfrak{b}^2 \times  d^{-\frac{1}{6}} i^{-\frac{1}{3}} -  2^{\frac{4}{3}}\mathfrak{b}^{\frac{5}{3}} i^{-\frac{1}{3}} d^{-\frac{1}{6}} \nonumber\\
         &\geq & \frac{29}{30} \mathfrak{b}^2 d^{-\frac{1}{6}} i^{-\frac{1}{3}}, \label{eq_n131}
\end{eqnarray}
where \eqref{eq_n130} holds since $\frac{i}{2\mathfrak{b}^2} \leq i-\mathfrak{b}^2$ because $i \geq \mathfrak{b}^2 +1 > 4$, and \eqref{eq_n131} holds since $\mathfrak{b} \geq 10^6$.
This proves \eqref{eq_c2}.
Moreover, by the definition of the event $E$, we also have that 
\begin{eqnarray}
   |\eta_{j_{\mathrm{min}}} -  \omega_{j_{\mathrm{min}}}| 
   &\stackrel{\textrm{Eq. } \eqref{eq_n28}}{=} & |\eta_{i- \mathfrak{b}^2} -  \omega_{i- \mathfrak{b}^2}| \nonumber\\ &\stackrel{\textrm{Eq. } \eqref{eq_n19}}{\leq}&   \mathfrak{b} (i-\mathfrak{b}^2)^{-\frac{1}{3}} d^{-\frac{1}{6}}\nonumber\\
    &\leq & \mathfrak{b} \left(\frac{i}{2\mathfrak{b}^2}\right)^{-\frac{1}{3}} d^{-\frac{1}{6}} \label{eq_n132}\\
    &\leq & 2^{\frac{1}{3}}\mathfrak{b}^{\frac{5}{3}} i^{-\frac{1}{3}} d^{-\frac{1}{6}}\nonumber\\
    &\leq &\frac{1}{30} \mathfrak{b}^2 d^{-\frac{1}{6}} i^{-\frac{1}{3}},  \label{eq_c1}
\end{eqnarray}
where \eqref{eq_n132} holds since $\frac{i}{2\mathfrak{b}^2} \leq i-\mathfrak{b}^2$ because $i \geq \mathfrak{b}^2 +1 > 4$, and \eqref{eq_c1} holds since $\mathfrak{b} \geq 10^6$.
Thus, by definition \eqref{eq_n26}, Inequality \eqref{eq_c1} implies that $\eta_{j_{\mathrm{min}}}  \in[b_{\mathrm{min}}, b_{\mathrm{max}}]$. 
Again, by the definition of the event $E$, we also have that 
\begin{eqnarray}
   |\eta_{j_{\mathrm{max}}} -  \omega_{j_{\mathrm{max}}}| 
   &\stackrel{\textrm{Eq. } \eqref{eq_n28}}{=}&  |\eta_{i+ \mathfrak{b}^2} -  \omega_{i+ \mathfrak{b}^2}| \nonumber\\ &\stackrel{\textrm{Eq. } \eqref{eq_n19}}{\leq} &  \mathfrak{b} (i+\mathfrak{b}^2)^{-\frac{1}{3}} d^{-\frac{1}{6}}\nonumber\\
    &\leq &\mathfrak{b} i^{-\frac{1}{3}} d^{-\frac{1}{6}}\nonumber\\
    &\leq & \frac{1}{30} \mathfrak{b}^2 d^{-\frac{1}{6}} i^{-\frac{1}{3}} \label{eq_c1b}
\end{eqnarray}
where \eqref{eq_c1b} holds since $\mathfrak{b} \geq 10^6$.
Thus, by definition \eqref{eq_n26}, Inequality \eqref{eq_c1b} implies that $\eta_{j_{\mathrm{max}}}  \in[a_{\mathrm{min}}, a_{\mathrm{max}}]$.
\end{proof}
\noindent
We use Proposition \ref{prop_n1} to define three different sets which we will use in the next steps of the proof.
By Proposition \ref{prop_n1}, if $\mathfrak{b}^2 < i < d-\mathfrak{b}^2$, whenever the event $E$ occurs we have that $ \eta_{j_{\mathrm{max}}}  \in[a_{\mathrm{min}}, a_{\mathrm{max}}]$ and  $ \eta_{j_{\mathrm{min}}}  \in[b_{\mathrm{min}}, b_{\mathrm{max}}]$.
Consider any $a,b$ such that $a_{\mathrm{min}} \leq a \leq  a_{\mathrm{max}}$  and $b_{\mathrm{min}} \leq b \leq  b_{\mathrm{max}}$.
 Define the sets 
\begin{itemize}
\item \begin{equation}\label{eq_n25}
    S_0(a,b):= \{ \eta \in \mathcal{W}_d : \eta_{j_{\mathrm{max}}} = a, \eta_{j_{\mathrm{min}}} = b\},
\end{equation}

\item  
\begin{equation}\label{eq_n22}
S_3(a,b; y) :=  \{\eta \in \mathcal{W}_d : \eta_i-\eta_{i+1} = y\} \cap S_0(a,b) \qquad \textrm{ for any } y \leq s\frac{1}{8\mathfrak{b}^4\sqrt{d} }, 
\end{equation}

\item 
\begin{equation*}
S_4(a,b) :=  \{\eta \in \mathcal{W}_d : \eta_i-\eta_{i+1} \geq s\} \cap S_0(a,b),
\end{equation*}

\end{itemize}
 where $\mathcal{W}_d$ was defined in \eqref{eq:WeylChamber}.

\paragraph{Step 3.} Next, we define a map from the ``bad'' set $S_3$ to the ``good'' set $S_4$, which, roughly speaking, will allow us to show that the good set has a much bigger volume and a much larger probability density than the bad set.
  More specifically, for any $y \leq s\frac{1}{\mathfrak{b}\sqrt{d}}$, we want to define a map $g:   S_2(a,b;y) \rightarrow S_4(a,b)$, such that its Jacobian $J_g(\eta)$ satisfies $\mathrm{det}(J_g(\eta)) \geq \Omega(\frac{1}{s})$ and 
\begin{equation}\label{eq_b3}
    \frac{f(\eta)}{f(g(\eta))} \leq (\mathfrak{b} \sqrt{d} )^2 \times y^2,
    \end{equation}
  for any $\eta \in  S_3(a,b;y) \cap E$.
  Towards this end, we consider the map $g:  \mathcal{W}_d \rightarrow \mathcal{W}_d$ such that 

  \begin{itemize}
  
  \item 
\begin{equation}\label{eq_g3}
      g(\eta)[j] = \eta_j \qquad \forall j \notin [j_{\mathrm{min}}, j_{\mathrm{max}}]
      \end{equation}
  
  \item 
\begin{equation}\label{eq_g4}
      g(\eta)[j_{\mathrm{max}}] = \eta_{j_{\mathrm{max}}} = a,
      \end{equation}

\item   
\begin{equation}\label{eq_g2}
    g(\eta)[j] = g(\eta)[j+1] + (1-\alpha) \times (\eta_j - \eta_{j+1}) \qquad \forall j \in[j_{\mathrm{min}}, j_{\mathrm{max}} -1]\backslash \{i\}
\end{equation}
  
  \item  
  \begin{equation}\label{eq_g1}
  g(\eta)[i] = g(\eta)[i+1] + \left(\frac{2}{s}(\eta_{i}- \eta_{i+1})+ 2\frac{1}{8\mathfrak{b}^4\sqrt{d}}\right) \times \frac{b-a -(\eta_{i}- \eta_{i+1})}{b-a},
  \end{equation}
\end{itemize}
where 
\begin{equation}\label{eq_n20}
\alpha := \frac{ \frac{2}{s}(\eta_{i}- \eta_{i+1}) +2\frac{1}{8\mathfrak{b}^4\sqrt{d}} }{b-a}.
\end{equation}

The following proposition provides some preliminary facts about the map $g$ which we will use when bounding its Jacobian determinant $\mathrm{det}(J_g(\eta))$ (Lemma \ref{prop_Jacobian}) and the density ratio $\frac{f(\eta)}{f(g(\eta))}$ (Lemma \ref{lemma_density_ratio}). 
\begin{proposition} \label{prop_map}
Suppose that  $\mathfrak{b}^2 < i < d-\mathfrak{b}^2$.
Then the following properties hold for $g$ for any $\eta \in E$:
\begin{itemize}

\item For any $z \in \mathcal{W}_d$ the pre-image $g^{-1}(\{z\}) := \{ \eta \in \mathcal{W}_d : g(\eta) = z\}$ has cardinality $|g^{-1}(\{z\}) |  \leq 2$. 

\item \begin{equation}\label{eq_n36}
    g(\eta)[j_{\mathrm{min}}] = \eta_{j_{\mathrm{min}}} = b,
\end{equation}

\item  $g(\eta)[i] -  g(\eta)[i+1]  \geq   \frac{1}{8\mathfrak{b}^4\sqrt{d}}$,  and hence  

\begin{equation}\label{eq_b5}
\frac{\eta_{i} - \eta_{i+1}}{g(\eta)[i] -  g(\eta)[i+1]}  \leq  8\mathfrak{b}^4 \sqrt{d}  \times (\eta_{i} - \eta_{i+1}) =  8 \mathfrak{b}^4 \sqrt{d}  \times y
\end{equation}
  for any $\eta \in  S_3(a,b;y)$ and any $y \leq s\frac{1}{8\mathfrak{b}^4\sqrt{d} }$.

\item  \begin{equation}\label{eq_b4}
g(\eta)[j] -  g(\eta)[j+1] \geq (1- \alpha)(\eta_{j} - \eta_{j+1}) \qquad \forall j \in [d].
\end{equation}
\end{itemize}

Moreover, we also have that 
\begin{equation}\label{eq_n42}
    b-a \geq  \mathfrak{b}^2 \times d^{-\frac{1}{6}} (\min(i, d-i))^{-\frac{1}{3}} \geq  \frac{1}{\sqrt{d}}.
\end{equation}

\end{proposition}

\begin{proof}
\noindent
{\em Cardinality of pre-image.}
 We will show that, given any vector $z \in \mathbb{R}^d$ there are {\em at most} two solutions $\eta \in \mathcal{W}_d$ to the equation $g(\eta) = z$.
To solve for $\eta$, we first solve for $\eta_{i}- \eta_{i+1}$ by solving equation \eqref{eq_g1} for $\eta_{i}- \eta_{i+1}$.
 As \eqref{eq_g1} is a quadratic equation in $\eta_{i}- \eta_{i+1}$, there are at most two solutions for $\eta_{i}- \eta_{i+1}$ to this equation.
This gives us at most two solutions for $\eta_{i}- \eta_{i+1}$ in terms of $g(\eta)[i] - g(\eta)[i+1]$.

Next, we show that for any one of these two solutions, which we denote by $\Delta \in \mathbb{R}$, there is at most one value $\eta \in \mathcal{W}_d$ such that $g(\eta) = z$ and $\eta_{i}- \eta_{i+1} = \Delta$.
 Specifically, to solve for this value $\eta$, we plug in the value of $\eta_{i}- \eta_{i+1} = \Delta$ to \eqref{eq_n20} to compute $ \alpha = \frac{ \frac{2}{s}\Delta +2\frac{1}{\mathfrak{b}\sqrt{d}} }{b-a}$, and  for every $j \in[j_{\mathrm{min}}, j_{\mathrm{max}}]$, plug in this value of $\alpha$ to \eqref{eq_g2}, to solve for $\eta_j-\eta_{j+1}$ in terms of $g(\eta)[j] - g(\eta)[j+1]$.

Finally, since $\eta_{j_{\mathrm{max}}} = a$ by \eqref{eq_g4}, we can compute  $\eta_j = a+ \sum_{\ell=j}^{j_{\mathrm{max}} -1} \eta_\ell - \eta_{\ell+1}$ for each $j \in[j_{\mathrm{min}}, j_{\mathrm{max}}-1]$.
Thus, given any vector $z \in \mathcal{W}_d$, and any $\Delta \in \mathbb{R}$, we can solve for at most one $\eta \in \mathcal{W}_d$ such that $g(\eta) = z$ and $\eta_{i}- \eta_{i+1} = \Delta$.
As we have already shown that for any $z \in \mathcal{W}_d$ there are {\em at most} two values of $\Delta \in \mathbb{R}$ such that $g(\eta) = z$ and $\eta_{i}- \eta_{i+1} = \Delta$, we must have that for any $z \in \mathcal{W}_d$ the equation  $g(\eta) = z$ has at most two solutions.  Therefore  $|g^{-1}(\{z\})| \leq 2$ for all $z \in \mathcal{W}_d$.

\medskip
\noindent
{\em Showing that $g(\eta)[j_{\mathrm{min}}] = b$ \eqref{eq_n36}.}
\begin{eqnarray*}
    g(\eta)[j_{\mathrm{min}}] &\stackrel{\textrm{ Eq. } \eqref{eq_g3}, \eqref{eq_g4}}{=}& a+ \sum_{\ell=j_{\mathrm{min}}}^{j_{\mathrm{max} -1}}  g(\eta)[\ell] -  g(\eta)[\ell+1]\\
&\stackrel{\textrm{Eq. } \eqref{eq_g2},  \eqref{eq_g1}}{=}& a+  \left(\frac{2}{s}(\eta_{i}- \eta_{i+1})+ 2\frac{1}{8\mathfrak{b}^4\sqrt{d}}\right) \times \frac{b-a -(\eta_i- \eta_{i+1})}{b-a} \nonumber\\  
& & + \sum_{\ell\in [j_{\mathrm{min}}, j_{\mathrm{max}} ] \backslash \{i\}}    (1-\alpha) \times (\eta_j - \eta_{j+1})\\
&\stackrel{\textrm{Eq. } \eqref{eq_n20}}{=}& b.
\end{eqnarray*}

\medskip
\noindent
{\em Showing \eqref{eq_b5}.}
Since  $y \leq s\frac{1}{8\mathfrak{b}^4\sqrt{d} }$ and $\eta \in  S_3(a,b;y)$, 
we have that
\begin{equation}\label{eq_n21}
    \eta_i-\eta_{i+1} \stackrel{\textrm{Eq. } \eqref{eq_n22}}{=} y \leq s\frac{1}{8\mathfrak{b}^4\sqrt{d} }.
    \end{equation}
Thus by \eqref{eq_g1},
  \begin{eqnarray}\label{eq_n23}
  g(\eta)[i] - g(\eta)[i+1] &\stackrel{\textrm{Eq. } \eqref{eq_g1}}{=} &\left(\frac{2}{s}(\eta_{i}- \eta_{i+1})+ 2\frac{1}{8\mathfrak{b}^4\sqrt{d}}\right) \times \frac{b-a -(\eta_{i}- \eta_{i+1})}{b-a}\nonumber\\
  &\stackrel{\textrm{Eq. } \eqref{eq_n21}}{\geq}& 2\frac{1}{8\mathfrak{b}^4\sqrt{d}} \times \frac{1}{2}\nonumber\\
  &=& \frac{1}{8\mathfrak{b}^4\sqrt{d}}.
  \end{eqnarray}
Hence,  
\begin{equation}
\frac{\eta_{i} - \eta_{i+1}}{g(\eta)[i] -  g(\eta)[i+1]}  \stackrel{\textrm{Eq. } \eqref{eq_n23}}{\leq}  8\mathfrak{b}^4 \sqrt{d}  \times (\eta_{i} - \eta_{i+1}) \stackrel{\textrm{Eq. } \eqref{eq_n22}}{=}   8\mathfrak{b}^4 \sqrt{d}  \times y,
\end{equation}
which proves \eqref{eq_b5}.

\medskip
\noindent
{\em Showing \eqref{eq_b4}.}
By \eqref{eq_g1}, we have
  \begin{eqnarray}
  g(\eta)[i] - g(\eta)[i+1] &\stackrel{\textrm{Eq. } \eqref{eq_g1}}{=} &\left(\frac{2}{s}(\eta_{i}- \eta_{i+1})+ 2\frac{1}{\mathfrak{b}\sqrt{d}}\right) \times \frac{b-a -(\eta_{i}- \eta_{i+1})}{b-a} \nonumber\\
  &\stackrel{\textrm{Eq. } \eqref{eq_n21}}{\geq}& \left(\frac{2}{s}(\eta_{i}- \eta_{i+1})+ 2\frac{1}{\mathfrak{b}\sqrt{d}}\right) \times \frac{1}{2} \nonumber\\
    &\geq & \frac{1}{s}(\eta_{i}- \eta_{i+1}) \nonumber\\
        &\geq & (\eta_{i}- \eta_{i+1}) \label{eq_n133}\\
        &\stackrel{\textrm{Eq. } \eqref{eq_n20}}{\geq}& (1-\alpha)\times (\eta_{i}- \eta_{i+1}), \label{eq_n24}
  \end{eqnarray}
  where \eqref{eq_n133} holds since $s\leq 1$, and \eqref{eq_n24} holds since $\alpha \geq 0$ by \eqref{eq_n20}.
  Thus, \eqref{eq_b4} holds for $j=i$ by \eqref{eq_n24}.
  Moreover, \eqref{eq_b4} holds for all  $j \in[j_{\mathrm{min}}, j_{\mathrm{max}}]\backslash \{i\}$ by \eqref{eq_g2} and  \eqref{eq_b4} holds for  all $j \notin [j_{\mathrm{min}}, j_{\mathrm{max}}]$ by \eqref{eq_g3}.
  Therefore \eqref{eq_b4} holds for all $j \in [d]$.

\noindent
  {\em Showing \eqref{eq_n42}.}

\begin{eqnarray}
    b-a  &\stackrel{\textrm{ Eq. } \eqref{eq_n36}, \,\, \eqref{eq_n25}}{=} &\eta_{j_{\mathrm{min}}} - \eta_{j_{\mathrm{max}}}\nonumber\\
    &\stackrel{\textrm{Prop. } \ref{prop_n1}}{\geq}& b_{\mathrm{min}} - a_{\mathrm{max}} \nonumber\\
    &\stackrel{\textrm{Eq. } \eqref{eq_n26}}{=}& \left(\omega_{j_{\mathrm{min}}} - \frac{1}{30} \mathfrak{b}^2 d^{-\frac{1}{6}} i^{-\frac{1}{3}}\right) - \left(\omega_{j_{\mathrm{max}}} + \frac{1}{30} \mathfrak{b}^2 d^{-\frac{1}{6}} i^{-\frac{1}{3}}\right)\nonumber\\
    &=& \omega_{j_{\mathrm{min}}} - \omega_{j_{\mathrm{max}}} - \frac{1}{15} \mathfrak{b}^2 d^{-\frac{1}{6}} i^{-\frac{1}{3}}\nonumber\\
      &\stackrel{\textrm{Prop. } \ref{prop_classical}, \, \, \textrm{Eq. } \eqref{eq_n28}}{\geq} &  2\mathfrak{b}^2 \times d^{-\frac{1}{6}} (i+  \mathfrak{b})^{-\frac{1}{3}} - \frac{1}{15} \mathfrak{b}^2 d^{-\frac{1}{6}} i^{-\frac{1}{3}}\nonumber\\
            &\stackrel{\textrm{Prop. } \ref{prop_classical}, \, \, \textrm{Eq. } \eqref{eq_n28}}{\geq} &  1.5\mathfrak{b}^2 \times d^{-\frac{1}{6}} (i+  \mathfrak{b})^{-\frac{1}{3}} - \frac{1}{15} \mathfrak{b}^2 d^{-\frac{1}{6}} i^{-\frac{1}{3}}\nonumber\\
          &\geq&  \mathfrak{b}^2 \times d^{-\frac{1}{6}} i^{-\frac{1}{3}}, \label{eq_n27}
\end{eqnarray}
where \eqref{eq_n27} holds since $\mathfrak{b} \geq 10^6$.

Recall that, without loss of generality, we may assume that $i \leq \frac{1}{2}d$,  since the GUE matrices $G$ and $-G$ have the same distribution and hence the joint eigenvalue distribution of the GUE is symmetric about $0$.
Thus we have
\begin{equation*}
    b-a \stackrel{\textrm{Eq. } \eqref{eq_n27}}{\geq}  \mathfrak{b}^2 \times d^{-\frac{1}{6}} (\min(i, d-i))^{-\frac{1}{3}} \geq  \frac{1}{\sqrt{d}}.
\end{equation*}
\end{proof}

\paragraph{Step 4.} Bounding the Jacobian determinant of the map $g$.

\begin{lemma}[\bf Jacobian determinant of $g$]\label{prop_Jacobian}
If  $y \leq s\frac{1}{8\mathfrak{b}^4\sqrt{d} }$ and $\eta \in  S_3(a,b;y) \cap E$, we have that
$$\mathrm{det}(J_g(\eta)) \geq \frac{1}{16 s}.$$
\end{lemma}
\begin{proof}

    Since $\frac{1}{b-a} \leq 2 \sqrt{d}$ by \eqref{eq_n42} of Proposition \ref{prop_map},  and $\mathfrak{b} > 100$ and $s<1$, we have that
\begin{equation}\label{eq_n29}
    2\frac{1}{\mathfrak{b}\sqrt{d} (b-a)} \stackrel{\textrm{Eq. } \eqref{eq_n42}  \textrm{ of Prop. } \ref{prop_map}}{\leq} \frac{1}{2s}.
\end{equation}
    Moreover, since $\eta \in E$ we have
\begin{eqnarray}
\eta_{i} - \eta_{i+1} &\leq&    \eta_{i - 8\mathfrak{b}} - \eta_{i + 8\mathfrak{b}} \nonumber\\
&\leq&  \omega_{i - 8\mathfrak{b}} - \omega_{i + 8\mathfrak{b}} + |\eta_{i - 8\mathfrak{b}} - \omega_{i - 8\mathfrak{b}}|  + |\eta_{i + 8\mathfrak{b}} - \omega_{i + 8\mathfrak{b}}| \nonumber\\
&\stackrel{\textrm{Prop. } \ref{prop_classical}, \, \, \textrm{Eq. } \eqref{eq_n19}}{\leq}&  16 \mathfrak{b}\times 2 \pi  d^{-\frac{1}{6}}  (i+2\mathfrak{b})^{-\frac{1}{3}} + 2 \mathfrak{b} \times (i-2\mathfrak{b})^{-\frac{1}{3}} d^{-\frac{1}{6}} \nonumber\\
&\leq& 32 \pi \mathfrak{b}\times  d^{-\frac{1}{6}}  i^{-\frac{1}{3}} + 4 \mathfrak{b} \times i^{-\frac{1}{3}} d^{-\frac{1}{6}} \label{eq_n54}\\
&\leq&  \frac{1}{10} \mathfrak{b}^2 \times d^{-\frac{1}{6}} i^{-\frac{1}{3}}  \label{eq_n56}\\
&\stackrel{\textrm{Prop. } \ref{prop_map}}{\leq}& \frac{1}{10}(b-a)  \label{eq_n57}
\end{eqnarray}
where \eqref{eq_n54} and \eqref{eq_n56} hold since $\mathfrak{b} \geq 10^6$.

    Therefore,
    \begin{eqnarray}\label{eq_n30}
    \frac{1}{s (b-a)}(\eta_{i}- \eta_{i+1}) &\stackrel{\textrm{Eq. } \eqref{eq_n57}}{\leq} & \frac{1}{10 s}.
    \end{eqnarray}
Consider the map $h: \mathcal{W}_{d} \rightarrow \mathbb{R}^{d}$, where   $h(\eta)[j] = \eta_j - \eta_{j+1}$ for $j \in[j_{\mathrm{min}}, j_{\mathrm{max}} -1]$ and $h(\eta)[j] = \eta_{j}$ for $j \in [1,d] \backslash [j_{\mathrm{min}}, j_{\mathrm{max}} -1]$.
The map $h$ is injective since for any $\Delta \in \mathbb{R}^d$ that is in the range of $h$ we can solve for the unique $\eta \in \mathcal{W}_d$ such that $h(\eta) = \Delta$.
 Specifically, the unique solution $\eta$, which we denote by $h^{-1}(\Delta)$, is given by 
 \begin{equation}\label{eq_n50}
 h^{-1}(\Delta)[j] := \eta_j = \begin{cases}
  \Delta_j, \qquad \qquad \qquad \qquad \qquad \, \,   j \in [1,d] \backslash [j_{\mathrm{min}}, j_{\mathrm{max}} -1]\\
  \Delta_{j_\mathrm{max}} + \sum_{r = 1}^{ j_{\mathrm{max}} - j} \Delta_{j_{\mathrm{max}} - r}, \qquad   j \in[j_{\mathrm{min}}, j_{\mathrm{max}} -1].
 \end{cases}
 \end{equation}
 Thus, \eqref{eq_n50} implies that $h$ is injective.
  Moreover,   \eqref{eq_n50} also implies that for every $j \in [d]$ and every $\Delta \in \mathbb{R}^d$ that is in the range of $h$,
  \begin{equation}\label{eq_n51}
  h^{-1}(\Delta)[j]- h^{-1}(\Delta)[j+1] = \Delta_j \qquad \qquad \forall  j \in[j_{\mathrm{min}}, j_{\mathrm{max}} -1].
    \end{equation}
Moreover, from \eqref{eq_g3}-\eqref{eq_g2} we have that for every $\eta \in \mathcal{W}_d$ and $j \in [d]$,
\begin{equation}\label{eq_n52}
\textrm{$g(\eta)[j]- g(\eta)[j+1]$ is a function of only $\eta_j- \eta_{j+1}$ and $\eta_i - \eta_{i+1}$,}
\end{equation}
and does not otherwise depend on any $\eta_\ell$ for $\ell \neq j$, $\ell \notin \{i, i+1 \}$.
Therefore, by \eqref{eq_n50}, \eqref{eq_n51} and \eqref{eq_n52},  for every $\Delta \in  \{h(\eta) :  \eta \in  S_3(a,b;y) \cap E \}$, and every $j \in [d]$, we have that 
\begin{equation}\label{eq_n53}
\textrm{$h(g(h^{-1}(\Delta)))[j]$ is a function of only $\Delta_j$ and $\Delta_i$,}
\end{equation}
 and does not otherwise depend on any $\Delta_\ell$ for $\ell \notin\{j, i\}$.
Thus,  for every $\Delta \in  \{h(\eta) :  \eta \in  S_3(a,b;y) \cap E \}$, we have that
\begin{equation} \label{eq_derivative1}
    \frac{\partial h \circ g\circ h^{-1}(\Delta)[\ell]}{\partial \Delta_j} \stackrel{\textrm{Eq. } \eqref{eq_n53}}{=} 0 \qquad \forall \ell \neq j, \ell \neq i.
\end{equation}
Moreover, we also have that, for every $\Delta \in \mathbb{R}^d$,
\begin{equation}\label{eq_derivative1_b}
    \frac{\partial h \circ g\circ h^{-1}(\Delta)[j]}{\partial \Delta_j} \stackrel{\textrm{Eq. } \eqref{eq_g4},\eqref{eq_g2}, \eqref{eq_n50}, \eqref{eq_n53}}{=} 1-\alpha \qquad \forall j \neq i, \quad  j \in  [j_{\mathrm{min}}, j_{\mathrm{max}}],
\end{equation}
\begin{align}\label{eq_derivative1_c}
    \frac{\partial h \circ g\circ h^{-1}(\Delta)[i]}{\partial \Delta_i} & \stackrel{\textrm{Eq. } \eqref{eq_g1},  \eqref{eq_n50}, \eqref{eq_n53}}{=} \frac{2}{s} -\frac{1}{4\mathfrak{b}^4\sqrt{d} (b-a)} -\frac{1}{s (b-a)}\Delta_i  \stackrel{\textrm{Eq. } \eqref{eq_n29},\, \eqref{eq_n30}}{\geq} \frac{1}{2s},
\end{align}
\begin{equation}\label{eq_derivative1_d}
    \frac{\partial h \circ g\circ h^{-1}(\Delta)[j]}{\partial \Delta_{j}} \stackrel{\textrm{Eq. } \eqref{eq_g3}}{=} 1 \qquad \forall j \notin  [j_{\mathrm{min}}, j_{\mathrm{max}}].
\end{equation}

\noindent
Moreover, since for any $\Delta \in  \{h(\eta) :  \eta \in  S_3(a,b;y) \cap E \}$,  we have $\Delta_i \leq s\frac{1}{8\mathfrak{b}^4\sqrt{d} }$ because $y \leq s\frac{1}{8\mathfrak{b}^4\sqrt{d} }$ and $h^{-1}(\Delta) \in  S_3(a,b;y)$, we also have that $\alpha= \frac{ \frac{2}{s}\Delta_i +2\frac{1}{8\mathfrak{b}^4\sqrt{d}} }{b-a} \leq  \frac{1}{(b-a)\mathfrak{b}^4\sqrt{d}}$.
Thus, the Jacobian matrix $J_{h \circ g \circ h^{-1}}(\Delta)$  has diagonal entries $1-\alpha \geq 1- \frac{1}{(b-a)\mathfrak{b}^2\sqrt{d}}$ for $j \in  [j_{\mathrm{min}}, j_{\mathrm{max}}-1] \backslash \{i\}$ by \eqref{eq_derivative1_b},  and $i$'th diagonal entry greater than or equal to $\frac{1}{2s}$ by \eqref{eq_derivative1_c}, and all other diagonal entries equal to $1$ by \eqref{eq_derivative1_d}.
Moreover, if one exchanges the $i$'th row and column of $J_{h \circ g \circ h^{-1}}(\Delta)$ with its first row and column,  by \eqref{eq_derivative1} the resulting matrix is a  $d\times d$ upper triangular matrix with the same determinant as  $J_{h \circ g \circ h^{-1}}(\Delta)$.
Thus, by Sylvester's formula, the determinant of $J_{h \circ g \circ h^{-1}}(\Delta)$ is equal to the product of its diagonal entries.
 Thus, for any $\Delta \in  \{h(\eta) :  \eta \in  S_3(a,b;y) \cap E \}$,
\begin{eqnarray}
\mathrm{det}(J_{h \circ g \circ h^{-1}}(\Delta)) &\geq &\frac{1}{2s} \left(1- \frac{1}{(b-a)\mathfrak{b}^4\sqrt{d}}\right)^{j_{\mathrm{max}} - j_{\mathrm{min}}- 2} \times 1 \nonumber\\
&\geq & \frac{1}{2s} \left(1- \frac{1}{(b-a)\mathfrak{b}^4\sqrt{d}}\right)^{2 \mathfrak{b}^2} \label{eq_n134}\\
&\stackrel{\textrm{Eq. } \eqref{eq_n42} \textrm{ of Prop. } \ref{prop_map}}{\geq}& \frac{1}{2s}\left(1- \frac{1}{\mathfrak{b}^2}\right)^{2 \mathfrak{b}^2} \nonumber\\
&\geq &\frac{1}{16 s}, \label{eq_g5}
\end{eqnarray}
where \eqref{eq_n134} holds because $j_{\mathrm{max}}-j_{\mathrm{min}} =  \min(i+\mathfrak{b}^2, d) - \max(i-\mathfrak{b}^2, 1) \leq 2 \mathfrak{b}^2$.
\eqref{eq_g5} holds since $\mathfrak{b}^2 >100$ and $(1-\frac{1}{100})^{200} > \frac{1}{16}$. 
Hence,

\begin{eqnarray*}
    \mathrm{det}(J_{g}(\eta)) &= &  \mathrm{det}(J_{  h^{-1} \circ h \circ g \circ h^{-1} \circ h}(\eta)) \\ 
    &=& \mathrm{det}( J_{h^{-1}}(\eta) \times J_{h \circ g \circ h^{-1}}(h(\eta))  \times J_{h}(\eta))\\
    &= & \mathrm{det}(J_{h^{-1}}(\eta)) \times \mathrm{det}(J_{h \circ g \circ h^{-1}}(h(\eta))  \times \mathrm{det}(J_{h}(\eta))\\
        &= & \mathrm{det}(J_{h}(\eta))^{-1} \times \mathrm{det}(J_{h \circ g \circ h^{-1}}(h(\eta))  \times \mathrm{det}(J_{h}(\eta))\\
    &=&\mathrm{det}(J_{h \circ g \circ h^{-1}}(h(\eta))  \times 1\\
    &\stackrel{\textrm{Eq. \eqref{eq_g5}}}{\geq}& \frac{1}{16s}.
\end{eqnarray*}
\end{proof}

\paragraph{Step 5.} {This step is a mean-field approximation for far-away eigenvalues.}
This mean-field approximation (Lemma \ref{lemma_mean_field}) will allow us to bound the component of the density ratio $\frac{f(\eta)}{f(g(\eta))}$ which corresponds to the gaps between ``far-away'' eigenvalues $\eta_i - \eta_j$ with indices $i,j$ at least roughly $i-j \geq \Omega(\mathfrak{b}) = \tilde{\Omega}(1)$ apart (Lemma \ref{lemma_density_ratio})

\begin{lemma}[\bf Mean-field approximation for far-away eigenvalues]\label{lemma_mean_field}

For any $y \leq s\frac{1}{8\mathfrak{b}^4\sqrt{d} }$ and any $\eta\in  S_3(a,b;y) \cap E$, we have that
\begin{equation}
\prod_{j \in [j_{\mathrm{min}} , j_{\mathrm{max}}] , \, \, \, \ell \notin [j_{\mathrm{min}} - 2\mathfrak{b} , j_{\mathrm{max}} + 2\mathfrak{b}]} \frac{| \eta_{j} - \eta_\ell|^2}{| g(\eta)[j] - g(\eta)[\ell]|^2}  \leq 2.
\end{equation}

\end{lemma}

\begin{proof}
Since  $y \leq s\frac{1}{8\mathfrak{b}^4\sqrt{d} }$ and $\eta \in  S_3(a,b;y)$, 
we have that
\begin{equation}\label{eq_n31}
    \eta_i-\eta_{i+1} \stackrel{\textrm{Eq. } \eqref{eq_n22}}{=} y \leq s\frac{1}{8\mathfrak{b}^4\sqrt{d} }.
    \end{equation}
Moreover, since $\eta \in E$, by Proposition \ref{prop_map} and \eqref{eq_n25} we have
\begin{equation}\label{eq_n32}
    b-a  \stackrel{\textrm{Prop. } \ref{prop_map}, \textrm{ Eq. } \eqref{eq_n25}}{=} \eta_{j_{\mathrm{min}}} - \eta_{j_{\mathrm{max}}}.
    \end{equation}
Consider any $j \in [j_{\mathrm{min}}, j_{\mathrm{max}}]$ and any $r \geq 2\mathfrak{b}$.
Then, since $\eta \in E$,  by the definition of the event $E$ \eqref{eq_n19} and Proposition \ref{prop_classical} we have that 
\begin{equation}\label{eq_b10}
    |\eta_j - \eta_{j+r} | = r \frac{1}{2\sqrt{d}} + \rho
\end{equation}
for some $\rho \geq 0$.
Moreover, we have that 
\begin{eqnarray}\label{eq_b9}
& & \!\!\!\!\!\!\!\!\!\!\!\!\!\!\!\!\!\!\!\!\!\!\!\!\!\!\!\!\!\!\!\!\!\!\!\!\!\!\!\!\!\!\!\!\!\!\!\! |(g(\eta)[j] - g(\eta)[j+r]) - (\eta_j - \eta_{j+r})|
\nonumber \\
&=&\left|\sum_{\ell=0}^{r-1} (g(\eta)[j+\ell] - g(\eta)[j+\ell+1]) - (\eta_{j+\ell} - \eta_{j+\ell+1})\right|\nonumber\\
&\stackrel{\textrm{Eq. } \eqref{eq_g2},\, \eqref{eq_g3}}{=}& \left |\sum_{\ell\in [j, j+r-1] \cap [j_{\mathrm{min}}, j_{\mathrm{max}}]} \alpha (\eta_{j+\ell} - \eta_{j+\ell+1}) \right |\nonumber\\
& \leq &\alpha \times (\eta_{j_{\mathrm{min}}} - \eta_{j_{\mathrm{max}}}) \nonumber\\
   &\stackrel{\textrm{Eq. } \eqref{eq_n32}}{=}& \alpha (b-a)    \nonumber\\
    &\stackrel{\textrm{Eq. } \eqref{eq_n20}}{=}&  \frac{2}{s}(\eta_{i}- \eta_{i+1}) +2\frac{1}{8\mathfrak{b}^4\sqrt{d}}    \nonumber\\
    &\stackrel{\textrm{Eq. } \eqref{eq_n31}}{\leq}& \frac{1}{2\mathfrak{b}^4\sqrt{d}}.
\end{eqnarray}
Thus, by \eqref{eq_b10} and \eqref{eq_b9}, for some $\zeta \in \mathbb{R}$ where $|\zeta| \leq \frac{1}{2\mathfrak{b}^4\sqrt{d}}$, we have 
\begin{eqnarray} 
 \frac{| \eta_{j} - \eta_{j+r}|^2}{| g(\eta)[j] - g(\eta)[j+r]|^2} 
& \stackrel{\textrm{Eq. } \eqref{eq_b10}, \, \, \eqref{eq_b9}}{\leq} & \frac{|  r \frac{1}{2\sqrt{d}} + \rho  |^2}{| r \frac{1}{2\sqrt{d}} + \rho + \zeta|^2}  \nonumber\\
&\leq & \left(1+ \frac{1}{r}   \times \frac{1}{\mathfrak{b}^4}\right)^2  \qquad \qquad \forall \, r \geq 2\mathfrak{b}, \label{eq_b7}
\end{eqnarray}
where \eqref{eq_b7} holds since  $|\zeta| \leq \frac{1}{4\mathfrak{b}^4\sqrt{d}}$ and $\rho\geq 0$.

Moreover, for every $\kappa \geq 1$ we have,
\begin{eqnarray}
\prod_{r=1}^{d}   (1+ \frac{1}{\kappa r}) &\leq& \prod_{r=1}^{d}   (1+ \frac{1}{r})^{\frac{2}{\kappa}} \label{eq_n70}\\
&=& \left( \prod_{r=1}^{d}   (1+ \frac{1}{r}) \right)^{\frac{2}{\kappa}} \nonumber\\
&\leq& (d+1)^{\frac{2}{\kappa}} \label{eq_n71}
\end{eqnarray}
where \eqref{eq_n70} holds since $1+ \frac{1}{r \kappa} \leq (1 + \frac{1}{r})^{\frac{2}{\kappa}}$ for all $r, \kappa \geq 1$, and \eqref{eq_n71} holds since $\prod_{r=1}^{d}   (1+\frac{1}{r}) = d+1$.

Plugging in $\kappa = \mathfrak{b}^4$ into \eqref{eq_n71}, we have
\begin{eqnarray}\label{eq_n72}
\prod_{r=1}^{d}   (1+ \frac{1}{\mathfrak{b}^4 r})
 &\stackrel{\textrm{Eq. } \eqref{eq_n71}}{\leq}& (d+1)^{\frac{2}{\mathfrak{b}^4}} \nonumber\\
&\stackrel{\textrm{Eq. } \eqref{eq_n37}}{\leq}& \left(d^{\frac{1}{\log(d)^{\log \log d}}}\right)^{\frac{2}{\mathfrak{b}^3}}  \nonumber\\
&=&\left(e^{\frac{\log(d)}{\log(d)^{\log \log d}}}\right)^{\frac{2}{\mathfrak{b}^3}}  \nonumber\\
 &\leq& e^{\frac{2}{\mathfrak{b}^3}}.
\end{eqnarray}
Therefore, we have 
\begin{eqnarray} 
\prod_{j \in [j_{\mathrm{min}} , j_{\mathrm{max}}] , \, \, \, \ell \notin [j_{\mathrm{min}} - 2\mathfrak{b} , j_{\mathrm{max}} + 2\mathfrak{b}]} \frac{| \eta_{j} - \eta_\ell|^2}{| g(\eta)[j] - g(\eta)[\ell]|^2}  &\stackrel{\textrm{Eq.  \eqref{eq_b7}}}{\leq}& \prod_{j \in [j_{\mathrm{min}} , j_{\mathrm{max}}]} \prod_{r = 2 \mathfrak{b}}^d \left(1+ \frac{1}{r} \times \frac{1}{\mathfrak{b}^4}\right)^2  \nonumber\\
&\stackrel{\textrm{Eq. } \eqref{eq_n72}}{\leq} & (e^{\frac{2}{\mathfrak{b}^3}})^{ j_{\mathrm{max}} -  j_{\mathrm{min}}} \nonumber\\
&\stackrel{\textrm{Eq. } \, \eqref{eq_n28}}{=}& (e^{\frac{2}{\mathfrak{b}^3}})^{2 \mathfrak{b}^2} \nonumber\\
& = & e^{\frac{4}{\mathfrak{b}}} \nonumber\\
& \stackrel{\textrm{Eq. } \eqref{eq_n37}}{\leq} 2\nonumber.
\end{eqnarray}
\end{proof}

\paragraph{Step 6.} Bounding the density ratio to show that $    \frac{f(\eta)}{f(g(\eta))} \leq \tilde{O}((\sqrt{d} \log d )^2 \times y^2)$.

\begin{lemma}\label{lemma_density_ratio}
For any $y \leq s\frac{1}{8\mathfrak{b}^4\sqrt{d} }$ and any $\eta\in  S_3(a,b;y) \cap E$, we have that

\begin{align*}
  \frac{f(\eta)}{f(g(\eta))} \leq 400 (8 \mathfrak{b}^4 \sqrt{d})^2 \times y^2,
\end{align*}
\end{lemma}

\begin{proof}
Since $\eta \in  S_3(a,b;y)$ and $y \leq s\frac{1}{8\mathfrak{b}^4\sqrt{d} }$, we have
\begin{equation}\label{eq_n33}
\eta_{i}- \eta_{i+1} \stackrel{\textrm{Eq. } \eqref{eq_n22}}{=} y \leq s\frac{1}{8\mathfrak{b}^4\sqrt{d} }.
\end{equation}
Thus, 
\begin{equation}\label{eq_n34}
\alpha \stackrel{\textrm{Eq. } \eqref{eq_n20}}{=} \frac{ \frac{2}{s}(\eta_{i}- \eta_{i+1}) +2\frac{1}{8\mathfrak{b}^4\sqrt{d}} }{b-a} \stackrel{\textrm{Eq. } \eqref{eq_n33}}{\leq}  \frac{1}{(b-a)\mathfrak{b}^4\sqrt{d}}.
\end{equation}
By Inequality \eqref{eq_n42} of Proposition \ref{prop_map}, we have
\begin{equation}\label{eq_n35}
    b-a  \stackrel{\textrm{Eq. } \eqref{eq_n42} \textrm{ of Prop. } \ref{prop_map}}{\geq} \frac{1}{2 \sqrt{d}}.
\end{equation}
Hence,
\begin{equation}\label{eq_b11}
    1-\alpha \stackrel{\textrm{Eq. } \eqref{eq_n34}}{\geq} 1- \frac{1}{(b-a)\mathfrak{b}^4\sqrt{d}} \stackrel{\textrm{Eq. } \eqref{eq_n35}}{\geq} 1- \frac{1}{2\mathfrak{b}^4} \geq 1- \frac{1}{\mathfrak{b}^2}.
\end{equation}
Moreover, we have that for every $\ell \in [j_{\mathrm{min}}+1,j_{\mathrm{max}}]$,
\begin{eqnarray}
  |\eta_{\ell}^2 - g(\eta)[\ell]^2| &=& |(\eta_{\ell} - g(\eta)[\ell])(\eta_{\ell} + g(\eta)[\ell])| \nonumber\\
 & & \leq  \left |\eta_{\ell} - g(\eta)[\ell]\right| \times \left|\eta_{1} + g(\eta)[1]\right|  \label{eq_n155}\\
 & & \stackrel{\textrm{Eq. } \eqref{eq_g3}}{=} |\eta_{\ell} - g(\eta)[\ell]|\times |\eta_{1} + \eta_{1}| \nonumber\\
& & \leq  |\eta_{\ell} - g(\eta)[\ell]| \times  6 \sqrt{d} \label{eq_n154}\\
& & =  \left|\eta_{j_{\mathrm{min}}} - g(\eta)[j_{\mathrm{min}}] - \left( \sum_{s = j_{\mathrm{min}}}^{\ell-1} \eta_{s} - \eta_{s+1} - (g(\eta)[s]-g(\eta)[s+1]) \right ) \right| \times  6 \sqrt{d}\nonumber\\
& & \stackrel{\textrm{ Eq. } \eqref{eq_n36} \textrm{ of Prop. } \ref{prop_map}}{=}  \left | \sum_{s = j_{\mathrm{min}}}^{\ell-1} \eta_{s} - \eta_{s+1} - (g(\eta)[s]-g(\eta)[s+1]) \right | \times  6 \sqrt{d}\nonumber\\
& & \stackrel{\textrm{ Eq. } \eqref{eq_b4} \textrm{ of Prop. } \ref{prop_map}}{\leq}  \left| \sum_{s = j_{\mathrm{min}}}^{\ell-1}  \alpha(\eta_{s} - \eta_{s+1}) \right | \times  6 \sqrt{d}\nonumber\\
 &=& \alpha \left| \eta_{j_{\mathrm{min}}} - \eta_{\ell} \right | \times  6 \sqrt{d}\nonumber\\
  &\leq & \alpha \left| \eta_{j_{\mathrm{min}}} - \eta_{j_{\mathrm{max}}} \right | \times  6 \sqrt{d} \label{eq_n156}\\
& & \stackrel{\textrm{ Eq. } \eqref{eq_n36} \textrm{ of Prop. } \ref{prop_map}, \textrm{ Eq. } \eqref{eq_g4}}{=} \alpha (b-a)\times  6 \sqrt{d}, \label{eq_n39}
\end{eqnarray}
where \eqref{eq_n155} and \eqref{eq_n156} hold since $\eta_1 \geq \eta_2 \geq \cdots \geq \eta_d$.   \eqref{eq_n154} holds since, whenever $\eta \in E$,
\begin{equation*}
|\eta_1| \leq |\omega_1| + |\eta_1 - \omega_1|  \stackrel{\textrm{Prop. } \eqref{prop_classical}}{\leq} 2 \sqrt{d} + |\eta_1 - \omega_1| \stackrel{\textrm{Eq. } \eqref{eq_n19}}{\leq} 3 \sqrt{d}.
\end{equation*}
Moreover, \eqref{eq_n39} also holds for $\ell = j_{\mathrm{min}}$ since $\eta_{j_{\mathrm{min}}}^2 - g(\eta)[j_{\mathrm{min}}]^2 = 0$ by Equation \eqref{eq_n36} of Proposition \ref{prop_map}.
Therefore, we have
\begin{equation}\label{eq_n40}
 \left|\eta_{\ell}^2 - g(\eta)[\ell]^2\right|    \stackrel{\textrm{Eq. } \eqref{eq_n39}}{\leq} \alpha (b-a)\times  6 \sqrt{d} \qquad \qquad \forall \ell \in [j_{\mathrm{min}},j_{\mathrm{max}}].
\end{equation}
Therefore, by the joint density formula for the eigenvalues of the GUE \eqref{eq_joint_density} we have
\begin{eqnarray}
  \frac{f(\eta)}{f(g(\eta))} 
 \!\!\! \!\!\!\!\!\!\!\!\!\!\!\! &\stackrel{\textrm{Eq. \eqref{eq_joint_density}}}{=} \!\!\! \!\!\!\!\!\! & \prod_{\ell<j,  :\, \, \, \ell,j \in [d]} \frac{| \eta_{\ell} - \eta_j|^2}{| g(\eta)[\ell] - g(\eta)[j]|^2} e^{-\frac{1}{2} \sum_{\ell=1}^{d}  (\eta_{\ell}^2 - g(\eta)[\ell]^2 )}\nonumber\nonumber\\
    &\stackrel{\textrm{Eq. \eqref{eq_g3}}}{=} & \prod_{\ell<j,  :\, \, \, \ell,j \in [d]} \frac{| \eta_{\ell} - \eta_j|^2}{| g(\eta)[\ell] - g(\eta)[j]|^2}  e^{-\frac{1}{2} \sum_{\ell=j_{\mathrm{min}}+1}^{j_{\mathrm{max}}-1}  (\eta_{\ell}^2 - g(\eta)[\ell]^2 )}\nonumber\nonumber\\
  &\stackrel{\textrm{Eq. } \eqref{eq_n36} \textrm{ of Prop. } \ref{prop_map}, }{\stackrel{\textrm{ Eq. } \eqref{eq_g3}, \eqref{eq_g4}}{=}} & \prod_{\ell<j,  :\, \, \, \ell,j \in [j_{\mathrm{min}} - 2 \mathfrak{b} , j_{\mathrm{max}} + 2 \mathfrak{b}]} \frac{| \eta_{\ell} - \eta_j|^2}{| g(\eta)[\ell] - g(\eta)[j]|^2} e^{-\frac{1}{2} \sum_{\ell=j_{\mathrm{min}}+1}^{j_{\mathrm{max}}-1}  (\eta_{\ell}^2 - g(\eta)[\ell]^2 )}\nonumber\\
&  & \times \quad  \prod_{j \in [j_{\mathrm{min}} , j_{\mathrm{max}}] , \, \, \, \ell \notin [j_{\mathrm{min}}  - 2 \mathfrak{b}  , j_{\mathrm{max}}  + 2 \mathfrak{b} ]} \frac{| \eta_{j} - \eta_\ell|^2}{| g(\eta)[j] - g(\eta)[\ell]|^2}\nonumber\\
  &\stackrel{\textrm{Lem. } \ref{lemma_mean_field}}{\leq} & \prod_{\ell<j,  :\, \, \, \ell,j \in [j_{\mathrm{min}} - 2 \mathfrak{b} , j_{\mathrm{max}} + 2 \mathfrak{b}]} \frac{| \eta_{\ell} - \eta_j|^2}{| g(\eta)[\ell] - g(\eta)[j]|^2} e^{-\frac{1}{2} \sum_{\ell=j_{\mathrm{min}}+1}^{j_{\mathrm{max}}-1}  (\eta_{\ell}^2 - g(\eta)[\ell]^2 )} \times 2\nonumber\\
 &\stackrel{\textrm{Eq. } \eqref{eq_b5},\eqref{eq_b4}  }{\stackrel{ \textrm{ of Prop. \ref{prop_map}}}{\leq}} &
\!\!\!\!\!\!\!\!\!\!\!\!\!\!   ( 8 \mathfrak{b}^4 \sqrt{d}  \times y)^2 \times \left(\frac{1}{1-\alpha}\right)^{(j_{\mathrm{max}} - j_{\mathrm{min}})(j_{\mathrm{max}} - j_{\mathrm{min}} +4 \mathfrak{b})} \times e^{-\frac{1}{2} \sum_{\ell=j_{\mathrm{min}}+1}^{j_{\mathrm{max}}-1}  (\eta_{\ell}^2 - g(\eta)[\ell]^2 )} \times 2 \nonumber\\
  &\stackrel{\textrm{Eq. } \eqref{eq_n40}}{\leq} & \!\!\!\!\!\!\!\!\!\!\!
  ( 8 \mathfrak{b}^4 \sqrt{d}  \times y)^2 \times \left(\frac{1}{1-\alpha}\right)^{(j_{\mathrm{max}} - j_{\mathrm{min}})(j_{\mathrm{max}} - j_{\mathrm{min}} +4 \mathfrak{b})} \times e^{\frac{1}{2}(j_{\mathrm{max}} - j_{\mathrm{min}})\alpha(b-a)\times 6\sqrt{d}} \times 2 \nonumber\\
  &\stackrel{\textrm{Eq. } \eqref{eq_n28}, \, \, \eqref{eq_n37}}{\leq}& ( 8 \mathfrak{b}^4 \sqrt{d}  \times y)^2 \times (1-\alpha)^{-2.1\mathfrak{b}^4} \times e^{\frac{1}{2}2\mathfrak{b}^2\alpha(b-a)\times 6\sqrt{d}} \times 2 \label{eq_n74}\\
     &\stackrel{\textrm{Eq. } \eqref{eq_b11},\,\, \eqref{eq_b9}}{\leq}& \left(1- \frac{1}{2\mathfrak{b}^4}\right)^{-2.1\mathfrak{b}^4}\times (8 \mathfrak{b}^4 \sqrt{d})^2 \times e^1 \times y^2 \label{eq_n75}\\
   &\leq & 2 e^{4.2} (8 \mathfrak{b}^4 \sqrt{d})^2 \times e^1 \times y^2 \label{eq_n76}\\
 &\leq & 400 (8 \mathfrak{b}^4 \sqrt{d})^2 \times y^2,  \label{eq_n181}
\end{eqnarray}
where \eqref{eq_n74} holds since 
$$j_{\mathrm{max}} -  j_{\mathrm{min}} \stackrel{\textrm{Eq. } \eqref{eq_n28}}{=}  \min(i+\mathfrak{b}^2, d) -  \max(i-\mathfrak{b}^2, 1) =  2 \mathfrak{b}^2 \stackrel{\textrm{Eq. } \eqref{eq_n37}}{\leq} 2.1 \mathfrak{b}^2 - 4 \mathfrak{b}.$$
\eqref{eq_n75} holds since  $1-\alpha \geq 1- \frac{1}{\mathfrak{b}^2}$ by \eqref{eq_b11}, and since $\alpha (b-a)\leq \frac{1}{2\mathfrak{b}^4\sqrt{d}}\leq \frac{1}{6\mathfrak{b}^2\sqrt{d}}$ by \eqref{eq_b9}.
\eqref{eq_n76} holds since $\left(1- \frac{1}{s}\right)^{-s} \leq 2^{\frac{1}{2.1}}e$ for any $s \geq 3$.
\end{proof}

\paragraph{Step 7.} Dealing with the eigenvalues near the edge of the spectrum.
In this step, we extend the results of the previous steps to the eigenvalues which are near the edge of the spectrum.
 Specifically, we consider the eigenvalues $\eta_i$ such that $i \in  [1, \mathfrak{b}^2] \cup [d-\mathfrak{b}^2, d]$. 
  Since the joint eigenvalue density function \eqref{eq_joint_density} is symmetric about 0, without loss of generality we may assume that $i \leq \mathfrak{b}^2$.
Define  $j_{\mathrm{min}}$,  $j_{\mathrm{max}}$ as in \eqref{eq_n28},  and define $a_{\mathrm{min}}$,   $a_{\mathrm{max}}$, $b_{\mathrm{min}}$ and $b_{\mathrm{max}}$ as in \eqref{eq_n26}.

\begin{proposition}\label{prop_n1_edge}
Suppose that the event $E$ occurs.
Then for all $i \in [1, \mathfrak{b}^2] \cup [d- \mathfrak{b}^2, d]$ we have
\begin{equation}\label{eq_c2_edge}
    \eta_{j_{\mathrm{min}}} - \eta_{j_{\mathrm{max}}} \geq    \frac{29}{30} \mathfrak{b}^2 d^{-\frac{1}{6}} \min(i, d-i)^{-\frac{1}{3}} \geq  \frac{29}{30} \mathfrak{b}^2 \frac{1}{\sqrt{d}}.
\end{equation}
Moreover, we also have that 
$ \eta_{j_{\mathrm{max}}}  \in[a_{\mathrm{min}}, a_{\mathrm{max}}]$ and  $ \eta_{j_{\mathrm{min}}}  \in[b_{\mathrm{min}}, b_{\mathrm{max}}]$.

\end{proposition}
\begin{proof}
Without loss of generality, we may assume that 
$i \leq \frac{1}{2}d$,  since the GUE matrices $G$ and $-G$ have the same distribution and hence the joint eigenvalue distribution of the GUE is symmetric about $0$.
If $E$ occurs, then by the definition of the event $E$ \eqref{eq_n19}, we have
\begin{eqnarray}
   \eta_{j_{\mathrm{min}}} - \eta_{j_{\mathrm{max}}} 
  & \stackrel{\textrm{Eq. } \eqref{eq_n28}}{=}&  \eta_{1} - \eta_{i+ \mathfrak{b}^2} \nonumber\\
   &\stackrel{\textrm{Eq. } \eqref{eq_n19}}{\geq} &    \omega_{i- \mathfrak{b}^2} - \omega_{i+ \mathfrak{b}^2} - 2\mathfrak{b} d^{-\frac{1}{6}} \nonumber\\
    & \stackrel{\textrm{Prop. } \ref{prop_classical}}{\geq}& \mathfrak{b}^2 \times  d^{-\frac{1}{6}} i^{-\frac{1}{3}} - 2\mathfrak{b} d^{-\frac{1}{6}} \nonumber\\
         &\geq & \frac{29}{30} \mathfrak{b}^2 d^{-\frac{1}{6}} i^{-\frac{1}{3}}, \label{eq_n47}
\end{eqnarray}
where \eqref{eq_n47} holds since $\mathfrak{b} \geq 10^6$.
This proves \eqref{eq_c2_edge}.
Moreover, by the definition of the event $E$, we also have that 
\begin{eqnarray}
  |\eta_{j_{\mathrm{min}}} -  \omega_{j_{\mathrm{min}}}| 
   &\stackrel{\textrm{Eq. } \eqref{eq_n28}}{=} & |\eta_{1} -  \omega_{1}| \nonumber\\
    &\stackrel{\textrm{Eq. } \eqref{eq_n19}}{\leq} &  \mathfrak{b} d^{-\frac{1}{6}}\nonumber\\
    &\leq & \frac{1}{30} \mathfrak{b}^2 d^{-\frac{1}{6}}, \label{eq_c1_edge}
\end{eqnarray}
where inequality \eqref{eq_c1_edge} holds since $\mathfrak{b} \geq 10^6$.
Thus, by definition \eqref{eq_n26}, Inequality \eqref{eq_c1_edge} implies that $\eta_{j_{\mathrm{min}}}  \in[b_{\mathrm{min}}, b_{\mathrm{max}}]$. 
Again, by the definition of the event $E$, we also have that 
\begin{eqnarray}
      \eta_{j_{\mathrm{min}}} - \eta_{j_{\mathrm{max}}} 
  & \stackrel{\textrm{Eq. } \eqref{eq_n28}}{=} & |\eta_{i+ \mathfrak{b}^2} -  \omega_{i+ \mathfrak{b}^2}| \nonumber\\ &\stackrel{\textrm{Eq. } \eqref{eq_n19}}{\leq}&   \mathfrak{b} (i+\mathfrak{b}^2)^{-\frac{1}{3}} d^{-\frac{1}{6}}\nonumber\\
    &\leq  &\mathfrak{b} i^{-\frac{1}{3}} d^{-\frac{1}{6}}\nonumber\\
    &\leq &\frac{1}{30} \mathfrak{b}^2 d^{-\frac{1}{6}} i^{-\frac{1}{3}} \label{eq_c1b_edge}
\end{eqnarray}
where \eqref{eq_c1b_edge} holds since $\mathfrak{b} \geq 10^6$.
Thus, by definition \eqref{eq_n26}, Inequality \eqref{eq_c1b_edge} implies that $\eta_{j_{\mathrm{max}}}  \in[a_{\mathrm{min}}, a_{\mathrm{max}}]$.
\end{proof}
\noindent
By Proposition \ref{prop_n1_edge}, if $i \leq \mathfrak{b}^2$,  whenever the event $E$ occurs we have that $ \eta_{j_{\mathrm{max}}}  \in[a_{\mathrm{min}}, a_{\mathrm{max}}]$ and  $ \eta_{j_{\mathrm{min}}}  \in[b_{\mathrm{min}}, b_{\mathrm{max}}]$.
 This fact allows us to extend the definition of the sets $S_0$, $S_3$, and $S_4$, whose definition requires that $ \eta_{j_{\mathrm{max}}}  \in[a_{\mathrm{min}}, a_{\mathrm{max}}]$ and  $ \eta_{j_{\mathrm{min}}}  \in[b_{\mathrm{min}}, b_{\mathrm{max}}]$, to the ``edge case'' where $i \leq \mathfrak{b}^2$.

Specifically, consider any $a,b$ such that $a_{\mathrm{min}} \leq a \leq  a_{\mathrm{max}}$  and $b_{\mathrm{min}} \leq b \leq  b_{\mathrm{max}}$.
 Recall from \eqref{eq_n25} and \eqref{eq_n22_edge} the definition of the sets $S_0$, $S_3$, and $S_4$; we extend these definitions to the ``edge case'' where $i \leq \mathfrak{b}^2$:
\begin{itemize}
\item \begin{equation}\label{eq_n25_edge}
    S_0(a,b):= \{ \eta \in \mathcal{W}_d : \eta_{j_{\mathrm{max}}} = a, \eta_{j_{\mathrm{min}}} = b\},
\end{equation}

\item  
\begin{equation}\label{eq_n22_edge}
S_3(a,b; y) :=  \{\eta \in \mathcal{W}_d : \eta_i-\eta_{i+1} = y\} \cap S_0(a,b) \qquad \textrm{ for any } y \leq s\frac{1}{8\mathfrak{b}^4\sqrt{d} }, 
\end{equation}

\item 
\begin{equation*}
S_4(a,b) :=  \{\eta \in \mathcal{W}_d : \eta_i-\eta_{i+1} \geq s\} \cap S_0(a,b),
\end{equation*}

\end{itemize}
 where $\mathcal{W}_d$ was defined in \eqref{eq:WeylChamber}.

  In place of the map $g$, we instead consider the map $\phi:  \mathcal{W}_d \rightarrow \mathcal{W}_d$ such that 

  \begin{itemize}

  \item 
\begin{equation}\label{eq_phi1}
      \phi(\eta)[j] = \eta_j \qquad \forall j > i
      \end{equation}

 \item \begin{equation}\label{eq_phi2}
      \phi(\eta)[i] =  \eta_{i+1} + \frac{2}{s}(\eta_{i}- \eta_{i+1})+ 2\frac{1}{8\mathfrak{b}^4\sqrt{d}}, \mbox{ and }
      \end{equation} 
      
       \item \begin{equation}\label{eq_phi3}
      \phi(\eta)[j] =  \phi(\eta)[j+1] + (\eta_{j}- \eta_{j+1}) \qquad \forall j < i.
      \end{equation} 
  
\end{itemize}

\begin{proposition} \label{prop_map_phi}
Suppose that  $ i \leq 2\mathfrak{b}^2$.
Then the following properties hold for $\phi$:
\begin{itemize}

\item $\phi$ is injective,

\item  $\phi(\eta)[i] -  \phi(\eta)[i+1]  \geq   \frac{1}{8\mathfrak{b}^4\sqrt{d}}$,  and hence  

\begin{equation}\label{eq_b5e}
\frac{\eta_{i} - \eta_{i+1}}{\phi(\eta)[i] -  \phi(\eta)[i+1]}  \leq  8\mathfrak{b}^4 \sqrt{d}  \times (\eta_{i} - \eta_{i+1}) =  8 \mathfrak{b}^4 \sqrt{d}  \times y
\end{equation}
  for any $\eta \in  S_3(a,b;y)$ and any $y \leq s\frac{1}{8\mathfrak{b}^4\sqrt{d} }$,

\item  \begin{equation}\label{eq_b4e}
\phi(\eta)[j] -  \phi(\eta)[j+1] \geq \eta_{j} - \eta_{j+1} \qquad \forall j \in [d].
\end{equation}
\end{itemize}

\end{proposition}

\begin{proof}

\medskip
\noindent
{\em Injectivity.}
To prove that $\phi$ is injective, we note that, given any vector $z \in \mathbb{R}^d$ we can find the {\em unique} $\eta \in \mathcal{W}_d$ such that $\phi(\eta) = z$ whenever such a value of $\eta$ exits.
We can do this by solving the system of linear equations given by  \eqref{eq_phi1}-\eqref{eq_phi3}:
First, we note that by \eqref{eq_phi1}, we can solve for $\eta_j$ for all $j >i$.
Then we can plug in the value we found for  $\eta_{i+1}$ into \eqref{eq_phi2} to solve for $\eta_i$.
Finally, we can use \eqref{eq_phi3} to solve for $\eta_j$ for all $j<i$ recursively, starting with $\eta_{i-1}$.

\medskip
\noindent
{\em Showing \eqref{eq_b5e}.}
Since  $y \leq s\frac{1}{8\mathfrak{b}^4\sqrt{d} }$ and $\eta \in  S_3(a,b;y)$,
we have that 
\begin{equation}
    \eta_i-\eta_{i+1} \stackrel{\textrm{Eq. } \eqref{eq_n22_edge}}{=} y \leq s\frac{1}{8\mathfrak{b}^4\sqrt{d} }.
    \end{equation}
Thus,
  \begin{eqnarray}\label{eq_n48}
  \phi(\eta)[i] - \phi(\eta)[i+1] &\stackrel{\textrm{Eq. } \eqref{eq_phi1}}{=}& \phi(\eta)[i] - \eta_{i+1} \nonumber\\
  &\stackrel{\textrm{Eq. } \eqref{eq_phi2}}{=}& \frac{2}{s}(\eta_{i}- \eta_{i+1})+ 2\frac{1}{8\mathfrak{b}^4\sqrt{d}} \nonumber\\
  &\geq & 2\frac{1}{8\mathfrak{b}^4\sqrt{d}} \nonumber\\
  &\geq & \frac{1}{8\mathfrak{b}^4\sqrt{d}}.
  \end{eqnarray}
Hence, 
\begin{equation*}
\frac{\eta_{i} - \eta_{i+1}}{\phi(\eta)[i] -  \phi(\eta)[i+1]}  \stackrel{\textrm{Eq. } \eqref{eq_n48}}{\leq}  8\mathfrak{b}^4 \sqrt{d}  \times (\eta_{i} - \eta_{i+1}) \stackrel{\textrm{Eq. } \eqref{eq_n22_edge}}{=}   8\mathfrak{b}^4 \sqrt{d}  \times y.
\end{equation*}
This proves \eqref{eq_b5e}.

\medskip
\noindent
{\em Showing \eqref{eq_b4e}.}
We have
  \begin{eqnarray}
  \phi(\eta)[i] - \phi(\eta)[i+1]  &\stackrel{\textrm{Eq. } \eqref{eq_phi1}}{=} & \phi(\eta)[i] - \eta_{i+1} \nonumber\\
  &\stackrel{\textrm{Eq. } \eqref{eq_phi2}}{=} & \frac{2}{s}(\eta_{i}- \eta_{i+1})+ 2\frac{1}{8\mathfrak{b}^4\sqrt{d}}  \nonumber\\
    &\geq & \frac{1}{s}(\eta_{i}- \eta_{i+1}) \nonumber\\
        &\geq & \eta_{i}- \eta_{i+1}, \label{eq_n49}
  \end{eqnarray}
  where \eqref{eq_n49} holds since $0 \leq s\leq 1$.
  Thus, \eqref{eq_b4e} holds for $j=i$.
  Moreover, \eqref{eq_b4e} holds for all  $j \neq i$ by \eqref{eq_phi1} and \eqref{eq_phi3}.
  Therefore \eqref{eq_b4e} holds for all $j \in [d]$.
\end{proof}

\begin{lemma}[\bf Jacobian determinant of $\phi$]\label{prop_Jacobian_phi}
Suppose that  $ i \leq \mathfrak{b}^2$.
If  $y \leq s\frac{1}{8\mathfrak{b}^4\sqrt{d} }$ and $\eta \in  S_3(a,b;y)$, we have that
$$\mathrm{det}(J_{\phi}(\eta)) = \frac{2}{s}.$$
\end{lemma}
\begin{proof}
Consider the map $h: \mathcal{W}_{d} \rightarrow \mathbb{R}^{d}$, where  for any $\eta \in \mathcal{W}_d$
\begin{equation}
h(\eta)[j] := \begin{cases}
\eta_j - \eta_{j+1} \qquad \qquad \textrm{for }j \leq i\\
\eta_{j} \qquad \qquad  \qquad \, \, \,  \, \, \textrm{for } j >i.
\end{cases}
\end{equation}
The map $h$ is injective, since for any $\Delta \in \mathbb{R}^d$ that is in the range of $h$ we can solve for the unique $\eta \in \mathcal{W}_d$ such that $h(\eta) = \Delta$.
 Specifically, the unique solution $\eta$, which we denote by $h^{-1}(\Delta)$, is given by 
 \begin{equation}\label{eq_n50.2}
 h^{-1}(\Delta)[j] := \eta_j = \begin{cases}
  \Delta_{i+1} + \sum_{r = 1}^{j - i +1} \Delta_{j+1 - r}, \qquad   j \leq i\\
    \Delta_j, \qquad \qquad \qquad \qquad \qquad \, \,   j > i.
 \end{cases}
 \end{equation}
   Moreover,   \eqref{eq_n50.2} also implies that for every $j \in [d]$ and every $\Delta \in \mathbb{R}^d$ that is in the range of $h$,
  \begin{equation}\label{eq_n51.2}
  h^{-1}(\Delta)[j]- h^{-1}(\Delta)[j+1] = \Delta_j \qquad \qquad \forall  j \leq i.
    \end{equation}
    
  \noindent  From \eqref{eq_phi1}-\eqref{eq_phi3} we have that for every $\eta \in \mathcal{W}_d$ and $j \in [d]$,
\begin{equation}\label{eq_n52.2}
\textrm{$\phi(\eta)[j]- \phi(\eta)[j+1]$ is a function of only $\eta_j- \eta_{j+1}$,}
\end{equation}
and does not otherwise depend on any $\eta_\ell$ for $\ell \neq j$.
Therefore, by \eqref{eq_n50.2}, \eqref{eq_n51.2} and \eqref{eq_n52.2},  for every $\Delta \in  \{h(\eta) :  \eta \in  S_3(a,b;y) \cap E \}$, and every $j \in [d]$, 
we have that 
\begin{equation}\label{eq_n53.2}
\textrm{$h(\phi(h^{-1}(\Delta)))[j]$ is a function of only $\Delta_j$,}
\end{equation}
 and does not otherwise depend on any $\Delta_\ell$ for $\ell \neq i$.
 Thus,  for every $\Delta \in  \{h(\eta) :  \eta \in  S_3(a,b;y) \cap E \}$, we have that
\begin{equation} \label{eq_derivative1.2}
    \frac{\partial h \circ \phi \circ h^{-1}(\Delta)[\ell]}{\partial \Delta_j} \stackrel{\textrm{Eq. } \eqref{eq_n53.2}}{=} 0 \qquad \forall \ell \neq j, \ell \neq i.
\end{equation}
 Moreover, we also have that
\begin{equation}\label{eq_derivative1_d.2}
    \frac{\partial h \circ \phi\circ h^{-1}(\Delta)[j]}{\partial \Delta_j} \stackrel{\textrm{Eq. } \eqref{eq_phi1}, \, \, \eqref{eq_phi3}}{=} 1 \qquad \forall j \neq i,
\end{equation}
\begin{equation}\label{eq_derivative1_c.2}
    \frac{\partial h \circ \phi\circ h^{-1}(\Delta)[i]}{\partial \Delta_i} \stackrel{\textrm{Eq. } \eqref{eq_phi2}}{=} \frac{2}{s}.
\end{equation}
Thus, by \eqref{eq_derivative1.2}, \eqref{eq_derivative1_d.2}, and \eqref{eq_derivative1_c.2}  the Jacobian $J_{h \circ \phi \circ h^{-1}}(\Delta)$  is a diagonal matrix, with $j$'th diagonal entries equal to $1$ for all $j \neq i$ and $i$'th entry equal to $\frac{2}{s}$.
Therefore,
\begin{align}\label{eq_g5b}
\mathrm{det}(J_{h \circ \phi \circ h^{-1}}(\Delta)) &\stackrel{\textrm{Eq. } \eqref{eq_derivative1.2},\, \eqref{eq_derivative1_d.2},\, \eqref{eq_derivative1_c.2} }{=} \frac{2}{s} \times 1 = \frac{2}{s}.
\end{align}
Hence,

\begin{eqnarray*}
    \mathrm{det}(J_{\phi}(\eta)) &= &  \mathrm{det}(J_{  h^{-1} \circ h \circ \phi \circ h^{-1} \circ h}(\eta)) \\ 
    &=& \mathrm{det}( J_{h^{-1}}(\eta) \times J_{h \circ \phi \circ h^{-1}}(h(\eta))  \times J_{h}(\eta))\\
    &= & \mathrm{det}(J_{h^{-1}}(\eta)) \times \mathrm{det}(J_{h \circ \phi \circ h^{-1}}(h(\eta))  \times \mathrm{det}(J_{h}(\eta))\\
        &= & \mathrm{det}(J_{h}(\eta))^{-1} \times \mathrm{det}(J_{h \circ \phi \circ h^{-1}}(h(\eta))  \times \mathrm{det}(J_{h}(\eta))\\
    &=&\mathrm{det}(J_{h \circ \phi \circ h^{-1}}(h(\eta))  \times 1\\
    &\stackrel{\textrm{Eq. \eqref{eq_g5b}}}{\geq}&  \frac{2}{s}.
\end{eqnarray*}
\end{proof}

\begin{lemma}\label{lemma_density_ratio_edge}
Suppose that  $ i \leq \mathfrak{b}^2$. 
For any $y \leq s\frac{1}{8\mathfrak{b}^4\sqrt{d} }$ and any $\eta\in  S_3(a,b;y) \cap E$, we have that

\begin{align*}
  \frac{f(\eta)}{f(\phi(\eta))} \leq 4 (8 \mathfrak{b}^4 \sqrt{d})^2 \times y^2.
\end{align*}
\end{lemma}

\begin{proof}

Since $\eta \in E$, we have that for all $j \in [d]$,
\begin{eqnarray}\label{eq_n68}
|\eta_j| &\stackrel{\textrm{Eq. } \eqref{eq_n19}}{\leq}& |\omega_j| + \mathfrak{b} d^{-\frac{1}{6}}\nonumber\\
 &\leq& \max(|\omega_1|, |\omega_d|) + \mathfrak{b} d^{-\frac{1}{6}}\nonumber\\
&\stackrel{\textrm{Prop. } \ref{prop_classical}}{\leq}& 2 \sqrt{d} + \mathfrak{b} d^{-\frac{1}{6}}\nonumber\\
&\stackrel{\textrm{Eq. } \eqref{eq_n37}}{\leq}& 3 \sqrt{d}.
\end{eqnarray}
Moreover, since  $y \leq s\frac{1}{8\mathfrak{b}^4\sqrt{d} }$ and $\eta \in  S_3(a,b;y)$, 
we have that
\begin{equation}\label{eq_n62}
    \eta_i-\eta_{i+1} \stackrel{\textrm{Eq. } \eqref{eq_n22_edge}}{=} y \leq s\frac{1}{8\mathfrak{b}^4\sqrt{d} }.
    \end{equation}
Thus, we have
\begin{eqnarray}
\phi(\eta)[i]- \phi(\eta)[i+1] &\stackrel{\textrm{Eq. } \eqref{eq_phi1},\, \, \eqref{eq_phi2}}{=}& \frac{2}{s}(\eta_{i}- \eta_{i+1})+ 2\frac{1}{8\mathfrak{b}^4\sqrt{d}} \label{eq_n66}\\
&\stackrel{\textrm{Eq. } \eqref{eq_n62}}{\leq}& 4\frac{1}{8\mathfrak{b}^4\sqrt{d}}. \label{eq_n61}
\end{eqnarray}
Hence, for all $j\leq i$,
\begin{eqnarray}
\phi(\eta)[j] - \eta_j &\stackrel{\textrm{Eq. } \eqref{eq_phi1}}{=}& \sum_{r= j}^i \phi(\eta)[r]- \phi(\eta)[r+1] - (\eta_r - \eta_{r+1})\nonumber\\
&\stackrel{\textrm{Eq. } \eqref{eq_phi3}}{=}& \phi(\eta)[i]- \phi(\eta)[i+1] - (\eta_i- \eta_{i+1}) \label{eq_n65}\\
&\leq& \phi(\eta)[i]- \phi(\eta)[i+1] \nonumber\\
&\stackrel{\textrm{Eq. } \eqref{eq_n61}}{\leq}& 4\frac{1}{8\mathfrak{b}^4\sqrt{d}}. \label{eq_n63}
\end{eqnarray}
Moreover, for all $j\leq i$,
\begin{eqnarray}
\phi(\eta)[j] - \eta_j &\stackrel{\textrm{Eq. } \eqref{eq_n65}}{=}&  \phi(\eta)[i]- \phi(\eta)[i+1] - (\eta_i- \eta_{i+1})\nonumber\\
&\stackrel{\textrm{Eq. } \eqref{eq_n66}}{=}& \left(\frac{2}{s}-1\right)(\eta_{i}- \eta_{i+1})+ 2\frac{1}{8\mathfrak{b}^4\sqrt{d}} \nonumber\\
&\geq& 0 \label{eq_n67},
\end{eqnarray}
where \eqref{eq_n67} holds since $s \leq 1$.
Therefore, for all $j\leq i$,
\begin{equation}\label{eq_n64}
|\phi(\eta)[j] - \eta_j| \stackrel{\textrm{Eq. } \eqref{eq_n63}, \,\, \eqref{eq_n67}}{\leq} 4\frac{1}{8\mathfrak{b}^4\sqrt{d}}.
\end{equation}
Thus,  for all $j\leq i$,
\begin{eqnarray}\label{eq_n69}
\left |\phi(\eta)[j]^2 - \eta_j^2 \right | &\stackrel{\textrm{Eq. } \eqref{eq_n67}}{=}& \left |\phi(\eta)[j] - \eta_j \right | \times \left |\phi(\eta)[j] + \eta_j \right | \nonumber\\
 &\stackrel{\textrm{Eq. } \eqref{eq_n68}}{\leq}& 4\frac{1}{8\mathfrak{b}^4\sqrt{d}} \times \left(4\frac{1}{8\mathfrak{b}^4\sqrt{d}} + 2|\eta_j| \right) \nonumber\\
  &\stackrel{\textrm{Eq. } \eqref{eq_n68}}{\leq}& 4\frac{1}{8\mathfrak{b}^4\sqrt{d}} \times \left(4\frac{1}{8\mathfrak{b}^4\sqrt{d}} + 6\sqrt{d} \right) \nonumber\\
  &\stackrel{\textrm{Eq. } \eqref{eq_n37}}{\leq}&  4\frac{1}{8\mathfrak{b}^4\sqrt{d}} \times 7 \sqrt{d} \nonumber\\
  &\leq&  \frac{4}{\mathfrak{b}^4}.
\end{eqnarray}
Moreover, for all $\ell, j \in [d]$, $\ell <j$, we have
\begin{eqnarray}\label{eq_n60}
      \phi(\eta)[j] - \phi(\eta)[\ell] &=& \sum_{r = j}^{\ell-1}      \phi(\eta)[r] - \phi(\eta)[r+1]\nonumber\\
        &\stackrel{\textrm{Eq. } \eqref{eq_b4e}}{\geq}& \sum_{r = j}^{\ell-1}      \eta_r - \eta_{r+1}\nonumber\\
        &=& \eta_{j}- \eta_{\ell}.
      \end{eqnarray}

\noindent
Therefore,
\begin{eqnarray}
\ \   \frac{f(\eta)}{f(\phi(\eta))}
    &\stackrel{\textrm{Eq. \eqref{eq_joint_density}}}{=} & \prod_{\ell<j,  :\, \, \, \ell,j \in [d]} \frac{| \eta_{\ell} - \eta_j|^2}{| \phi(\eta)[\ell] - \phi(\eta)[j]|^2} e^{-\frac{1}{2} \sum_{\ell=1}^{d}  (\eta_{\ell}^2 - g(\eta)[\ell]^2 )}\nonumber\\
    &\stackrel{\textrm{Eq. } \eqref{eq_n60}}{\leq}& \frac{| \eta_{i} - \eta_{i+1}|^2}{| \phi(\eta)[i] - \phi(\eta)[i+1]|^2} \times 1 \times  e^{-\frac{1}{2} \sum_{\ell=1}^{d}  (\eta_{\ell}^2 - g(\eta)[\ell]^2 )}\nonumber\\
   &\stackrel{\textrm{Eq. } \eqref{eq_phi1}}=&   \frac{| \eta_{i} - \eta_{i+1}|^2}{| \phi(\eta)[i] - \phi(\eta)[i+1]|^2} \times e^{-\frac{1}{2} \sum_{\ell=1}^{i}  (\eta_{\ell}^2 - \phi(\eta)[\ell]^2 )}\nonumber\\
  &\stackrel{\textrm{Eq. \eqref{eq_b5e} of Prop. \ref{prop_map_phi}}}{\leq}& ( 8 \mathfrak{b}^4 \sqrt{d}  \times y)^2  \times e^{-\frac{1}{2} \sum_{\ell=1}^{i}  (\eta_{\ell}^2 - \phi(\eta)[\ell]^2 )}\nonumber\\
    &\stackrel{\textrm{Eq. } \eqref{eq_n69}}{\leq}& ( 8 \mathfrak{b}^4 \sqrt{d}  \times y)^2  \times e^{i \times \frac{4}{\mathfrak{b}^4}} \nonumber\\
        &\leq& ( 8 \mathfrak{b}^4 \sqrt{d}  \times y)^2  \times e^{\mathfrak{b}^2 \times \frac{4}{\mathfrak{b}^4}} \label{eq_n73}\\
                &=& ( 8 \mathfrak{b}^4 \sqrt{d}  \times y)^2  \times e^{\frac{4}{\mathfrak{b}^2}} \nonumber\\
                   &\stackrel{\textrm{Eq. } \eqref{eq_n37}}{\leq} & 4 (8 \mathfrak{b}^4 \sqrt{d})^2 \times y^2.\nonumber\\
 \end{eqnarray}
 where \eqref{eq_n73} holds since $i \leq \mathfrak{b}^2$.
\end{proof}

\paragraph{Step 8.} Completing the proof.

\begin{proof}[Proof of Lemma \ref{lemma_GUE_gaps}]

\medskip
\noindent
{\em Bulk case ($\mathfrak{b}^2 < i < d-\mathfrak{b}^2$).}
Recalling the definition of $\mathcal{W}_d$ from \eqref{eq:WeylChamber}, by Proposition \ref{prop_map} we have that for any $z \in \mathcal{W}_d$, the pre-image $g^{-1}(\{z\}) := \{ \eta \in \mathcal{W}_d : g(\eta) = z\}$ has cardinality $|g^{-1}(\{z\}) |  \leq 2$.
Therefore, for any $s>0$ we have
\begin{equation}\label{eq_n43}
\int_{a_{\min}}^{a_{\max}} \int_{b_{\min}}^{b_{\max}} \int_0^{s\frac{1}{8\mathfrak{b}^4\sqrt{d}}} \int_{S_3(a,b;y) \cap E} f(g(\eta)) \mathrm{det}(J_g(\eta)) \mathrm{d} \eta \mathrm{d}y \mathrm{d} a \mathrm{d} b \stackrel{\textrm{Prop. } \ref{prop_map}}{\leq} 2 \times  \int_{\mathcal{W}_d} f(\eta) \mathrm{d} \eta = 2 \times  1,
\end{equation}
where the inequality holds since $|g^{-1}(\{z\}) |  \leq 2$ for all $z \in \mathcal{W}_d$, and the equality holds since since $f$ is a probability density.
Therefore, 
\begin{eqnarray}\label{eq_n44}
2 & \stackrel{\textrm{Eq. } \eqref{eq_n43}}{\geq}& \int_{a_{\min}}^{a_{\max}} \int_{b_{\min}}^{b_{\max}} \int_0^{s\frac{1}{8\mathfrak{b}^4\sqrt{d}}} \int_{S_3(a,b;y) \cap E} \frac{f(g(\eta))}{f(\eta)}  \mathrm{det}(J_g(\eta))\times f(\eta) \mathrm{d} \eta \mathrm{d}y \mathrm{d} a \mathrm{d} b\nonumber\\
&\stackrel{\textrm{Lem. } \ref{lemma_density_ratio} \, \& \, \ref{prop_Jacobian}}{\geq}& \int_{a_{\min}}^{a_{\max}} \int_{b_{\min}}^{b_{\max}} \int_0^{s\frac{1}{8\mathfrak{b}^4\sqrt{d}}} \int_{S_3(a,b;y) \cap E} \frac{1}{400 (8 \mathfrak{b}^4 \sqrt{d})^2 \times y^2} \times \frac{1}{16s} \times f(\eta) \mathrm{d} \eta \mathrm{d}y \mathrm{d} a \mathrm{d} b\nonumber\\
&\geq &\int_{a_{\min}}^{a_{\max}} \int_{b_{\min}}^{b_{\max}} \int_0^{s\frac{1}{8\mathfrak{b}^4\sqrt{d}}} \int_{S_3(a,b;y) \cap E} \frac{1}{400   s^2} \times \frac{1}{16s} \times f(\eta) \mathrm{d} \eta \mathrm{d}y \mathrm{d} a \mathrm{d} b \label{eq_n157}\\
&=& \frac{1}{6400   s^3}  \int_{a_{\min}}^{a_{\max}} \int_{b_{\min}}^{b_{\max}} \int_0^{s\frac{1}{8\mathfrak{b}^4\sqrt{d}}} \int_{S_3(a,b;y) \cap E}  f(\eta) \mathrm{d} \eta \mathrm{d}y \mathrm{d} a \mathrm{d} b,
\end{eqnarray}
\eqref{eq_n157} holds since the bounds of integration for $y$ are $0\leq y \leq s\frac{1}{8\mathfrak{b}^4\sqrt{d}}$. 
Therefore,
\begin{equation}\label{eq_f1}
 \int_{a_{\min}}^{a_{\max}} \int_{b_{\min}}^{b_{\max}} \int_0^{s\frac{1}{8\mathfrak{b}^4\sqrt{d}}} \int_{S_3(a,b;y) \cap E}  f(\eta) \mathrm{d} \eta \mathrm{d}y \mathrm{d} a \mathrm{d} b \stackrel{\textrm{Eq. } \eqref{eq_n44}}{\leq} 12800 \times s^3.
\end{equation}
Hence,
\begin{eqnarray}\label{eq_n46}
   & &\!\!\!\!\!\!\!\!\!\!\!\!\!\!\!\!\!\!\!\!\!\!\!\!\!\!\!\!\!\! \mathbb{P}\left(\eta_i - \eta_{i+1} \leq s\frac{1}{8\mathfrak{b}^4\sqrt{d} }\right) 
    \leq \mathbb{P}\left(\left\{\eta_i - \eta_{i+1} \leq  s\frac{1}{8\mathfrak{b}^4\sqrt{d} }\right\} \cap E\right)+ \mathbb{P}(E^c)\nonumber\\
        &\stackrel{\textrm{Eq. } \eqref{eq_joint_density}}{=}& \int_{\left \{\eta \in \mathcal{W}_d \, \, \, : \, \, \,\eta_i - \eta_{i+1} \leq  s\frac{1}{8\mathfrak{b}^4\sqrt{d} } \right\} \cap E} f(\eta) \mathrm{d} \eta + \mathbb{P}(E^c)\label{eq_n159}\\
&\stackrel{\textrm{Prop. } \ref{prop_n1}}{=} & \int_{a_{\min}}^{a_{\max}} \int_{b_{\min}}^{b_{\max}} \int_{\left \{\eta \in \mathcal{W}_d \, \, \, : \, \, \,\eta_i - \eta_{i+1} \leq  s\frac{1}{8\mathfrak{b}^4\sqrt{d} } \right\} \cap E \cap \{\eta \in \mathcal{W}_d \, \, \, : \, \, \, \eta_{j_{\mathrm{max}}} = a, \, \, \eta_{j_{\mathrm{min}}} = b\}} f(\eta) \mathrm{d} \eta \mathrm{d} a \mathrm{d} b\nonumber\\
& & + \mathbb{P}(E^c) \label{eq_n158}\\
    &\stackrel{\textrm{Eq. } \eqref{eq_n22}}{=}&\int_{a_{\min}}^{a_{\max}} \int_{b_{\min}}^{b_{\max}} \int_0^{s\frac{1}{8\mathfrak{b}^4\sqrt{d}}} \int_{S_3(a,b;y) \cap E} f(\eta) \mathrm{d} \eta \mathrm{d}y \mathrm{d} a \mathrm{d} b + \mathbb{P}(E^c)\nonumber\\
    &\stackrel{\textrm{Eq. \eqref{eq_f1}}}{\leq}&  12800\times s^3 + \mathbb{P}(E^c) \nonumber\\
        &\stackrel{\textrm{Eq. } \eqref{eq_rigidity_1}}{\leq}& 12800 \times s^3 + \frac{1}{d^{1000}}.
\end{eqnarray}
 Here \eqref{eq_n159} holds since for every $\eta \in \mathcal{W}_d$, $f(\eta)$ is the joint probability density function of the GUE eigenvalues $\eta$ \eqref{eq_joint_density}.
 \eqref{eq_n158} holds since, by Proposition \ref{prop_n1}, we have $ \eta_{j_{\mathrm{max}}}  \in[a_{\mathrm{min}}, a_{\mathrm{max}}]$ and  $ \eta_{j_{\mathrm{min}}}  \in[b_{\mathrm{min}}, b_{\mathrm{max}}]$ whenever the event $E$ occurs.

Redefining the universal constant $L$ (and hence redefining $\mathfrak{b}$), we get that
\begin{equation*}
    \mathbb{P}\left(\eta_i - \eta_{i+1} \leq s \frac{1}{\mathfrak{b}\sqrt{d}}\right) \stackrel{\textrm{Eq. } \eqref{eq_n46}}{\leq} s^{3} + \frac{1}{d^{1000}} \qquad \forall s>0,
\end{equation*}
which proves  Lemma \ref{lemma_GUE_gaps} for any $\mathfrak{b}^2 \leq i \leq d-\mathfrak{b}^2$.

\medskip
\noindent
{\em Edge case ($\min(i, d-i) \leq \mathfrak{b}^2$).}
Since the joint density of the eigenvalues  \eqref{eq_joint_density} is symmetric about $0$, without loss of generality we may assume that $i\leq \mathfrak{b}^2$.

The proof of Lemma \ref{lemma_GUE_gaps} for the edge case $i\in \mathfrak{b}^2$ follows exactly the same steps as for the bulk case ($\mathfrak{b}^2 \leq i \leq d-\mathfrak{b}^2$), if we replace the map $g$ with the map $\phi$, Proposition \ref{prop_n1} with Proposition \ref{prop_n1_edge}, Proposition \ref{prop_map} with Proposition \ref{prop_map_phi}, Lemma \ref{lemma_density_ratio} with Lemma \ref{lemma_density_ratio_edge}, and Lemma \ref{prop_Jacobian} with Lemma  \ref{prop_Jacobian_phi}.
\end{proof}

\section*{Acknowledgments}
OM was supported in part by an NSF CCF-2104528 award and a Google Research Scholar award. 
NV was supported in part by an NSF CCF-2112665 award.

\bibliographystyle{plain}
\bibliography{DP}

\newpage
\appendix

\newpage

\section{Tightness of the upper bound in Theorem \ref{thm_rank_k_covariance_approximation_new}} \label{appendix_tightness}
In this section, we show that the upper bound in Theorem \ref{thm_rank_k_covariance_approximation_new} is tight up to lower-order terms.

{
\paragraph{The case when $k=d$.}  To see why our bound in Theorem \ref{thm_rank_k_covariance_approximation_new} is tight when $k=d$, one can plug in $\sigma_{d+1} = 0$ into the r.h.s. of our utility bound which gives a bound of $\sqrt{\mathbb{E}[\|\hat{M}-M\|_F^2]} \leq \tilde{O}(d  \sqrt{T})$.
 Since $\hat{M}-M = (G + G^\ast) \times \sqrt{T}$ where $G$ has  iid Gaussian entries, we have that $\|\hat{M}-M\|_F = \Theta(d  \sqrt{T})$ w.h.p. from standard matrix concentration bounds.}
 
{
\paragraph{The case when $k=1$.} To see why our bound is tight when $k=1$,  consider the case when $M$ has top eigenvalue $\sigma_1$ and all other eigenvalues $\sigma_2= \cdots = \sigma_d$ where $\sigma_1$ is very large ($\sigma_1 \rightarrow \infty$) and  $\frac{\sigma_1}{\sigma_1- \sigma_2} = c$ for any constant $c>0$.
In this case, the eigenvalue repulsion terms $\frac{1}{\gamma_i(t) - \gamma_j(t)}$ in the eigenvalue evolution equations \eqref{eq_DBM_eigenvalues} are higher-order which scale as $\frac{1}{\sigma_1}$ as $\sigma_1 \rightarrow \infty$.
Thus, from \eqref{eq_DBM_eigenvalues} we have that 
\begin{equation}\label{eq_tight_1}
\hat{\sigma}_1 - \sigma_1 + g_1
\end{equation}
with probability $1$ as $\sigma_1 \rightarrow \infty$, where $g_1 \sim N(0,T)$.}

{
In a similar manner, we have that the terms $\frac{1}{(\gamma_i(t)- \gamma_j(t))^2}$ in the eigenvector evolution equations \eqref{eq_DBM_eigenvectors} are higher-order terms which scale as $\frac{1}{\sigma_1^2}$ as $\sigma_1 \rightarrow \infty$.
Denote by $v_1,\ldots, v_d$ the  eigenvectors of $M$ and $\hat{v}_1$ the top eigenvector of $\hat{M}$.
Thus, we have from  \eqref{eq_DBM_eigenvectors} that
 \begin{equation}\label{eq_tight_2}
 \sigma_1(\hat{v}_1 - v_1) \rightarrow \sum_{i=2}^d g_i \times c \times v_i
 \end{equation}
 with probability $1$ as $\sigma_1 \rightarrow \infty$ where $g_2,\ldots, g_d \sim N(0,T)$.
Thus, from \eqref{eq_tight_1} and \eqref{eq_tight_2} we have that
\begin{equation*}
\|\hat{M}_1 - M_1\|_F = \|\hat{\sigma}_1 \hat{v}_1 \hat{v}_1^\ast -  \sigma_1 v_1 v_1^\ast\|_F \rightarrow  \sqrt{\sum_{i=2}^d  g_i^2 \times c^2} = \Theta(\sqrt{d} \times c \sqrt{T})
\end{equation*}
with probability $1$ as $\sigma_1 \rightarrow \infty$.
In other words, for $\sigma_1$ large enough we have that $\|\hat{M}_1 - M_1\|_F = \Theta(\sqrt{d}  \frac{\sigma_1}{\sigma_1- \sigma_2} \sqrt{T})$ w.h.p.}

{\paragraph{The case when $1\leq k<d$.} The above example, which was given for $k=1$, can be generalized to any $k<d$ by setting $M$ to have top-$k$ eigenvalues $\sigma_1 = \cdots = \sigma_k$, and the remaining eigenvalues $\sigma_{k+1}= \cdots = \sigma_d$, and taking $\sigma_k \rightarrow \infty$ where $\frac{\sigma_k}{\sigma_{k}- \sigma_{k+1}} = c$ for any constant $c>0$.
In this case, we get that, for $\sigma_k$ large enough, $\|\hat{M}_k - M_k\|_F = \Theta(\sqrt{k} \sqrt{d} \frac{\sigma_k}{\sigma_k- \sigma_{k+1}} \sqrt{T})$ w.h.p.
 Thus, for any $k\leq d$, our bound is tight up to factors of $(\log d)^{\log \log d}$ hidden in the $\tilde{O}$ notation.}

\section{Eigengap-free utility bounds in a weaker Frobenius norm metric} \label{appendix_gap_free_bounds_in_weaker_metric}

In this Section, we show how one can extend our main result in Theorem \ref{thm_rank_k_covariance_approximation_new} to obtain eigengap-free utility bounds on the weaker Frobenius metric $\|\hat{M}_k - M\|_F^2 - \|M_k - M\|_F^2$:

\begin{theorem}[\bf Eigengap-free utility bound in a Weaker Frobenius Metric]\label{thm_weaker_Frobenius_metric}
Suppose we are given $k>0$, $T>0$, and a  Hermitian matrix  $M \in \mathbb{C}^{d \times d}$ with eigenvalues  $\sigma_1 \geq \cdots \geq \sigma_d \geq 0$.   
Let $\hat{M} := M + \sqrt{T}[(W_1 + \mathfrak{i}W_2) + (W_1 + \mathfrak{i}W_2)^\ast]$ where $W_1, W_2 \in \mathbb{R}^{d \times d}$ have entries which are independent $N(0,1)$ random variables.
 Denote, respectively, by  $\sigma_1 \geq \cdots \geq \sigma_d$ and  $\hat{\sigma}_1\geq \ldots \geq \hat{\sigma}_d \geq 0$ the eigenvalues of $M$ and $\hat{M}$, and by $V$ and  $\hat{V}$ the matrices whose columns are the corresponding eigenvectors of $M$ and $\hat{M}$.
Moreover, let  $M_k := V \Gamma_k V^\ast$ and  $\hat{M}_k := \hat{V} \hat{\Gamma}_k \hat{V}^\ast$ be the rank-$k$ approximations of $M$ and $\hat{M}$, where $\Gamma_k := \mathrm{diag}(\sigma_1,\ldots, \sigma_k,0,\ldots,0)$  and $\hat{\Gamma}_k := \mathrm{diag}(\hat{\sigma}_1,\ldots, \hat{\sigma}_k,0,\ldots,0)$.
Suppose that $\sigma_1 \leq d^{50}$.
Then we have
$$ 
 \sqrt{\mathbb{E}[\|\hat{M}_k - M\|_F^2 - \|M_k - M\|_F^2]} \leq \tilde{O}\left(\sqrt{kd} \cdot \sqrt{T}\right).
$$
\end{theorem}

\noindent
The following steps can be used to extend the proof of Theorem \ref{thm_rank_k_covariance_approximation_new} to obtain the eigengap-free utility bounds on the weaker Frobenius metric in Theorem \ref{thm_weaker_Frobenius_metric}:

\begin{enumerate}
\item {\bf Applying Ito's lemma to the weaker Frobenius norm metric.}
When bounding the stronger utility metric $\left\|\hat{M}_k -  M_k\right \|_F^2$ in the proof of Theorem \ref{thm_rank_k_covariance_approximation_new} we apply Ito's lemma (Lemma  \ref{lemma_ito_lemma_new}) to the function $f(Y) = \|Y\|_F^2$.  
 If we only wish to bound the weaker utility metric $\|\hat{M}_k - M\|_F^2 - \|M_k - M\|_F^2$, we can instead apply Ito's Lemma to the function
 \begin{equation}\label{eq_n104}
 g(Y) := \|Y - M\|_F^2.
 \end{equation}
 %

 %
   For conciseness, with slight abuse of notation, in steps that apply Ito's lemma we will view $g: \mathbb{C}^{d^2} \rightarrow \mathbb{R}$ as a function that takes as input the vectorized matrix $Y \in \mathbb{C}^{d \times d} \equiv  \mathbb{C}^{d^2}$, and express the derivatives of $g$ in this notation.
   Then, recalling from   \eqref{eq_n96} and \eqref{eq_n45} that $\Psi(T) = \hat{M}_k$ and $\Psi(0) = M_k$, we have
\begin{eqnarray} \label{eq_n105}
 \|\hat{M}_k - M\|_F^2 - \|M_k - M\|_F^2  \stackrel{\textrm{Eq. } \eqref{eq_n104}}{=} g(\Psi(T)) - g(\Psi(0))\qquad \qquad \qquad \qquad \qquad \qquad  \nonumber\\
%
%
\stackrel{\textrm{Ito's Lemma (Lem.  \ref{lemma_ito_lemma_new})}}{=}  \int_0^T  \left( \frac{1}{2} (\mathrm{d}\Psi(t))^\ast  \nabla^2 g(\Psi(t)) \mathrm{d}\Psi(t)   + (\nabla g(\Psi(t)))^\ast \mathrm{d}\Psi(t) \right ) \mathrm{d}t,
\end{eqnarray}
\color{black}
 where
\begin{equation} \label{eq_W2}
\nabla g(Y) [ij] = \frac{\partial}{\partial Y_{ij}} g(Y) = 2Y_{ij} -2M_{ij}, \qquad \textrm{ and }
\end{equation}
 \begin{equation*}
 \nabla^2 g(Y)[ij, \alpha \beta] = \frac{\partial}{\partial Y_{ij} \partial Y_{\alpha \beta}} g(Y) = \begin{cases} 2 & \textrm{ for } (i,j)=(\alpha, \beta)\\ 0 & \textrm{otherwise}. \end{cases}
 \end{equation*}
 
\item {\bf Canceling the eigengap terms.}
The extra term $-2M_{ij}$ in the first derivative \eqref{eq_W2} leads to cancellations of the terms in the utility bound which depend on the eigenvalue gap.
 To see why, we first note that from the proof of Theorem \ref{thm_rank_k_covariance_approximation_new}, we have
   \begin{eqnarray} 
  \mathrm{d} \Psi(t) &\stackrel{\textrm{Eq. } \eqref{eq_ito_derivative}}{=}& \sum_{i=1}^d \lambda_i(t) \mathrm{d}(u_i(t) u_i^\ast(t))) + (\mathrm{d}\lambda_i(t)) (u_i(t) u_i^\ast(t)) \nonumber\\
   &\stackrel{\textrm{Eq. } \eqref{eq_n102}}{=}&  \frac{1}{2} \sum_{i=1}^d     \sum_{j \neq i} \frac{\lambda_i(t)- \lambda_j(t)}{\gamma_i(t) - \gamma_j(t)}(u_i(t) u_j^\ast(t)\mathrm{d}B_{ij}(t) + u_j(t) u_i^\ast(t)\mathrm{d}B_{ij}^\ast(t)) \nonumber\\
& & \qquad \qquad  - \sum_{i=1}^d   \sum_{j\neq i} \frac{ \lambda_i(t) - \lambda_j(t) }{(\gamma_i(t)- \gamma_j(t))^2} u_i(t) u_i^\ast(t) \mathrm{d}t \nonumber\\
&+&  \sum_{i=1}^d  (\mathrm{d}\lambda_i(t)) (u_i(t) u_i^\ast(t)) \nonumber\\
   &\stackrel{\textrm{Eq. } \eqref{eq_n6}, \, \eqref{eq_n45}}{=}&  \frac{1}{2} \sum_{i=1}^d     \sum_{j \neq i} \frac{\lambda_i(t)- \lambda_j(t)}{\gamma_i(t) - \gamma_j(t)}(u_i(t) u_j^\ast(t)\mathrm{d}B_{ij}(t) + u_j(t) u_i^\ast(t)\mathrm{d}B_{ij}^\ast(t)) \nonumber\\
& & \qquad \qquad  - \sum_{i=1}^d   \sum_{j\neq i} \frac{ \lambda_i(t) - \lambda_j(t) }{(\gamma_i(t)- \gamma_j(t))^2} u_i(t) u_i^\ast(t) \mathrm{d}t \nonumber\\
&+&     \sum_{i=1}^k \left(\mathrm{d}B_{i i}(t) +  2 \sum_{j \neq i} \frac{1}{\gamma_i(t) - \gamma_j(t)} \mathrm{d}t \right) u_i(t) u_i^\ast(t) \label{eq_n106}
  \end{eqnarray}

\noindent  
  Therefore, we have
  \begin{eqnarray}
& & \!\!\!\!\!\!\!\!\!\!\!\!\!\!\!\!\!\!\!\!\!\!\!\!\!\!\!\!\!\!\!\!\!\!\!\!\!\!\!\!\!\! 
 (\mathrm{d}\Psi(t))^\ast  \nabla^2 g(\Psi(t)) \mathrm{d}\Psi(t)  \stackrel{\textrm{Eq. } \eqref{eq_W2}}{=} 2 \sum_{\alpha, \beta \in [d]}   |\mathrm{d}\Psi(t) [\alpha, \beta]|^2 \nonumber\\
  &=& 2 \langle \mathrm{d}\Psi(t), \,\, \mathrm{d}\Psi(t) \rangle \nonumber\\
  &\stackrel{\textrm{Eq. } \eqref{eq_n106}}{=}& \frac{1}{2} \sum_{i=1}^d     \sum_{j \neq i} \frac{(\lambda_i(t)- \lambda_j(t))^2}{(\gamma_i(t) - \gamma_j(t))^2} \|u_i(t) u_j^\ast(t)\mathrm{d}B_{ij}(t) + u_j(t) u_i^\ast(t)\mathrm{d}B_{ij}^\ast(t)\|_F^2 \nonumber\\
  &+&     \sum_{i=1}^k |\mathrm{d}B_{i i}(t)|^2  \|u_i(t) u_i^\ast(t)\|_F^2 \label{eq_n108}\\
  &=& \sum_{i=1}^d     \sum_{j \neq i} \frac{(\lambda_i(t)- \lambda_j(t))^2}{(\gamma_i(t) - \gamma_j(t))^2} \mathrm{d}t +   k \mathrm{d}t,  \label{eq_n107}
  \end{eqnarray}
  where \eqref{eq_n108} holds since $(\mathrm{d}t)^2 = 0$,  and $\mathrm{d}B_{ij} \mathrm{d}B_{\ell r} = 0$ whenever $(i,j) \notin\{(\ell, r), (r, \ell)$,  and since 
  \begin{equation}\label{eq_n109}
  \langle u_i(t) u_j^\ast(t), u_j(t) u_i^\ast(t) \rangle = \mathrm{tr}(u_j(t) u_i^\ast(t) u_j(t) u_i^\ast(t)) = 0
  \end{equation}
   for all $i \neq j$ because $u_i(t)$ and $u_j(t)$ are orthogonal eigenvectors.
  \eqref{eq_n107} holds since 
  \begin{eqnarray}
 & & \|u_i(t) u_j^\ast(t)\mathrm{d}B_{ij}(t) + u_j(t) u_i^\ast(t)\mathrm{d}B_{ij}^\ast(t)\|_F^2 = \nonumber\\
  & & \qquad \qquad  \mathrm{tr}(u_i(t) u_j^\ast(t) u_j(t) u_i^\ast(t))|\mathrm{d}B_{ij}(t)|^2 + 
   \mathrm{tr}(u_i(t) u_j^\ast(t) u_j(t) u_i^\ast(t))|\mathrm{d}B_{ij}^\ast(t)|^2 \nonumber\\
    & & \qquad \qquad \qquad \qquad+   \langle u_i(t) u_j^\ast(t), \, \, u_j(t) u_i^\ast(t) \rangle  | \mathrm{d}B_{ij}^\ast(t)|^2 +  \langle u_i(t) u_j^\ast(t) \, \, u_i(t) u_j^\ast(t) \rangle |\mathrm{d}B_{ij}(t)|^2 \nonumber\\
   &&  \qquad  \,\,\,\,\, \stackrel{\textrm{Eq. } \eqref{eq_n109}}{=} \mathrm{tr}(u_i(t) u_j^\ast(t) u_j(t) u_i^\ast(t))|\mathrm{d}B_{ij}(t)|^2 + 
   \mathrm{tr}(u_i(t) u_j^\ast(t) u_j(t) u_i^\ast(t))|\mathrm{d}B_{ij}^\ast(t)|^2 + 0 \nonumber\\
       & &  \qquad \qquad  = 2 \mathrm{d}t  \label{eq_n110},
\end{eqnarray}
and \eqref{eq_n110} holds since $|\mathrm{d}B_{ij}^\ast(t)|^2 = \mathrm{d}t$ for all $i,j \in [d]$, and since $u_i(t)$ and $u_j(t)$ are orthonormal eigenvectors.

Moreover, recall that 
\begin{equation}\label{eq_n112}
M \stackrel{\textrm{Eq. } \eqref{eq_n95}}{=}  \Phi(0) \stackrel{\textrm{Eq. } \eqref{eq_n92}}{=}  \sum_{i=1}^d \gamma_i(0) u_i(0) u_i^\ast(0)
\end{equation}
 and
\begin{equation}\label{eq_n113}
\Psi(t) \stackrel{\textrm{Eq. } \eqref{eq_n96}, \, \, \eqref{eq_n45}}{=} \sum_{i=1}^k \gamma_i(t) u_i(t) u_i^\ast(t) \qquad \forall t \geq 0.
\end{equation}
Then we have
\begin{eqnarray}
& & \!\!\!\!\!\!\!\!\!\!\!\!\!\!\!\!\!\!\!\!\!\!\!\!\!\!\!\!\! \mathbb{E}\left[ (\nabla g(\Psi(t)))^\ast \mathrm{d}\Psi(t) \right ]   \stackrel{\textrm{Eq. } \eqref{eq_W2}}{=}  \mathbb{E}\left[  \left \langle 2 \Psi(t) - 2M \, , \,  \mathrm{d}\Psi(t)  \right  \rangle \right ]  \nonumber\\
&\!\!\!\!\!\!\!\!\!\!\!\!\!\!\!\! \stackrel{\textrm{Eq. } \eqref{eq_n106}}{=}&\!\!\!\!\!\!\!\!\!\!\!\! \mathbb{E}\left[ \left \langle  2 \Psi(t) - 2M \, \, , \, \, \frac{1}{2} \sum_{i=1}^d     \sum_{j \neq i} \frac{\lambda_i(t)- \lambda_j(t)}{\gamma_i(t) - \gamma_j(t)}(u_i(t) u_j^\ast(t)\mathrm{d}B_{ij}(t) + u_j(t) u_i^\ast(t)\mathrm{d}B_{ij}^\ast(t))  \right \rangle \right ] \nonumber\\
& & \qquad \qquad  -  \mathbb{E}\left[ \left \langle  2 \Psi(t) - 2M \, \, , \, \,  \sum_{i=1}^d   \sum_{j\neq i} \frac{ \lambda_i(t) - \lambda_j(t) }{(\gamma_i(t)- \gamma_j(t))^2} u_i(t) u_i^\ast(t) \mathrm{d}t \right \rangle \right ] \nonumber\\
&+&     \mathbb{E}\left[ \left \langle  2 \Psi(t) - 2M   \, \, , \, \,   \sum_{i=1}^d  (\mathrm{d}\lambda_i(t)) (u_i(t) u_i^\ast(t))     \right \rangle \right ]  \nonumber\\
& = &  0 \, \, \,  - \, \, \,  \mathbb{E}\left[ \left \langle  2 \Psi(t) - 2M \, \, , \, \,  \sum_{i=1}^d   \sum_{j\neq i} \frac{ \lambda_i(t) - \lambda_j(t) }{(\gamma_i(t)- \gamma_j(t))^2} u_i(t) u_i^\ast(t) \mathrm{d}t \right \rangle \right ] \label{eq_n111}\\
&+&     \mathbb{E}\left[ \left \langle  2 \Psi(t) - 2M   \, \, , \, \,   \sum_{i=1}^d  (\mathrm{d}\lambda_i(t)) (u_i(t) u_i^\ast(t))     \right \rangle \right ]  \nonumber\\
& \stackrel{\textrm{Eq.}\, \eqref{eq_n112},}{\stackrel{\eqref{eq_n113}}{=}} &\!\!\!\!\! \mathbb{E}\left[ \left \langle   2 \sum_{i=1}^k \gamma_i(t) u_i(t) u_i^\ast(t) - 2\sum_{i=1}^d \gamma_i(0) u_i(0) u_i^\ast(0), \, \,  \sum_{i=1}^d   \sum_{j\neq i} \frac{ \lambda_i(t) - \lambda_j(t) }{(\gamma_i(t)- \gamma_j(t))^2} u_i(t) u_i^\ast(t) \mathrm{d}t \right \rangle \right ] \nonumber\\
&+&     \mathbb{E}\left[ \left \langle  2 \Psi(t) - 2M   \, \, , \, \,   \sum_{i=1}^d  (\mathrm{d}\lambda_i(t)) (u_i(t) u_i^\ast(t))     \right \rangle \right ] \nonumber\\
&=&\mathbb{E}\bigg[   2\sum_{i=1}^{d} \sum_{j\neq i} \frac{\lambda_i(t) - \lambda_j(t)}{(\gamma_i(t) - \gamma_j(t))^2} \times \sum_{\ell=1}^k \gamma_\ell(t)  \left \langle u_i(t) u_i^\ast(t), \, \, u_{\ell}(t) u_{\ell}^\ast(t) \right \rangle \mathrm{d}t \nonumber\\ 
 && - \quad    2\sum_{i=1}^{d} \sum_{j\neq i} \frac{\lambda_i(t) - \lambda_j(t)}{(\gamma_i(t) - \gamma_j(t))^2} \times \sum_{\ell=1}^d \gamma_\ell(0)  \left \langle u_i(0) u_i^\ast(0), \, \, u_{\ell}(t) u_{\ell}^\ast(t) \right \rangle \mathrm{d}t \bigg] \nonumber\\
 &+&     \mathbb{E}\left[ \left \langle  2 \Psi(t) - 2M   \, \, , \, \,   \sum_{i=1}^d  (\mathrm{d}\lambda_i(t)) (u_i(t) u_i^\ast(t))     \right \rangle \right ], \label{eq_n118}
\end{eqnarray}
where \eqref{eq_n111} holds since $\mathbb{E}[\mathrm{d}B_{ij}(t)] = 0$ and $\mathrm{d}B_{ij}(t)$ is independent of $\Psi(\tau)$, $\gamma_\ell(\tau)$, $\lambda_\ell(\tau)$, and $u_i(\tau)$ for all $\tau \leq t$ and all $i, j, \ell \in [d]$, and $M$ is a constant matrix.

\bigskip
\bigskip
\bigskip

Therefore, we have
 \begin{eqnarray} \label{eq_W1}
 & & \!\!\!\!\!\!\!\!\!\!\!\!\!\!\!\!\!\!\!\!\!\!\!\!\!\!\!\!\!\!\!\!\!\!\!\!\!\!\!\!\!\!\!\!\!\!\!\! \mathbb{E}\left[ \|\hat{M}_k - M\|_F^2 - \|M_k - M\|_F^2 \right]\nonumber\\
 &\stackrel{\textrm{Eq. } \eqref{eq_n105}}{=}& \mathbb{E}\left[\int_0^T  \frac{1}{2} (\mathrm{d}\Psi(t))^\ast  \nabla^2 g(Y) \mathrm{d}\Psi(t)   + (\nabla g(Y))^\ast \mathrm{d}\Psi(t) \, \, \mathrm{d}t  \right ] \nonumber\\
 &\stackrel{\textrm{Eq. } \eqref{eq_n107},  \eqref{eq_n118}}{=}& \mathbb{E}\bigg[ \int_{0}^{T}   \sum_{i=1}^{d}  \sum_{j \neq i}  \frac{(\lambda_i(t) - \lambda_j(t))^2}{(\gamma_i(t) - \gamma_j(t))^2}  \mathrm{d}t \nonumber\\
&&+ \quad  \int_{0}^{T}   2\sum_{i=1}^{d} \sum_{j\neq i} \frac{\lambda_i(t) - \lambda_j(t)}{(\gamma_i(t) - \gamma_j(t))^2} \times \sum_{\ell=1}^k \gamma_\ell(t)  \left \langle u_i(t) u_i^\ast(t), \, \, u_{\ell}(t) u_{\ell}^\ast(t) \right \rangle \mathrm{d}t \nonumber\\ 
 && - \quad    \int_{0}^{T}   2\sum_{i=1}^{d} \sum_{j\neq i} \frac{\lambda_i(t) - \lambda_j(t)}{(\gamma_i(t) - \gamma_j(t))^2} \times \sum_{\ell=1}^d \gamma_\ell(0)  \left \langle u_i(0) u_i^\ast(0), \, \, u_{\ell}(t) u_{\ell}^\ast(t) \right \rangle \mathrm{d}t \bigg] \nonumber\\
&& + \quad   \mathbb{E}\left[  \int_{0}^{T}   \left \langle  2 \Psi(t) - 2M   \, \, , \, \,   \sum_{i=1}^d  (\mathrm{d}\lambda_i(t)) (u_i(t) u_i^\ast(t))     \right \rangle \right ],
  \end{eqnarray}

\noindent
  The term $\mathbb{E}\left[  \int_{0}^{T}   \left \langle  2 \Psi(t) - 2M   \, \, , \, \,   \sum_{i=1}^d  (\mathrm{d}\lambda_i(t)) (u_i(t) u_i^\ast(t))     \right \rangle \right ] = \tilde{O}(kd \cdot T)$ in \eqref{eq_W1} can be bounded using the same steps as \eqref{eq_a1}.  Thus, we have
 \begin{eqnarray}
& &\!\!\!\!\!\!\!\!\!\!\!\!\!\!\!\!\!\!\!\!\!\!\!\!\!\!\!\!\!\!\!\!\!\!\! \mathbb{E}\left[ \|\hat{M}_k - M\|_F^2 - \|M_k - M\|_F^2 \right]\nonumber\\ 
 &\stackrel{\textrm{Eq. } \eqref{eq_W1}}{=}& \mathbb{E}\bigg[ \int_{0}^{T} \bigg(   \sum_{i=1}^{d}  \sum_{j \neq i}  \frac{(\lambda_i(t) - \lambda_j(t))^2}{(\gamma_i(t) - \gamma_j(t))^2}  \nonumber\\
 && + \quad 2\sum_{i=1}^{d} \sum_{j\neq i} \frac{\lambda_i(t) - \lambda_j(t)}{(\gamma_i(t) - \gamma_j(t))^2} \times \sum_{\ell=1}^k \gamma_\ell(t)  \left \langle u_{\ell}(t) u_{\ell}^\ast(t) , \, \,   u_i(t) u_i^\ast(t) \right \rangle \nonumber\\ 
 && - \quad 2\sum_{i=1}^{d} \sum_{j\neq i} \frac{\lambda_i(t) - \lambda_j(t)}{(\gamma_i(t) - \gamma_j(t))^2} \times \sum_{\ell=1}^d \gamma_\ell(0)  \left \langle u_{\ell}(0) u_{\ell}^\ast(0), \, \, u_i(t) u_i^\ast(t) \right \rangle \bigg ) \mathrm{d}t \bigg] + \tilde{O}(kd)  \nonumber\\
 &=& \mathbb{E}\bigg[ \int_{0}^{T}  \bigg( \sum_{i=1}^{d}  \sum_{j \neq i}  \frac{(\lambda_i(t) - \lambda_j(t))^2}{(\gamma_i(t) - \gamma_j(t))^2}     + 2\sum_{i=k+1}^{d} \sum_{j\neq i} \frac{(\lambda_j(t)-\lambda_i(t)) \times \gamma_i(t)  }{(\gamma_i(t) - \gamma_j(t))^2} \nonumber \\
&& - \quad    2\sum_{i=k+1}^{d} \sum_{j\neq i} \frac{(\lambda_j(t)-\lambda_i(t)) \times \gamma_i(t)  }{(\gamma_i(t) - \gamma_j(t))^2} \bigg) \mathrm{d}t\bigg ] + \mathbb{E}\left[\int_0^T \mathcal{H}(t)\mathrm{d}t\right] + \tilde{O}(kd \cdot T)  \label{eq_n115}\\
 &=& \mathbb{E}\bigg[  \int_{0}^{T}   \sum_{i=1}^{k}  \sum_{j \neq i}  \frac{(\gamma_i(t) - \gamma_j(t))^2}{(\gamma_i(t) - \gamma_j(t))^2}  \mathrm{d}t \bigg]  + \mathbb{E}\left[\int_0^T \mathcal{H}(t)\mathrm{d}t\right] + \tilde{O}(kd \cdot T)  \nonumber\\
 &=& \tilde{O}(kd \cdot T),  \label{eq_n116}
 \end{eqnarray}
 where \eqref{eq_n115} is obtained by making small-$t$ approximations $ \gamma_i(0) \approx \gamma_i(t)$ and $u_i(0) \approx u_i(t)$, and  $\mathcal{H}(t)$ are the higher-order terms which remain after making these approximations.

\color{black}

\item {\bf Bounding the higher-order terms.} 
More specifically, the higher-order terms are
 \begin{eqnarray}
\mathcal{H}(t) \!\!\! &=& \!\!\! \sum_{i=1}^{d} \sum_{j\neq i} \frac{\lambda_i(t) - \lambda_j(t)}{(\gamma_i(t) - \gamma_j(t))^2} \times \sum_{\ell=1}^d   (\gamma_\ell(0) - \gamma_\ell(t) ) \left \langle  u_{\ell}(t) u_{\ell}^\ast(t), \, \, u_i(t) u_i^\ast(t) \right \rangle \mathrm{d}t \nonumber\\
&& +  2\sum_{i=1}^{d} \sum_{j\neq i} \frac{\lambda_i(t) - \lambda_j(t)}{(\gamma_i(t) - \gamma_j(t))^2} \times \sum_{\ell=1}^d   \gamma_\ell(0) \left \langle u_{\ell}(0) u_{\ell}^\ast(0) -  u_{\ell}(t) u_{\ell}^\ast(t), \, \, u_i(t) u_i^\ast(t) \right \rangle \mathrm{d}t \label{eq_n119}\\
&= & 2\sum_{i=1}^{d} \sum_{j\neq i} \frac{\lambda_i(t) - \lambda_j(t)}{(\gamma_i(t) - \gamma_j(t))^2} \times  (\gamma_i (0) - \gamma_i(t) ) \mathrm{d}t \nonumber\\
& & +  2\sum_{i=1}^{d} \sum_{j\neq i} \frac{\lambda_i(t) - \lambda_j(t)}{(\gamma_i(t) - \gamma_j(t))^2} \times \sum_{\ell=1}^d   \gamma_\ell(0) \left \langle u_{\ell}(0) u_{\ell}^\ast(0) -  u_{\ell}(t) u_{\ell}^\ast(t), \, \, u_i(t) u_i^\ast(t) \right \rangle \mathrm{d}t, \qquad \quad \label{eq_W3}
\end{eqnarray}
where \eqref{eq_n119} holds since $ \langle  u_i(t) u_i^\ast(t), \, \, u_{\ell}(t) u_{\ell}^\ast(t)  \rangle = 0$ for $\ell \neq i$, and $ \langle  u_i(t) u_i^\ast(t),$ $\, \, u_i(t) u_i^\ast(t)  \rangle = 1$.

 To bound the first term on the r.h.s. of \eqref{eq_W3}, we use the fact that the two-time joint distribution of Dyson Brownian motion, $f(\gamma(0), \gamma(t))$, is symmetric in the sense that it depends only on the quantities $\{|\gamma_i(t) - \gamma_j(0)|\}_{1\leq i,j \leq d}$ (see e.g. \cite{tao2012topics}), which implies that $\mathbb{E}[\sum_{i=1}^{d} \sum_{j\neq i} \frac{\lambda_i(t) - \lambda_j(t)}{(\gamma_i(t) - \gamma_j(t))^2} \times  (\gamma_i (0) - \gamma_i(t) ) \mathrm{d}t]=0$.
The second term can be bounded in a similar manner.

After bounding these higher-order terms, one gets the eigenvalue gap-free bound 
 \begin{eqnarray}
 \mathbb{E}[\|\hat{M}_k - M\|_F - \|M_k - M\|_F] &\leq& \sqrt{\mathbb{E}[\|\hat{M}_k - M\|_F^2 - \|M_k - M\|_F^2]} \label{eq_n117}\\
 & \stackrel{\textrm{Eq. } \eqref{eq_n116}}{\leq} & \tilde{O}(\sqrt{kd} \cdot \sqrt{T}) \nonumber,
 \end{eqnarray}
 where \eqref{eq_n117} holds since $\|\hat{M}_k - M\|_F \geq \|M_k - M\|_F \geq 0$ and since $(a-b)^2 =a^2+b^2-2ab \leq a^2 - b^2$ for any $a\geq b \geq 0$.

\end{enumerate}

 \begin{remark}[\bf Tightness of weaker metric bound]\label{remark_tightness_weaker_metric}
Theorem \ref{thm_weaker_Frobenius_metric}, which provides a bound on the weaker metric $\| \hat{M}_k - M \|_F^2 - \|  M_k - M \|_F^2$, is tight for any $d > 0$ and any $k<d$.

To see why, we note that, in Appendix \ref{appendix_tightness}, for any $d$ and  $k<d$, we construct a matrix $M$ such that $\|\hat{M}_k - M_k\|_F = \Theta(\sqrt{k} \sqrt{d} \frac{\sigma_k}{\sigma_k- \sigma_{k+1}} \sqrt{T})$ w.h.p.
This matrix $M$ is assumed to have top-$k$ eigenvalues $\sigma_1 = \cdots = \sigma_k$, and the remaining eigenvalues $\sigma_{k+1}= \cdots = \sigma_d$.
If we set $\sigma_{k+1}= \cdots = \sigma_d = 0$ in this construction, $M$  is a rank-$k$ matrix, and we obtain 
\begin{equation}\label{eq_R1}
\|\hat{M}_k - M_k\|_F = \Theta\left(\sqrt{k} \sqrt{d} \frac{\sigma_k}{\sigma_k- \sigma_{k+1}} \sqrt{T}\right) =  \Theta(\sqrt{k} \sqrt{d} \sqrt{T})
\end{equation}
 w.h.p.
 Moreover, since $M$ is rank $k$ we have $M = M_k$ and hence that
\begin{equation}\label{eq_R2}
\| \hat{M}_k - M \|_F^2 - \|  M_k - M \|_F^2 = \|\hat{M}_k - M\|_F^2.
\end{equation}
Thus, plugging \eqref{eq_R2} into \eqref{eq_R1}, we must also have that
$$\sqrt{\| \hat{M}_k - M \|_F^2 - \|  M_k - M \|_F^2} = \sqrt{\|\hat{M}_k - M\|_F^2} =  \Theta(\sqrt{k} \sqrt{d} \sqrt{T})$$
w.h.p.
This implies that our bound in Theorem B.2 on the weaker metric $\| \hat{M}_k - M \|_F^2 - \|  M_k - M \|_F^2$ must also be tight for any $d>0$ and any $k<d$.
\end{remark}
\color{black}

\section{Proof outline for low-rank subspace recovery problem (Theorem \ref{thm_rank_k_subspace})}\label{Subpsace_recovery_outline}

To prove Theorem \ref{thm_rank_k_subspace}, one can follow the same outline as the first part of the proof outline of Theorem \ref{thm_rank_k_covariance_approximation_new} given in Section \ref{sec_SDE_utility_overview}.
For simplicity, we set $T=1$ in this outline.

\paragraph{Constructing a rank-$k$ projection-matrix-valued diffusion. }  As in  Section \ref{sec_SDE_utility_overview}, we consider the continuous-time matrix diffusion $\Phi(t) = M + B(t),$ whose eigenvalues $\gamma_i(t)$ and eigenvectors $u_i(t)$, $i\in [d]$, evolve over time.
Here, $B(t) := W(t) + W(t)^\ast$, where $W(t)$ is a $d \times d$ matrix where the real part (and complex part) of each entry is an independent standard Brownian motion with distribution $N(0, tI_d)$ at time $t$.
We let $\Phi(t) = U(t) \Gamma(t) U(t)^\top$  be a spectral decomposition of the symmetric matrix $\Phi(t)$ at every time $t\geq 0$, and we define a rank-$k$ matrix-valued stochastic process $\Theta(t):= U(t) \Lambda U(t)^\top$ where $\Lambda$ is a diagonal matrix with some specified eigenvalues $\lambda_1 \geq \cdots \geq \lambda_d$ which are fixed at every time $t$.
 
The main difference is that, to bound the utility for the subspace recovery problem, we need $\Theta(t)$ to be a rank-$k$ projection matrix (instead of a rank-$k$ matrix with eigenvalues roughly equal to the top-$k$ eigenvalues of $M$).
Towards this, we set $\lambda_i = 1$ for  $i \leq k$ and $\lambda_i = 0$ otherwise (in place of the values $\lambda_i = \sigma_i$ for all $i \leq k$ used in Section \ref{sec_SDE_utility_overview}).
We obtain an equation for the utility $\mathbb{E}[\|\hat{V}_k\hat{V}_k^\ast -  V_k V_k^\ast  \|_F^2]$ of the subspace recovery problem which has the same r.h.s. as Equation \eqref{eq_t7_2} in Section \ref{sec_SDE_utility_overview}, but with $\lambda_i = 1$ for  $i \leq k$ and $\lambda_i = 0$ in place of the previous choice of $\lambda$'s,
    \begin{eqnarray}\label{eq_t7_2_b} %
& & \!\!\!\!\!\!\!\!\!\!\!\!\!\!\!\!\!\!\!\!\!\!\!\!\!\!\!\!\!\!\!\!\!\!\!\! \mathbb{E}[\|\hat{V}_k\hat{V}_k^\ast -  V_k V_k^\ast  \|_F^2] = \mathbb{E}\left[\left\|\Theta(T) -  \Theta(0)\right \|_F^2 \right]\nonumber\\
 &=&   \sum_{i=1}^{d} \int_0^{T} \mathbb{E}\left[ \sum_{j \neq i}  \frac{(\lambda_i - \lambda_j)^2}{(\gamma_i(t)-\gamma_j(t))^2} \right]
     +    T \mathbb{E}\left[\left(\sum_{j\neq i} \frac{\lambda_i - \lambda_j}{(\gamma_i(t)-\gamma_j(t))^2}\right)^2   \right]\mathrm{d}t.
\end{eqnarray}

\paragraph{Bounding the eigenvalue gaps.} Recall that, in Section \ref{sec_SDE_utility_overview}, to obtain an upper bound on the utility from \eqref{eq_t7_2} for the covariance matrix approximation problem (where, roughly, $\lambda_i = \sigma_i$ for $i\leq k$), we had to first show a lower bound on the gap terms  $\gamma_i(t) - \gamma_j(t)$ in \eqref{eq_t7_2} for all $i,j \leq k$, $i \neq j$.
For the subspace recovery problem, we only need to bound $\gamma_i(t) - \gamma_j(t)$ for $i\leq k<j$, as all the other gap terms in \eqref{eq_t7_2_b} cancel since $\lambda_i - \lambda_j = 0$ whenever either $i,j \geq k$ or $i,j < k$.

To bound these gap terms, we may apply Weyl's inequality (Lemma \ref{lemma_weyl}), which says that $\gamma_i(t) - \gamma_j(t) \geq \gamma_i(0) - \gamma_j(0) - \|B(t)\|_2$ for all $t$.
Thus, since $\|B(t)\|_2 = O(\sqrt{d})$ w.h.p. for all $t\in [0,T]$ (by Lemma \ref{lemma_spectral_martingale_b}), for any $i \neq j$ we obtain a bound of  $\gamma_i(t) - \gamma_j(t) \geq \Omega(\gamma_i(0) - \gamma_j(0))$, whenever the initial eigenvalues $\gamma(0)$, that is, the eigenvalues of the input $M$, satisfy  $$\gamma_i(0) - \gamma_j(0) = \sigma_i - \sigma_j \geq \Omega(\sqrt{d}).$$
For the subspace recovery problem, where we only require a bound on the gaps  $\gamma_i(t) - \gamma_j(t)$ for $i\leq k<j$, it is sufficient to assume a bound $\sigma_{k} - \sigma_{k+1} \geq \Omega(\sqrt{d})$ on only the initial $k$'th eigenvalue gap in order to apply Weyl's inequality.

\paragraph{Completing the proof. } Plugging in $\lambda_i = 1$ for $i \leq k$ and $\lambda_i = 0$ otherwise to \eqref{eq_t7_2_b}, and simplifying, we get a utility bound for the subspace recovery problem, 
$$ 
\sqrt{\mathbb{E}\left[\left\|\hat{V}_k\hat{V}_k^\ast -  V_k V_k^\ast \right \|_F^2\right]} \leq \tilde{O}\left(\sqrt{\sum_{i=1}^k \sum_{j= k+1}^d \frac{1}{(\sigma_i - \sigma_j)^2}}\right).
$$
whenever the $k$'th initial eigenvalue gap satisfies $\sigma_k - \sigma_{i+1} \geq \Omega(\sqrt{d})$.

Note that, as we have only used Weyl's inequality, a deterministic bound, to bound the eigenvalue gaps, there is no need to bound the inverse second moments $\mathbb{E}\left[\frac{1}{(\gamma_i(t) - \gamma_j(t))^2}\right]$ of the eigenvalue gaps.
Thus, the same proof for the subspace recovery utility bound (Theorem \ref{thm_rank_k_subspace}) applies in both the real-symmetric and complex-Hermitian cases.

\section{Additional discussion of deterministic-bound approaches}\label{sec_challenges}

Recall from Section \ref{sec_deterministic_perturbation_bounds} that one approach to bounding the quantity $\|\hat{V} \Sigma_k \hat{V}^\ast - V \Sigma_k V^\ast\|_F$ is to decompose 
\begin{equation}\label{eq_n79}
V \Sigma_k V^\ast = \sum_{i=1}^{k-1} (\sigma_i- \sigma_{i+1})V_i V_i^\ast + \sigma_k V_k V_k^\ast,
\end{equation}
and apply the Davis-Kahan theorem \cite{davis1970rotation} (see \eqref{eq_DK}) to each projection matrix $V_i V_i^\ast$ (Inequality \eqref{eq_DK_1}).
Here we give additional steps used to derive Inequality \eqref{eq_DK_1}:

\begin{eqnarray}
    & & \!\!\!\!\!\!\!\!\!\!\!\!\!\!\! \!\!\!\!\!\! \!\!\!\!\!\!\!\!\! \!\!\!\!\!\!\!\!\! \|\hat{V} \Sigma_k \hat{V}^\ast - V \Sigma_k V^\ast\|_F\nonumber\\
    &\stackrel{\textrm{Eq. } \eqref{eq_n79}}{=}& \left\|\sum_{i=1}^{k-1} (\sigma_i- \sigma_{i+1})\hat{V}_i \hat{V}_i^\ast + \sigma_k \hat{V}_k \hat{V}_k^\ast - \left(\sum_{i=1}^{k-1} (\sigma_i- \sigma_{i+1})V_i V_i^\ast + \sigma_k V_k V_k^\ast\right) \right\|_F \nonumber\\  \nonumber\\
&=&    \left\|\sum_{i=1}^{k-1} (\sigma_i - \sigma_{i+1}) (\hat{V}_i \hat{V}_i^\ast - V_i V_i^\ast) + \sigma_k (\hat{V}_k \hat{V}_k^\ast - V_k V_k^\ast) \right\|_F \nonumber\\
    &\leq & \sum_{i=1}^{k-1} (\sigma_i - \sigma_{i+1}) \|\hat{V}_i \hat{V}_i^\ast - V_i V_i^\ast\|_F + \sigma_k \|\hat{V}_k \hat{V}_k^\ast - V_k V_k^\ast\|_F \nonumber\\
    & \leq & \sum_{i=1}^{k-1} (\sigma_i - \sigma_{i+1}) \frac{\sqrt{i} \sqrt{d}}{\sigma_i - \sigma_{i+1}} + \sigma_k   \frac{\sqrt{k} \sqrt{d}}{\sigma_k - \sigma_{k+1}} \nonumber\\
    & = & O\left(k^{1.5} \sqrt{d} + \frac{\sigma_k}{\sigma_k - \sigma_{k+1}} \sqrt{k} \sqrt{d}\right).  \nonumber
\end{eqnarray}

\noindent
If one only wishes to bound the quantity  $\|M- \hat{V}_k \hat{\Sigma}_k \hat{V}_k^\ast\|_F - \|M- V_k \Sigma_k V_k^\ast\|_F$ (which is bounded above by $\|\hat{V} \hat{\Sigma}_k \hat{V}^\ast - V \Sigma_k V^\ast\|_F$), 
it is also possible to use deterministic trace inequalities. 
This is the approach taken in \cite{dwork2014analyze}, which applies the fact that 
\begin{equation} \label{eq_trace}
\mathrm{tr}(X) \leq \mathrm{rank}(X)\|X\|_2 \qquad \qquad \forall X \in \mathbb{R}^{d \times d}
\end{equation}
 to show that $$\|M- \hat{V}_k \hat{\Sigma}_k \hat{V}_k^\ast\|_F^2 - \|M- V_k \Sigma_k V_k^\ast\|_F^2 \leq O(k \|M - V_k \Sigma_k V_k^\ast\|_2 \|E\|_2 + k \|E\|_2^2).$$
The r.h.s. depends on $\sigma_{k+1} = \|M - V_k \Sigma_k V_k^\ast\|_2$, and is therefore not invariant to scalar multiplications of $M$.
However, one can obtain a scalar-invariant bound on the quantity $\|M- \hat{V}_k \hat{\Sigma}_k \hat{V}_k^\ast\|_F - \|M- V_k \Sigma_k V_k^\ast\|_F$ by plugging in $\|M - V_k \Sigma_k V_k^\ast\|_2 \leq \|M - V_k \Sigma_k V_k^\ast\|_F$ and plugging in the high-probability bound  $\|E\|_2 = O(\sqrt{d})$.
This leads to a bound of $\|M- \hat{V}_k \hat{\Sigma}_k \hat{V}_k^\ast\|_F - \|M- V_k \Sigma_k V_k^\ast\|_F \leq O(k \sqrt{d})$.
In the special case where $k=d$, this bound is $O(d^{1.5})$, and thus is not tight since we have $\|M- \hat{V}_k \hat{\Sigma}_k \hat{V}_k^\ast\|_F - \|M- V_k \Sigma_k V_k^\ast\|_F = \|\hat{M} - M \|_F = \|E\|_F = O(d)$ w.h.p.
Roughly, the additional factor of $\sqrt{k} = \sqrt{d}$ incurred in their bound is due to the fact that the matrix trace inequality \eqref{eq_trace} their analysis relies on gives a bound in terms of the spectral norm, even though they only need a bound in terms of the Frobenius norm-- which can (in the worst case) be larger than the spectral norm by a factor of $\sqrt{k}$.

\section{Proof of Lemma \ref{lemma_spectral_martingale_b}}\label{sec_proof_of_lemma_spectral_martingale_b}
\begin{proof}[Proof of Lemma \ref{lemma_spectral_martingale_b}]
To prove Lemma \ref{lemma_spectral_martingale_b} we use Doob's submartingale inequality.
Towards this end, let  $\mathcal{F}_s$ be the filtration generated by $B(s)$.
First, we note that $\exp(\|B(t)\|_2)$ is a submartingale for all $t \geq 0$; that is, $\mathbb{E}[\exp(\|B(t)\|_2) |  \mathcal{F}_s] \geq  \mathrm{exp}(\| B(s)\|_2 )$ for all $0\leq s \leq t$.
This is because for all $s \leq t$, we have
\begin{eqnarray}
\mathbb{E}[\exp(\|B(t)\|_2) |  \mathcal{F}_s] &=& \mathbb{E}\left[\mathrm{exp}\left(\sup_{v \in \mathbb{R}^d: \|v\|_2 = 1} v^\top B(t) v\right) \, \, \bigg | \, \,\mathcal{F}_s\right] \nonumber\\
&\geq & \mathrm{exp}\left(\mathbb{E}\left[\sup_{v \in \mathbb{R}^d: \|v\|_2 = 1} v^\top B(t) v \, \, \bigg | \, \,\mathcal{F}_s\right]\right) \label{eq_n136}\\
&\geq &\mathrm{exp}\left(\sup_{v \in \mathbb{R}^d: \|v\|_2 = 1}  \mathbb{E}\left[v^\top B(t) v \, \,  | \, \, \mathcal{F}_s\right]\right) \nonumber \\
&=&  \mathrm{exp}\left(\sup_{v \in \mathbb{R}^d: \|v\|_2 = 1}  \mathbb{E}\left[v^\top (B(t) - B(s)) v  + v^\top B(s) v \, \,  | \, \, \mathcal{F}_s\right] \right)  \nonumber \\
&=&  \mathrm{exp}\left(\sup_{v \in \mathbb{R}^d: \|v\|_2 = 1}  \mathbb{E}\left[v^\top (B(t) - B(s)) v | \, \, \mathcal{F}_s\right]  +  \mathbb{E}\left[v^\top B(s) v \, \,  | \, \, \mathcal{F}_s\right] \right)  \label{eq_n137} \\
&=& \mathrm{exp}\left( \sup_{v \in \mathbb{R}^d: \|v\|_2 = 1}  \mathbb{E}\left[v^\top B(s) v \, \,  | \, \, \mathcal{F}_s\right] \right)  \nonumber \\
&=& \mathrm{exp}\left(\sup_{v \in \mathbb{R}^d: \|v\|_2 = 1}  v^\top B(s) v \right),  \nonumber \\
&=& \mathrm{exp}(\| B(s)\|_2 ),   \nonumber 
\end{eqnarray}
where \eqref{eq_n136} holds by Jensen's inequality since $\exp(\cdot)$ is convex, and \eqref{eq_n137} holds since $v^\top (B(t) - B(s)) v$ is independent of $\mathcal{F}_s$ and is distributed as $N(0,2(t-s))$.
Thus, by Doob's submartingale inequality, for any $\beta>0$  (we will choose the value of $\beta$ later to optimize our bound) we have,
\begin{eqnarray}
    \mathbb{P}\left(\sup_{t \in [0,T]}\|B(t)\|_2 > 2\sqrt{T}(\sqrt{d} + \alpha)\right) & = &     \mathbb{P}\left(\sup_{t \in [0,T]}\frac{\beta}{2\sqrt{T}}\|B(t)\|_2 - \beta\sqrt{d} > \beta \alpha\right) \nonumber \\
    &= &\mathbb{P}\left(\sup_{t \in [0,T]}\exp\left(\frac{\beta}{2\sqrt{T}}\|B(t)\|_2 - \beta\sqrt{d}\right) > \exp(\beta \alpha)\right) \nonumber\\
    & \leq & \frac{\mathbb{E}[\exp(\frac{\beta}{2\sqrt{T}}\|B(t)\|_2 - \beta\sqrt{d})]}{\exp(\beta\alpha)}  \label{eq_n138}\\
        & = & \frac{\int_{0}^{\infty} \mathbb{P}[\exp(\frac{\beta}{2\sqrt{T}}\|B(t)\|_2 - \beta\sqrt{d})>x] \mathrm{d}x}{\exp(\beta \alpha)} \nonumber\\
                & = & \frac{\int_{0}^{\infty} \mathbb{P}[\frac{1}{2}\|B(t)\|_2 - \sqrt{d}> \beta^{-1} \log(x)] \mathrm{d}x}{\exp(\beta \alpha)} \nonumber\\
                & \leq & \frac{\int_{0}^{\infty} 2 e^{-\beta^{-2}\log^2(x)} \mathrm{d}x}{\exp(\beta \alpha)} \label{eq_n139}\\
                                                                                               & = & \frac{2\sqrt{\pi} \beta e^{\frac{1}{4}\beta^2}}{\exp(\beta \alpha)} \nonumber\\
                                                                & \leq   &\frac{2\sqrt{\pi}e^{\frac{1}{2}\beta^2}}{\exp(\beta \alpha)} \nonumber\\
                &= & 2\sqrt{\pi} e^{\frac{1}{2}\beta^2 - \beta \alpha}, \nonumber
\end{eqnarray}
where \eqref{eq_n138} holds by Doob's submartingale inequality, and \eqref{eq_n139} holds by Lemma \ref{lemma_concentration}.
Setting $\beta = \alpha$, we have
\begin{eqnarray*}
    \textstyle \mathbb{P}\left(\sup_{t \in [0,T]}\|B(t)\|_2 > \sqrt{T}(\sqrt{d} + \alpha)\right) \leq 2\sqrt{\pi} e^{-\frac{1}{2}\alpha^2}.
    \end{eqnarray*}
\end{proof}

\section{Comparison between the Frobenius distance metric and a weaker Frobenius metric}\label{sec_strong_weak_metric_comparison}

 The metrics in Theorems \ref{thm_utility} and \ref{thm_rank_k_covariance_approximation_new} measure the Frobenius distance to the solution of the optimization problem $\min_{Z \in \mathbb{C}^{d \times d}} \|  Z - A\|_F$ subject to $Z$ being a Hermitian matrix of rank at most $k$.
 In a similar vein, the metric in Theorem \ref{thm_rank_k_subspace} measures the Frobenius distance to the solution of the optimization problem $\min_{P \in \mathbb{C}^{d\times d}} \|  P - V_k V_k^\ast\|_F$ subject to $P$ being a projection matrix of rank at most $k$.
 The above distance metrics hold several advantages over metrics which measure the difference in the value of an objective function, such as the metric used in Theorem B.1, which measures the difference $g(\mathcal{A}(M)) - g(M_k)$ in the value of the objective function $g(X) := \|X - M\|_F$ at the mechanism output $\mathcal{A}(M)$ and the optimal solution $M_k$.
 Specifically,
 
\begin{itemize}

\item In many applications one wishes to recover the rank-$k$ matrix $M_k$ which minimizes the Frobenius distance to a given input matrix $M$ (see e.g. \cite{davenport2016overview}).
These include statistics applications where the top-$k$ eigenvectors correspond to the most significant features in the data.
It also includes signal processing applications where one wishes to extract a signal of rank-$k$ from a superposition of signals arising from multiple sources.
In such applications, one wishes to output a matrix $\mathcal{A}(M)$ which minimizes the Frobenius distance $\|\mathcal{A}(M) - M_k\|_F$ to the best rank-$k$ matrix $M_k$. 
A bound on $\|M-\mathcal{A}(M)\|_F - \|M-M_k\|_F$ may be insufficient as it does not imply a bound on the Frobenius distance $\|\mathcal{A}(M) - M_k\|_F$.

\item The Frobenius distance metric is a stronger metric in the sense that an upper bound on $\|\mathcal{A}(M)- M_k\|_F$ implies the same upper bound on  $\|M-\mathcal{A}(M)\|_F - \|M-M_k\|_F$ (the reverse direction is not true in general).
 Moreover, in many applications the recovered matrix $\mathcal{A}(M)$ is post-processed by applying a given function $f$.
 For instance, in many machine learning applications, the recovered matrix $\mathcal{A}(M)$ may be used to reduce the dimension and/or normalize a dataset before plugging the data into a machine learning model (see e.g. \cite{james2013introduction}).
 If the post-processing function $f$ is $L$-Lipschitz (with respect to the Frobenius norm), a bound on $\| \mathcal{A}(M) -M_k\|_F \leq b$ immediately implies a bound of $\| f(\mathcal{A}(M)) -f(M_k)\|_F \leq L \times b$ on the Frobenius norm error of the post-processed matrix $f(\mathcal{A}(M))$.
 In contrast, a bound of $\|M-\mathcal{A}(M)\|_F - \|M-M_k\|_F \leq b$ does not imply that $\|f(M)-f(\mathcal{A}(M))\|_F - \|f(M)-f(M_k)\|_F \leq L \times b$.

\item The Frobenius distance metric for the subspace recovery problem used in Theorem \ref{thm_rank_k_subspace} is the same metric (up to a constant factor) as the metric used in the classical eigenspace perturbation results of \cite{davis1970rotation}, as well as more recent work which obtain improved eigenspace perturbation bounds for certain random matrix perturbations \cite{o2018random}, and thus provides for a more direct comparison.
 The Frobenius distance metric $\|\mathcal{A}(M) -M_k\|_F$ for the covariance approximation problem used in Theorems \ref{thm_utility} and \ref{thm_rank_k_covariance_approximation_new} is a close analog to the Frobenius distance metric for the subspace recovery problem $\|\mathcal{A}(M) -V_k V_k^\ast\|_F$, and thus provides a more direct comparison to the classical matrix perturbation results.

\end{itemize}
However, we note that the Frobenius distance metric $\|\mathcal{A}(M) - M_k\|_F$ is only uniquely defined for matrices $M$ whose $k$’th eigenvalue gap $\sigma_k - \sigma_{k+1}$ is strictly greater than $0$.
This is because the optimal solution $M_k$ to  $\min_{Z \in \mathbb{C}^{d \times d}} \|  Z - A\|_F$ s.t. $Z$ being Hermitian of rank at most $k$, is not unique for matrices $M$ where $\sigma_k - \sigma_{k+1} = 0$.
 To see why, note that if $M=I$, then any rank-$k$ matrix $Z$ with all non-zero eigenvalues equal to $1$ is a minimizer of the quantity $\| Z - I\|_F$.

 Moreover, as we show in Appendix \ref{appendix_tightness},  for any $k<d$ and any value of $c \geq 1$ one can construct a $d \times d$ rank-$k$ matrix $M$ with gap ratio $\frac{\sigma_k}{\sigma_k- \sigma_{k+1}} = c$ for which $\|\hat{M}_k - M_k\|_F = \Theta(\sqrt{k} \sqrt{d} \frac{\sigma_k}{\sigma_k- \sigma_{k+1}} \sqrt{T})$ w.h.p., whenever $\hat{M} = M + (G + G^\ast) \times \sqrt{T}$ where $G$ has iid standard Gaussian entries.
Thus, any high-probability upper bound on $\|\hat{M}_k - M_k\|_F$ must depend on the eigenvalue gap ratio $\frac{\sigma_k}{\sigma_k- \sigma_{k+1}}$.
This is in contrast to the weaker metric $\|\hat{M}_k - M\|_F - \|M_k - M\|_F$ which allows for eigenvalue-gap free bounds (see e.g. Theorem 7 of \cite{dwork2014analyze}, or our Theorem \ref{thm_weaker_Frobenius_metric})

\section{High-probability utility bounds}\label{sec_high_probability}

\paragraph{High-probability utility bounds with sub-linear growth in the probability parameter}
While the bound in Theorem \ref{thm_rank_k_covariance_approximation_new} holds in expectation, 
it is possible to use our techniques to prove high probability bounds. 
The simplest approach is to plug in the expectation bound of Theorem \ref{thm_rank_k_covariance_approximation_new} into Chebyshev's inequality, which says that $P(|\| \hat{M}_k - M_k \|_F - \mathbb{E}[\| \hat{M}_k - M_k \|_F ]|\geq s) \leq \tilde{O}\left(\frac{E(\| \hat{M}_k - M_k \|_F^2)}{s^2} \right )$ for all $s>0$.
 This gives a bound of $P(\|\hat{M}_k - M_k \|_F \geq \tilde{\Omega}(s \cdot \sqrt{kd} \frac{\sigma_k}{\sigma_k - \sigma_{k+1}}  \sqrt{T})) \leq \frac{1}{s^2}$ for all $s>0$.
  In other words, we have that $\|\hat{M}_k - M_k \|_F \leq \tilde{O}\left(\sqrt{kd} \frac{\sigma_k}{\sigma_k - \sigma_{k+1}}  \sqrt{T}  \frac{1}{\sqrt{s}}\right )$ w.h.p. $1-s$ for all $s>0$.

\paragraph{High-probability utility bounds with logarithmic growth in the probability parameter}

   It is an interesting open problem whether one can strengthen the high-probability bounds implied by Theorem \ref{thm_rank_k_covariance_approximation_new}, which have a sub-linear growth factor $\frac{1}{\sqrt{s}}$ in the probability parameter $\frac{1}{s}$, to high-probability bounds which grow logarithmically in $\frac{1}{s}$.
 To show a bound of $\|\hat{M}_k - M_k \|_F \leq \tilde{O}(\sqrt{kd} \frac{\sigma_k}{\sigma_k - \sigma_{k+1}}  \sqrt{T}  \log(\frac{1}{s}))$ with high probability $1-s$ for $s>0$, we would need to show that
 \begin{equation}\label{eq_n180}
 \mathbb{P}\left(\|\hat{M}_k - M_k \|_F > \tilde{O}\left (\sqrt{kd} \frac{\sigma_k}{\sigma_k - \sigma_{k+1}}  \sqrt{T}  \times s \right) \right) \leq e^{-s}
 \end{equation}
 for $s>0$. 
 
 The main difficulty in extending our proof to obtain the exponential decay $e^{-s}$ on the r.h.s. of \eqref{eq_n180}, which would imply high-probability utility bounds which grow logarithmically in the probability parameter $\frac{1}{s}$,  is that our results rely on bounds on the eigenvalue gaps of Dyson Brownian motion. 
 While we show in Theorem \ref{thm:eigenvalue_gap} that these eigenvalue gaps satisfy a bound of \(\eta_i - \eta_{i+1} \leq \tilde{O}\left(\frac{s}{\sqrt{d}}\right)\) with probability \(O(s^3)\) that decays at a polynomial rate \(s^3\) in the probability parameter \(s\) (for the complex-valued GUE random matrix), this probability does not decay at an exponential rate (one can easily verify that the exponential decay rate does not hold, e.g., by examining the \(d=2\) case, and we are not aware of any results which show an exponential decay for larger \(d\)).

    \section{List of notation}\label{sec_glossary_of_notation}
     In this section, we list key notations used in the proofs.
     
    \begin{itemize}
 \item  The Hermitian matrices $M$ and $\hat{M}$ (See statement of Theorems \ref{thm_utility}, \ref{thm_rank_k_covariance_approximation_new})

\begin{itemize}   
 \item $\sigma_1 \geq \cdots \geq \sigma_d$ denote the eigenvalues of $M$.
 \item $\Gamma_k := \mathrm{diag}(\sigma_1,\ldots, \sigma_k,0,\ldots,0)$.
 \item $V$ and $\hat{V}$ denotes the orthogonal matrix whose columns are the corresponding eigenvectors of $M$.
   \end{itemize} 

\item   $B(t)$, a Hermitian matrix-valued Brownian motion with zero initial condition  (see Section \ref{sec_DBM})
\item  $\Phi(t)= M + B(t)$  (see \eqref{eq_n120})
\begin{itemize}
\item $\gamma_1(t) \geq \cdots \geq \gamma_d(t)$ denote the eigenvalues of  $\Phi(t)$ 
\item $\Gamma(t) = \mathrm{diag}(\gamma_1(t), \ldots, \gamma_d(t))$
\item  $U(t)$ is the unitary matrix whose columns $u_1(t), \ldots, u_d(t)$ are the eigenvectors of $\Phi(t)$.
\end{itemize}
\item $\Psi(t)$, a rank-$k$ Hermitian matrix-valued stochastic process (see \eqref{eq_n96}, \eqref{eq_n45})

\item The matrix-valued stochastic processes $X(t)$ ( see \eqref{eq_n3}), R(t) and $Q(t)$ (see \eqref{eq_n5})
   
     \item   $\mathcal{W}_d := \{(x_1, \ldots, x_d) \in \mathbb{R}^d:  x_1 \geq \cdots \geq x_d\}$ (See \eqref{eq:WeylChamber})
 
   \item $\eta_1\geq \cdots \geq \eta_d$ denotes the eigenvalues of the  GOE/GUE $G+G^\ast$ where $G$ is a random matrix with i.i.d. complex Gaussian entries, or, more, generally, the eigenvalues of a perturbation of the GOE/GUE matrix $M+G + G^\ast$ (See Theorem \ref{thm:eigenvalue_gap}).

\item Depending on the context, $f(\cdot)$ denotes
\begin{itemize}
\item the joint probability density $f(\eta)$ of the eigenvalues $\eta$ of the GOE/GUE random matrix (see \eqref{eq_joint_density}).
\item The objective function $f(\cdot) = \|\cdot \|_F^2$, or, more generally, a generic objective function (see Ito's lemma (Lemma \ref{lemma_ito_lemma_new}))
\end{itemize}

\item Depending on the context, $g(\cdot)$ denotes
\begin{itemize}
\item A map $g : \mathcal{W}_d \rightarrow \mathcal{W}_d$ defined in \eqref{eq_g3}-\eqref{eq_g1} 
\item An objective function  $g(Y) = \|Y - M\|_F^2$ (see \eqref{eq_n104})
\end{itemize}

\item The map $\phi : \mathcal{W}_d \rightarrow \mathcal{W}_d$ defined in \eqref{eq_phi1}-\eqref{eq_phi3}

\item The set-valued functions $S_0(\cdot)$, $S_3(\cdot)$, $S_4(\cdot)$ (see \eqref{eq_n22} and \eqref{eq_n25})

  \item The ``classical'' eigenvalue locations $\omega_1 \geq \cdots \geq \omega_d$ (see \eqref{eq_a7})

\item The quantities $t_0$ (see \eqref{eq_t0}), $\alpha$ (see \eqref{eq_alpha}), $\mathfrak{b}$ (see \eqref{eq_n37}), $j_{\mathrm{min}}$, $j_{\mathrm{max}}$ (see \eqref{eq_n28}).

\item $E_\alpha$: A rare ``bad'' event when one or more eigenvalue gaps are unusually small (see \eqref{eq_n77})

   \end{itemize}

{\bf Standard matrix and complex analysis notation:}

For any matrices $A, B \in \mathbb{C}^{d \times d}$,
\begin{itemize}
\item For any matrix $A$ we denote by $A_{ij}$ or $A[i,j]$ the $(i,j)$'th entry of the matrix $A$.
\item $A^\ast$ denotes the conjugate-transpose of $A$
\item  Trace: $\mathrm{tr}(A)$ denotes the trace of $A$.
\item The Frobenius inner product: $\langle A , B \rangle := \mathrm{tr}(A^\ast B)$ denotes the Frobenius inner product of $A$ and $B$
\item The Frobenius norm $\|A\|_F = \sqrt{\langle A , A \rangle}$
\item The spectral norm $\| A \|_2 := \sup_{x \in \mathbb{C}^d} \frac{\|A x\|}{\|x\|}$.
\item Real component: $\mathcal{R}(A)$ denotes the matrix whose entries are the real components of the entries of $A$
\item Imaginary component: $\mathcal{I}(A)$ denotes the matrix whose entries are the imaginary components of the entries of $A$
\end{itemize}

\end{document}